\documentclass[12pt]{article}
\usepackage{setspace}
\RequirePackage{amsthm,amsmath,amsfonts,amssymb}       
\usepackage{mathrsfs}
\usepackage{enumerate}
\usepackage{bm}
\usepackage{multirow}
\font\myfont=cmr12 at 20pt
\usepackage{diagbox}
\usepackage{xcolor}
\usepackage{natbib}
\usepackage{bbm, dsfont}
\usepackage{makecell}
\usepackage{upgreek}
\usepackage{mathtools}
\usepackage{dsfont,makecell,chngcntr}
\usepackage{float}
\usepackage{multirow,stfloats,booktabs}
\usepackage[english]{babel}
\usepackage{graphicx} 
\usepackage{makecell}
\usepackage{tabularx,array}
\usepackage[    protrusion=true,
expansion=true,
final,
babel
]{microtype}
\usepackage[colorlinks,
linkcolor=red,
anchorcolor=blue,
citecolor=blue
]{hyperref}

\def\tr{\mathop{\text{tr}}\kern.2ex}

\def\cb{{\mathbf c}}
\def\ab{{\mathbf a}}

\def\supp{\mathop{\text{supp}}}

\long\def\comment#1{}

\def\tr{\mathop{\text{Tr}}}

\def\cS{{\mathcal{S}}}

\def\diag{\operatorname{diag}}

\newcommand{\bel}{\begin{eqnarray}\label}
\newcommand{\eel}{\end{eqnarray}}
\newcommand{\bes}{\begin{eqnarray*}}
	\newcommand{\ees}{\end{eqnarray*}}

\let\hat\widehat
\let\tilde\widetilde

\def\EE{{\mathbb E}}

\def\supp{\mathop{\text{supp}\kern.2ex}}

\def\sgn{\operatorname{sgn}}

\def\tr{{\rm{Tr}}}

\def\supp{\mathop{\text{supp}}}

\def\M{{\mathrm M}}

\def\tr{\mathrm{Tr}}

\def \F {\rm F}

\renewcommand{\baselinestretch}{1.0}

\def\sep{\;\big|\;}

\def \hbSigma {\widehat{\bSigma}}
\def \bLambda {\mathbf{\Lambda}}

\def \F {\text{F}}

\def \btheta {\boldsymbol{\theta}}
\def \bTheta {\mathbf{\Theta}}

\def\##1\#{\begin{align}#1\end{align}}
\def\$#1\${\begin{align*}#1\end{align*}}
\usepackage{enumitem}

\theoremstyle{plain}

\newcommand{\eb}{\mathbf{e}}
\newcommand{\fb}{\mathbf{f}}

\newcommand{\tb}{\mathbf{t}}
\newcommand{\ub}{\mathbf{u}}
\newcommand{\vb}{\mathbf{v}}

\newcommand{\xb}{\mathbf{x}}
\newcommand{\yb}{\mathbf{y}}

\newcommand{\bd}{\bm{d}}

\newcommand{\bv}{\bm{v}}

\newcommand{\bx}{\bm{x}}

\newcommand{\Ab}{\mathbf{A}}
\newcommand{\Bb}{\mathbf{B}}
\newcommand{\Cb}{\mathbf{C}}
\newcommand{\Db}{\mathbf{D}}
\newcommand{\Eb}{\mathbf{E}}
\newcommand{\Fb}{\mathbf{F}}
\newcommand{\Gb}{\mathbf{G}}
\newcommand{\Hb}{\mathbf{H}}
\newcommand{\Ib}{\mathbf{I}}

\newcommand{\Mb}{\mathbf{M}}

\newcommand{\Ob}{\mathbf{O}}
\newcommand{\Pb}{\mathbf{P}}
\newcommand{\Qb}{\mathbf{Q}}
\newcommand{\Rb}{\mathbf{R}}

\newcommand{\Ub}{\mathbf{U}}
\newcommand{\Vb}{\mathbf{V}}

\newcommand{\Xb}{\mathbf{X}}
\newcommand{\Yb}{\mathbf{Y}}
\newcommand{\Zb}{\mathbf{Z}}

\newcommand{\bW}{\bm{W}}
\newcommand{\bX}{\bm{X}}

\newcommand{\bZ}{\bm{Z}}

\newcommand{\cA}{\mathcal{A}}
\newcommand{\cB}{\mathcal{B}}

\newcommand{\cD}{\mathcal{D}}
\newcommand{\cE}{\mathcal{E}}

\newcommand{\cN}{\mathcal{N}}

\newcommand{\cP}{\mathcal{P}}

\newcommand{\cV}{\mathcal{V}}

\newcommand{\II}{\mathbb{I}}

\newcommand{\PP}{\mathbb{P}}

\newcommand{\RR}{\mathbb{R}}

\newcommand{\bSigma}{\bm{\Sigma}}

\newcommand{\bOmega}{\bm{\Omega}}

\DeclareMathOperator{\Var}{{\rm Var}}

\DeclareMathOperator*{\Cov}{\rm Cov}

\theoremstyle{plain}

\newtheorem{theorem}{Theorem}[section]
\newtheorem{lemma}[theorem]{Lemma}
\newtheorem{corollary}[theorem]{Corollary}

\theoremstyle{remark}

\newtheorem{example}{Example}

\newtheorem{remark}{Remark}
\newtheorem{assumption}{Assumption}

\newcommand{\blind}{1}
\newcommand{\litreview}{0}
\begin{document}

\def\spacingset#1{\renewcommand{\baselinestretch}%
{#1}\small\normalsize} \spacingset{1}


\if1\blind
{
  \title{\bf Dimension Reduction for Large-Scale Federated Data: Statistical Rate and Asymptotic Inference}
   \author{Shuting Shen, Junwei Lu, and Xihong Lin 
  \thanks{Shuting Shen is Assistant Professor of Statistics \& Data Science at the National University of Singapore ({\em shuting\_shen@nus.edu.sg}), Junwei Lu ({\em junweilu@hsph.harvard.edu}) is  Assistant Professor at the Department of Biostatistics at Harvard T.H. Chan School of Public Health, and Xihong Lin is Professor of Biostatistics at Harvard T.H. Chan School of Public Health and Professor of Statistics at Harvard University ({\em xlin@hsph.harvard.edu}).  This work was supported by the National Institutes of Health grants R35-CA197449, U01-HG009088, U01HG012064, U19-CA203654,  and P30 ES000002. The paper was finished while Shuting Shen was a Ph.D. student at the Department of Biostatistics at Harvard University. } \\
    }
    \date{}
  \maketitle
} \fi

\if0\blind
{
  \bigskip
  \bigskip
  \bigskip
  \begin{center}
  { {\myfont Dimension Reduction for Large-Scale Federated Data: Statistical Rate and 
  Asymptotic Inference}}
\end{center}
  \medskip
} \fi

\bigskip
\begin{abstract}
In light of the rapidly growing large-scale data in federated ecosystems, the traditional principal component analysis (PCA) is often not applicable due to privacy protection considerations and large computational burden. Algorithms were proposed to lower the computational cost, but few can handle both high dimensionality and massive sample size under distributed settings. In this paper, we propose the FAst DIstributed (FADI) PCA method for federated data when both  the dimension $d$ and the sample size $n$ are ultra-large, by simultaneously performing parallel computing along  $d$  and distributed computing along $n$. Specifically, we utilize $L$ parallel copies of $p$-dimensional fast sketches to divide the computing burden along $d$ and aggregate the results distributively along the split samples. We present a general framework applicable to multiple statistical problems, and establish comprehensive theoretical results under the general framework. We show that FADI accelerates the computation while enjoying the same non-asymptotic error rate as the traditional PCA when $Lp \ge d$.  We also derive inferential results that characterize the asymptotic distribution of FADI, and show a phase-transition phenomenon as $Lp$ increases. We perform extensive simulations to empirically validate our theoretical findings, and apply FADI to the 1000 Genomes data to study the population structure.
\end{abstract}




\newpage
\spacingset{1.9} 

\vspace{-20pt}
\section{Introduction}\label{sec: intro} 

Widely employed for dimension reduction, principal component analysis (PCA) finds applications in various scientific fields, including network studies \citep{abbe2017entrywise}, statistical genetics \citep{reich2008pca} and finance \citep{Pasini2017PRINCIPALCA}. Parameter estimation in many statistical models is based on PCA, such as spectral clustering in graphical models \citep{abbe2018sbmrev},  
and clustering with subsequent k-means refinement in Gaussian mixture models \citep{chen2020spectral}.
 When it comes to real data analysis, however, several shortcomings of the traditional PCA hinder its application to large-scale datasets. First, the high dimensionality and large sample size of modern big data can make the computation of PCA infeasible. {For instance}, PCA is frequently employed to address ancestry confounding in Genome-Wide Association Studies (GWAS) \citep{price2006population}, yet large biobanks, such as the UK Biobank \citep{ukbiobank2015cathie}, often contain hundreds of thousands to millions of Single Nucleotide Polymorphisms (SNPs)  and  subjects, necessitating more scalable algorithms for efficient computation. 
 Second, large-scale datasets in many applications are stored in federated  ecosystems, where data cannot leave individual warehouses due to privacy protection considerations \citep{BioMe1, dey2022efficient, BioVu1}.
 This calls for federated learning methods \citep{jordan2019communication,li2020federated} that provide efficient and privacy-protected  strategies for joint analysis across data warehouses without the need to exchange individual-level data. 
In view of those limitations, efforts have been made in recent years on  developing fast PCA and distributed PCA algorithms. 
 

Specifically, the existing fast PCA algorithms use the full-sample data and apply random projection to speed up calculation.
For instance,
\citet{halkofinding2011} proposed to estimate the $K$ leading eigenvectors of a $d \times d$ matrix ($K \ll d$) using Gaussian random sketches, which decreases the computation time
by a factor of $O(d)$ at the cost of increasing the statistical error by a factorial power of $d$. \citet{halkofinding2011} tried to enhance estimation accuracy by employing ``subspace iteration'' via taking power of the original matrix. However, this method is not practically applicable to federated data as it necessitates
numerous rounds of data communications. \citet{chen2016integrating} modified \citet{halkofinding2011}'s method by repeating fast sketching and showed the algorithm is consistent when the number of i.i.d. sketches goes to infinity. However, they did not discuss the finite sample statistical rates, and their results are limited to deterministic matrices without accounting for data randomness.
  As all of these methods use the full data, they have two major limitations. First, though allowing for large $d$, the existing fast PCA methods are not scalable to large sample sizes $n$.  
Second,
they are not applicable to federated data when data in different sites cannot be shared.

The existing distributed PCA algorithms reduce the computation burden by partitioning the full data  ``horizontally'' or ``vertically''. 
The horizontal partition splits the data 
over the sample size $n$, whereas the vertical partition splits the data over the dimension $d$. 
\citet{fandistributed2019}  considered the horizontally distributed PCA where they estimated the $K$ leading eigenvectors of the $d \times d$ population covariance matrix by applying traditional PCA to each data split and aggregating the results across sites. They showed when the number of splits is not too large, the error rate is of the same order as the traditional PCA. Since traditional PCA is used for each partition, the computational complexity is at least of order $O(d^3)$, which will be computationally difficult for large $d$.
\citet{Kargupta2001distributed} considered vertical partition and developed a method that collects local principal components (PCs) and then reconstructs  global PCs by linear transformations. However, there is no theoretical guarantee on the error rate,
and the method may fail for correlated variables.

 { Apart from the above applications in parameter estimation, inference also constitutes an important part of PCA.}  For example, when studying the ancestry groups of  whole genome data under the mixed membership models, while the estimation error rate guarantees the overall misclustering rate for all subjects, one may be interested in testing whether two individuals of interest share the same ancestry membership profile and assessing the associated uncertainty \citep{fan2022simple}. 
 Furthermore, despite the rich literature depicting the asymptotic distribution of traditional PCA estimators under different statistical models \citep{anderson1963asymp, johnstone2001spiked, baik2005phase, paul2007asymptotics, fan2017eigenasymp, Chen2019matcompinf, yan2021heteroPCA, fan2022simple}, distributional characterization of fast PCA methods and distributed PCA methods are not well-studied. For instance,  \citet{halkofinding2011} and  \citet{fandistributed2019} provided error bound for the fast PCA and distributed PCA algorithms, but with no characterization of the asymptotic distribution and hence no evaluation of the testing efficiency. \citet{Yang2021HowTR} analyzed the convergence in probability for various sketching methods involving random projections, yet they did not provide inferential analysis on the estimator. In independent work by \citet{zhang2022perturbation}, the error bounds and asymptotic distribution of the ``subspace iteration'' sketching estimator were derived. However, their method is not efficiently applicable to federated data, and their model assumptions primarily focus on error matrices with independent entries, which may not hold for correlated entry-wise errors.

{
 
  In view of the gaps in the existing literature, we propose in this paper a scalable and computationally efficient FAst DIstributed (FADI) PCA method applicable to federated data that could be large in both $d$ and $n$. 
  Specifically,  to obtain the $K$-leading PCs of a $d\times d$ matrix $\Mb$ from the observed data distributed across $m$ sites, we apply $L$ {\it parallel} copies of $p$-dimensional fast sketching to each local split, which serves to distribute the computing burden along $d$ across $L$ {\it parallel} machines. The {\it parallel } fast sketching results are subsequently aggregated across the data splits to leverage the information from the complete data. Finally, we aggregate the PC results across parallel fast sketches to restore statistical accuracy. These two levels of aggregations offer distinct advantages: the first aggregation across data partitions ensures the robustness of our method regardless of the number of machines $m$, while the second aggregation across parallel sketches reduces statistical errors.

  We will show that FADI has computational complexities of smaller magnitudes than the existing methods (see Table~\ref{tab:comparison intro}), { while achieving the same asymptotic efficiency as the traditional PCA.  }
 {Moreover, we establish FADI under a general framework that covers multiple statistical models}, including the spiked covariance model, the Gaussian mixture models (GMM), the degree-corrected mixed membership (DCMM) model, and the incomplete matrix inference model. For the clarity of presentation, we focus on the spiked covariance model and the GMM in the main text as illustrative applications and will discuss the other two in Supplementary Materials~A. We consider the horizontally distributed setting for the spiked covariance model, and the vertically distributed setting for GMM. Further elaboration on each model can be found in Section~\ref{sec: problem setting}.

\if 1 \litreview {
\noindent{\bf {Related Literature}.} { (add into main content not separate subsection)} Prior research has made advancements in inferential analyses using the traditional PCA across various statistical models, including the DCMM model \citep{fan2022simple}, the matrix completion problem \citep{Chen2019matcompinf}, and high-dimensional data with heteroskedastic noise and missingness under the spiked covariance model \citep{yan2021heteroPCA}.
However, those works were all based upon the traditional PCA approach and considered no distributed data setting, and hence will suffer from low computational efficiency when the data are high-dimensional or distributively stored across different sites. 

The statistical properties of sketching methods have also garnered recent research interest. \citet{Yang2021HowTR} analyzed the convergence in probability for various sketching methods involving random projections, although they did not provide inferential analysis on the estimator. In independent work by \citet{zhang2022perturbation}, the error bounds and asymptotic distribution of the ``subspace iteration'' estimator were derived. However, their method is not efficiently applicable to federated data, and their model assumptions primarily focus on error matrices with independent entries, which may not hold for correlated entry-wise errors.
Besides, sketching methods have found applications beyond PCA. For instance, \citet{sun2018nips} proposed sketching algorithms for testing combinatorial properties of large-scale graphical models by randomly subsampling the graph. In the context of linear regression, \citet{Ahfock2020sketchreg} investigated the properties of sketching methods and provided inferential results for the regression coefficients based upon sketched data.
} \fi

We summarize the major contributions of our paper as follows. 
  First, {the existing fast PCA or distributed PCA methods either handle high dimensions $d$ or large sample sizes $n$, but not both. FADI allows both $n$ and $d$ to be large.  It improves over fast PCA \citep{halkofinding2011} by achieving scalability for large $n$ through data splitting and accommodating federated data.  It improves over distributed PCA \citep{fandistributed2019} by allowing for large $d$ using multiple fast sketches. 
  Due to the fact that variables are usually dependent, it is challenging to achieve parallel computing along $d$ and distributed computing along $n$ simultaneously.
 To address this challenge, FADI splits the data along $n$ and untangles the variable dependency along $d$ by dividing the high-dimensional data into $L$ copies of $p$-dimensional fast sketches. 
 We establish theoretical error bounds to show that FADI is as accurate as the traditional PCA so long as $Lp \gtrsim d$. 
 Second, 
 we provide distributional guarantees on the FADI estimator to facilitate inference, which is absent in previous literature on  fast or distributed PCA methods. Specifically, we depict the trade-off between computational complexity and testing efficiency by studying FADI's asymptotic distribution under the regimes $Lp \ll d$ and $Lp \gg d$ respectively, and show a phase-transition phenomenon. 
   Third, we propose FADI under a general framework  applicable to multiple statistical models.
 We provide a comprehensive investigation of FADI's performance
 both methodologically and theoretically under the general framework, and  illustrate the results with specific  statistical models. { In comparison, the existing distributed methods mainly focus on estimating the covariance structure of independent samples \citep{fandistributed2019}.}
 The rest of the paper is organized as follows. Section~\ref{sec: problem setting} introduces the concrete problem setups. Section~\ref{sec: mtd} discusses FADI's implementation details, as well as its complexity and modifications when $K$ is unknown. Section~\ref{sec: theory} presents the theoretical results on the statistical rates and asymptotic normality.
Section~\ref{sec: simulations} shows the empirical evaluation of FADI's finite sample performance and comparisons with  several existing methods.
The application of FADI to the 1000 Genomes Data is given in Section~\ref{sec: real data},
 followed by discussions in Section~\ref{sec: discussion}.



\section{Eigenspace Estimation for Low-Rank Matrix}\label{sec: problem setting}
We first introduce some useful notations. For a vector $\vb$, we denote by $\|\vb\|_2$ the $\ell_2$-norm, and  $\|\vb\|_{\infty}$ the $\ell_{\infty}$-norm. For a matrix $\Ab = [\Ab_{ij}] \in \RR^{m \times n}$, denote by $\Ab = \Ub \mathbf{\Lambda} \Vb^{\top} = \sum_{j = 1}^K \sigma_j \ub_j \vb_j^{\top}$ its singular value decomposition (SVD). We use $\sigma_j(\Ab)$ (respectively $\lambda_j(\Ab)$) to represent the $j$-th largest singular value (respectively eigenvalue) of $\Ab$, and $\sigma_{\max}(\Ab)$ or $\sigma_{\min}(\Ab)$ (respectively $\lambda_{\max}(\Ab)$ or $\lambda_{\min}(\Ab)$) stands for the largest or smallest singular value (respectively eigenvalue) of $\Ab$. Denote by $\sgn(\Ab) = \sum_{\sigma_j > 0}\ub_j \vb_j^{\top}$ { the matrix signum}, by $\|\Ab\|_2$ the matrix spectral norm, $\|\Ab\|_{\F}$ the Frobenius norm,  and $\|\Ab\|_{2,\infty}  = \sup_{\|\xb\|_{2} = 1} \|\Ab \xb\|_{\infty} = \max_{i} \|\Ab^{\top} \eb_i\|_2$  the 2-to-$\infty$ norm, where $\{\eb_i\}_{i = 1}^m \subseteq \RR^m$ is the canonical basis. For two orthonormal matrices $\Vb, \Ub \in \RR^{n_1 \times n_2} $ with $n_1 > n_2$, we define the metric $\cD(\Ub, \Vb) = \|\Ub\Ub^{\top} - \Vb\Vb^{\top}\|_{\F}$.  For an integer $ n  $, define $[n] = \{1,2,\ldots, n\}$. { For matrices $A_1, \ldots, A_k$, let $\diag(A_1, \ldots, A_k)$ denote the block diagonal matrix with $A_1, \ldots, A_k$ as its diagonal blocks.} We use $c$ and $C$ to represent generic constants, whose values may vary from place to place.

In this paper, we aim to estimate the eigenspace of the rank-$K$ symmetric 
matrix $\Mb \in \RR^{d \times d}$, whose eigen-decomposition is $\Mb = \Vb \mathbf{\Lambda} \Vb^{\top}$,  where $\mathbf{\Lambda} = \operatorname{diag}(\lambda_1,\ldots, \lambda_K)$, $|\lambda_1| \ge |\lambda_2| \ge \ldots \ge |\lambda_K| > 0$ and $\Vb$ is the stacked $K$ leading eigenvectors. Note that {when $\Mb$ is asymmetric, we can deploy the ``symmetric dilation'' trick \citep{chen2020spectral}
to fit it into the setting. Denote by $\Delta = |\lambda_K|$ the eigengap, and assume without loss of generality that $\lambda_1 > 0 $.  $\widehat{\Mb}$ is a corrupted version of $\Mb$ obtained from observed data, with $\Eb = \widehat\Mb - \Mb$ being the error matrix. 
Our goal is to estimate the column space of $\Vb$ from 
$\widehat\Mb$  distributively and scalably.
The following two examples provide concrete statistical setups.

\begin{example}[Spiked Covariance Model \citep{johnstone2001spiked}]\label{ex: spiked gaussian}{Let
$\bX_1, \ldots, \bX_n \in \RR^d$ be i.i.d. { sub-Gaussian} random vectors with $\EE(\bX_i) = \mathbf{0}$ and $\EE(\bX_i \bX_i^{\top}) = \mathbf{\Sigma}$. We assume $\{\bX_i \}_{i=1}^n$ are i.i.d. for the simplicity of presentation and will generalize the theoretical results to non-i.i.d. data in Section~\ref{sec: thm err rate bound}. We assume the following decomposition for the covariance matrix: $\mathbf{\Sigma} =  \Vb \mathbf{\Lambda} \Vb^{\top}+\sigma^2 \Ib_d$, where $\Vb \in \RR^{d \times K}$ is the stacked $K$ leading eigenvectors and $\bLambda = \diag(\lambda_1, \ldots, \lambda_K)$ with $\lambda_1 \ge \ldots \ge \lambda_K > 0$. Assume that the data are split along the sample size $n$ and stored on $m$ different sites. Denote by $\{\bX_i^{(s)}\}_{i=1}^{n_s}$  the  sample split of size $n_s$ on the $s$-th site, and by $\Xb^{(s)} = (\bX_1^{(s)}, \ldots, \bX_{n_s}^{(s)})^{\top}$ the corresponding data matrix split ($s=1,\ldots, m$ and $\sum_{s=1}^m n_s = n$). Denote by $\Xb = (\bX_1, \ldots, \bX_n)^{\top}$ the full $n \times d$ data matrix. Then $\Mb = \Vb \mathbf{\Lambda} \Vb^{\top}$, and $\widehat{\Mb} = \widehat{\mathbf{\Sigma}} - \hat{\sigma}^2 \Ib_d$, where $\widehat{\mathbf{\Sigma}} = \frac{1}{n} \sum_{i=1}^n \bX_i \bX_i^{\top}$ is the sample covariance matrix and $\hat{\sigma}^2$ is a consistent estimator for $\sigma^2$. 
}
\end{example}

\begin{example}[Gaussian Mixture Models (GMM) \citep{chen2020spectral}]\label{ex: GMM}{
Let $\bW_1, \ldots, \bW_d \in \RR^n$ be independent samples with $\bW_j$ $(j \in [d])$ generated from one of  $K$  Gaussian distributions with means $\btheta_k  \in \RR^n$ ($k \in [K]$). Specifically, for $j \in [d]$, $\bW_j$ is associated with a  membership label $k_j \in [K]$,
 and $\bW_j \sim {\cN}(\sum_{k = 1}^K \btheta_{k} \II\{k_j = k\}, \Ib_n
)$.  Our goal is to recover the unknown membership labels $k_j$'s.  Denote $\Xb = (\bW_1, \ldots, \bW_d) = (\bX_1, \ldots, \bX_n)^{\top}$, where $\bX_i$ is the $i$-th row of $\Xb$.  Without loss of generality, we order $\bW_j$'s such that $\EE(\Xb) = \bTheta \Fb^{\top}$, 
\vspace{-15pt}
{\singlespace
$$
\text{where } \quad \bTheta = (\btheta_1, \ldots, \btheta_K) \in \RR^{n \times K},\quad \Fb = \diag(\mathbf{1}_{d_1}, \ldots, \mathbf{1}_{d_K}) \in \RR^{d \times K},
$$}{ with $d_k$ denoting the number of samples 
with mean $\btheta_k$ and $\mathbf{1}_{d_k} \in \RR^{d_k}$ denoting vector with all entries equal to 1.}
Then we define $\Mb = \EE[\Xb^\top\Xb] - n\Ib_d = \Fb \bTheta^{\top}\bTheta\Fb^{\top}$ and $\widehat\Mb = \Xb^{\top}\Xb - n\Ib_d$.  { Recall $\Mb = \Vb \mathbf{\Lambda} \Vb^{\top}$.} Since $\Vb$ and $\Fb$ share the same column space, we can recover the memberships from $\Vb$.
We consider the regime where $n > d$, and assume there exists a constant $C>0$ such that $\max_k d_k \le C \min_k d_k$ and $ \sigma_{1}(\bTheta) \le  C\sigma_{K}(\bTheta)$. We consider the vertically distributed setting where the data are split along the dimension $n$ on $m$ sites. Denote by $\Xb^{(s)} = (\bX_1^{(s)}, \ldots, \bX_{n_s}^{(s)} )^{\top}$ the data split on the $s$-th site of size $n_s$ ($s \in [m]$).

}
\end{example}
We primarily illustrate with the above two examples for readability and will provide additional applications to the  degree-corrected mixed membership (DCMM) model and the incomplete matrix inference model in Supplementary Materials~A. 
 

}
\vspace{-20pt}
\section{Fast Distributed Principal Component Analysis}\label{sec: mtd}
In this section, we present the FADI algorithm and its application to different examples.  We then provide the computational complexities of FADI and compare it with the existing methods. We also discuss how to estimate the rank $K$ when it is unknown.
\vspace{-20pt}
\subsection{Overview and Intuition}
\label{sec: FADI-overview}

For a given matrix $\widehat\Mb \in \RR^{d\times d}$, the computational cost of the traditional PCA on $\widehat\Mb$ is $O(d^3)$. In the case where $\widehat\Mb$ is computed from observed data, e.g., the sample covariance matrix $\hbSigma = \frac{1}{n} \sum_{i = 1}^n \bX_i \bX_i^{\top}$, extra computational burden comes from calculating $\widehat\Mb$, e.g., $O(nd^2)$ flops for computing the sample covariance matrix. Hence performing traditional PCA for large-scale data with high dimensions and huge sample sizes can be considerably expensive. 

To reduce the computational cost when $d$ is large, the most straightforward idea is to reduce the data dimension. One popular method for dimension reduction is random sketching \citep{halkofinding2011}. For instance, for a low-rank matrix $\Mb$  of rank $K$, its column space can be represented by a low-dimensional fast sketch $\Mb\bOmega \in \RR^{d \times p}$,
where $\mathbf{\Omega} \in \RR^{d\times p}$ is a random Gaussian matrix with $K < p \ll d$. 
In practice, $\Mb$ is usually replaced by an almost low-rank corrupted matrix $\widehat\Mb$ calculated from observed data. Traditional fast PCA methods then consider performing random sketching on $\widehat\Mb$ instead, and  use the full sample to obtain the fast sketch $\widehat\Yb= \widehat\Mb \mathbf{\Omega} \approx \Vb \mathbf{\Lambda} \Vb^{\top} \mathbf{\Omega}$ that almost maintains the same left singular space as $\Mb = \Vb \mathbf{\Lambda} \Vb^{\top}$.  It is hence reasonable to estimate $\Vb$ by performing SVD on the $d\times p$ matrix $\widehat\Yb$ that has a much smaller computational cost  than directly performing PCA on $\widehat\Mb$. However, one major drawback of this approach is that information might be lost due to fast sketching. Furthermore, the method is not scalable when $n$ is large or the data are federated. This motivates us to propose FADI, where we repeat the fast sketching  multiple times on each local split and aggregate the results to reduce the statistical error. 
Specifically,  
assume the data are stored across $m$ sites, and we have the decomposition $\widehat\Mb = \sum_{s=1}^m \widehat\Mb^{(s)}$, where $\widehat\Mb^{(s)}$ is the component that can be computed locally on the $s$-th site ($s \in [m]$). { Then instead of applying random sketching directly to $\widehat\Mb$, FADI computes in parallel the local fast sketching for each component $\widehat\Mb^{(s)}$  and aggregates the results across $m$ sites, which will reduce the cost of computing $\widehat\Mb \bOmega$ by a factor of $1/m$. } Note that this representation of $\widehat\Mb$ is legitimate in many models.   Taking Example~\ref{ex: spiked gaussian} for instance, define $\widehat\Mb^{(s)} = \frac{1}{n}(\Xb^{(s) \top} \Xb^{(s)})  - (\hat\sigma^2/m) \Ib_d$, and we have $\widehat\Mb = \hbSigma - \hat\sigma^2 \Ib_d = \sum_{s=1}^m \widehat\Mb^{(s)}$. 
We will further elaborate on the decomposition 
in Section~\ref{sec: step 0}. 

{ 


}

 
\vspace{-20pt}

\subsection{General Algorithmic Framework}\label{sec: alg}
Figure~\ref{fig: FADI illustrate} illustrates the fast distributed PCA (FADI) algorithm:

\begin{figure}[htbp]
		\centering

			\includegraphics[width=1\textwidth]{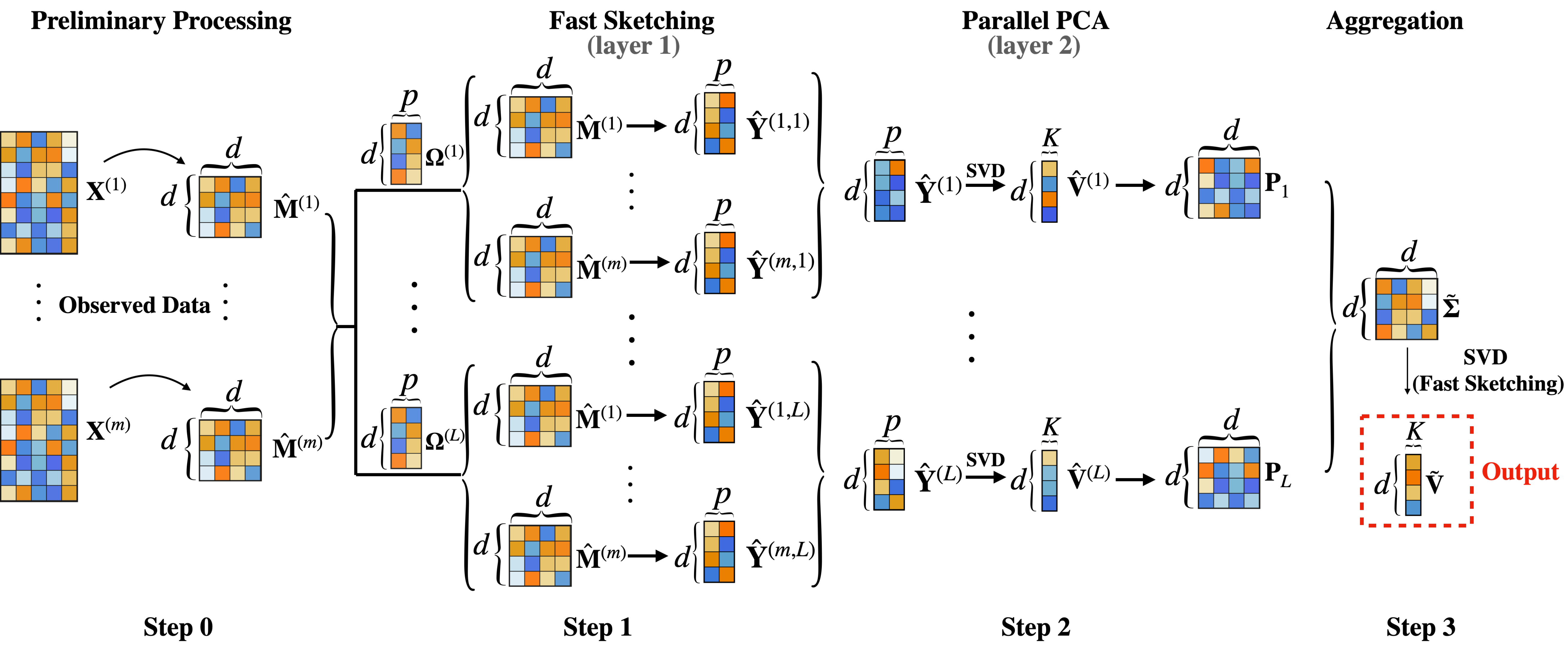}
		\caption{\small Illustration of FADI.   { Here $\{\Xb^{(s)}\}_{s=1}^m$ are the raw data stored distributively on $m$ sites, and $\widehat\Mb^{(s)}$ is the $s$-th component of $\widehat\Mb$ that can be calculated from $\Xb^{(s)}$. $\widehat\Yb^{(\ell)} = \sum_{s \in [m]} \widehat\Yb^{(s,\ell)}$ is the $\ell$-th copy of the fast sketch by aggregating the fast sketches distributively for each data split.
		}}\label{fig: FADI illustrate}
	\end{figure}


In Step 0, we perform preliminary processing on the raw data to produce $\{\widehat\Mb^{(s)}\}_{s=1}^m$. We will elaborate on the case-specific preprocessing in Section~\ref{sec: step 0}.

{ In Step 1,  
we generate $L$ i.i.d. Gaussian test matrices $\{\bOmega^{(\ell)}\}_{\ell =1}^L$, and for each $\ell \in [L]$, we apply $\bOmega^{(\ell)}$ distributively to $\widehat\Mb^{(s)}$ for each $s \in [m]$ and obtain the $\ell$-th fast sketch of $\widehat\Mb^{(s)}$ as  $ \widehat\Yb^{(s,\ell)} = \widehat\Mb^{(s)} \bOmega^{(\ell)}$. We send $\widehat\Yb^{(s,\ell)}$ ($s=1,\cdots, m$) to the $\ell$-th parallel server for aggregation.}

In Step 2, on the $\ell$-th server, the random sketches $\widehat\Yb^{(s,\ell)}$ ($s=1,\cdots, m$) from the $m$ split datasets  corresponding to the  $\ell$-th test matrix  $\mathbf{\Omega}^{(\ell)}$ will be collected and added up to get the $\ell$-th fast sketch: $\widehat\Yb^{(\ell)} = \sum_{s=1}^m \widehat\Yb^{(s,\ell)}$ $(\ell \in [L])$. We next compute in parallel the top $K$ left singular vectors $\widehat{\Vb}^{(\ell)}$ of $\widehat\Yb^{(\ell)}$ and send the $\widehat{\Vb}^{(\ell)}$'s to the central processor for aggregation.

In Step 3, on the central processor, calculate  $\widetilde{\mathbf{\Sigma}} =\frac{1}{L} \sum_{\ell = 1}^L \widehat{\Vb}^{(\ell)} \widehat{\Vb}^{(\ell) \top} = \frac{1}{L}\sum_{\ell = 1}^L \mathbf{P}_{\ell}$, where  $\Pb_{\ell} = \widehat{\Vb}^{(\ell)} \widehat{\Vb}^{(\ell) \top}$
is the projection matrix of $\widehat{\Vb}^{(\ell)}$. 
We next calculate the $K$ leading eigenvectors $\widetilde{\Vb}$ of $\widetilde{\mathbf{\Sigma}}$, which will serve as the final estimator of $\Vb$.

To further improve the computational efficiency, we might conduct another fast sketching in Step 3 to compute $\widetilde\Vb$: we apply the power method \citep{halkofinding2011} to $\widetilde\bSigma$ by calculating $\widetilde{\Yb} = \widetilde{\mathbf{\Sigma}}^q \mathbf{\Omega}^{\rm F}= \left(\frac{1}{L} \sum_{\ell = 1}^{L} \widehat{\Vb}^{(\ell)} \widehat{\Vb}^{(\ell) \top} \right)^q \mathbf{\Omega}^{\rm F}$ for $q\ge 1$, where $\mathbf{\Omega}^{\rm F} \in \RR^{d \times p'}$ is a Gaussian test matrix with dimension $p'$ that can be set different from $p$ for { computational efficiency}. Here, $\widetilde{\Yb}$ is calculated iteratively: $\widetilde{\Yb}_{(i)} = \frac{1}{L} \sum_{\ell = 1}^{L} \left(\widehat{\Vb}^{(\ell)} \widehat{\Vb}^{(\ell) \top} \widetilde{\Yb}_{(i-1)} \right)$ for $i = 1, \ldots, q$, where $ \widetilde{\Yb}_{(0)} = \mathbf{\Omega}^{\rm F}$ and $\widetilde{\Yb} = \widetilde{\Yb}_{(q)}$. 
We denote by $\widetilde\Vb^{\F}$ the leading $K$ left singular vectors of $\widetilde\Yb$.  We will show in Section~\ref{sec: theory} when $q$ is properly large, the distance between $\widetilde\Vb$ and $\widetilde\Vb^{\F}$ is negligible.  
\begin{remark}\label{rmk: different bias reductions}
We refer to Theorem~\ref{thm: error bound} for the choice of $p$ and $L$. In general, taking $p = 2K$ is sufficient. 
{ 
Different error reduction strategies are used in Steps~1 and 3 to accommodate different data structures: it is not desirable to replace the power method in Step 3 by repeated sketching.
This is because aggregating parallel sketching results incurs additional $O(d^3)$ flops for the SVD, which the fast sketching in Step 3 aims to avoid. Conversely, the power method is unsuitable for Step 1 because applying it to the distributed data $\widehat\Mb = \sum_{s=1}^m \widehat\Mb^{(s)}$ requires repeated data communication among the $m$ splits for iterative computation of $(\sum_{s=1}^m \widehat\Mb^{(s)})^q \bOmega^{\F}$, which is usually infeasible due to communication constraints in distributed systems.  In addition, applying the power method locally to each split and aggregating the results subsequently will increase the statistical error significantly when the local sample size is insufficient (see Supplementary Materials~B.4 for further details).}
\end{remark}
\vspace{-20pt}
\subsection{Case-Specific Processing of Raw Data}\label{sec: step 0}
In this section, we discuss the calculation of $\widehat\Mb$ in Step 0 specifically for each example.

\noindent  \textbf{Example~\ref{ex: spiked gaussian}: }
Recall that in Step 0 of FADI, to obtain $\widehat\Mb$, we need a consistent estimator of the residual variance $\sigma^2$. Denote by $S = \{i_1, i_2, \ldots, i_{K'}\} \subseteq [d]$ an arbitrary index set of size $K' \ge K + 1$. Then we estimate $\sigma^2$ by $\hat\sigma^2 = \lambda_{\min} (\hbSigma_S)$, where $\hat\bSigma_S$ is a $K' \times K'$ principal submatrix of $\hbSigma$ computed using only data columns in the set $S$, and 
can be easily computed distributively (see Figure~S6 in the supplement for reference).
Then for $s \in [m]$, we have $\widehat\Mb^{(s)} = \frac{1}{n}(\Xb^{(s) \top} \Xb^{(s)})  - (\hat\sigma^2/m)\Ib_d$.  Note that since computing $\widehat\Mb^{(s)}\bOmega = \frac{1}{n}\Xb^{(s) \top}( \Xb^{(s)} \bOmega) - m^{-1}\hat\sigma^2 \bOmega$ is much faster than first computing $\widehat\Mb^{(s)}$ then computing $\widehat\Mb^{(s)}\bOmega$, { we will calculate $\widehat\Mb^{(s)}\bOmega$ by calculating $\Xb^{(s)} \bOmega$ first rather than directly computing $\widehat\Mb^{(s)}$.}
  \vspace{5pt}

\noindent  \textbf{Example~\ref{ex: GMM}: }Recall that the data $\{\bW_j\}_{j=1}^d \subseteq \RR^{n}$ are vertically distributed across $m$ sites, and $\{\Xb^{(s)}\}_{s=1}^m$ are the corresponding data splits. 
For the $s$-th site, we have $\widehat\Mb^{(s)} =  \Xb^{(s)\top} \Xb^{(s)} - (n/m)\Ib_d$, and for $\ell \in [L]$, we compute $\widehat\Yb^{(s,\ell)}$ by $ \Xb^{(s)\top} (\Xb^{(s)}\bOmega^{(\ell)}) - (n/m)\bOmega^{(\ell)}$.

\vspace{-20pt}
\subsection{Computational Complexity}\label{sec: compu cost}
In this section, we provide the communicational and computational complexities of FADI for each example given in Section~\ref{sec: problem setting}. The complexity of each step is listed in Table \ref{table: compu cost}.
{\singlespace
\begin{table}[ht]
    \centering
     \resizebox{0.8\textwidth}{!}{%
    \begin{tabular}{c c c c c}
    \hline
    \hline
    & \multicolumn{2}{c}{Communication } & \multicolumn{2}{c}{Computation }\\
    \hline
         &Example~\ref{ex: spiked gaussian}  & Example~\ref{ex: GMM} &Example~\ref{ex: spiked gaussian}  & Example~\ref{ex: GMM}  \\
         \hline
         \textbf{Step 0} & $O(m K^2)$ & N/A &  \thead{$\hbSigma_S: O(\frac{K^2 n}{m} + K^2m)$\\ $\hat\sigma^2: O(K^3)$} &O(1)\\
         \hline
         \textbf{Step 1} & $O(mpd)$ & $O(mpd)$ & \thead{$\widehat\Yb^{(s,\ell)}: O(\frac{dnp}{m})$} & \thead{$\widehat\Yb^{(s,\ell)}: O(\frac{dnp}{m})$} \\
         \hline
         \textbf{Step 2} & $O(LKd)$ &$O(LKd)$ & \thead{$\widehat\Yb^{(\ell)}: O(mdp)$\\ $\widehat\Vb^{(\ell)}: O(dp^2)$} & \thead{$\widehat\Yb^{(\ell)}: O(mdp)$\\ $\widehat\Vb^{(\ell)}: O(dp^2)$} \\
         \hline
         \multirow{2.5}{*}{\textbf{Step 3}} & \multirow{2.5}{*}{{N/A}}&  \multirow{2.5}{*}{{N/A}}& \thead{$\widetilde\Vb:O( d^2 p L + d^3)$}& \thead{$\widetilde\Vb:O( d^2 p L + d^3)$} \\
          \cline{4-5}
         & & & \multicolumn{2}{c}{\thead{$\widetilde\Vb^{\F}: O( d K p' L q +dp^{\prime2})$}}\\
         \hline
         \textbf{Total} &$O(mpd + LKd)$&$O(mpd + LKd)$& \thead{$O(\frac{dnp}{m} + d K p' L q   +{ dmp})$} &\thead{$O(\frac{dnp}{m} + d K p' L q + { dmp})$} \\
         \hline
         \hline
    \end{tabular}}
    \caption{\small Complexity for Examples \ref{ex: spiked gaussian} and \ref{ex: GMM}. For the simplicity of presentation, we assume $\max_{s \in [m]} n_s \asymp n/m$. In Step 3,
    we recommend computing $\widetilde\Vb^{\rm F}$ instead of $\widetilde\Vb$ in practice. The total complexity in the last line refers to the total computational cost for $\widetilde\Vb^{\rm F}$.}
    \label{table: compu cost}
\end{table}
}
 { The optimal choice of $L$ depends on whether the priority is estimation or inference. For consistent estimation of $\Vb$, we recommend choosing $L = d/p$. For inferential analysis, we suggest choosing $L = \log d \cdot d/p$ when sufficient computational resources are available to balance valid inference and optimal estimation rates. Alternatively, $L$ should be chosen based on the maximum available computational budget when time constraints are a concern. Case-specific recommendations for $L$ are provided in Remarks~\ref{rmk: col spiked gaussian large L} and \ref{rmk: col GMM large L}, as well as Corollaries~\ref{col: guassian case} and \ref{col: GMM L small}.}
 When $p \asymp (K \vee \log d)$, $L \asymp d/p$, $p'\asymp K$  and $q \asymp \log d$, the total computational cost is $O\big(dn(K\vee\log d)/m + d^2K\log d + { md(K\vee\log d)}\big)$.

\begin{table}[htbp]
    \centering
     \resizebox{0.65\textwidth}{!}{%
    \begin{tabular}{c|c|c}
    \hline
    \hline
         Method & Error Rate & Computational Complexity  \\
         \hline
         FADI & $O( \sqrt{{Kr}/{n}} )$ & $O\left(dn(K\vee\log d)/m + d^2K\log d { + d(K \vee \log d) m}\right)$\\
         \hline
         Traditional PCA & $O( \sqrt{{Kr}/{n}})$ & $O(d^2n + d^3)$\\
         \hline
           Fast PCA & $O( \sqrt{{Kdr}/{n}})$ & $O(dnK + d^2K)$\\ 
           \hline
         Distributed PCA & $O( \sqrt{{Kr}/{n}} )$ & $O(d^2 n/m + d^3 { + d^2 m})$\\
         \hline
         \hline
    \end{tabular}}
    \caption{\small   Error rates and computational complexities for FADI,
    traditional PCA, fast PCA (one sketching) \citep{halkofinding2011} and distributed PCA \citep{fandistributed2019} for Example~\ref{ex: spiked gaussian}, where the error rate is evaluated by $\big(\EE |\cD(\,\cdot\,, \Vb)|^2\big)^{1/2}$. Here $r = \operatorname{tr}(\mathbf{\Sigma}) / \|\mathbf{\Sigma}\|_2$ 
    and $m$ is the number of sites. For FADI, we take $p \asymp (K \vee \log d)$, $L \asymp d/p$, $p' \asymp K$ and $q \asymp \log d $.
    }
    \label{tab:comparison intro}
\end{table}

To compare FADI with existing works, we provide in Table \ref{tab:comparison intro} the theoretical error rates and the computational complexities of FADI against different PCA methods under Example~\ref{ex: spiked gaussian}.
We illustrate by Example~\ref{ex: spiked gaussian} because the existing distributed PCA methods mainly consider this setting
\citep{fandistributed2019}.
Under the distributed setting, FADI has a much lower computational complexity than the other three methods, while enjoying the same error rate as the traditional full-sample PCA. In comparison, the distributed PCA in \citep{fandistributed2019} is slowed down by applying traditional PCA to each data split.
 The fast PCA algorithm in \citep{halkofinding2011} has suboptimal computational complexity and theoretical error rate due to their downstream projection that hinders aggregation.

\vspace{-20pt}
\subsection{Estimation of the Rank  \texorpdfstring{$K$}{K}}\label{sec: k unknown}	

 FADI requires inputting the rank $K$ of $\Mb$. In practice, if we are only interested in estimating the leading PCs, the exact value of $K$ is not needed as long as
 the fast sketching dimensions,
 $p$ and $p'$, are sufficiently larger than $K$. Yet knowing the exact value of $K$ will improve the computational efficiency as well as facilitate inference on PCs. In fact, the estimation of $K$ can be incorporated into Step 2 and Step 3 of FADI. Specifically, for the $\ell$-th parallel server ( $\ell \in [L]$), after performing the SVD $\widehat\Yb^{(\ell)} = \hat{\Vb}_p^{(\ell)} \widehat\bLambda_p^{(\ell)} \hat{\Ub}_p^{(\ell) \top}$, we estimate $K$ by 
 \vspace{-15pt}
{\singlespace 
\begin{equation}\label{eq: est K ell}
    \hat{K}^{(\ell)} = \min  \{k < p: \sigma_{k+1}(\widehat\Yb^{(\ell)}) - \sigma_{p}(\widehat\Yb^{(\ell)}) \le \sqrt{p}\mu_0 \},
\end{equation}}where $\mu_0 > 0$ is a user-specified parameter (we refer to Theorem~\ref{thm: est k} for the choice of $\mu_0$). 
 Then send all the left singular vectors $\hat{\Vb}_p^{(\ell)}$ and $\hat{K}^{(\ell)}, \ell \in [L]$ to the central processor. Finally, on the central processor, take $\hat{K} = 
 \lceil
 \operatorname{median}\big\{\hat{K}^{(1)},\hat{K}^{(2)},\ldots, \hat{K}^{(L)}\big\}
 \rceil$ 
 as the estimator for $K$, and obtain $\widetilde\Vb_{\hat{K}}$ (respectively $\widetilde\Vb_{\hat{K}}^{\F}$) by performing PCA (respectively powered fast sketching) on the aggregated average of $\{\widehat\Vb_{\hat{K}}^{(\ell)}\}_{\ell \in [L]}$ and taking the $\hat{K}$ leading PCs, where $\widehat\Vb_{\hat{K}}^{(\ell)}$ is the $\hat{K}$ leading PCs of $\widehat\Yb^{(\ell)}$. We will show in Theorem~\ref{thm: est k} that $\hat{K}$ is a consistent estimator of $K$.
\vspace{-20pt}
\section{Theoretical Guarantees on the FADI Estimator}\label{sec: theory}
\vspace{-15pt}
\subsection{Theoretical Bound on Error Rates}\label{sec: thm err rate bound}

We need the following condition to guarantee that the error term converges at a proper rate.

\begin{assumption}[Convergence of $\|\Eb\|_2$]\label{asp: tail prob bound}
Recall that ${\Eb} = \widehat{\Mb} - \Mb$ is the error matrix. Assume that $\|\Eb\|_2$ is sub-exponential, and there exists a rate $r_1(d)$ such that  
\vspace{-15pt}
{\singlespace
$$\|\|\Eb\|_2\|_{\psi_1} = \sup_{q \ge 1} q^{-1}\left(\EE\|\Eb\|_2^q\right)^{1/q}\lesssim r_1(d).$$}
\end{assumption}
\begin{remark}\label{rmk: Eb subexp const}
By {standard probability theory}, 
we know that there exists a constant $c_e > 0$ such that for any $t > 0$ we have $\PP(\|\Eb\|_2 \ge t) \le { 2\exp\left(-c_e t/r_1(d)\right)}$ and $\|\Eb\|_2 = O_P\left(r_1(d)\right)$.
\end{remark}
We will conduct a variance-bias decomposition on the error rate $\cD(\widetilde{\Vb}, \Vb)$.
We first introduce the intermediate matrix ${\mathbf{\Sigma}}^\prime = \EE_{\mathbf{\Omega}} \big( \widehat{\Vb}^{(\ell)}\widehat{\Vb}^{(\ell)\top}  \big)$, 
where the expectation is taken with respect to $\mathbf{\Omega}$. Let $\Vb^\prime$ be the top $K$ eigenvectors of ${\mathbf{\Sigma}}^\prime$. 
Then for the FADI PC estimator $\widetilde{\Vb}$, we have the following ``variance-bias'' decomposition of the error rate:
\vspace{-15pt}
{\singlespace
$$
\cD(\widetilde{\Vb}, \Vb) \le \underbrace{ \cD(\widetilde\Vb, {\Vb}^{\prime})}_{\text{variance}} + \underbrace{ \cD(\Vb^\prime, {\Vb})}_{\text{bias}}.
$$}

\noindent Conditional on all the available data, the first term characterizes the statistical randomness of $\widetilde{\Vb}$ due to  fast sketching, whereas the second bias term is deterministic and depends on all the information provided by the data. Intuitively, since $\widetilde{\mathbf{\Sigma}} = \frac{1}{L} \sum_{\ell = 1}^{L} \widehat{\Vb}^{(\ell)} \widehat{\Vb}^{(\ell) \top}$ converges to the conditional expectation ${\mathbf{\Sigma}}^\prime$, $\widetilde\Vb$ will also converge to $\Vb^\prime$. Hence the first variance term goes to 0 asymptotically. 
As for the second bias term, let $\hat\Vb$ be the top $K$ eigenvectors of $\widehat\Mb$, and we have: $\cD(\Vb^\prime, {\Vb}) \le \cD(\hat\Vb, {\Vb}) + \cD(\Vb^\prime, \widehat{\Vb})$.
We can see that the first term is the error rate for the traditional PCA, whereas the second term is the bias caused by fast sketching. We can show 
that the second term is 0 with high probability and is hence negligible,
and the bias of the FADI estimator is of the same order as the error rate of the traditional PCA. Namely, the bias of FADI mainly comes from $\hat\Vb$, which is due to the  information we can get from the available data.
The following theorem gives the overall error rate of the FADI PC estimator.  Its proof is given in  Supplementary Materials~E.2. 

\begin{theorem}\label{thm: error bound}{\it
Under Assumption \ref{asp: tail prob bound}, if $p \ge \max(2K, K+7)$ and $ (\log d)^{-1}\sqrt{p/d} \Delta/r_1(d)$ $\ge C$ for some large enough constant $C > 0$, we have 
\vspace{-15pt}
{\singlespace
\begin{equation}\label{eq: error main}
      \left( \EE |\cD(\widetilde\Vb, \Vb)|^2 \right)^{1/2} \lesssim { \frac{\sqrt{K}}{\Delta} r_1(d)} + {\sqrt{\frac{Kd }{\Delta^2 pL}} r_1(d)}.
\end{equation}
}

\noindent Furthermore, 
under the conditions that $p \ge \max(2K, K+ 8q -1) $ and $ (\log d)^{-1}\sqrt{p/d} \Delta/r_1(d) \ge C$, there exists some constant $\eta > 0$ such that
\vspace{-15pt}
{\singlespace
\begin{equation}\label{eq: error VF Vk}
\begin{aligned}
   \left( \EE  |\cD(\widetilde\Vb^{\rm F}, \Vb)|^2 \right)^{1/2} & \! \lesssim  \! { \frac{\sqrt{K}}{\Delta} r_1(d)}\! +\! {\sqrt{\!\frac{Kd }{\Delta^2 pL}} r_1(d)} \!+\! \sqrt{\!\frac{Kd}{p'}}\!\!\left(\! \eta q^2  \!\sqrt{\frac{d }{\Delta^2p}}r_1(d) \!\right)^{q}\!\!\!,
   \end{aligned}
\end{equation}}
}
\end{theorem}
On the RHS of \eqref{eq: error main}, the first term is the bias term, while the second term is the variance term. When $Lp \asymp d$, the variance term will be of the same order as the bias term, which is the error rate of the traditional PCA. As for \eqref{eq: error VF Vk}, the first term and the second term on the RHS are the same as 
the bias and the variance terms in 
\eqref{eq: error main}, while the third term comes from the additional fast sketching.
{ When we further impose that $ \sqrt{p/d}  \Delta/r_1(d) \ge C \log^3 d $ for some large enough constant $C > 0$, if we properly choose $q$ to be 
{\singlespace 
\begin{equation}\label{eq: q range}
    3\log d / \log \big(\sqrt{p/d} \Delta/(\eta r_1(d) )\big) + 3 \le q \le  \big(\sqrt{p/d}  \Delta/(\eta r_1(d)) \big)^{1/3},
\end{equation}}the third term in \eqref{eq: error VF Vk} will be negligible. For $d \ge 3$, the LHS of \eqref{eq: q range} is upper bounded by $\log d /\log\log d + 3$, and hence taking $q = \lceil \log d /\log\log d + 3\rceil \lesssim \log d$ is sufficient in practice.}
Based upon Theorem~\ref{thm: error bound}, we provide the case-specific error rate for each example given in Section~\ref{sec: problem setting} in the following corollary. The proof is deferred to Supplementary Materials~E.3.

\begin{corollary}\label{prop: err rate terms}
{\it For Examples~\ref{ex: spiked gaussian} and \ref{ex: GMM}, we have the following error bounds for each case under corresponding regularity conditions.
\begin{itemize}
    \item {Example~\ref{ex: spiked gaussian}: } 
    Define $\kappa_1 = (\lambda_1 + \sigma^2)/\Delta$ and recall $r = \operatorname{tr}(\mathbf{\Sigma}) / \|\mathbf{\Sigma}\|_2$. Under the conditions that $d \ge 3$, $p' \ge \max(2K, K + 7)$, { $q = \lceil \log d /\log\log d + 3\rceil$}, $p \ge \max(2K,  K + 8q -1)$ and { $n \ge C (rd/{p})\kappa_1^2 \log^6 d$} for some large enough constant $C > 0$, it holds that
    \vspace{-15pt}
    {\singlespace
    \begin{equation}\label{eq: err bd exm 1}
            \left( \EE |\cD(\widetilde\Vb^{\rm F}, \Vb)|^2 \right)^{1/2} \lesssim  \kappa_1 \sqrt{\frac{Kr}{n}} + \kappa_1  \sqrt{\frac{Kd r}{npL}}.
    \end{equation}}

        \item {Example~\ref{ex: GMM}: } Under the conditions that { $d \ge 3$ and $ \Delta_0^2 \ge C K\log^3 d  \max\left\{ d \log^3  d/p , \sqrt{n/p}\right\}$}
        for some large enough constant $C > 0$,  where $\Delta_0 = \|\bTheta\|_2$,  taking $p' \ge \max(2K, K + 7)$, { $q = \lceil \log d /\log\log d + 3\rceil$} and $p \ge \max(2K, K + 8q -1)$, it holds that
    \vspace{-15pt}
    {\singlespace 
    \begin{equation}\label{eq: err bd exm 3}
         \begin{aligned}
        \left( \EE |\cD(\widetilde\Vb^{\rm F}, \Vb)|^2 \right)^{1/2} &\!\!\!\lesssim \left(\frac{K}{\Delta_0} \!+\!\frac{K}{\Delta_0^2} \sqrt{\frac{Kn}{d}}\right)\!+  \!\sqrt{\frac{d}{pL}}\left(\frac{K}{\Delta_0} \!+\!\frac{K}{\Delta_0^2} \sqrt{\frac{Kn}{d}}\right).
    \end{aligned}
    \end{equation}}
\end{itemize}
}
\end{corollary}

\begin{remark}\label{rmk: err rate case-specific}
{ The results of Example~\ref{ex: spiked gaussian} can be extended to heterogeneous residual variance models for non-i.i.d. data. For example, each sample $\bX_i$ can have a different covariance matrix $\bSigma_i = \EE(\bX_i \bX_i^{\top}) = \Vb \bLambda_i \Vb^{\top} + \Db_i$,
where $\bLambda_i$ and $\Db_i$ are diagonal and can vary across samples. This is useful when population covariances differ across different sites but share the same $K$-leading eigenspace. Supplementary Materials~E.4 details how FADI handles more heterogeneous settings by redefining $\widehat\Mb$, including this example as a special case. Specifically, we consider the more general setting where $\bSigma_i$ varies across samples and $n^{-1} \sum_{i=1}^n \bSigma_i$ converges to a limiting matrix $\bSigma = \Db + \Vb \bLambda \Vb^{\top} \in \RR^{d \times d}$ with a spiked structure. Corollary~E.5 provides the statistical rate of FADI under this heterogeneous setting. }
\end{remark}
 In Example~\ref{ex: spiked gaussian}, when $Lp \gtrsim d$, our error rate in \eqref{eq: err bd exm 1} achieves optimality \citep{fandistributed2019}. In the distributed data setting, we impose the condition $n/r \gtrsim d/p$, while \citet{fandistributed2019}'s distributed PCA requires $n/r > m$. In our approach, the additional factor $d/p$ can be interpreted as the number of ``vertical splits'' along the dimension $d$, playing a similar role as the extra factor $m$ (the number of splits along $n$) in \citet{fandistributed2019}'s method. Both \citet{fandistributed2019}'s and our scaling conditions involve an extra factor in exchange for reduced computational costs under the distributed setting.
 As for Example~\ref{ex: GMM}, our estimation rate in \eqref{eq: err bd exm 3} is the same as in \citep{chen2020spectral}. 
When the rank $K$ is unknown and estimated by FADI, the following theorem shows that under appropriate conditions, our estimator $\hat{K}$ presented in Section~\ref{sec: k unknown} recovers the true $K$ with high probability.
{
\begin{theorem}\label{thm: est k}
{\it 
Under Assumption \ref{asp: tail prob bound}, define $\eta_0 = 480 c_e^{-1} \sqrt{d/(\Delta^2 p)} r_1(d)\log d$, where $c_e > 0$ is the constant defined in Remark~\ref{rmk: Eb subexp const}. When $d \ge 2$, $2K \le p \ll d(\log d)^{-2}$ and $ \eta_0 \le (32\log d)^{-{2}/{(p-K+1)}}$, if we choose $\mu_0$ such that $\Delta \eta_0/24 \le \mu_0 \le \Delta \sqrt{\eta_0}/12$, then with probability at least $1-  O(d^{-(L \wedge 20)/2})$, $\hat{K} = K$.
}
\end{theorem}
We defer the proof to Supplementary Materials~E.5. 
We provide case-specific choices of  $\mu_0$ in the following corollary, whose proof is in Supplementary Materials~E.6.

\begin{corollary}\label{prop: est K}{\it
We specify the choice of $\mu_0$ for Examples \ref{ex: spiked gaussian} and \ref{ex: GMM}, 
\begin{itemize}
    \item {Example~\ref{ex: spiked gaussian}: } Under the conditions that $2K \le p \ll (\log d)^{-2}d$, $n \gg \kappa_1^2 rd/p (\log d)^4$, $(\lambda_1 + \sigma^2) \ll \left(\sqrt{np}/(d \log d)\right)^{1/4}$ and $\Delta \gg \left({\sigma^{-2} (np)^{-1/2}}d\log d\right)^{1/3}$, if we take $\mu_0 = \left(d(np)^{-1/2} \log d\right)^{3/4}/12$, with probability at least $1- O\left(d^{-(L \wedge 20)/2}\right) $, we have $\hat{K} = K$.
    
    \item {Example~\ref{ex: GMM}: } Under the conditions that $2K \le p \ll (\log d)^{-2}d$ and $ {K (\log d)^3}\sqrt{n/p} \ll \Delta_0^2 \ll {nK/d} (\log d)^2$, if we take $\mu_0 = d(\log d)^2\sqrt{n/p}/12$, with probability at least $1- O\left(d^{-(L \wedge 20)/2}\right) $, we have $\hat{K} = K$. 
    
\end{itemize}
}
\end{corollary}
}
\vspace{-20pt}
\subsection{Inferential Analysis: Intuition and Assumptions}\label{sec: inf intuition and asps}

In Section~\ref{sec: thm err rate bound}, we discuss the theoretical error bounds
and present the bias-variance decomposition for the FADI estimator $\widetilde\Vb^{\F}$. From \eqref{eq: error VF Vk}, we can see that when $Lp \gg d$, the bias term will be the leading term, and the dominating error comes from $\cD(\widehat\Vb, \Vb)$, whereas when $Lp \ll d$, the variance term will be the leading term and the main error derives from $\cD(\widetilde\Vb^{\F}, \widehat\Vb)$. This offers insight into conducting inference on the estimator and implies a possible phase transition in the asymptotic distribution. Before moving on to further discussions, we state the following assumption to ensure that the bias of $\widehat\Mb$ is negligible. 
\begin{assumption}[Statistical Rate for the Biased Error Term]\label{asp: stat rate biased error}{
For the error matrix $\Eb$ we have the decomposition $\Eb = \Eb_0 + \Eb_{b}$, where $\EE(\Eb_0) = \mathbf{0}$ and $\Eb_{b}$ is the biased error term satisfying
$ \lim_{d \rightarrow \infty} \PP\big(\|\Eb_b\|_2 {\le} r_2(d)\big) = 1$ with $r_2(d) = o\big(r_1(d)\big)$. }
\end{assumption}

In fact, we will later show in Section~\ref{sec: inf L much larger} and Section~\ref{sec: inf L much smaller} that the leading term for the distance between $\widetilde\Vb^{\F}$ and $\Vb$ takes on two different forms under the two regimes: 
\vspace{-15pt}
{\singlespace
$$
 \begin{array}{ll}
       \widetilde\Vb^{\F}\Hb - \Vb \approx \Pb_{\perp} \Eb_0 \Vb \bLambda^{-1} &, \quad \text{ if } Lp \gg d ; \\
       \widetilde\Vb^{\F}\Hb - \Vb \approx \Pb_{\perp} \Eb_0 \bOmega \Bb_{\bOmega} L^{-1}&, \quad  \text{ if } Lp \ll d,
\end{array}
$$}
%

\noindent where $\Hb$ is {some orthogonal aligning matrix,
$\Pb_{\perp} = \Ib - \Pb_{\Vb}= \Ib - \Vb \Vb^{\top}$,
$\mathbf{\Omega} = \frac{1}{\sqrt{p}}(\mathbf{\Omega}^{(1)}, \ldots, \mathbf{\Omega}^{(L)}) \in \RR^{d \times Lp}$ and $\Bb_{\mathbf\Omega} = (\Bb^{(1) \top}, \ldots, \Bb^{(L)\top})^{\top}$ with $\Bb^{(\ell)} = (\bLambda \Vb^{\top}\mathbf{\Omega}^{(\ell)}/\sqrt{p})^{\dagger} \in \RR^{p \times K}$ for $\ell = 1,\ldots,L$. Here $(\cdot)^{\dagger}$ stands for the Moore-Penrose pseudo inverse. To get an intuitive understanding on the form of the leading error term, let's start with the regime $Lp \gg d$ where $\cD(\widetilde\Vb^{\F}, \Vb) \approx \cD(\widehat\Vb, \Vb)$ and consider the case where $\{|\lambda_k|\}_{k =1}^K$ are well-separated such that $\Hb \approx \Ib_K$. Following basic algebra, we have 
{\singlespace \small
$$
    \widetilde\Vb^{\F} - \Vb  \approx \widehat\Vb - \Vb \approx \Pb_{\perp} ( \widehat\Vb - \Vb ) = \Pb_{\perp}(\widehat\Mb \widehat\Vb \widehat\bLambda^{-1} - \Mb \Vb \bLambda^{-1}) \approx \Pb_{\perp} (\widehat\Mb - \Mb)\Vb \bLambda^{-1} =  \Pb_{\perp} \Eb_0 \Vb \bLambda^{-1}, 
$$
}where $\widehat\bLambda$ is the $K$-leading eigenvalues of  $\widehat\Mb$ corresponding to $\widehat\Vb$, and the second approximation is due to the fact that $\widehat\Vb$ and $\Vb$ are fairly close and $\Pb_{\Vb}(\widehat\Vb  - \Vb)$ will be negligible. 

Now we turn to the scenario $Lp \ll d$, where the error mainly comes from $\widetilde\Vb^{\F} - \widehat\Vb$. For a given $\ell \in [L]$, denote $\Yb^{(\ell)} = \Mb \bOmega^{(\ell)} = \Vb\bLambda \widetilde\bOmega^{(\ell)}$, where $\widetilde\bOmega^{(\ell)} = \Vb^{\top}\bOmega^{(\ell)}$ is also a Gaussian test matrix. Intuitively, $p^{-1}\tilde{\mathbf{\Omega}}^{(\ell)} \tilde{\mathbf{\Omega}}^{(\ell)\top} \approx \Ib_K$   when $p$ is much larger than $K$. Hence $\tilde{\mathbf{\Omega}}^{(\ell)}$  acts like an orthonormal matrix scaled by $\sqrt{p}$,
and the rank-$K$ truncated SVD for  $\widehat\Yb^{(\ell)}/\sqrt{p}$ and $\Yb^{(\ell)}/\sqrt{p}$ will approximately be $\widehat\Vb^{(\ell)} \widehat\bLambda (\widetilde\bOmega^{(\ell)}/\sqrt{p})$ and $\Vb \bLambda (\widetilde\bOmega^{(\ell)}/\sqrt{p})$ respectively. Then following similar arguments as when $Lp \gg d$, we have
\vspace{-15pt}
{\singlespace
\begin{align*}
    &\widehat\Vb^{(\ell)} - \Vb \approx \Pb_{\perp} \left((\widehat\Yb^{(\ell) } /\sqrt{p})(\widetilde\bOmega^{(\ell)}/\sqrt{p})^{\top}\widehat\bLambda^{-1} - (\Yb^{(\ell) } /\sqrt{p})(\widetilde\bOmega^{(\ell)}/\sqrt{p})^{\top}\bLambda^{-1}   \right)\\
    &\quad  \approx \Pb_{\perp} \left(\widehat\Yb^{(\ell) }/\sqrt{p} - \Yb^{(\ell) }/\sqrt{p}\right)(\widetilde\bOmega^{(\ell)}/\sqrt{p})^{\top}\bLambda^{-1} \approx \Pb_{\perp} \Eb_0 (\bOmega^{(\ell)}/\sqrt{p})\Bb^{(\ell)},
\end{align*}}
 
\noindent where the last approximation is because when $\widetilde\bOmega^{(\ell)}/\sqrt{p}$ is almost orthonormal we have $\Bb^{(\ell)} = (\bLambda \widetilde\bOmega^{(\ell)}/\sqrt{p})^{\dagger} \approx (\widetilde\bOmega^{(\ell)}/\sqrt{p})^{\top}\bLambda^{-1}$. Then aggregating the results over $\ell \in [L]$  we have 
\vspace{-15pt}
{\singlespace
$$\widetilde\Vb^{\F} - \Vb \approx \frac{1}{L} \sum_{\ell = 1}^L  \left\{\widehat\Vb^{(\ell)} - \Vb \right\} \approx  \frac{1}{L} \sum_{\ell = 1}^L  \Pb_{\perp} \Eb_0 (\bOmega^{(\ell)}/\sqrt{p})\Bb^{(\ell)} = \Pb_{\perp} \Eb_0 \bOmega \Bb_{\bOmega}L^{-1}.$$}

\noindent It is worth noting that 
\vspace{-15pt}
{\singlespace 
\begin{equation}\label{eq: converge of cov two scenarios}
\frac{1}{L} \bOmega \Bb_{\bOmega} \approx \frac{1}{L} \Big(\sum_{\ell = 1}^L (\bOmega^{(\ell)}/\sqrt{p})(\bOmega^{(\ell)}/\sqrt{p})^{\top}\Big) \Vb \bLambda^{-1} \rightarrow  \Vb \bLambda^{-1},    
\end{equation}
}when $Lp \gg d $, which demonstrates the consistency of the leading term across different regimes of $Lp$. To unify the notations, we denote the leading term for $\widetilde\Vb^{\F}\Hb - \Vb$ by 
\vspace{-15pt}
{\singlespace
$$
\cV (\Eb_0) = \left\{\begin{array}{ll}
     \Pb_{\perp} \Eb_0 \Vb \bLambda^{-1} &, \quad \text{ if } Lp \gg d ; \\
     \Pb_{\perp} \Eb_0 \bOmega \Bb_{\bOmega} L^{-1} &, \quad  \text{ if } Lp \ll d.
\end{array}\right.
$$}

\noindent Before we formally present the theorems, we introduce the following extra regularity conditions necessary for studying the asymptotic features of the eigenspace estimator. 

\begin{assumption}[Incoherence Condition]\label{asp: incoh}{
For the eigenspace of the true matrix $\Mb$, we assume 
$
\|\Vb\|_{2,\infty} \le \sqrt{{\mu K}/{d}}
$,}
where $\mu \ge 1$ may change with $d$.
\end{assumption}
\begin{assumption}[Statistical Rates for Eigenspace Convergence]\label{asp: eigenspace}{ 
For the unbiased error term $\Eb_0$ and the traditional PCA estimator $\widehat\Vb$, we have the following statistical rates
\vspace{-15pt}
{\singlespace
$$
\lim_{d \rightarrow \infty} \PP\big(\|\widehat{\Vb} \operatorname{sgn}(\widehat{\Vb}^{\top} \Vb) - \Vb\|_{2,\infty} { \le} r_3(d)\big) = 1,
\,\,
 \lim_{d \rightarrow \infty}  \PP\big(\|\Eb_0(\Ib_d - \widehat\Vb\widehat\Vb^{\top})\Vb\|_{2,\infty} {\le} r_4(d)\big) = 1.
$$}}
\end{assumption}
\begin{assumption}[Central Limit Theorem]\label{asp: clt}{ 
For the leading term $\cV (\Eb_0)$ and any $j \in [d]$, it holds that $\mathbf{\Sigma}_j^{-1/2}\cV (\Eb_0)^{\top} \eb_j \overset{d}{\rightarrow} {\cN}(\mathbf{0}, \Ib_K)
$, where $\bSigma_j = \Cov(\cV (\Eb_0)^{\top} \eb_j | \bOmega)$ when $Lp \ll d$ and $\bSigma_j = \Cov(\cV (\Eb_0)^{\top} \eb_j )$ when $Lp \gg d$.}
\end{assumption}

Assumption \ref{asp: incoh} is the incoherence condition 
to guarantee that the information of the eigenspace is uniformly spread.
In Assumption \ref{asp: eigenspace} , $r_3(d)$ bounds the row-wise estimation error for the eigenspace, while $r_4(d)$ characterizes the row-wise convergence rate of the residual error term {projected onto the spaces spanned by $\hat\Vb_{\perp}$ and $\Vb$ consecutively}, i.e., $\|\Eb_0(\Ib_d - \widehat\Vb\widehat\Vb^{\top})\Vb\|_{2,\infty} = \|\Eb_0\Pb_{\widehat\Vb_{\perp}}\Pb_{\Vb}\|_{2,\infty}$. Assumption \ref{asp: clt} states that the leading term satisfies the central limit theorem (CLT). { These assumptions are for the general framework and will be translated into case-specific conditions for concrete examples.}
With the above assumptions in place, we are ready to present the formal inferential results. 
\vspace{-15pt}
\subsection{Inference When \texorpdfstring{$Lp \gg d$}{Lp gg d}}\label{sec: inf L much larger}
 We first define $\Hb = \Hb_2 \Hb_1 \Hb_0$ to be the alignment matrix between $\widetilde\Vb^{\F}$ and $\Vb$, where $\Hb_2 = \operatorname{sgn}(\widetilde{\Vb}^{\text{F}\top} \widetilde{\Vb})$, $\Hb_1 = \operatorname{sgn}(\widetilde{\Vb}^{\top} \widehat{\Vb})$ and $\Hb_0 = \operatorname{sgn}(\widehat{\Vb}^{\top} {\Vb})$. The following theorem provides the distributional guarantee of FADI when $Lp \gg d$.
\begin{theorem}\label{thm: leading term L big}{\it
    When { $d \ge 3$ and} $Lp \gg d$, under Assumptions \ref{asp: tail prob bound} - \ref{asp: clt}, recall $\bSigma_j = \Cov\big(\cV (\Eb_0)^{\top}\eb_j\big)$  for $j \in [d]$. Define $r(d) = \Delta^{-1}\big(\sqrt{\frac{Kd }{ pL}} r_1(d) + r_3(d)r_1(d) \!+\! \sqrt{\frac{\mu K}{d\Delta^2}}r_1(d)^2\!+\! r_2(d)\!+\! r_4(d)\big)$, and assume that there exists a statistical rate $\eta_1(d)$ such that 
    \vspace{-15pt}
    {\singlespace
    $$\min_{j \in [d]} \lambda_K \big(\bSigma_j \big) \gtrsim \eta_1(d) \quad \text{and} \quad \eta_1(d)^{-1/2} r(d) = o(1).$$
}
    
 \noindent If $\Delta^{-1}r_1(d)(\log d)^2\sqrt{d/p} = o(1)$ and we take 
 \vspace{-25pt}
 {\singlespace
    $$q \ge 2 + \log (Ld)/\log\log d, \quad p' \ge \max(2K, K+7) \quad\text{and}\quad p \ge \max(2K, K+8q-1),$$
    \begin{equation}\label{eq: general clt vk large L}
          \text{we have}\quad   \mathbf{\Sigma}_j^{-1/2}(\widetilde{\Vb}^{\rm F} \Hb - \Vb)^{\top} \eb_j \overset{d}{\rightarrow} {\cN}(\mathbf{0}, \Ib_K), \quad \forall j \in [d].
    \end{equation}}}
\end{theorem}
\begin{remark}\label{rmk: asymp normal large L}
The proof 
is deferred to Supplementary Materials~E.10. Here $\eta_1(d)$ guarantees that the asymptotic covariance
is positive definite, and the rate $r(d)$ bounds the remainder term stemming from fast sketching approximation and eigenspace misalignment. 
We will see in the concrete examples that the asymptotic covariance of the FADI estimator under the regime $Lp \gg d$ is the same as that of the traditional PCA estimator.  Namely, we can  increase the number of repeated sketches in exchange for the same testing efficiency as the traditional PCA. { We also discuss in Supplementary Materials~E.11 how incorporating additional sparsity assumptions on $\Vb$ can further improve estimation rate when $Lp \gg d$.}
We present the corollaries of Theorem~\ref{thm: leading term L big} for Examples \ref{ex: spiked gaussian} and \ref{ex: GMM} as follows. 
\end{remark}
Recall the set $S$ of size $K'$ defined in Section~\ref{sec: step 0} for estimating $\hat\sigma^2$. Denote by $\bSigma_{S}$ the population covariance matrix corresponding to $\hbSigma_S$ and by $\delta = \lambda_K(\bSigma_S) - \sigma^2$ the eigengap of $\bSigma_S$. Define $\tilde{\sigma}_1 = \|\mathbf{\Sigma}_{S}\|_2$.
We have the following corollary of Theorem~\ref{thm: leading term L big} for Example~\ref{ex: spiked gaussian}.
\begin{corollary}[Spiked Covariance Model]\label{col: inf gaussian case L big}{\it 
Assume that $\{\bX_i\}_{i=1}^n$ are i.i.d. multivariate Gaussian. 
If we take $K' = K+1$, $p' \ge \max(2K, K+7)$, $q \ge 2 + \log(Ld)/\log\log d$ and $p \ge \max(2K, K + 8q -1)$, then when { $d \ge 3$ }and  $Lp \gg  {Kdr}\kappa_1^2{\lambda_1}/{\sigma^2}$, under Assumption \ref{asp: incoh} and the conditions that 
\vspace{-15pt}
{\singlespace
$$n \gg \max \Big(\kappa_1^4(\log d)^4 r^2 {\lambda_1}/{\sigma^2}, \big(\kappa_1 {\lambda_1}/{\sigma^2}\big)^6\Big) \,\, \text{and} \,\, { K \ll \min\left(\big({\tilde\sigma_1}/{\delta}\big)^{-2}\kappa_1 r, \mu^{-{2}/{3}}\kappa^{-{4}/{3}}_1d^{{2}/{3}}\right),}$$  }

\noindent we have that \eqref{eq: general clt vk large L} holds.
Furthermore, we have 
\vspace{-15pt}
{\singlespace
\begin{equation}\label{eq: col gaussian L large 2}
    \widetilde{\bSigma}_j^{-1/2}(\widetilde{\Vb}^{\rm F} \Hb - \Vb)^{\top} \eb_j \overset{d}{\rightarrow} {\cN}(\mathbf{0}, \Ib_K), \quad \forall j \in [d],
\end{equation}}

\noindent where $\widetilde\bSigma_j  = \frac{\sigma^2}{n} \bLambda^{-1} \Vb^{\top}\bSigma \Vb \bLambda^{-1}$ is a simplification of $\bSigma_j$ under Example~\ref{ex: spiked gaussian}. Besides, if we define $\widetilde\bLambda = \widetilde\Vb^{{\rm F}\top}\widehat\Mb \widetilde\Vb^{\rm F}$ and estimate $\widetilde\bSigma_j$ by $\hbSigma_j = \frac{1}{n}( \hat\sigma^2 \widetilde\bLambda^{-1} + \hat\sigma^4 \widetilde\bLambda^{-2} )$, then we have
\vspace{-15pt}
{\singlespace
\begin{equation}\label{eq: est cov spiked gaussian L large}
    {\hbSigma}_j^{-1/2}(\widetilde{\Vb}^{\rm F}  - \Vb\Hb^{\top})^{\top}  \eb_j \overset{d}{\rightarrow} {\cN}(\mathbf{0}, \Ib_K), \quad \forall j \in [d].
\end{equation}}

}
\end{corollary}
\begin{remark}\label{rmk: col spiked gaussian large L}
The proof is in Supplementary Materials~E.12.   We  compute $\widetilde\bLambda$
 distributively across the $m$ data splits, and the cost for computing $\hbSigma_j$ is $O(ndK/m)$. 
We recommend taking $p = \lceil\sqrt{d}\rceil$, $L = \lceil \kappa_1^2 K d^{3/2}\log d \rceil $ and $q = \lceil\log d\rceil \gg 2 + \log (Ld)/\log\log d$ for optimal computational efficiency, where the total computation cost will be $O(K^3 d^{5/2} (\log d)^2)$. Our asymptotic covariance matrix is the same as that of the traditional PCA estimator under the incoherence condition \citep{fan2017eigenasymp}. Specifically, \citet{fan2017eigenasymp} studied the asymptotic distribution of the traditional PCA estimator by assuming that the spiked eigenvalues are well-separated and diverging to infinity, which is not required by our paper.  
Our scaling conditions are stronger than the estimation results in Corollary~\ref{prop: err rate terms} to cancel out the additional randomness induced by fast sketching and allow for efficient inference.
\end{remark}

 Denote by $\mu_{\theta} = \Delta_0^{-1}\sqrt{n/K} \|\bTheta\|_{2,\infty}$ the incoherence parameter for the Gaussian means. Then we have the following corollary 
 for Example~\ref{ex: GMM}.
\begin{corollary}[Gaussian Mixture Models]\label{col: GMM L big}{\it
When $d \ge 3$ and $Lp \gg d$, if we take $q \ge 2 + \log(Ld)/\log \log d$, $p \ge \max(2K, K + 8q -1)$ and $p' \ge \max(2K, K+7)$, under the conditions that
\vspace{-15pt}
{\singlespace
$${ K = o(d)}, \quad n \gg d^2, \quad K\sqrt{n}(\log d)^2 \ll \Delta_0^2 \ll  \frac{n^{4/3}}{\mu_{\theta}^2d} \quad \text{and} \quad L \gg \frac{Kd^2}{p},$$ }

\noindent \eqref{eq: general clt vk large L} holds. Furthermore, denote $\widetilde{\bSigma}_j = \mathbf{\Lambda}^{-1}\Vb^{\top}\big\{ \Fb\bTheta^{\top}\bTheta\Fb^{\top} + n\Ib_d \big\} \Vb \mathbf{\Lambda}^{-1}$, we have
\vspace{-15pt}
{\singlespace
\begin{equation}\label{eq: col GMM L big 2}
    \widetilde{\bSigma}_j^{-1/2}(\widetilde{\Vb}^{\rm F} \Hb - \Vb)^{\top} \eb_j \overset{d}{\rightarrow} {\cN}(\mathbf{0}, \Ib_K), \quad \forall j \in [d].
\end{equation}}

\noindent If we define $\widetilde\bLambda = \widetilde\Vb^{{\rm F}\top}\widehat\Mb \widetilde\Vb^{\rm F}$ and estimate $\widetilde\bSigma_j$ by $\hbSigma_j = \widetilde\bLambda^{-1} + n \widetilde\bLambda^{-2}$,  we have 
\vspace{-15pt}
{\singlespace
\begin{equation}\label{eq: est cov gmm L large}
{\hbSigma}_j^{-1/2}(\widetilde{\Vb}^{\rm F}  - \Vb\Hb^{\top})^{\top}  \eb_j \overset{d}{\rightarrow} {\cN}(\mathbf{0}, \Ib_K), \quad \forall j \in [d].
\end{equation}}

} 
\end{corollary}
\begin{remark}\label{rmk: col GMM large L}
Please refer to Supplementary Materials~E.13 for the proof. 
The distributive computation cost of $\hbSigma_j$ is $O(ndK/m)$. 
 We recommend taking $p = \lceil\sqrt{d}\rceil$, $L = \lceil K d^{3/2}\log d\rceil $ and $q = \lceil\log d\rceil$, 
 with total complexity of $O(K^3 d^{5/2} (\log d)^2)$. In Corollary~\ref{col: GMM L big}, the scaling condition for $n$ is $n \gg d^2$ compared to $n > d$ in Corollary~\ref{prop: err rate terms}, where the extra factor $d$ is to guarantee fast enough convergence rate of the remainder term for inference. It can be verified that the Cramér-Rao lower bound for unbiased estimators of $\Vb^{\top}\eb_j$ is $\bLambda^{-1}$, and thus
 when $\Delta_0$ is large enough, the asymptotic efficiency of $\widetilde\Vb^{\F}$ is 1 under the regime $Lp \gg d$.
\end{remark}

\vspace{-20pt}

\subsection{Inference When \texorpdfstring{$Lp \ll d$}{Lp ll d}}\label{sec: inf L much smaller}
Similar as when $Lp \gg d$, we first redefine the alignment matrix between $\widetilde\Vb^{\F}$ and $\Vb$ as $\Hb = \Hb_1 \Hb_0$, where $\Hb_1 = \operatorname{sgn}(\widetilde{\Vb}^{\text{F}\top} \widetilde{\Vb}) $ and $\Hb_0 = \operatorname{sgn}(\widetilde{\Vb}^{\top} \Vb)$. Then we have the following theorem characterizing the limiting distribution for $\widetilde\Vb^{\F}$. 


\begin{theorem}\label{thm: leading term}{\it 
     For the case when $Lp \ll d$, under Assumptions \ref{asp: tail prob bound}, \ref{asp: stat rate biased error}, \ref{asp: incoh} and \ref{asp: clt}, for $j \in [d]${,} recall $\mathbf{\Sigma}_j = \operatorname{Cov}(\cV (\Eb_0)^{\top}\eb_j | \mathbf{\Omega})$ and assume that there exists a statistical rate $\eta_2(d)$ such that
\vspace{-15pt}
{\singlespace
      $$\lim_{d \rightarrow \infty} \!\!\PP_{\mathbf{\Omega}} \Big(\!\min_{j \in [d]}\lambda_K\big(\bSigma_j\big) { \ge} \eta_2(d)\!\! \Big) \!= 1, \, \frac{d^2 r_1(d)^4(\log d)^4}{p^2 \Delta^4 \big( \eta_2(d)\!\!\wedge \!\!(\log d)^{-1}\big)} = o(1) \, \text{and} \, \frac{d r_2(d)^2}{Lp \Delta^2 \eta_2(d)} = o(1).$$}
      
     \noindent Then if we take $K(\log d)^2 \ll p \asymp p' \lesssim d/(\log d)^2$ and $q \ge \log d$ we have 
    \vspace{-15pt}
    {\singlespace \begin{equation}\label{eq: general clt for vk small L}
            \mathbf{\Sigma}_j^{-1/2}(\widetilde{\Vb}^{\rm F} \Hb - \Vb)^{\top} \eb_j \overset{d}{\rightarrow} {\cN}(\mathbf{0}, \Ib_K),\quad \forall j \in [d].
    \end{equation}}}
\end{theorem}
\begin{remark}
Theorem~\ref{thm: leading term} states that under 
proper scaling conditions,
the FADI estimator still enjoys asymptotic normality even when $Lp \ll d$.
The rate $\eta_2(d)$ is usually at least of order $(d/\lambda_1^2Lp)\lambda_{\min}(\Cov(\Eb_0\eb_j))$. In comparison, the rate $\eta_1(d)$ in Theorem~\ref{thm: leading term L big} is usually of order $\lambda_1^{-2}\lambda_{\min}(\Cov(\Eb_0\eb_j))$, suggesting a larger variance and lower testing efficiency of FADI at $Lp \ll d$ than at $Lp \gg d$.
The proof is deferred to Supplementary Materials~E.7.  
\end{remark}


The following corollaries of Theorem~\ref{thm: leading term} provide case-specific distributional guarantee for Examples~\ref{ex: spiked gaussian} and \ref{ex: GMM} under the regime $Lp \ll d$.


\begin{corollary}[Spiked Covariance Model]\label{col: guassian case}{\it
Assume  $\{\bX_i\}_{i=1}^n$ are i.i.d. multivariate Gaussian.  When $Lp \ll {\lambda_1^{-2} \Delta^2 d} $, if we take $K' = K+1$, $K(\log d)^2 \ll p \asymp p' \lesssim d/(\log d)^2$ and $q \ge \log d$,   under Assumption \ref{asp: incoh} and the conditions that 
\vspace{-15pt}
{\singlespace
$$\quad {n} \gg \max\Big( \frac{\kappa_1^4\lambda_1^2 d r^2 L}{p\sigma^4}, \frac{\lambda_1^2 \tilde{\sigma}_1^6K^2}{\Delta^2 \delta^4 \sigma^4} \Big)(\log d)^4 \quad \text{and} \quad  { \frac{K\lambda_1^2}{\Delta^2}\sqrt{\frac{\mu}{d}} = o(1) },$$ }

\noindent we have that \eqref{eq: general clt for vk small L} holds. Furthermore, if we define $\widetilde\bSigma_j = \frac{\sigma^2}{nL^2}\Bb_{\mathbf{\Omega}}^{\top}\mathbf{\Omega}^{\top}\bSigma\mathbf{\Omega}\Bb_{\mathbf{\Omega}}$, we have 
\vspace{-15pt}
{\singlespace
 \begin{equation}\label{eq: col gaussian L small 2}
    \widetilde\bSigma_j^{-1/2}(\widetilde{\Vb}^{\rm F} \Hb - \Vb)^{\top} \eb_j \overset{d}{\rightarrow} {\cN}(\mathbf{0}, \Ib_K),\quad \forall j \in [d].
\end{equation}}

\noindent If we further assume { ${\sigma^{-2}}{\lambda_1\kappa_1^4}\sqrt{{d^2 r}/{(np^2L)}} = o(1)$} and estimate $\widetilde\bSigma_j$ by $\widehat\bSigma_j = \frac{\hat\sigma^2}{nL^2}\widehat\Bb_{\mathbf{\Omega}}^{\top}\mathbf{\Omega}^{\top}\hbSigma\mathbf{\Omega}\widehat\Bb_{\mathbf{\Omega}}$, where  $\widehat\Bb_{\mathbf{\Omega}} = (\widehat\Bb^{(1)\top},\ldots, \widehat\Bb^{(L)\top})^{\top}$ with $\widehat\Bb^{(\ell)} = (\widetilde\Vb^{{\rm F}\top}\widehat\Yb^{(\ell)}/\sqrt{p})^{\dagger}$ for $\ell \in [L]$,  we have
\vspace{-15pt}
{\singlespace
\begin{equation}\label{eq: est cov spiked gaussian small L}
{\hbSigma}_j^{-1/2}(\widetilde{\Vb}^{\rm F}  - \Vb\Hb^{\top})^{\top}  \eb_j \overset{d}{\rightarrow} {\cN}(\mathbf{0}, \Ib_K), \quad \forall j \in [d].
\end{equation}}

}
\end{corollary}
\begin{remark}\label{rmk: col spiked gaussian small L}
The proof is in Supplementary Materials~E.8. 
For the computation of $\hbSigma_j$, apart from $\widehat\Vb^{(\ell)}$, the $\ell$-th machine on layer 2 (see Figure~\ref{fig: FADI illustrate}) will send $\mathbf{\Omega}^{(\ell)}$ and $\widehat\Yb^{(\ell)}$ to the central processor, and the total communication cost for each server is $O(dp)$. On the central processor,
the total computational cost of $\Bb_{\mathbf{\Omega}}$ will be $O(dpKL)$. Then we will compute 
$\mathbf{\Omega}^{\top} \hbSigma \mathbf{\Omega} = \frac{1}{\sqrt{p}}\mathbf{\Omega}^{\top}(\widehat\Yb^{(1)},\ldots, \widehat\Yb^{(L)}) + \hat\sigma^2\bOmega^{\top} \bOmega$ 
with total cost $O\big(d(Lp)^2\big) = o(d^3)$. Compared to Corollary~\ref{col: inf gaussian case L big},
Corollary~\ref{col: guassian case} has stronger scaling conditions on the sample size $n$ to compensate for the extra variability due to less fast sketches. As indicated by \eqref{eq: converge of cov two scenarios}, the asymptotic covariance matrix of Corollary~\ref{col: GMM L small} is consistent with Corollary~\ref{col: GMM L big}. 
\end{remark}

\begin{corollary}[Gaussian Mixture Models]\label{col: GMM L small}{\it 
When $Lp \ll d$, if we take $K(\log d)^2 \ll p \asymp p' \lesssim d/(\log d)^2$ and $q \ge \log d$, we have that \eqref{eq: general clt for vk small L} holds under the conditions that 
\vspace{-20pt}
{\singlespace
$${ \sqrt{\frac{K}{d}}\log d = O(1)}, \quad n \gg \frac{d^3L}{p}, \quad \text{and} \quad K (\log d)^2 \sqrt{\frac{dnL}{p}} \ll \Delta_0^2 \ll \min\left(n,\frac{n^{4/3}}{\mu_{\theta}^2d}\right).$$}

\noindent Furthermore, if we define 
$
    \widetilde\bSigma_j = L^{-2}\Bb_{\mathbf{\Omega}}^{\top}\mathbf{\Omega}^{\top}\Big(\Fb \bTheta^{\top} \bTheta\Fb^{\top} + n\Ib_d\Big)\mathbf{\Omega}\Bb_{\mathbf{\Omega}},
$
then we have 
\vspace{-15pt}
{\singlespace
 \begin{equation}\label{eq: col GMM L small 2}
    \widetilde\bSigma_j^{-1/2}(\widetilde{\Vb}^{\rm F} \Hb - \Vb)^{\top} \eb_j \overset{d}{\rightarrow} {\cN}(\mathbf{0}, \Ib_K), \quad \forall j \in [d].
\end{equation}}

\noindent If we further assume $d^4\Delta_0^2 \ll {KLp^2n^2}$ and estimate $\widetilde\bSigma_j$ by $\hbSigma_j = \frac{1}{L^2}\widehat\Bb_{\mathbf{\Omega}}^{\top}\mathbf{\Omega}^{\top}\Big(\widehat\Mb + n\Ib_d\Big)\mathbf{\Omega}\widehat\Bb_{\mathbf{\Omega}}$, where  $\widehat\Bb_{\mathbf{\Omega}} = (\widehat\Bb^{(1)\top},\ldots, \widehat\Bb^{(L)\top})^{\top}$ with $\widehat\Bb^{(\ell)} = (\widetilde\Vb^{{\rm F}\top}\widehat\Yb^{(\ell)}/\sqrt{p})^{\dagger}$ for $\ell \in [L]$, we have 
\vspace{-15pt}
{\singlespace
\begin{equation}\label{eq: est cov gmm L small}
{\hbSigma}_j^{-1/2}(\widetilde{\Vb}^{\rm F}  - \Vb\Hb^{\top})^{\top}  \eb_j \overset{d}{\rightarrow} {\cN}(\mathbf{0}, \Ib_K), \quad \forall j \in [d].
\end{equation}}

}
\end{corollary}

\begin{remark}\label{rmk: col GMM L small}
The proof of Corollary~\ref{col: GMM L small} is deferred to Supplementary Materials~E.9. Computation of $\hbSigma_j$ is very similar to Example~\ref{ex: spiked gaussian} as described in Remark~\ref{rmk: col spiked gaussian small L},  and the total computational cost is $O(d(Lp)^2) = o(d^3)$. The stronger scaling conditions are the trade-off for higher computational efficiency with less fast sketches.  
\end{remark}

\vspace{-20pt}

\section{Numerical Results}\label{sec: simulations}
In this section, we conduct extensive simulation studies to assess the performance of FADI
under each example given in Section~\ref{sec: problem setting} and compare it with several existing methods.

\vspace{-15pt}
\subsection{Example~\ref{ex: spiked gaussian}: Spiked Covariance Model}\label{sec: exm1 simu}
We generate $\{\bX_i\}_{i=1}^n$ i.i.d. from ${\cN}(\mathbf{0},\mathbf{\Sigma})$, 
 where $\bSigma = \Vb \bLambda \Vb^{\top} + \sigma^2 \Ib_d$.
We consider $K= 3$, 
$n=20000$ and set $d = 500, 1000, 2000$ respectively to study the asymptotic properties of the FADI estimator under different settings. To ensure the incoherence condition is satisfied, we set $\Vb$ to be the left singular vectors of a $d \times K$ i.i.d. Gaussian matrix. We take $\bLambda = \diag(6,4,2)$ and $\sigma^2 = 1$. 
We split the data into $m = 20$ subsamples, and set  $K' = 6$, $p = p' = 12$ and $q = 7$ 
to compute $\widetilde\Vb^{\F}$. We set $L$ at a range of values by taking the ratio $Lp/d \in \{0.2,0.6,0.9,1,1.2,2,5,10\}$ for each setting and compute the asymptotic covariance via Corollary~\ref{col: inf gaussian case L big} and Corollary~\ref{col: guassian case} correspondingly.
We define $\tilde\vb = \hbSigma_1^{-1/2}(\widetilde\Vb^{\F} - \Vb \Hb^{\top})^{\top}\eb_1$, where 
$\Hb = \sgn(\widetilde\Vb^{\F\top}\Vb)$, and calculate the coverage probability by empirically evaluating $\PP\big(\|\tilde\vb\|_2^2  \le \chi_3^2(0.95)\big)$ with $\chi_3^2(0.95)$ being the 0.95 quantile of the Chi-squared distribution with degrees of freedom equal to 3. Results
are shown in Figure~\ref{fig: exm1 simu}. Figure~\ref{fig: exm1 simu}(a) shows that as $Lp/d$ increases, the error rate of FADI converges to that of the traditional PCA. From Figure~\ref{fig: exm1 simu}(b) we can see that when $Lp/d$ is approaching 1 from the left, the computational efficiency drops due to the cost of computing $\hbSigma_1$. For Figure~\ref{fig: exm1 simu}(c), convergence towards the nominal 95\% level can be observed when $Lp/d$ is much smaller or much larger than 1, while the valley at $Lp/d$ around 1 is consistent with the theoretical conditions on $Lp/d$ in Section~\ref{sec: theory} and implies a possible phase-transition phenomenon on the distributional convergence of FADI. Note that the empirical coverage is closer to the nominal level 0.95 at $d = 2000$ than at $d \in \{500, 1000\} $,{  which might be caused by the vanishing of some error terms for approximation of the asymptotic covariance matrix as $d$ grows larger. }
The good Gaussian approximation of $\tilde\vb_1$  is further validated by Figure~\ref{fig: exm1 simu}(d), { where $\tilde\vb_1$ is the first entry of $\tilde\vb$}. { Based upon the low computational efficiency and poor empirical coverage at $Lp/d$ around 1, we recommend conducting inference based on FADI at regimes $Lp \gg d$ and $Lp \ll d$ only. In particular, we suggest the regime $Lp \gg d$ if priority is given to higher testing efficiency, and the regime $Lp \ll d$ if one needs valid inference with faster computation.}

We also compare FADI with traditional and distributed PCA \citep{fandistributed2019}. 
Results over 100 Monte Carlos are given in Table \ref{tab: dmn err rate n runtime}. We can see that FADI significantly outperforms both distributed PCA and traditional PCA in terms of computation time under the distributed setting. Specifically, FADI enjoys similar error rates to traditional PCA and distributed PCA while being computationally much faster, ranging from 65 to 717 times faster  than traditional PCA and 8.4 to 80.5 times faster than distributed PCA for a range of $d$ and $n$. Its computational advantage is more pronounced as $d$ and $n$ increase.  

\begin{figure}[ht]
		\centering
    \begin{tabular}{cc}
       \quad {\small (a) Error Rate} & \quad {\small (b) Running Time} \\
       \!\!\!\includegraphics[height=0.27\textwidth]{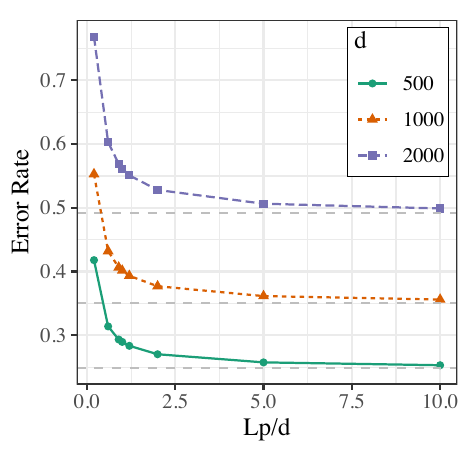} & \!\!\!\!\includegraphics[height=0.27\textwidth]{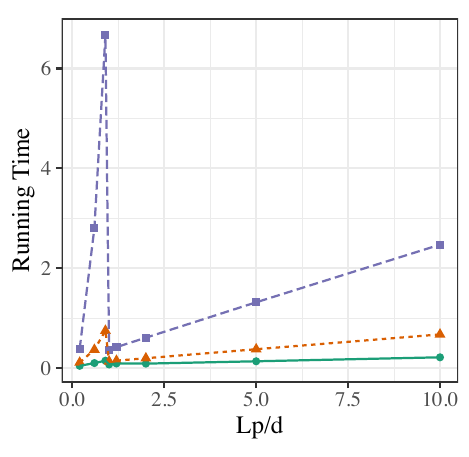} \\
        \quad \quad \!\!{\small (c) Coverage Probability} & \quad{\small (d) Q-Q Plot }\\
       \!\!\!\!\includegraphics[height=0.27\textwidth]{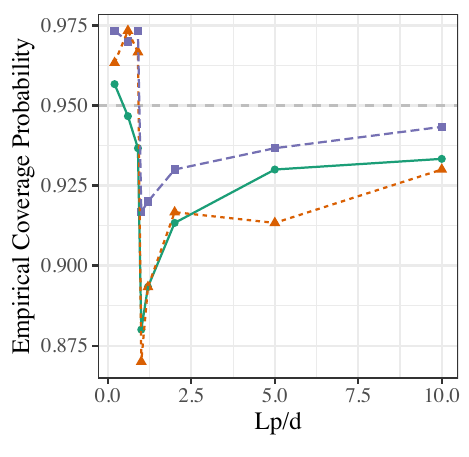} & \!\!\!\!\includegraphics[height=0.27\textwidth]{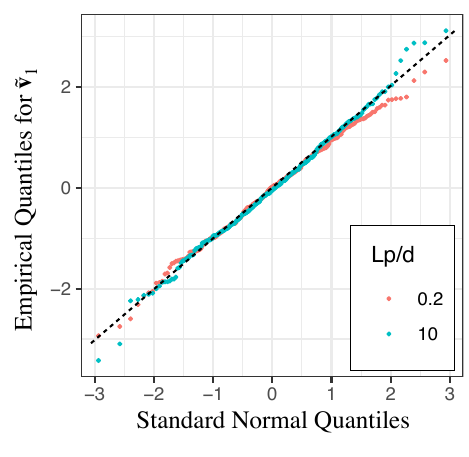}     
  \end{tabular}
		\caption{\small Performance of FADI under different settings for Example~\ref{ex: spiked gaussian} (with 300 Monte Carlos). (a) Empirical error rates of $\cD(\widetilde\Vb^{\rm F}, \Vb)$, where the grey dashed lines represent the error rates for the traditional PCA estimator $\widehat\Vb$; (b) Running time (in seconds) under different settings (including the computation time of $\hbSigma_1$). For the traditional PCA, the running time is 4.86 seconds at $d = 500$, 20.95 seconds at $d = 1000$ and 99.23 seconds at $d=2000$; (c) Empirical coverage probability; (d) Q-Q plot for $\tilde\vb_1$ at $Lp/d \in \{0.2, 10\}$.}\label{fig: exm1 simu}
	\end{figure}

  \begin{table}[htbp]
  \centering
    \resizebox{0.9\textwidth}{!}{%
    \begin{tabular}{c c c c | r r r |r r r}
      \hline
      \hline
      \multicolumn{4}{c|}{Parameters} & \multicolumn{3}{c|}{Error rate} & \multicolumn{3}{c}{Running time (seconds)}  \\
      \hline
      {$d$}&{$n$}&{$m$}&{$L$} & {FADI} & {Traditional} & {Distributed } &  {FADI} & {Traditional} & {Distributed }   \\
      \hline
400 & 30000 & 15 & 40 & 0.068 & 0.065 & 0.065 & 0.07 & 4.53  & 0.59\\
400 & 60000 & 30 & 40 & 0.048 & 0.046  & 0.046  & 0.05 & 8.84  & 0.60\\
400 & 100000 & 50 & 40 & 0.037 & 0.036  & 0.036 & 0.05 & 14.84  & 0.62\\
800 & 100000 & 50 & 80 & 0.052 & 0.050  & 0.050  & 0.10 & 55.76  & 3.66\\
800 & 5000 & 50 & 80 & 0.230 & 0.220 & 0.230 & 0.05 & 3.76  & 2.56\\
800 & 25000 & 50 & 80 & 0.106 & 0.103  & 0.103  & 0.07 & 15.07  & 2.82\\
800 & 50000 & 50 & 80 & 0.073 & 0.070  & 0.070  & 0.07 & 28.68  & 3.23\\
1600 & 30000 & 15 & 160 & 0.134 & 0.130  & 0.130  & 0.31 & 80.72  & 27.02\\
1600 & 60000 & 30 & 160 & 0.095 & 0.092  & 0.092  & 0.35 & 150.75  & 27.29\\
1600 & 100000 & 50 & 160 & 0.074 & 0.071 & 0.071  & 0.34 & 243.83  & 27.38\\
      \hline
      \hline
    \end{tabular}}
     \caption{\small Comparison of the empirical error rates (of $\cD(\cdot, \Vb)$) and the running times (in seconds) between FADI,  traditional full sample PCA and distributed PCA  \citep{fandistributed2019} 
     at $\mathbf{\Sigma} = \operatorname{diag}(50,25,12.5, 1,\ldots, 1)$.  For FADI, $p = p' = 12$, $K=3$, $K' = 4$, $\Delta = 11.5$ and $q=7$.
  }\label{tab: dmn err rate n runtime}
  \end{table}

\vspace{-20pt}
\subsection{Example~\ref{ex: GMM}: Gaussian Mixture Models}\label{sec: exm3 simu}
Under this setting, we take $K = 3$, fix the Gaussian vector dimension at $n = 20000$ and set $\Delta_0^2 = n^{2/3}$. Then we generate the Gaussian means by $\btheta_k \overset{\text{i.i.d.}}{\sim} N\left(\mathbf{0}, \frac{\Delta^2_0}{2n} \Ib_n \right)$, $k \in [K]$. We set $d = 500, 1000, 2000$ respectively and generate independent Gaussian samples $\{\bW_i\}_{i=1}^d \in \RR^n$ from a mixture of Gaussian with means $\btheta_k, k \in [K]$
under different settings. We assign each cluster $k \in [K]$ with $d/K$ Gaussian samples. We divide the data vertically along $n$ into $m = 20$ splits, set $p = p' = 12$ and $q = 7$ for the final powered fast sketching. We take the ratio $Lp/d \in \{0.2,0.6,0.9,1,1.2,2,5,10\}$ for each setting and compute the asymptotic covariance via Corollary~\ref{col: GMM L big} and Corollary~\ref{col: GMM L small}.
We define $\tilde\vb = \hbSigma_1^{-1/2}(\widetilde\Vb^{\F} - \Vb \Hb^{\top})^{\top}\eb_1$ where $\hbSigma_1$ is the asymptotic covariance for the first row of $\widetilde\Vb^{\F}$ and $\Hb = \sgn(\widetilde\Vb^{\F\top}\Vb)$ is the alignment matrix, and calculate the empirical coverage probability by empirically evaluating $\PP\big(\|\tilde\vb\|_2^2  \le \chi_3^2(0.95)\big)$.
We perform 300 Monte Carlo simulations and the results under different settings are shown in Figure~\ref{fig: exm3 simu}. We can see that the error rate of FADI gets closer to that of traditional PCA estimator as $Lp/d$ increases while  FADI greatly outperforms the traditional PCA in terms of running time under different settings. Note that here $d$ is the sample size, and the decreasing of error rates with increasing $d$ and fixed $n$ (at the same $Lp/d$ ratio) is consistent with Corollary~\ref{prop: err rate terms}. Similar to Example~\ref{ex: spiked gaussian},
we can see from Figure~\ref{fig: exm3 simu}(b) the running time is large due to the calculation of $\hbSigma_1$ at $Lp/d$ approaching 1 from the left, and we do not recommend inference at this regime. Validation of the inferential properties are shown in Figure~\ref{fig: exm3 simu}(c) and Figure~\ref{fig: exm3 simu}(d).

\begin{figure}[ht]
		\centering
		\begin{tabular}{cc}
       \quad \quad {\small (a) Error Rate} & \quad  {\small (b) Running Time} \\
       \includegraphics[height=0.27\textwidth]{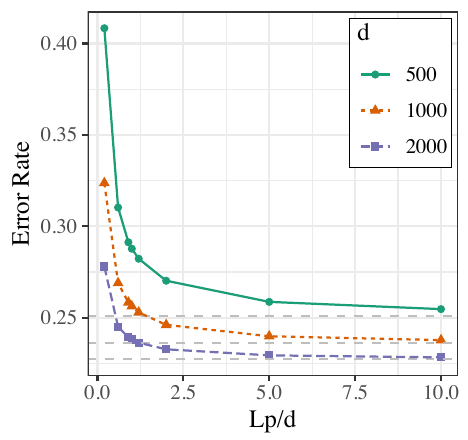} & \includegraphics[height=0.27\textwidth]{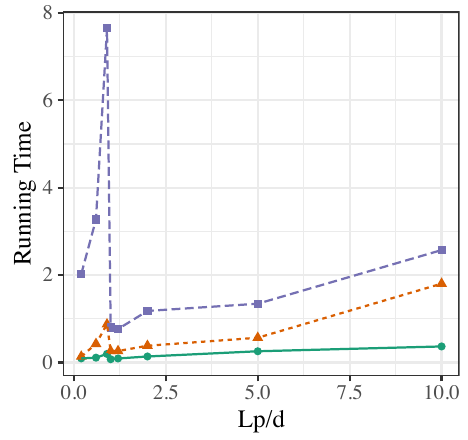}\\
       \quad \,\,{\small (c) Coverage Probability} & \quad  {\small (d) Q-Q Plot} \\
       \includegraphics[height=0.27\textwidth]{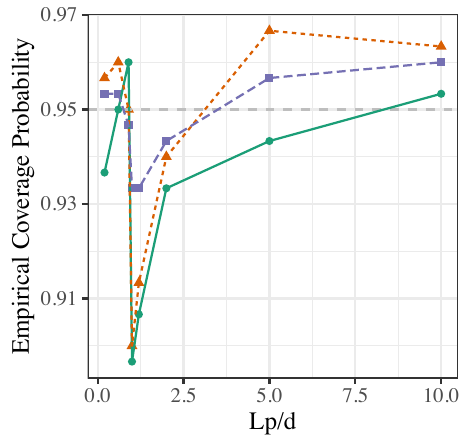} & \includegraphics[height=0.27\textwidth]{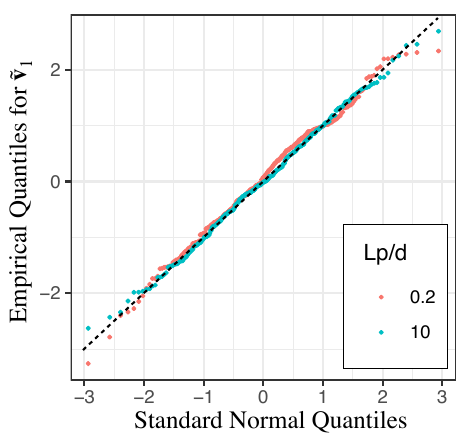}
  \end{tabular}
		\caption{\small \small Performance of FADI under different settings for Example~\ref{ex: GMM}. (a) Empirical error rates of $\cD(\widetilde\Vb^{\rm F}, \Vb)$;
  (b) Running time (in seconds) under different settings. For the traditional PCA, the running time is 5.43 seconds at $d = 500$, 23.32 seconds at $d = 1000$ and 105.58 seconds at $d=2000$; (c) Empirical coverage probability; 
  (d) Q-Q plot for $\tilde\vb_1$ at $Lp/d \in \{0.2, 10\}$.}\label{fig: exm3 simu}
	\end{figure}
 \vspace{-20pt}
\section{Application to the 1000 Genomes Data}\label{sec: real data}

In this section, we apply FADI and the existing methods to the 1000 Genomes Data \citep{10002015global}. 
We use phase 3 of the 1000 Genomes Data and focus on common variants 
with minor allele frequencies larger than or equal to 0.05. 
There are 2504 subjects in total, and 168,047 independent variants after the linkage disequilibrium (LD) pruning.  As we are interested in the ancestry principal components to capture population structure,  the sample size $n$ is the number of independent variants after LD pruning ($n=168,047$), and the dimension $d$ is the number of subjects ($d=2504$) \citep{price2006population}. The data were collected from 7 super populations: (1) \textbf{AFR}: African; (2) \textbf{AMR}: Admixed American; (3) \textbf{EAS}: East Asian; (4) \textbf{EUR}: European; (5) \textbf{SAS}: South Asian; (6) \textbf{PUR}: Puerto Rican and (7) \textbf{FIN}: Finnish; and 26 sub-populations.

For the estimation of the principal components, we assume that the data follow the spiked covariance model specified in Example~\ref{ex: spiked gaussian}. We also perform additional inferential analysis that we defer to Supplementary Materials~D. We set $K' = 27$, $p=50$, $p' = 100$, $q=3$, $m=100$ and $L = 80$.
For the estimation of the number of spikes, we take the thresholding parameter $\mu_0 = \left(d(np)^{-1/2} \log d\right)^{3/4}/12$. The estimated number of spikes from FADI is { $\hat{K} = 26$}, which is close to 25, the number of self-reported ethnicity groups minus 1, i.e., $K=26-1$. The results of the 4 leading PCs are shown in Figure~\ref{fig: PCA 1000g}, where a clear separation can be observed among different super-populations.
{ We compare the computational times of different methods for analyzing the 1000 Genomes Data. FADI takes 5.6 seconds at $q = 3$, whereas the traditional PCA method takes 595.4 seconds and the distributed PCA method \citep{fandistributed2019} takes  120.2 seconds.
These results show that FADI greatly outperforms the existing PCA methods in terms of computational time.  
}
  \begin{figure}[ht]
		\centering
		\begin{tabular}{ccc}
			\includegraphics[height=0.3\textwidth]{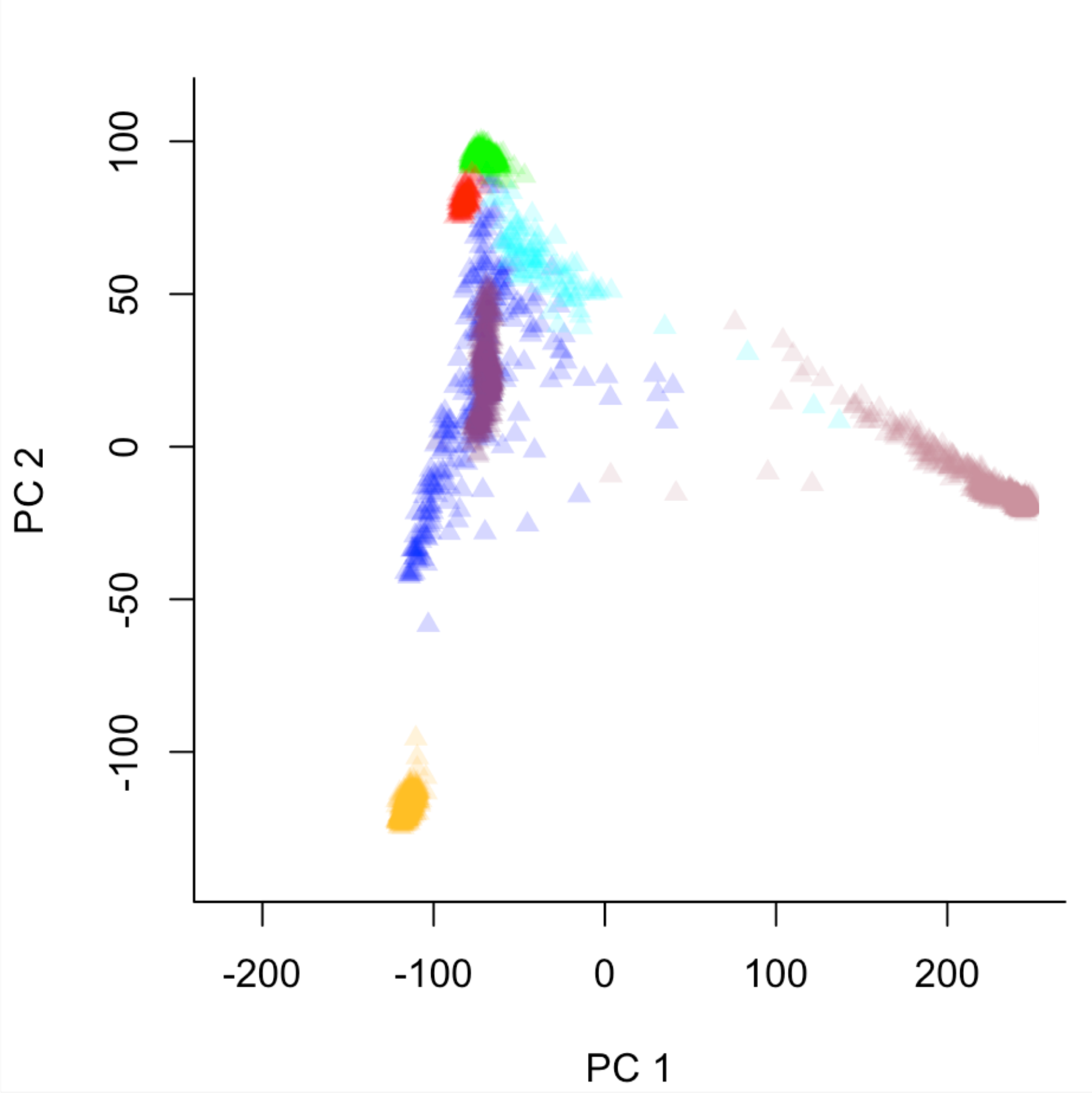}
			&\includegraphics[height=0.3\textwidth]{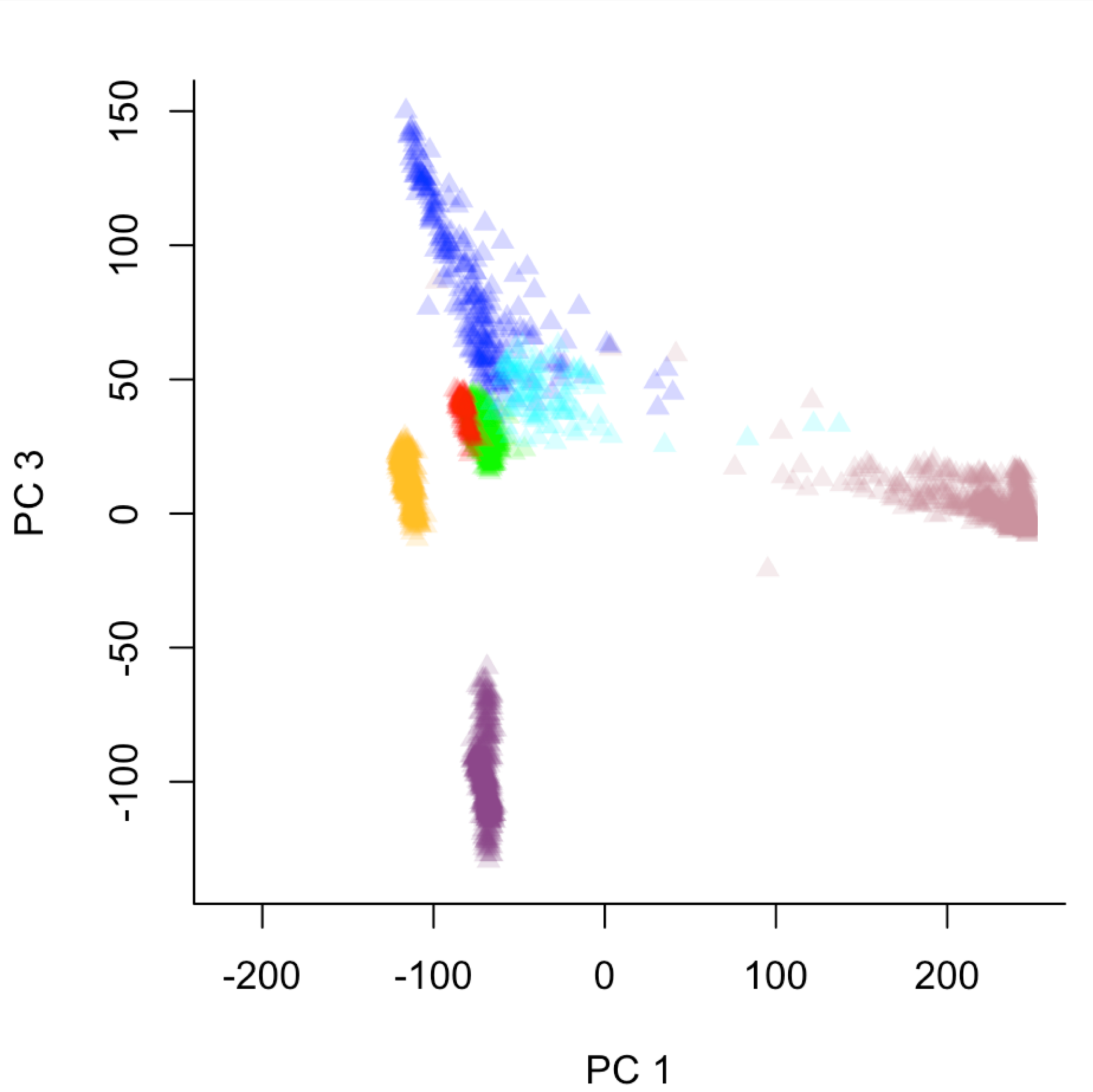}
			&\includegraphics[height=0.3\textwidth]{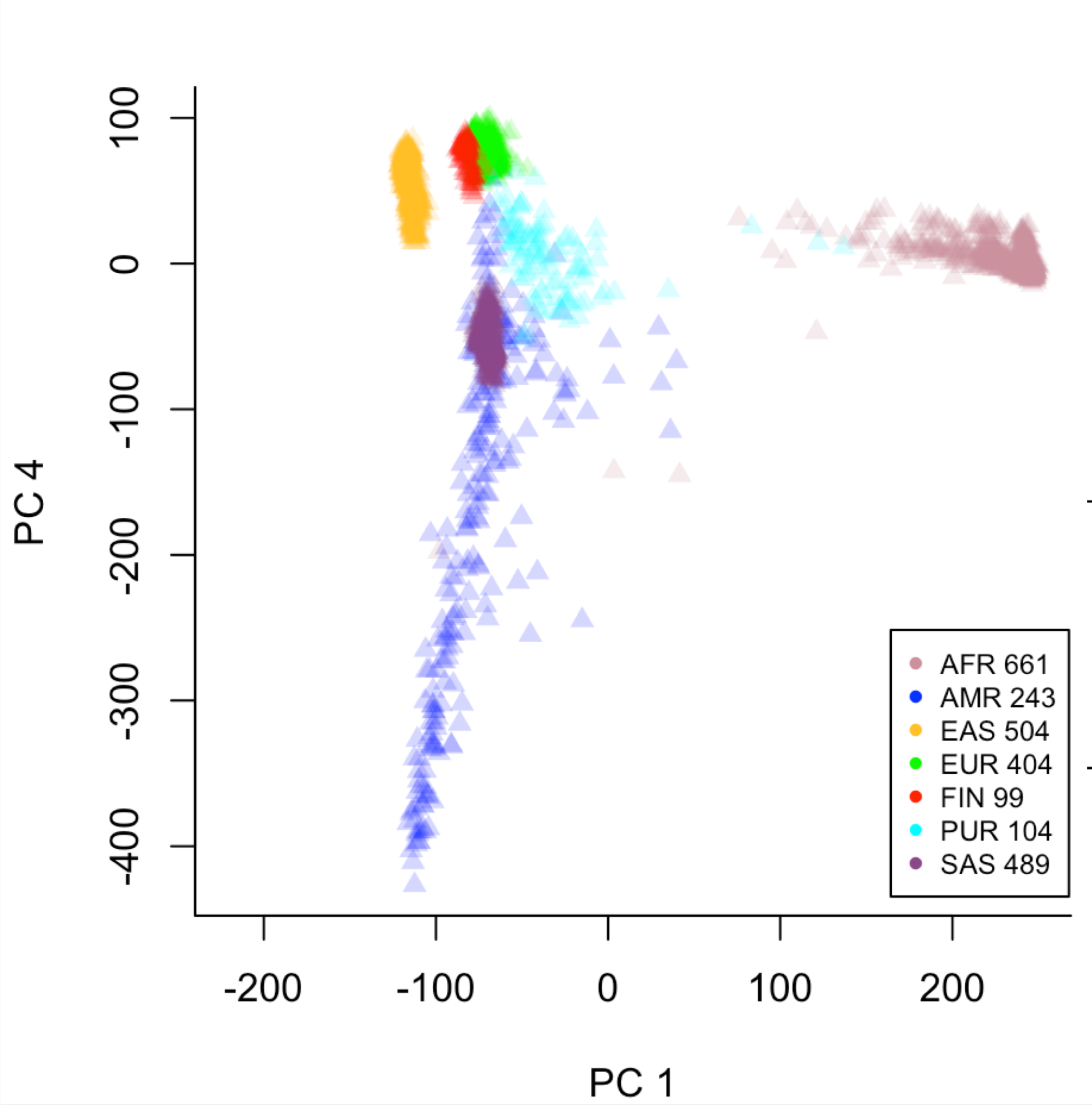}\\
			{\small (a) PC 1 versus PC 2 }& {\small (b) PC 1 versus PC 3} &{\small (c) PC 1 versus PC 4}\\
			\includegraphics[height=0.3\textwidth]{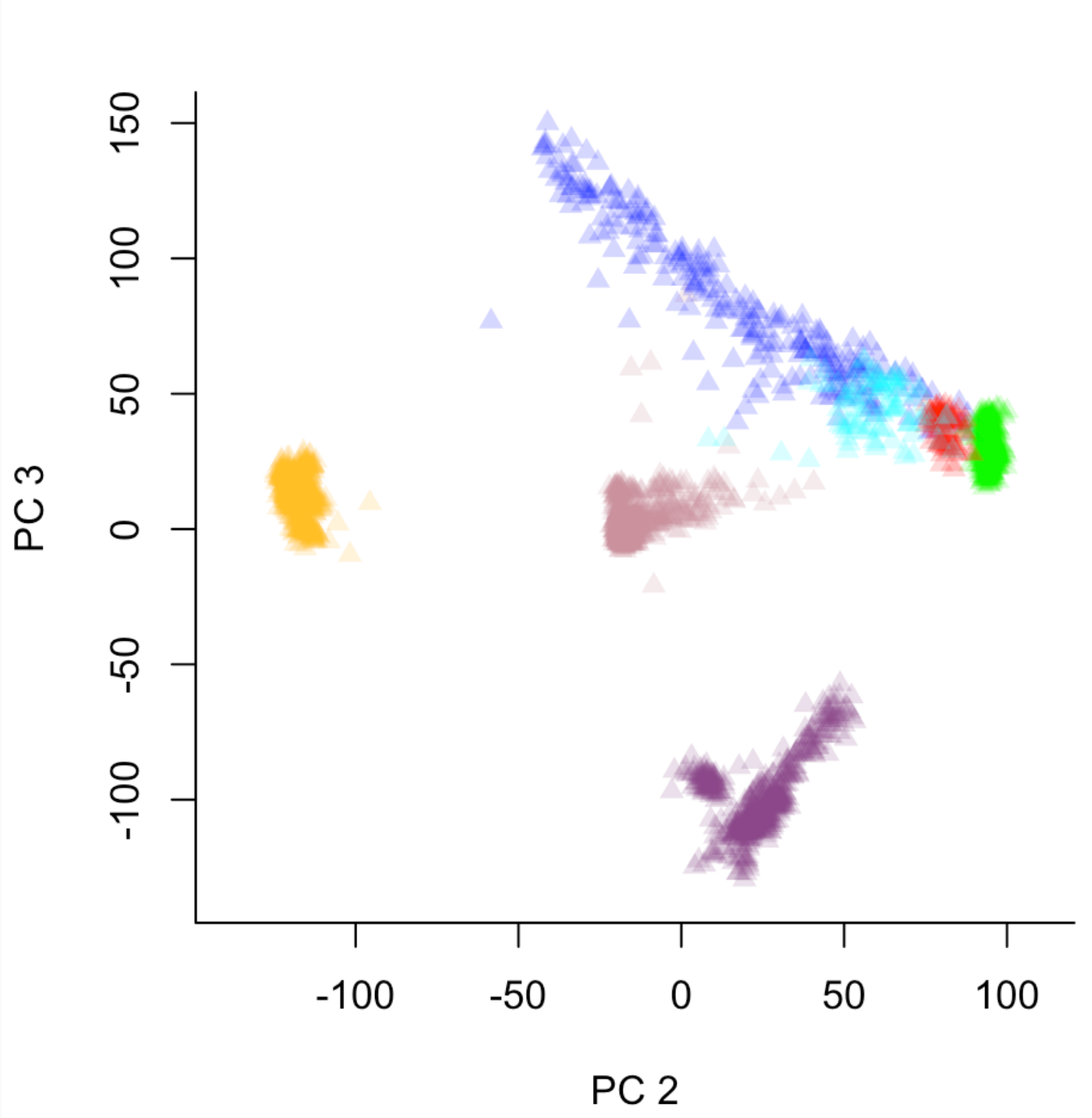}
			&\includegraphics[height=0.3\textwidth]{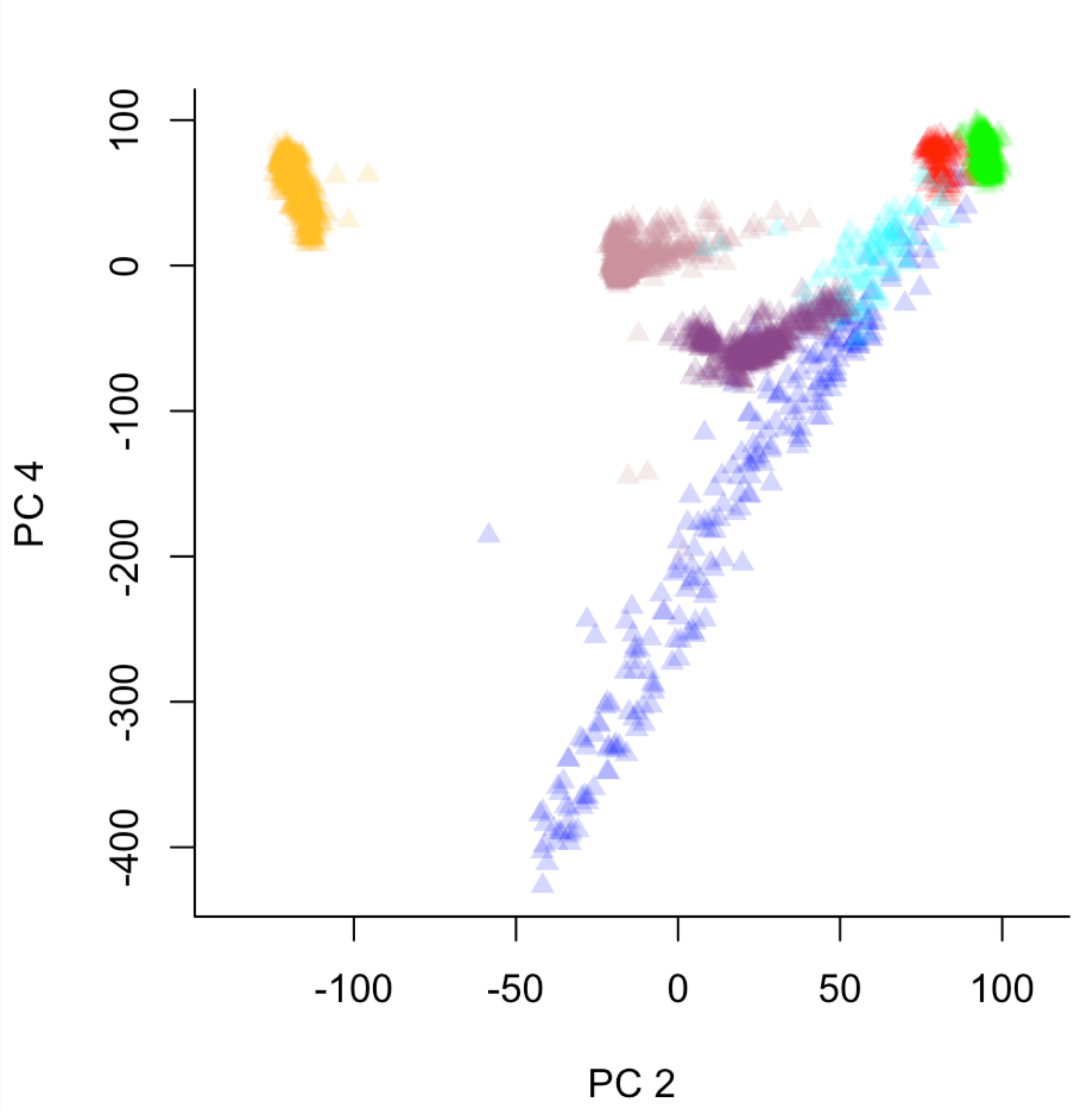}
			&\includegraphics[height=0.3\textwidth]{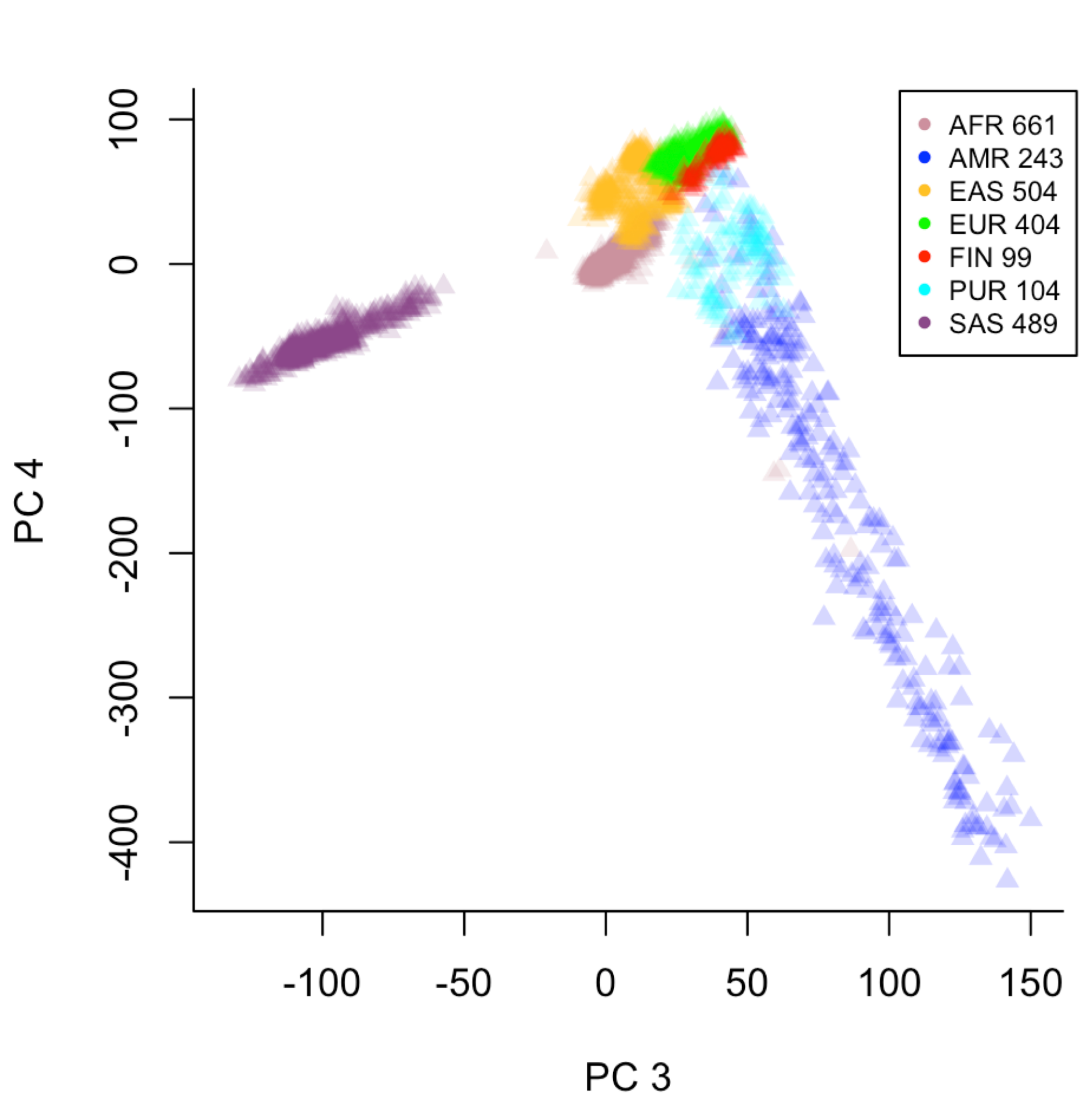}\\
			{\small (d) PC 2 versus PC 3 }& {\small (e) PC 2 versus PC 4 }&{\small (f) PC 3 versus PC 4}\\
		\end{tabular}
		\caption{\small The top 4 PCs of the 1000 Genomes Data. For the first two PCs, PC 1  separates African (AFR) super-population from the others, whereas PC 2  separates East Asian (EAS) from the others. As for PC 3 and PC 4, South Asian (SAS) and Ad Mixed American (AMR) are well separated from the rest of the super-populations by PC 3, while PC 4 presents some additional separation. }\label{fig: PCA 1000g}
	\end{figure}

\vspace{-20pt}

\section{Discussion}\label{sec: discussion}

{
In this paper, we develop a FAst DIstributed PCA algorithm (FADI) to address the challenges posed by high-dimensional PCA computations, offering a compelling balance between computational efficiency and result accuracy.
The main idea is to apply distributed-friendly random sketches so as to reduce the data dimension and aggregate the results from multiple sketches to improve the statistical accuracy and accommodate federated data.


In contrast to the computationally expensive traditional full sample PCA, which is not suitable for federated data, FADI significantly reduces computational costs and is well-suited for large federated data. While existing distributed PCA methods \citep{fandistributed2019} can handle federated data by applying traditional PCA to each data split, they lack scalability when the dimension $d$ is large. On the other hand, existing fast PCA methods \citep{halkofinding2011,chen2016integrating} use random sketches on full data, allowing for large $d$ but lack scalability for large sample sizes $n$ and are not applicable to federated data.
FADI addresses the limitations of both distributed PCA and fast PCA methods, offering significant scalability when both $d$ and $n$ are large or when dealing with federated data. It achieves computational scalability by computing multiple random sketches to split datasets and efficiently aggregating the results across them.
Theoretical analysis shows that FADI enjoys the same non-asymptotic error rate as the traditional PCA when the number of repeated sketches $L$ is of order $d/p$, which is also affirmed by extensive simulation studies. We also establish distributional guarantee for the FADI estimator and perform numerical experiments to validate the potential phase-transition phenomenon in distributional convergence.

Fast PCA algorithms using random sketches usually require the data to have certain ``almost low-rank'' structures, without which the approximation  might not be accurate \citep{halkofinding2011}. It is of future research interest to investigate whether the proposed FADI approach can be extended to non-low-rank settings. In Step 3 of FADI, we aggregate local estimators by taking a simple average over the projection matrices. It would be of future research interest to explore the performance of other weighted averages.
}
\subsection*{Disclosure Statement}
The authors report there are no competing interests to declare.

	\bibliographystyle{apalike}
		\bibliography{main.bib}

\newpage

\renewcommand*{\theHsection}{\thesection}
\renewcommand*{\theHsubsection}{\thesubsection}

\setcounter{section}{0}
\renewcommand{\thesection}{\Alph{section}}
\renewcommand{\theequation}{\Alph{section}.\arabic{equation}}
\setcounter{page}{1}
\renewcommand\thefigure{S\arabic{figure}}\setcounter{figure}{0}
\renewcommand\thetable{S\arabic{table}}\setcounter{table}{0}

\title{\singlespace
\begin{center}
\textit{\large Supplementary Materials to }
\end{center}
\begin{center}
{\Large Dimension Reduction for Large-Scale Federated Data: Statistical Rate and Asymptotic Inference}

\author{}
\date{}
\end{center}}
\maketitle
This file contains the supplementary materials to the paper ``Dimension Reduction for Large-Scale Federated Data: Statistical Rate and Asymptotic Inference''. Section~\ref{sec: add models} presents additional applications of FADI under the degree-corrected mixed membership (DCMM) model and the incomplete matrix inference model. In Section~\ref{sec: additional simus} we provide numerical results for Example~\ref{ex: SBM} and Example~\ref{ex: missing mat} along with some additional simulation results for Example~\ref{ex: spiked gaussian} under the genetic setting. In Section~\ref{sec: real data inference}, we present additional real data application of FADI to the 1000 Genomes Data for inference under Example~\ref{ex: SBM}. In Section~\ref{sec: proof main theory}, we present the proofs for the main theorems, propositions and corollaries given in Section~\ref{sec: theory} of the main paper. In Section~\ref{sec: proof tec lems} we give the proofs of some technical lemmas useful for the proofs of the main theorems. In Section~\ref{sec: supp}, we present the modified version of Wedin's theorem, which is used in several proofs.  Section~\ref{sec: supp figs} provides the supplementary figures deferred from the main paper.

\section{Additional Applications to Other Statistical Models}\label{sec: add models}
As mentioned in Section~\ref{sec: intro}, FADI is developed within a general framework that encompasses multiple statistical models. In the main text, we present applications of FADI to the spiked covariance model and the Gaussian mixture models (GMM) for the purpose of illustration. In this section, we provide additional applications of FADI to the degree-corrected mixed membership (DCMM) model and the incomplete matrix inference model. The specific model setups are provided below.
\begin{example}[Degree-Corrected Mixed Membership (DCMM) Model \citep{fan2022simple}]\label{ex: SBM}{ 
Let $\Xb  \in \RR^{d \times d}$ be a symmetric adjacency matrix for an undirected graph of $d$ nodes, where $\Xb_{ij} = 1$ if nodes $i, j \in [d]$ are connected and $\Xb_{ij} = 0$ otherwise. Assume $\Xb_{ij}$'s are independent for $i \le j$ and $\EE (\Xb) = \mathbf{\Theta} \mathbf{\Pi} \Pb \mathbf{\Pi}^{\top} \mathbf{\Theta}$,
where $\mathbf{\Theta} = \operatorname{diag}(\theta_1, \ldots, \theta_d)$ stands for the degree heterogeneity matrix, $\mathbf{\Pi} = (\boldsymbol{\pi}_1,\ldots,\boldsymbol{\pi}_d)^{\top} \in \RR^{d \times K}$ is the stacked community assignment probability vectors  
and $\Pb \in \RR^{K \times K}$ is a symmetric rank-$K$ matrix
with constant entries $\Pb_{kk'} \in (0,1)$ for $k, k' \in [K]$. 
Then $\Mb = \EE (\Xb) =  \mathbf{\Theta} \mathbf{\Pi} \Pb \mathbf{\Pi}^{\top} \mathbf{\Theta}$ and $\widehat\Mb = \Xb$.\footnote{In the case where self-loops are absent, $\Xb$ will be replaced by $\Xb' = \Xb - \diag(\Xb)$ and $\Eb$ will be replaced by $\Eb' = \Eb - \diag(\Xb)$.
Our theoretical results hold for both cases. } The goal is to infer the community membership profiles $\mathbf{\Pi}$.
 Recall $\Mb = \Vb \mathbf{\Lambda} \Vb^{\top}$. Since $\Vb$ and $\bTheta \mathbf{\Pi}$ share the same column space, we can make inference on $\mathbf{\Pi}$
through $\Vb$.\footnote{To address the degree heterogeneity, one can perform the SCORE normalization to cancel out $\bTheta$ \citep{jin2015fast}. }
In this paper, we assume that there exist constants $C \ge c >0$ such that $\sigma_{K}(\mathbf{\Pi}) \ge  c \sqrt{d/K}$, $c \le \lambda_{K}(\Pb) \le \lambda_{1}(\Pb) \le C K$ and $\max_i \theta_i \le C \min_i \theta_i$, where we define $\theta = \max_i \theta_i^2$ as the rate of signal strength.    We assume that the adjacency matrix is distributed across $m$ sites, where on the $s$-th site we observe the connectivity matrix $\Xb^{(s)} \in \RR^{d \times d_s}$ and $\Xb = (\Xb^{(1)}, \ldots, \Xb^{(m)})$. }
\end{example}
\begin{example}[Incomplete Matrix Inference \citep{Chen2019matcompinf}]\label{ex: missing mat}{
Assume that $\Mb = \Vb \bLambda \Vb^{\top}$ is a symmetric rank-$K$ matrix, and $\cS \subseteq [d] \times [d]$ is a subset of indices. We only observe the perturbed entries of $\Mb$ in the subset $\cS$. Specifically, for $i\le j$, we denote $\delta_{ij} = \delta_{ji} = \II\{(i,j) \in \cS\}$, and $\delta_{ij} \overset{\text{i.i.d}}{\sim} \operatorname{Bernoulli}(\theta)$ is an indicator for whether the $(i,j)$-th entry is missing. Then for $i,j \in [d]$, the observation for $\Mb_{ij}$ is
$\Xb_{ij}=(\Mb_{ij} + \varepsilon_{ij})\delta_{ij}$, where $\varepsilon_{ij} = \varepsilon_{ji}$ are i.i.d. random variables satisfying $\EE(\varepsilon_{ij}) = 0$, $\EE(\varepsilon_{ij}^2) = \sigma^2$ and $\sup_{i\le j}|\varepsilon_{ij}| \lesssim \sigma \log d$.\footnote{We can generalize the results to sub-Gaussian error $\varepsilon_{ij}$'s with variance proxy $\sigma^2$ by taking the truncated error $\varepsilon_{ij}^t = \varepsilon_{ij} \II\{|\varepsilon_{ij}| \le 4 \sigma \sqrt{ \log d}\}$, and by the maximal inequality for sub-Gaussian random variables we know that with probability at least $1 - O(d^{-6})$, $\varepsilon_{ij} = \varepsilon_{ij}^t, \forall i,j \in [d]$, and the theorems can be generalized with minor modifications.} Then to adjust for scaling, we define the observed data as $\widehat\Mb = [\widehat\Mb_{ij}] = \hat{\theta}^{-1}[\Xb_{ij}]$,
where $\hat\theta = 2|\cS|/\big(d(d+1)\big)$.\footnote{In practice, we can estimate $\Vb$ by $\Xb$ rather than by $\widehat\Mb = \hat{\theta}^{-1}\Xb$, since the two matrices share exactly the same eigenvectors. However, we need  the factor $\hat{\theta}^{-1}$ to preserve correct scaling for the estimation of eigenvalues as well as the follow-up inference. Please see Theorem~\ref{thm: est k} and Corollary~\ref{col: missing mat} for more details. } Consider the distributed setting where the data are split along $d$ on $m$ servers, where $\Xb^{(s)} \in \RR^{d \times d_s}$ stands for the observations on the $s$-th server and $\widehat\Mb = \hat\theta^{-1}(\Xb^{(1)}, \ldots, \Xb^{(m)})$. The goal is to infer
$\Vb$ from $\widehat\Mb$ in the presence of  missing data. 
}
\end{example}
In Example~\ref{ex: SBM} and Example~\ref{ex: missing mat}, the sample size $n$ coincides with the dimension $d$, and we consider the distributed settings along $d$. Table~\ref{table: complexity and par choice intro add exms} presents the computational complexities and parameter choices of FADI for these two examples.

\begin{table}[htb]
    \centering
    \begin{tabular}{c c c c}
    \hline
    \hline
         &Complexity & $p$ & $L$ \\
         \hline
         DCMM model & \thead{$O(d^2 p/m + d K p L \log d)$} & $\sqrt{d}$ & $\sqrt{d}$\\
         \hline
         Incomplete matrix inference & \thead{$O(d^2 p/m + d K p L \log d)$} & $\sqrt{d}$ & $\sqrt{d}$\\
         \hline
         \hline
    \end{tabular}
    \caption{\small Computational complexities and parameter choice  of FADI
    for PCA estimation under Example~\ref{ex: SBM} and Example~\ref{ex: missing mat},  where $K$ is the rank of $\Mb$, $d$ is the dimension of $\Mb$, $m$ is the number of data splits, $p$ is the fast sketching dimension and $L$ is the number of repeated sketches.
    } 
    \label{table: complexity and par choice intro add exms}
\end{table}

\subsection{Raw Data Processing}
The preliminary data processing for generating $\widehat\Mb$ in Step 0 of FADI is presented as follows. 

\noindent \textbf{Example~\ref{ex: SBM}: } Recall that the adjacency matrix is stored distributively on $m$ sites, and for the $s$-th site we observe the connectivity matrix $\Xb^{(s)}$.  Then for $s \in [m]$, define 
$\widehat\Mb^{(s)} = (\eb_s^{\top} \otimes \Ib_d) \diag(\Xb^{(1)},\ldots, \Xb^{(m)})$, where $\otimes$ is the Kronecker product, and $\{\eb_s\}_{s=1}^m \subseteq \RR^{m}$ is the canonical basis for $\RR^m$. Namely, $\widehat\Mb^{(s)}$ is the $s$-th observation $\Xb^{(s)}$ augmented by zeros, and $\widehat\Mb = \sum_{s=1}^m \widehat\Mb^{(s)} = (\Xb^{(1)}, \ldots, \Xb^{(m)}) = \Xb$. No preliminary computation is needed.

\noindent  \textbf{Example~\ref{ex: missing mat}: }
 Recall that we observe the split data $\{\Xb^{(s)}\}_{s=1}^m$ with missing entries on $m$ servers. Define $\widehat\Mb^{(s)} = \hat\theta^{-1} (\eb_s^{\top} \otimes \Ib_d) \diag(\Xb^{(1)},\ldots, \Xb^{(m)})$ for the $s$-th server, where $\hat\theta = 2|\cS|/\big(d(d+1)\big)$, then we have $\widehat\Mb = \sum_{s=1}^m \widehat\Mb^{(s)} = \hat\theta^{-1}(\Xb^{(1)}, \ldots, \Xb^{(m)})$.  

\subsection{Complexity Analysis}
\begin{table}[htbp]
    \centering
    \begin{tabular}{c c c c c}
    \hline
    \hline
    & \multicolumn{2}{c}{Communication } & \multicolumn{2}{c}{Computation }\\
    \hline
         &Example~\ref{ex: SBM}  & Example~\ref{ex: missing mat} &Example~\ref{ex: SBM}  & Example~\ref{ex: missing mat}  \\
         \hline
         \textbf{Step 0} & N/A & $O(1)$ &  N/A & $O(\frac{d^2} {m})$\\
         \hline
         \textbf{Step 1} & $O(mpd)$ & $O(mpd)$ & \thead{$\widehat\Yb^{(s,\ell)}: O(\frac{d^2 p}{m})$} & \thead{$\widehat\Yb^{(s,\ell)}: O(\frac{d^2 p}{m})$} \\
         \hline
         \textbf{Step 2} & $O(LKd)$ &$O(LKd)$ & \thead{$\widehat\Yb^{(\ell)}: O(mdp)$\\ $\widehat\Vb^{(\ell)}: O(dp^2)$} & \thead{$\widehat\Yb^{(\ell)}: O(mdp)$\\ $\widehat\Vb^{(\ell)}: O(dp^2)$} \\
         \hline
         {\textbf{Step 3}} & {{N/A}}&  {{N/A}} & \multicolumn{2}{c}{\thead{$\widetilde\Vb^{\F}: O( d K p' L q +dp^{\prime2})$}}\\
         \hline
         \textbf{Total} &$O(mpd + LKd)$&$O(mpd + LKd)$& \thead{$O(\frac{d^2 p}{m}+ d K p' L q { +mdp})$}  & \thead{$O(\frac{d^2 p}{m}+ d K p' L q{ +mdp})$}  \\
         \hline
         \hline
    \end{tabular}
    \caption{\small Communication and computational costs for Example~\ref{ex: SBM} and Example~\ref{ex: missing mat}. For the simplicity of presentation, we assume  $\max_{s \in [m]} d_s \asymp d/m$. We only recommend computing $\widetilde\Vb^{\rm F}$ instead of $\widetilde\Vb$ for Example~\ref{ex: SBM} and Example~\ref{ex: missing mat}.}
    \label{table: compu cost add exms}
\end{table}

Table~\ref{table: compu cost add exms} provides the complexities of each step for Example~\ref{ex: SBM} and Example~\ref{ex: missing mat}. When $m$ can be customized, we recommend taking $m \asymp \sqrt{d}$ for optimal efficiency. Since direct SVD on $\widetilde\bSigma$ will induce computational cost of order $d^3$ and we only suggest $\widetilde\Vb^{\F}$ as the eigenspace estimator.  If we take $p \asymp \sqrt{d}$, $L \asymp d/p$, $p' \asymp K$ and  $q \asymp \log d$, the total computational cost will be $O(d^{5/2}/m + K^2 d^{3/2} \log d)$. Computational costs for the inferential procedures will be discussed in Section~\ref{sec: inf Lp gg d add exms}.
\subsection{Statistical Rates and Rank Estimation}
Below, we present the corollary of Theorem~\ref{thm: error bound} that illustrates the error rates of FADI in Example~\ref{ex: SBM} and Example~\ref{ex: missing mat}. The proof is deferred to Section~\ref{sec: proof prop err rates}.
\begin{corollary}\label{prop: err rate terms add exms}
{\it For Example~\ref{ex: SBM} and Example~\ref{ex: missing mat}, we have the following error bounds under corresponding regularity conditions.
\begin{itemize}
    \item {Example~\ref{ex: SBM}: } Suppose $d \ge 3$ and $\theta { \ge} K^2 d^{-1/2+\epsilon}$ for some constant $\epsilon >0$. If we take $p' \ge \max(2K, K + 7)$, $p  \gtrsim \sqrt{d}$ and { $q = \lceil \log d /\log\log d + 3\rceil$}, it holds that
    \begin{equation}\label{eq: err bd exm 2}
            \left( \EE |\cD(\widetilde\Vb^{\rm F}, \Vb)|^2 \right)^{1/2} \lesssim K\sqrt{\frac{K}{d\theta}} + K \sqrt{\frac{K}{pL\theta}}.
    \end{equation}
    \item {Example~\ref{ex: missing mat}: } Define $\kappa_2 = |\lambda_1|/\Delta$. Suppose $d \ge 3$, $\theta \ge  d^{-1/2 + \epsilon}$ for some constant $\epsilon > 0$, $\sigma/\Delta \ll (\log d)^{-3} d^{-1}\sqrt{p\theta}$, $\|\Vb\|_{2, \infty} \le \sqrt{\mu K/d} $ for some $\mu \ge 1$ and { $\kappa_2\mu K \ll (\log d)^{-3} d^{\epsilon/2} $}, if we take $p' \ge \max(2K, K + 7)$, $p \gtrsim \sqrt{d}$ and { $q = \lceil \log d /\log\log d + 3\rceil$}, it holds that
    \begin{equation}\label{eq: err bd exm 4}
        \begin{aligned}
        \left( \EE |\cD(\widetilde\Vb^{\rm F}, \Vb)|^2 \right)^{1/2} &\!\!\lesssim \sqrt{K} \left(\frac{\kappa_2 \mu K}{\sqrt{d\theta}} \!+\! \sqrt{\frac{d\sigma^2}{\Delta^2 \theta}}\right) \!\!+\!  \sqrt{\frac{Kd}{pL}} \!\!\left(\frac{\kappa_2 \mu K}{\sqrt{d\theta}} \!+\! \sqrt{\frac{d\sigma^2}{\Delta^2 \theta}}\right).
    \end{aligned}
    \end{equation}

\end{itemize}
}
\end{corollary}

For Example~\ref{ex: SBM}, our estimation rate in \eqref{eq: err bd exm 2} matches the inferential results in \citep{fan2022simple}. Section~\ref{sec: inf Lp gg d add exms} gives a detailed comparison with the method in \citep{fan2022simple} in terms of the limiting distributions. For Example~\ref{ex: missing mat}, our error rate in \eqref{eq: err bd exm 4} matches the results in \citep{chen2020spectral}. Recall we show in Theorem~\ref{thm: est k} that when the rank $K$ is unknown, it can be recovered with high probability by properly choosing the thresholding parameter $\mu_0$. Corollary~\ref{prop: est K add exms} specifies the choice of $\mu_0$ for Example~\ref{ex: SBM} and Example~\ref{ex: missing mat}. Please refer to Section~\ref{sec: proof of prop est K} for the proof.

\begin{corollary}\label{prop: est K add exms}{\it
For Examples \ref{ex: SBM} and \ref{ex: missing mat}, we specify the choice of $\mu_0$ under certain regularity conditions.
\begin{itemize} 
    \item {Example~\ref{ex: SBM}: } Define $\hat{\theta} = d^{-2} \sum_{i \le j} \widehat{\Mb}_{ij}$, then under the condition that $\theta \ge K^2 d^{-1/2+\epsilon}$ for some constant $\epsilon > 0$ and $\sqrt{d} \lesssim p \ll (\log d)^{-2} d$,  if we take $\mu_0 = (\hat{\theta}/p)^{1/2} d \log d/12$, with probability at least $1- O\left(d^{-(L \wedge 20)/2}\right) $, we have $\hat{K} = K$. 
    \item {Example~\ref{ex: missing mat}: } 
    When $\theta \ge d^{-1/2 + \epsilon}$ for some constant $\epsilon > 0$, $\|\Vb\|_{2,\infty} \le \sqrt{\mu K/d}$ for some $\mu \ge 1$, $\kappa_2^2\mu^2 K \ll (\log d)^2$, $\sqrt{d} \lesssim p \ll (\log d)^{-2} d$ and $(p\theta)^{-1/4}\sqrt{d\sigma/\Delta}\log d = o(1)$, if we take $\mu_0 =  d\hat\sigma_0\log d (p\hat\theta)^{-1/2}/12$, where $\hat\sigma_0 = \big( \sum_{(i,j) \in \cS} (\hat\theta\widehat\Mb_{ij})^2 / |\cS|\big)^{1/2}$, then with probability at least $1- O\left(d^{-(L \wedge 20)/2}\right) $, we have $\hat{K} = K$.
\end{itemize}
}
\end{corollary}

\subsection{Inferential Results When \texorpdfstring{$Lp \gg d$}{Lp gg d}} \label{sec: inf Lp gg d add exms}
In this section, we provide the inferential results of Example~\ref{ex: SBM} and Example~\ref{ex: missing mat} based on Theorem~\ref{thm: leading term L big}.
\subsubsection{Degree-Corrected Mixed Membership Models}
\begin{corollary}\label{col: sbm}{\it 

When $d \ge 3$ and $\theta \ge K^2 d^{-1/2+\epsilon}$ for some constant $\epsilon >0$ and $K = o(d^{1/32})$, if we take $p \gtrsim \sqrt{d}$, $p' \ge \max(2K,K+7)$, $L \gg K^5 d^2/p$ and $q \ge 2 + \log (Ld)/\log \log d $,  then \eqref{eq: general clt vk large L} holds.
Furthermore, if we denote $\widetilde{\bSigma}_j = \mathbf{\Lambda}^{-1}\Vb^{\top} \diag \big(\left[\Mb_{jj'}(1-\Mb_{jj'})\right]_{j' \in [d]}\big) \Vb \mathbf{\Lambda}^{-1}$, we have
\begin{equation}\label{eq: col SBM 2}
    \widetilde{\bSigma}_j^{-1/2}(\widetilde{\Vb}^{\rm F} \Hb - \Vb)^{\top} \eb_j \overset{d}{\rightarrow} {\cN}(\mathbf{0}, \Ib_K), \quad \forall j \in [d].
\end{equation}
Besides, define $\widetilde{\Mb} =(\widetilde\Vb^{\rm F} \widetilde\Vb^{{\rm F}\top})\widehat\Mb(\widetilde\Vb^{\rm F}\widetilde\Vb^{{\rm F}\top})$ and  $\widetilde{\mathbf{\Lambda}} = \widetilde\Vb^{{\rm F}\top}\widehat\Mb\widetilde\Vb^{\rm F}$, then if we estimate $\widetilde\bSigma_j$ by $\widehat{\bSigma}_j = \widetilde{\mathbf{\Lambda}}^{-1} \widetilde\Vb^{{\rm F}\top} \diag\big([\widetilde\Mb_{jj'}(1-\widetilde\Mb_{jj'})]_{j' \in [d]}\big)\widetilde\Vb^{\rm F} \widetilde{\mathbf{\Lambda}}^{-1} $,
we have 
\begin{equation}\label{eq: est cov SBM}
{\hbSigma}_j^{-1/2}(\widetilde{\Vb}^{\rm F}  - \Vb\Hb^{\top})^{\top}  \eb_j \overset{d}{\rightarrow} {\cN}(\mathbf{0}, \Ib_K), \quad \forall j \in [d].
\end{equation}

}

\end{corollary}
\begin{remark}\label{rmk: col sbm}
The proof is deferred to Section~\ref{sec: proof col sbm}. We can obtain $\widetilde\bLambda$ by computing $\widetilde\Vb^{\F \top} \Xb^{(s)}$ in parallel for $s \in [m]$, and  the computational cost for $\widehat\bSigma_j$ is $O(d^2K/m)$. To achieve the optimal computational efficiency, we would take $p =\lceil \sqrt{d}\rceil$ and $L =\lceil K^5 d^{3/2} \log d\rceil$. Hence taking $q = \lceil \log d \rceil$ is sufficient, and the total computational cost will be $O(K^7 d^{5/2}(\log d)^2)$. Inferential analyses on the membership profiles has received attention in previous works \citep{fan2022simple, shenlu2020sbm}. \citet{fan2022simple} studied the asymptotic normality of the spectral estimator under the DCMM model with complicated assumptions on the eigen-structure (see Conditions 1, 3, 6, 7 in their paper). In comparison, we only impose non-singularity conditions on the membership profiles, but have a stronger scaling condition on the signal strength to facilitate the divide-and-conquer process. Our asymptotic covariance is almost the same as \citet{fan2022simple}'s, suggesting the same level of asymptotic efficiency. 
\end{remark}
\subsubsection{Incomplete Matrix Inference}
\begin{corollary}\label{col: missing mat}{
When $d \ge 3$ and $Lp \gg \kappa_2^2 Kd^2$ and $\theta \ge d^{-1/2 + \epsilon}$ for some constant $\epsilon >0$, if we take $p' \ge \max(2K, K+7)$, $p \gtrsim \sqrt{d}$ and $q \ge 2 + \log(Ld)/\log\log d$, then under Assumption \ref{asp: incoh} and the conditions that 
$$\kappa_2^6 K^3\mu^3=o(d^{1/2}) \quad \text{and} \quad \sigma/\Delta \ll \sqrt{\theta/d} \cdot \min\left(\big(\kappa_2^2\sqrt{\mu K} + \kappa_2\sqrt{K \log d}\big)^{-1}, \sqrt{p/d} \right),$$
we have that \eqref{eq: general clt vk large L} holds. Furthermore, if we denote $\widetilde{\bSigma}_j = \mathbf{\Lambda}^{-1}\Vb^{\top}\diag\big([\Mb_{jj'}^2(1-\theta)/\theta + \sigma^2/\theta ]_{j'=1}^d\big) \Vb \mathbf{\Lambda}^{-1}$, we have
\begin{equation}\label{eq: missing mat 2}
    \widetilde{\bSigma}_j^{-1/2}(\widetilde{\Vb}^{\rm F} \Hb - \Vb)^{\top} \eb_j \overset{d}{\rightarrow} {\cN}(\mathbf{0}, \Ib_K), \quad \forall j \in [d].
\end{equation}
 Define 
  $\widetilde\bLambda = \widetilde\Vb^{{\rm F}\top}\widehat\Mb \widetilde\Vb^{\rm F}$ and $\widetilde\Mb = \widetilde\Vb^{\rm F}\widetilde\bLambda\widetilde\Vb^{{\rm F}\top}$. If we estimate  $\sigma^2$ by $\hat\sigma^2 = \sum_{(i,i') \in \cS} (\hat\theta\widehat\Mb_{ii'}-\widetilde\Mb_{ii'})^2 / |\cS|$ and $\widetilde\bSigma_j$ by $\hbSigma_j = \widetilde\bLambda^{-1}\widetilde\Vb^{{\rm F}\top}\diag\big([\widetilde\Mb_{jj'}^2(1-\hat\theta)/\hat\theta + \hat\sigma^2/\hat\theta ]_{j'=1}^d\big) \widetilde\Vb^{\rm F} \widetilde\bLambda^{-1}$, we have 
  \begin{equation}\label{eq: est cov mis mat}
  {\hbSigma}_j^{-1/2}(\widetilde{\Vb}^{\rm F}  - \Vb\Hb^{\top})^{\top}  \eb_j \overset{d}{\rightarrow} {\cN}(\mathbf{0}, \Ib_K), \quad \forall j \in [d].
  \end{equation}
  
}
\end{corollary}
\begin{remark}\label{rmk: col missing mat}
Please see Supplementary Materials~\ref{sec: proof col missing mat} for the proof of Corollary~\ref{col: missing mat}. We compute $\widetilde\bLambda$ by calculating $\widetilde\Vb^{\F \top} \Xb^{(s)}$ in parallel, and then $\widetilde\bLambda$ can be communicated across servers at low cost for computing $\hat\sigma^2$.
The total computational cost for calculating $\hbSigma_j$ is $O(d^2 K/m)$.  We recommend taking $p = \lceil \sqrt{d} \rceil$, $L = \lceil \kappa_2^2 K d^{3/2}\log d\rceil$ and $q = \lceil \log d\rceil$, and the total computational cost will be $O(K^3 d^{5/2}(\log d)^2)$. \citet{Chen2019matcompinf} studied the incomplete matrix inference problem through penalized optimization, and their testing efficiency is the same as ours.
\end{remark}

 We do not have distributional results for Examples~\ref{ex: SBM} and \ref{ex: missing mat} under the regime $Lp \ll d$. An intuitive explanation would be that the information contained in each entry is independent for Example~\ref{ex: SBM} and Example~\ref{ex: missing mat}, and when $Lp \ll d$, too much information will be lost from the $d \times d$ graph or matrix. In comparison, we can still recover information from Examples~\ref{ex: spiked gaussian} and \ref{ex: GMM} under the regime $Lp \ll d$ due to the correlation structure of the matrix. 
\section{Additional Simulation Results}\label{sec: additional simus}
In this section we present the simulation results for Example~\ref{ex: SBM} and Example~\ref{ex: missing mat}, and we provide some additional simulation results for Example~\ref{ex: spiked gaussian} to evaluate the performance of FADI under the  genetic settings.

\subsection{Example~\ref{ex: SBM}: Degree-Corrected Mixed Membership Models}\label{sec: exm2 simu}
 We consider the mixed membership model without degree heterogeneity for the simulation, i.e., $\bTheta = \sqrt{\theta} \Ib_d$, and $\Mb = \theta \mathbf{\Pi} \Pb \mathbf{\Pi}^{\top}$. For two preselected nodes $j,j' \in [d]$, we test 
 ${\rm H}_0: \boldsymbol{\pi}_j = \boldsymbol{\pi}_{j'}$ vs. ${\rm H}_1: \boldsymbol{\pi}_j \neq  \boldsymbol{\pi}_{j'}$ by testing

whether $ \Vb^{\top} (\eb_j - \eb_{j'}) = 0$.  
To simulate the data, we set $\theta = 0.9$, $K = 3$, and set the membership profiles $\mathbf{\Pi}$ and the connection probability matrix $\Pb$ to be 
{\singlespace
$$
    \boldsymbol{\pi}_j=\left\{\begin{array}{ll}(1,0,0)^{\top} & \text { if } 1 \le j \le \lfloor d/6 \rfloor \\ (0,1,0)^{\top} & \text { if } \lfloor d/6 \rfloor <j \le \lfloor d/3 \rfloor \\ (0,0,1)^{\top} & \text{ if } \lfloor d/3 \rfloor < j \le \lfloor d/2 \rfloor \\ (0.6,0.2,0.2)^{\top} & \text{ if } \lfloor d/2 \rfloor < j \le \lfloor 5d/8 \rfloor \\ (0.2,0.6,0.2)^{\top} & \text{ if } \lfloor 5d/8 \rfloor < j \le \lfloor 3d/4 \rfloor \\ (0.2,0.2,0.6)^{\top} & \text{ if } \lfloor 3d/4 \rfloor < j \le \lfloor 7d/8 \rfloor \\ (1/3,1/3,1/3)^{\top} & \text{ if } \lfloor 7d/8 \rfloor < j \le \lfloor d \rfloor \end{array}\right., \quad \Pb = \begin{pmatrix}
1 & 0.2 & 0.1\\
0.2 & 1 & 0.2\\
0.1 & 0.2 & 1 \end{pmatrix}.
$$
}

We test the performance of FADI under $d \in \{ 500, 1000, 2000\}$ respectively, and under each setting of $d$, we take $m=10$, $p = p' = 12$, $q = 7$ and set $L$ by the ratio $Lp/d \in \{ 0.2, 0.6, 0.9, 1, 1.2, 2, 5, 10\}$. For each setting, we conduct 300 independent Monte Carlo simulations. { To perform the test, with minor modifications of Corollary~\ref{col: sbm}, we can show that 
\begin{equation}\label{eq: col simu SBM cov}
    \widetilde\bSigma_{j,j'}^{-1/2} (\widetilde\Vb^{\F}\Hb - \Vb)^{\top}(\eb_j - \eb_{j'}) \overset{d}{\rightarrow} {\cN}(\mathbf{0}, \Ib_K),
\end{equation}
where the asymptotic covariance is defined as $\widetilde\bSigma_{j,j'} = \widetilde\bSigma_j + \widetilde\bSigma_{j'}$
and can be consistently estimated by $\hbSigma_{j,j'} = \hbSigma_j + \hbSigma_{j'}$.
} 
We first preselect two nodes, which we denote by $j$ and $j'$, with membership profiles both equal to $(0.6,0.2,0.2)^{\top}$ and calculate the empirical coverage probability of $\PP\big(\|\tilde{\bd}\|_2^2  \le \chi_3^2(0.95)\big)$, where $\tilde\bd = \hbSigma_{j,j'}^{-1/2}\widetilde\Vb^{\F \top}(\eb_j - \eb_{j'})$.
We also evaluate the power of the test by choosing two nodes with different membership profiles equal to $(0.6, 0.2,0.2)^{\top}$ and $(1/3,1/3,1/3)^{\top}$ respectively, which we denote by $j$ and $k$. We empirically calculate the power $\PP\big(\|\tilde{\bd'}\|_2^2 \ge \chi_3^2(0.95)  \big)$, where $\tilde\bd' = \hbSigma_{j,k}^{-1/2}\widetilde\Vb^{\F \top}(\eb_j - \eb_k)$. Under the regime $Lp/d < 1$, we calculate the asymptotic covariance referring to Theorem  \ref{thm: leading term} by $$\widehat\bSigma_{j,j'} =L^{-2}\widehat\Bb_{\bOmega}^{\top}\bOmega{\top} \diag\left([\widetilde\Mb_{jk}(1 - \widetilde\Mb_{jk}) + \widetilde\Mb_{j'k}(1 - \widetilde\Mb_{j'k})]_{k=1}^d\right) \bOmega \widehat\Bb_{\bOmega},$$ where $\widehat\Bb_{\mathbf\Omega} = (\widehat\Bb^{(1) \top}, \ldots, \widehat\Bb^{(L)\top})^{\top}$ with $\widehat\Bb^{(\ell)} = (\widetilde\Vb^{{\rm F}\top} \widehat\Yb^{(\ell)}/\sqrt{p})^{\dagger} \in \RR^{p \times K}$ for $\ell = 1,\ldots,L$. We also apply k-means to $\widetilde\Vb^{\rm F}$ to differentiate different membership profiles and compare the misclustering rate with the traditional PCA. The results of different settings are shown in Figure~\ref{fig: simu 1 exm 2}. We can see from Figure~\ref{fig: simu 1 exm 2}(d) that under the regime $Lp/d < 1$, the empirical coverage probability is zero under all settings, which validates the necessity of $Lp/d \gg 1$ for performance guarantee. Figure~\ref{fig: simu 1 exm 2}(f) demonstrates the asymptotic normality of $\tilde\bd_1$ at $Lp/d = 10$ and poor Gaussian approximation of FADI at $Lp/d = 0.2$, { where $\tilde\bd_1$ is the first entry of $\tilde\bd$}. 

We also compare FADI with the SIMPLE method \citep{fan2022simple} on the membership profile inference under the DCMM model. The SIMPLE method conducted inference directly on the traditional PCA estimator $\widehat\Vb$ and adopted a one-step correction to the empirical eigenvalues for calculating the asymptotic covariance matrix. We compare the inferential performance of FADI at $Lp/d = 10$ with the SIMPLE method (under 100 independent Monte Carlos), and summarize the results in Table~\ref{tab: inf comp fadi simple}, where the running time includes both the PCA procedure and the computation time of $\hbSigma_{j,j'}$. Compared to the SIMPLE method, our method has a similar coverage probability and power but is  computationally more efficient.

  \begin{table}[htbp]
  \centering
    \begin{tabular}{c c c| r r | r r| r r}
      \hline
      \hline
      \multicolumn{3}{c|}{Parameters}& \multicolumn{2}{c|}{Coverage probability} & \multicolumn{2}{c|}{Power} & \multicolumn{2}{c}{Running time (seconds)}  \\
      \hline
      {$d$}& $p$ & $L$ & {FADI} & SIMPLE & {FADI} & SIMPLE & {FADI} & SIMPLE  \\
      \hline
500 & 12 & 417 & 0.91 & 0.92 & 0.87 & 0.88 & 0.21  & 0.73\\
1000 & 12 & 833 & 0.94 & 0.94  & 1.00  & 1.00 & 0.69  & 6.77\\
2000 & 12 & 1667 & 0.95 & 0.98  & 1.00 & 1.00 & 2.61  & 59.42\\
      \hline
      \hline
    \end{tabular}
    \caption{\small Comparison of the coverage probability, power and running time (in seconds) between FADI and SIMPLE \citep{fan2022simple} under different settings of $d$.  In all settings, we take $m = 10$, $p = p' = 12$, $q=7$ and set $Lp/d =10$ for FADI.}\label{tab: inf comp fadi simple}
  \end{table}
 \begin{figure}[ht]
		\centering
  \begin{tabular}{ccc}
      \quad \quad (a) Error Rate &  \quad (b) Misclustering Rate & \quad \!\!(c) Running Time \\
      \includegraphics[height = 0.3\textwidth]{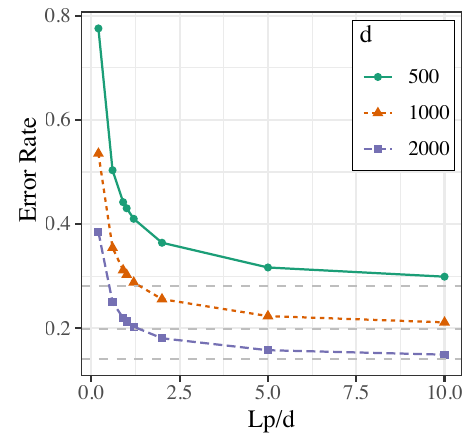} & \!\!\!\!\!\!\includegraphics[height = 0.3\textwidth]{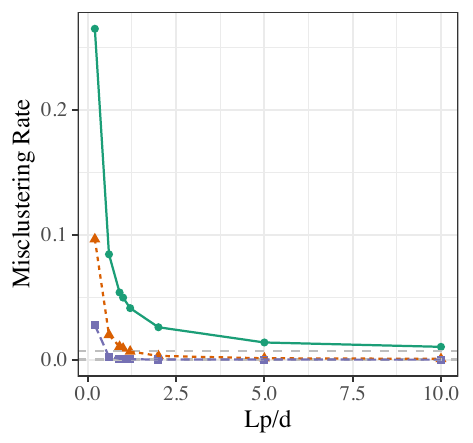} & \!\!\!\!\!\!\!\includegraphics[height = 0.3\textwidth]{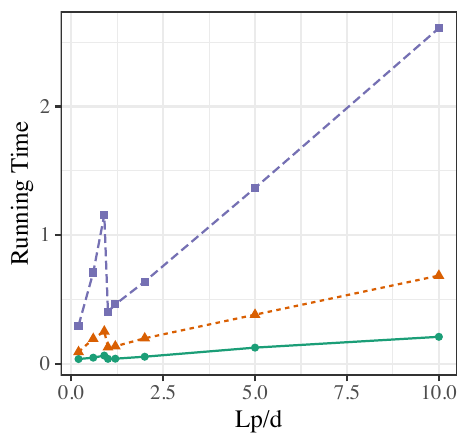}\\
      \quad (d) Coverage Probability &\quad  (e) Power & \quad (f) Q-Q Plot \\
      \includegraphics[height = 0.3\textwidth]{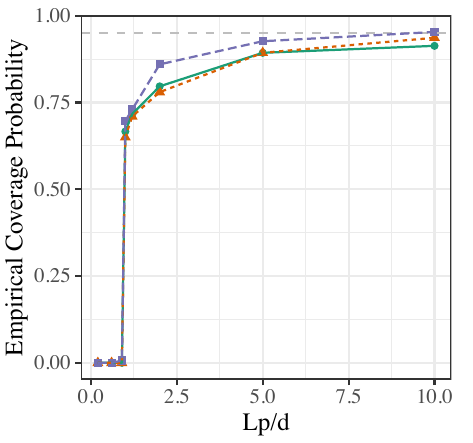} & \!\!\!\!\!\! \includegraphics[height = 0.3\textwidth]{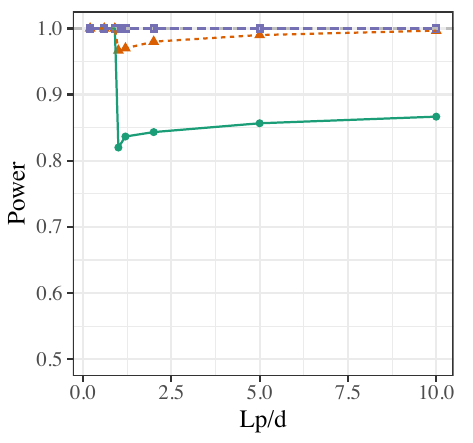} &\!\!\!\!\!\!\!\! \includegraphics[height = 0.3\textwidth]{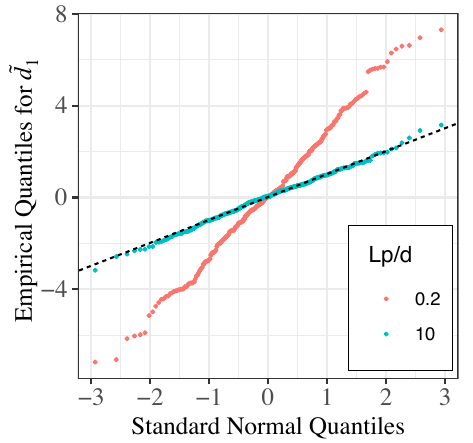}
  \end{tabular}
   \caption{\small Performance of FADI under different settings for Example~\ref{ex: SBM}. (a) Empirical error rates of $\cD(\widetilde\Vb^{\rm F}, \Vb)$; (b) Misclustering rate for $\widetilde\Vb^{\rm F}$ by K-means with grey dashed lines representing the misclustering rates for the traditional PCA estimator $\widehat\Vb$; (c) Running time (in seconds) under different settings (including computing $\hbSigma_{j,j'}$). For the traditional PCA, the running time is 0.43 seconds at $d = 500$, 3.77 seconds at $d = 1000$ and 32.62 seconds at $d=2000$; (d) Empirical coverage probability ($1 -$ Type I error); (e) Power of the test; (f) Q-Q plot for $\tilde\bd_1$ at $Lp/d \in \{0.2, 10\}$.    }\label{fig: simu 1 exm 2}
	\end{figure}
	

\subsection{Example~\ref{ex: missing mat}: Incomplete Matrix Inference}\label{sec: exm4 simu}
For the true matrix $\Mb$, we consider $K = 3$, take $\Vb$ to be the $K$ left singular vectors of a pregenerated $d \times K$ i.i.d. Gaussian matrix, and take $\bLambda = \diag(6, 4, 2)$. We { consider the distributed setting $m=10$, and} set the dimension at $d \in \{500, 1000, 2000\}$ respectively, and set $\theta = 0.4$ and $\sigma = 8/d$ for each setting. Then we generate the entry-wise noise by $\varepsilon_{ij} \overset{\text{i.i.d.}}{\sim} {\cN}(0, \sigma^2)$ for $i \le j$, and subsample non-zero entries of $\Mb$ with probability $\theta = 0.4$. Under each setting, we perform FADI at $p=p'=12$, $q=7$ and  $Lp/d \in \{0.2,0.6,0.9,1,1.2,2,5,10\}$ for the computation of $\widetilde\Vb^{\F}$. Define $\tilde\vb = \hbSigma_1^{-1/2}(\widetilde\Vb^{\F} - \Vb \Hb^{\top})^{\top}\eb_1$ with $\hbSigma_1$ being the asymptotic covariance for $\widetilde\Vb^{\F \top} \eb_1$ defined in Corollary~\ref{col: missing mat} and $\Hb = \sgn(\widetilde\Vb^{\F\top}\Vb)$, and empirically calculate the coverage probability, i.e., $\PP\big(\|\tilde\vb\|_2^2  \le \chi_3^2(0.95)\big)$. Similar as in Section~\ref{sec: exm2 simu}, for the regime $Lp < d$, we refer to Theorem~\ref{thm: leading term} and calculate $\hbSigma_1$ by
$$
\hbSigma_1 = L^{-2} \widehat\Bb_{\bOmega}^{\top} \bOmega^{\top} \diag\big([\widetilde\Mb_{1j}^2(1-\hat\theta)/\hat\theta + \hat\sigma^2/\hat\theta ]_{j=1}^d\big) \bOmega \widehat\Bb_{\bOmega}.
$$
Results over 300 Monte Carlo simulations are provided in Figure~\ref{fig: exm4 simu}. Figure~\ref{fig: exm4 simu}(a) illustrates that the error rate of FADI is almost the same as the traditional PCA as $Lp/d$ gets larger, and Figure~\ref{fig: exm4 simu}(b) shows that the computational efficiency of FADI greatly outperforms the traditional PCA for large dimension $d$. We can observe from Figure~\ref{fig: exm4 simu}(c) that the confidence interval performs poorly at $Lp/d < 1$ with the coverage probability equal to 1, which is consistent with the theoretical conditions in Corollary~\ref{col: missing mat} for distributional convergence. Figure~\ref{fig: exm4 simu}(d) shows the good Gaussian approximation of FADI at $Lp/d=10$, and the results at $Lp/d = 0.2$ is consistent with Figure~\ref{fig: exm4 simu}(c).  
\begin{figure}[ht]
		\centering
		\begin{tabular}{cc}
       \quad (a) Error Rate &  \quad (b) Running Time \\
  \includegraphics[height=0.3\textwidth]{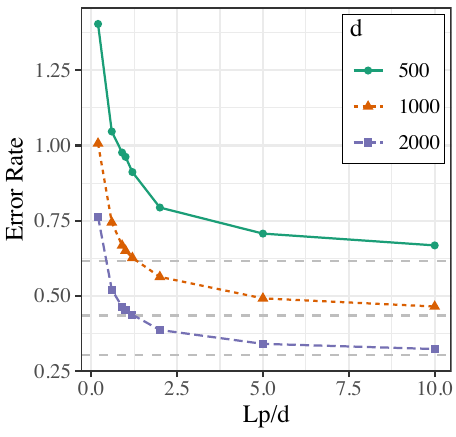} & \includegraphics[height=0.3\textwidth]{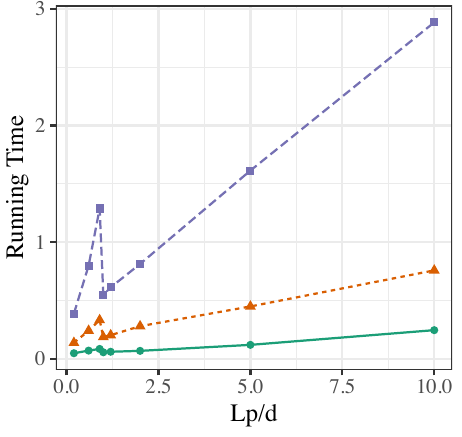} \\
 \quad (c) Coverage Probability & \quad (d) Q-Q Plot \\
 \includegraphics[height=0.3\textwidth]{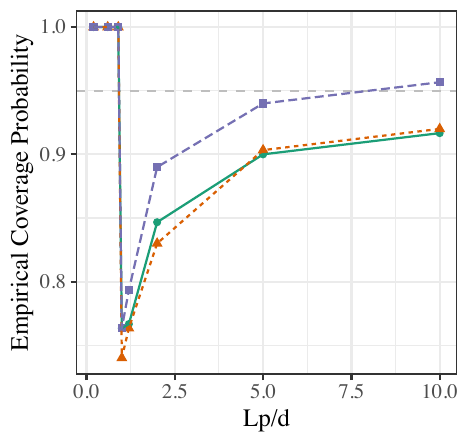} & \includegraphics[height=0.3\textwidth]{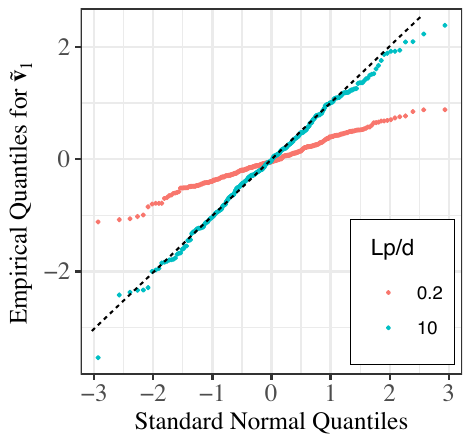}
  \end{tabular}
		\caption{\small Performance of FADI under different settings for Example~\ref{ex: missing mat}. (a) Empirical error rates of $\cD(\widetilde\Vb^{\rm F}, \Vb)$ with traditional PCA error rates as the reference; (b) Running time (in seconds) under different settings (including the computational time of $\hbSigma_1$). For the traditional PCA, the running time is 0.42 seconds at $d = 500$, 3.48 seconds at $d = 1000$ and 30.62 seconds at $d=2000$; (c) Empirical coverage probability; (d) Q-Q plot for $\tilde\vb_1$ at $Lp/d = 10$.}\label{fig: exm4 simu}
	\end{figure}
\subsection{Additional Results for Example~\ref{ex: spiked gaussian} in the Genetic Setting}\label{sec: exm1 simu gene setting}

 Section~\ref{sec: exm1 simu} compares FADI with several existing methods under a relatively large eigengap.
 In practice, the eigengap of the population covariance matrix may not be large. To assess different methods in a more realistic scenario, we imitate the setting of the 1000 Genomes Data, where we take the number of spikes $K= 20$, $\sigma^2 = 0.4$ and the eigengap to be $\Delta = 0.2$. We generate the data by $\{\bX_i\}_{i=1}^n \overset{\text{i.i.d.}}{\sim} {\cN}(\mathbf{0},\mathbf{\Sigma})$, where
 $$\mathbf{\Sigma} = \operatorname{diag} (2.4, 1.2, \underbrace{0.6, \ldots, 0.6}_{K-2}, 0.4 \ldots, 0.4).$$
 The dimension is $d = 2504$ and the sample size is $n=160,000$. Error rates and running times using different algorithms are compared  under different number of splits $m$ for the sample size $n$. For FADI, we take $L = 75$, $p=p' = 40$ and $q = 7$. 


 { 
 Table \ref{tab: err rate realistic} shows that the number of sample splits $m$ has little impact on the error rate of FADI as expected, while the error rate of \citet{fandistributed2019}'s distributed PCA increases as $m$ increases.  FADI is much faster than the other two methods
 in all the practical settings when the eigengap is small.
 This suggests that in practical problems where the sample size is large and the eigengap is small, FADI not only enjoys much higher computational efficiency compared to the existing methods, but also gives stable estimation for different sample splits along the sample size $n$.    Although the settings of small eigengap are of major interest in this section, we still conduct simulations where the eigengap increases gradually to see how it affects the performance of FADI. Table \ref{tab: err rate n runtime realistic incr det} shows that as the eigengap gets larger, the error rate of FADI gets closer to that of the traditional full sample PCA, whereas the error rate ratio of distributed PCA 
  to FADI gets below 1, but are still above 0.9 when the eigengap is larger than 1. As to the running time, FADI outperforms the other two methods in all the settings. In summary, when the eigengap grows larger, the performance of the three algorithms becomes similar to what we see in Section~\ref{sec: exm1 simu}.
 }
 
\begin{table}[htbp]
\centering
    \begin{tabular}{  c c c c|c}
      \hline
      \hline
 & FADI & Traditional PCA & Distributed PCA & $m$ \\
      \hline
Error Rate & 2.296 & 1.811 (0.79) & 2.629 (1.15)  & 10\\
& 2.294 & 1.811 (0.79) & 3.412 (1.49)  & 20\\
& 2.294 & 1.811 (0.79) & 3.955 (1.72)  & 40\\
& 2.294 & 1.811 (0.79) & 4.215 (1.84)  & 80\\
      \hline
Running Time & 5.76 & 983.86 (170.8) & 189.76 (32.9) & 10\\
& 3.82 & 992.09 (259.8) & 144.18 (37.8) & 20\\
& 2.86 & 972.47 (339.5) & 119.29 (41.6) & 40\\
& 2.37 & 968.43 (408.5) & 99.39 (41.9)  & 80\\
      \hline
      \hline
    \end{tabular}
    \caption{\small Comparison of the error rates {and running times (in seconds)} among FADI, full sample PCA and distributed PCA   \citep{fandistributed2019}, using different numbers of sample splits $m$ in the genetic setting. Values in the parentheses represent the error rate ratios or the computational time ratios of each method with respect to FADI.
    }\label{tab: err rate realistic}
  \end{table}


  \begin{table}[htbp]
  \centering
    \begin{tabular}{ c c c c |c}
      \hline
      \hline
       & FADI & Traditional PCA & Distributed PCA& Eigengap\\
      \hline
Error Rate&1.28 & 1.06 (0.82) & 1.57 (1.22) & 0.4\\
  &0.77 & 0.65 (0.85) & 0.71 (0.92) & 0.8\\
&0.48 & 0.42 (0.88) & 0.43 (0.90) & 1.6\\
&0.31 & 0.29 (0.92) & 0.29 (0.93) & 3.2\\
      \hline
Running Time&2.76 & 925.15 (334.7) & 115.29 (41.7)  & 0.4\\
 &2.77 & 916.52 (331.4)& 114.76 (41.5)  & 0.8\\
&2.69 & 922.85 (342.7) & 114.75 (42.6)  & 1.6\\
&2.77 & 919.20 (332.2) & 115.26 (41.7)  & 3.2\\
      \hline
      \hline
    \end{tabular}
    \caption{\small Comparison of the error rates and running times (in seconds) among FADI, full sample PCA and distributed PCA \citep{fandistributed2019} for different eigengaps $\Delta$ in the genetic setting. The number of sample splits $m$ is 40 for FADI and distributed PCA. 
    The settings of the other parameters are the same as those in Table \ref{tab: err rate realistic}. 
    }\label{tab: err rate n runtime realistic incr det}
  \end{table}
{
\subsection{Performance of Powered Sketching Method on Example~\ref{ex: spiked gaussian} for Computing PCs for Local Splits}\label{sec: add simu local power mtd}
As discussed in Remark~\ref{rmk: different bias reductions} in Section~\ref{sec: alg}, the power method cannot be directly applied to distributed data $\widehat\Mb = \sum_{s=1}^m \widehat\Mb^{(s)}$ in Step 1 due to data communication restrictions. In this section, we evaluate the performance of applying the power method independently to each data split with subsequent aggregation. Specifically, we generate $m$ i.i.d. Gaussian test matrices $\{ \bOmega^{(s)} \}_{s=1}^m \subset \RR^{d \times p}$ and compute $\widehat\Vb^{(s)}$ locally as the top $K$ left singular vectors of $\widehat\Yb^{(s)} = (\widehat\Mb^{(s)})^q \bOmega^{(s)}$. Then, we apply the same aggregation method in Step 3 to $\{\widehat\Vb^{(s)}\}_{s = 1}^{m}$ to obtain the final estimator
$\widetilde\Vb^{\F}$. 


Table~\ref{tab: compare FADI power trad} presents numerical comparisons between FADI, the power method, and traditional (centralized) PCA on Example~\ref{ex: spiked gaussian}. We generate the data from the same covariance structure as specified in Section~\ref{sec: exm1 simu}, and divide the sample equally into $m$ splits. The results show that the power method yields significantly higher error rates than FADI and traditional PCA when the sample size per split is smaller than or equal to the dimension $d$, while the running times of FADI and the power method are similar across settings. FADI's performance remains stable across different values of $m$, with consistent error rates as $m$ increases, indicating its suitability for distributed data settings where sample sizes per split are limited. 
\begin{table}[htbp]
    \centering
     \resizebox{1\textwidth}{!}{%
     
    \begin{tabular}{ccc| r c r| r c r}
    \hline
    \hline
       \multicolumn{3}{c|}{Parameters} & \multicolumn{3}{c|}{Error Rate} &   \multicolumn{3}{c}{Running time (seconds)}\\
         \hline
         $d$ & $n_s$ & $m$ & FADI & Power Method & Centralized  & FADI & Power Method & Centralized \\
         \hline
         500 & 1000 & 20 & 0.270 & 0.261 &0.249 & 0.076 & 0.089 & 5.284 \\
         500 & 500 & 40 & 0.270 & 0.285 &0.249 & 0.071 & 0.063 & 5.284 \\
         500 & 250 & 80 & 0.270 & 0.356 &0.249 & 0.065 & 0.059 & 5.284 \\
         1000 & 1000 & 20 &0.377 & 0.410 & 0.351 &0.163 & 0.179 & 22.520 \\
         1000 & 500 & 40 &0.377 & 0.553 & 0.351 &0.152 & 0.114 & 22.520 \\
         1000 & 250 & 80 &0.377 & 0.710 & 0.351 &0.147 & 0.103 & 22.520 \\
         2000 & 1000 & 20 &0.528 & 0.853 & 0.491 &0.527 & 0.407 & 104.462\\
         2000 & 500 & 40 &0.528 & 1.095 & 0.491 &0.492 & 0.231 & 104.462\\
         2000 & 250 & 80 &0.528 & 1.160 & 0.491 &0.476 & 0.198 & 104.462\\
         \hline
         \hline
    \end{tabular}}
    \caption{ Comparison of FADI with the power method and the centralized PCA without divide and conquer on Example~\ref{ex: spiked gaussian} under different number of sample splits $m$, where $n_s$ denotes the sample size per data split. We set the total sample size at $n = 20000$ and $Lp/d = 2$ for FADI across the settings. FADI remained close to the centralized PCA and outperformed the power method in terms of the error rate when the sample size per split is smaller than or equal to the number of variables, i.e., $n_s \le d$, while the running time for FADI is comparable to that of the power method across all settings, suggesting FADI as a more suitable method for distributed data setting to allow for insufficient sample size per split. }
    \label{tab: compare FADI power trad}
\end{table}
\subsection{Estimation of \(K\) Using Various Aggregation Methods}\label{sec: unknown K aggregate}
In Section~\ref{sec: k unknown}, we propose an aggregation method for estimating the number of spikes \(K\) when it is unknown, using the median of local estimators from parallel sketchings. Theorem~\ref{thm: est k} provides theoretical guarantees that this aggregation method recovers the true value with high probability. In practice, due to the randomness in the sketching approximation, aggregating across multiple sketchings can reduce the probability of failure in recovering the true \(K\). In this section, we evaluate different aggregation methods for estimating \(K\) and compare them with the non-aggregation approach that estimates \(K\) using a single sketching. For the non-aggregation approach,  we randomly select \(\ell \in [L]\) and estimate \(K\) by \(\hat{K}^{(\ell)}\) as defined in \eqref{eq: est K ell}. For the aggregation methods, in addition to the median approach proposed in Section~\ref{sec: k unknown}, we consider the majority voting method and the mean singular value approach. Specifically, the majority voting method selects the estimate \(\hat{K}^{\rm MV}\) that appears most frequently among \(\{\hat{K}^{(1)}, \hat{K}^{(2)}, \dots, \hat{K}^{(L)}\}\):
\[
\hat{K}^{\rm MV} = \min \left\{\text{Mode}\left(\hat{K}^{(1)}, \hat{K}^{(2)}, \dots, \hat{K}^{(L)}\right)\right\},
\]
where, in the case of multimodality in \(\{\hat{K}^{(1)}, \hat{K}^{(2)}, \dots, \hat{K}^{(L)}\}\), \(\hat{K}^{\rm MV}\) is set to be the smallest mode. The following corollary of Theorem~\ref{thm: est k} demonstrates that the majority voting method performs as well as the median aggregation method.
\begin{corollary}\label{col: est K majority vote}
     Define $\eta_0 = 480 c_e^{-1} \sqrt{d/(\Delta^2 p)} r_1(d)\log d$. Under the same conditions as Theorem~\ref{thm: est k}, 
    if we choose $\mu_0$ such that $\Delta \eta_0/24 \le \mu_0 \le \Delta \sqrt{\eta_0}/12$, then with probability at least $1-  O(d^{-(L \wedge 20)/2})$, $\hat{K}^{\rm MV} = K$.
\end{corollary}
The proof of Corollary~\ref{col: est K majority vote} follows from a straightforward modification of Theorem~\ref{thm: est k}, noting that $\PP(\hat{K}^{\rm MV} = K) \ge \PP(\sum_{\ell=1}^L \II\{\hat{K}^{(\ell)} = K\} > L / 2) $.

For the mean singular value approach, we average the singular values obtained from each local sketching: \(\bar{\sigma}_{k} = \frac{1}{L}\sum_{\ell=1}^L \sigma_k (\hat\Yb^{(\ell)})\), \(k = 1, \ldots, p\), and then estimate \(K\) as
$$
 \hat{K}^{\rm MS} = \min  \{k < p: \bar\sigma_{k+1} - \bar\sigma_{p}\le \sqrt{p}\mu_0 \}.
$$
We numerically evaluate the application of the above methods to FADI for estimating \(K\) in Example~\ref{ex: spiked gaussian}, where the data are generated following the same setup as in Section~\ref{sec: exm1 simu}. We compute the empirical recovery rates of the estimator over 300 Monte Carlo simulations with \(d = 1000\), \(n = 20,000\), \(K = 3\), \(m = 20\), \(q = 3\), and \(p' = p = 12\). The results are presented in Table~\ref{tab: K aggregate}. The aggregation methods consistently outperform the non-aggregation method that uses only one sketching. While the recovery rates are similar among the three aggregation approaches, the mean singular value approach requires communication of the local singular values for each sketch to the central processor, whereas the median aggregation method and the majority voting method only require the transfer of the local estimators \(\hat{K}^{(\ell)}\)'s. Therefore, we recommend the median and majority voting methods for aggregation in distributed settings.
\begin{table}[htbp]
    \centering
    { \begin{tabular}{c|c c c c}
    \hline
    \hline
    $Lp/d$ & Single Sketching & Median  & Majority Voting & Mean Singular Value\\
    \hline
    0.5 & 0.940 & 1.000 & 1.000 & 1.000\\
1.0 & 0.920 & 1.000 & 1.000 & 1.000\\
2.0 & 0.947 & 1.000 & 1.000 & 1.000\\
\hline
\hline
    \end{tabular}}
    \caption{ Empirical recovery rates of different aggregation methods in FADI for estimating the unknown number of spikes \(K\) with \(d = 1000\), \(n = 20000\), \(K = 3\), $m = 20$, $q = 3$ and \(p' = p = 12\). The aggregation methods uniformly outperform the non-aggregation method that uses only a single sketching.} 
    \label{tab: K aggregate}
\end{table}
\section{Impact of Over-Specification and Under-Specification of \(K'\) in Application to the 1000 Genomes Data}\label{sec: K prime 1000 g}
In Section~\ref{sec: real data}, we apply FADI to estimate the principal eigenspace of the 1000 Genomes Data, assuming the data follow the spiked covariance model described in Example~\ref{ex: spiked gaussian}. For the estimation of residual variance \(\sigma^2\), we set \(K' = \hat{K} + 1 = 27\), where \(\hat{K}\) is the estimator of \(K\) obtained from the FADI algorithm. In practice, due to data randomness, perturbations in \(\hat{K}\) may lead to over-specification or under-specification of \(K'\). In this section, we evaluate the robustness of the estimator \(\hat{\sigma}^2\) with respect to different choices of \(K'\) in the 1000 Genomes Data. We vary \(K'\) from 21 to 30 and randomly sample different sets of \(K'\) individuals from the 2504 subjects to estimate \(\sigma^2\). This sampling process is repeated 100 times, and we compute the sample mean and standard deviation of \(\hat{\sigma}^2\). We present the results in Table~\ref{tab: k prime 1000 g}.  The performance of FADI in estimating \(\sigma^2\) remains stable across different choices of \(K'\), indicating the robustness of FADI to slight over-specification or under-specification of \(K'\). This robustness may be due to the slow decay of the eigenvalues of the sample covariance matrix beyond the top 10, as reflected by Figure~\ref{fig: sample eigen 1000 g}.
\begin{table}[htbp]
    \centering
    \resizebox{1\textwidth}{!}{%
     
     \begin{tabular}{c|c c c c c c c c c c }
    \hline
    \hline
      $K'$ &  21 & 22	& 23 & 24 & 25 & 26 & 27 & 28 & 29 & 30 \\
      \hline
${\rm mean}(\hat\sigma^2)$ & 0.7881 & 0.7887 & 0.7831 & 0.7837 & 0.7815 & 0.7789 & 0.7811 & 0.7798 & 0.7810 & 0.7782 \\
${\rm SD}(\hat\sigma^2)$ & 0.0163 & 0.0178 & 0.0179 & 0.0158 & 0.0167 & 0.0189 & 0.0139 & 0.0176 & 0.0140 & 0.0169 \\
\hline
\hline
    \end{tabular}}
    \caption{ Sample mean and standard deviation of \(\hat{\sigma}^2\) using FADI under different choices of \(K'\) for the 1000 Genomes Data. The estimation of \(\sigma^2\) remains stable across across different settings of $K'$.}
    \label{tab: k prime 1000 g}
\end{table}
\begin{figure}
    \centering
    \includegraphics[width=0.7\linewidth]{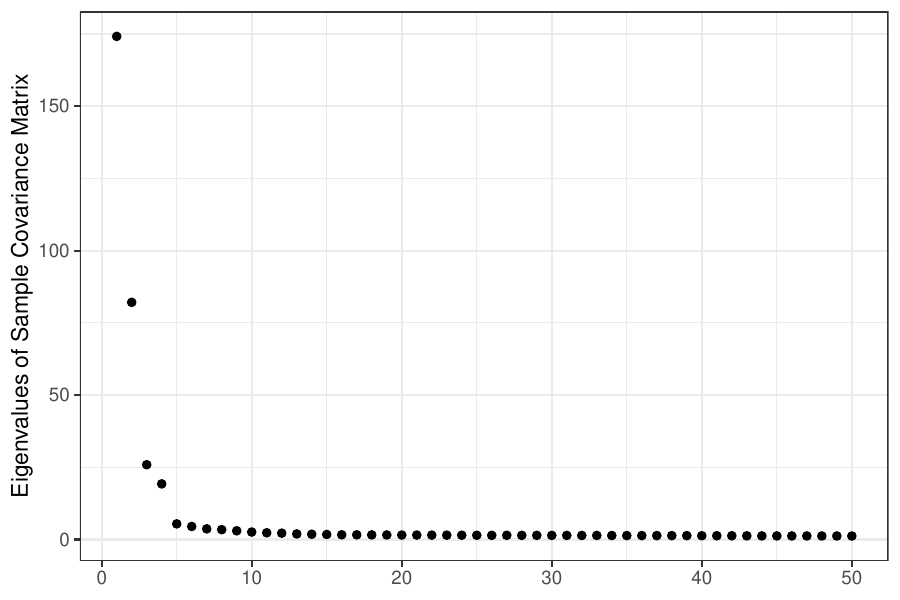}
    \caption{ Top 50 eigenvalues of the sample covariance matrix for the 1000 Genomes Data, showing a slow decay after the top 10.}
    \label{fig: sample eigen 1000 g}
\end{figure}
}
\section{Additional Application to the 1000 Genomes Data: Inference on Ancestry Membership Profiles}\label{sec: real data inference}
In this section, we consider the application of FADI to the 1000 Genomes Data for inferential analysis under the model specified by Example~\ref{ex: SBM}. We use the same pruned data described in Section~\ref{sec: real data}, and generate an undirected graph from the 1000 Genomes Data.

To increase the randomness for better fitting of the model setting in Example~\ref{ex: SBM}, we sample 1000 out of the total 168047 variants for generating the graph. More specifically, we treat each subject as a node, and for each given pair of subjects $(i,j)$, we define a genetic similarity score $s_{ij} = \sum_{k = 1}^{1000} \II\left\{x_{ik} = x_{jk}\right\}$, where $x_{ik}$ refers to the genotype of the $k$-th variant for subject $i$. We denote by $s^{0.95}$ the 0.95 quantile of $\{s_{ij}\}_{i<j}$. Subjects $i$ and $j$ are connected if and only if $s_{ij} > s^{0.95}$. Denote by $\Ab$ the adjacency matrix (allowing no self-loops). We include only four super populations: AFR, EAS, EUR and SAS, with 2058 subjects in total. We are interested in testing whether two given subjects $i$ and $j$ belong to the same super population, i.e., ${\rm H}_0: \Vb_i = \Vb_j$ vs. ${\rm H}_1 : \Vb_i \ne \Vb_j$.  We divide the adjacency matrix equally into $m=10$ splits, and perform FADI with $p = 50$, $p' = 50$, $q=3$ and $L = 1000$. The rank estimator from FADI is $\hat{K} = 4$ by setting $\mu_0 = ({\hat{\theta}}/{p})^{1/2} d \log d/12$, where $\hat\theta$ is the average degree estimator defined in Section~\ref{sec: step 0}. We can see the estimated rank is consistent with the number of super populations. We apply K-means clustering to the FADI estimator $\widetilde\Vb_{\hat{K}}^{\F}$, and calculate the misclustering rate by treating the self-reported ancestry group as the ground truth. The misclustering rate of FADI is 0.135, with computation time of  3.7 seconds.
 In comparison, 
 the misclustering rate for the traditional PCA method is 0.134 with computation time of 26.5 seconds, and the correlation between the top four PCs for the traditional PCA and FADI are 0.997, 0.994, 0.994 and 0.996 respectively. \\
 
To conduct pairwise inference on the ancestry membership profiles, we preselect 16 subjects, with 4 subjects from each super population. We apply Bonferroni correction to correct for the multiple comparison issue and set the level at $0.05\times \binom{16}{2}^{-1} = 4.17\times 10^{-4}$.  We estimate the asymptotic covariance matrix by Corollary~\ref{col: sbm} and correct $\widetilde\Mb$ by setting entries larger than 1 to 1 and entries smaller than 0 to 0. The pairwise p-values are summarized in Figure~\ref{fig: 1000g p values}. The computational time for computing the covariance matrix is 0.31 seconds. We can see that most of the comparison results are consistent with the true ancestry groups, while the inconsistency could be due to the mixed memberships of certain subjects and the unaccounted sub-population structures.
\begin{figure}[htbp]
		\centering
		\includegraphics[width=1\textwidth]{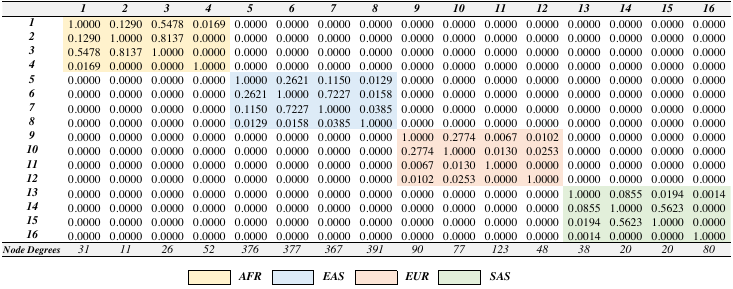} 
		\caption{\small p-values for pairwise comparison among 16 preselected subjects. For subjects pair $(i,j)$, p-value is defined as  $\PP\big(\chi_{\hat{K}}^2 > \|\tilde{\bd}\|_2^2\big)$, where $\chi_{\hat{K}}^2$ is Chi-squared distribution with degrees of freedom equal to $\hat{K}$, and $\tilde\bd = \hbSigma_{i,j}^{-1/2}\widetilde\Vb^{\rm F}_{\hat{K}}(\eb_i - \eb_j)$ with $\hbSigma_{i,j}$ being the asymptotic covariance matrix defined in Section~\ref{sec: exm2 simu}.  }\label{fig: 1000g p values}
	\end{figure}

\section{Proof of Main Theoretical Results}\label{sec: proof main theory}
In this section we provide proofs of the theoretical results in Section~\ref{sec: theory}. For the inferential results, we will present proofs of the theorems under the regime $Lp \ll d$ first, which takes into consideration the extra variability caused by the fast sketching, and then give proofs of the theorems under the regime $Lp \gg d$ where the fast sketching randomness is negligible. 
\subsection{Unbiasedness of Fast Sketching With Respect to \texorpdfstring{$\widehat\Mb$}{Mbh} }\label{sec: proof lemma bias control}
 We show by the following Lemma \ref{lm: bias control} that the fast sketching is unbiased with respect to $\widehat\Mb$ under proper conditions. 
\begin{lemma}\label{lm: bias control}{
Let $\hat\Vb_d \hat{\mathbf{\Lambda}}_d \hat{\Vb}_d^{\top}$ be the eigen-decomposition of $\widehat{\Mb}$, and let $\widehat\Vb = (\hat\vb_1, \ldots, \hat\vb_K)$ be the stacked $K$ leading eigenvectors of $\widehat\Mb$ corresponding to the eigenvalues with largest magnitudes.  When $\|{\mathbf{\Sigma}}^\prime - \hat\Vb \hat\Vb^{\top} \|_2 < 1/2$, we have that $\operatorname{Col}(\Vb^\prime) = \operatorname{Col}(\hat\Vb)$, where $\operatorname{Col}(\cdot)$ denotes the column space of the matrix.}
\end{lemma}
\begin{proof}
We will first show that $\hat\Vb_d^{\top}{\mathbf{\Sigma}}^\prime \hat\Vb_d$ is diagonal. For any $j \in [d]$, we let $\Db_j = \Ib_d - 2 \eb_j \eb_j^{\top}$,
and recall we denote the eigen-decomposition of $\widehat{\Mb}$ by $\widehat{\Mb} = \hat\Vb_d \hat{\mathbf{\Lambda}}_d \hat\Vb_d^{\top}$. Then conditional on $\widehat{\Mb}$ we have
\begin{align*}
    &\hat\Vb_d \Db_j \hat\Vb_d^{\top} \widehat\Yb^{(\ell)} \widehat\Yb^{(\ell) \top} \hat\Vb_d \Db_j \hat\Vb_d^{\top} =  \hat\Vb_d \Db_j \hat\Vb_d^{\top} \hat\Vb_d \hat{\mathbf{\Lambda}}_d \hat\Vb_d^{\top} \mathbf{\Omega}^{(\ell)} \mathbf{\Omega}^{(\ell)\top} \hat\Vb_d \hat{\mathbf{\Lambda}}_d \hat\Vb_d^{\top} \hat\Vb_d \Db_j \hat\Vb_d^{\top}\\
     &\quad =  \hat\Vb_d  \hat{\mathbf{\Lambda}}_d (\Db_j \hat\Vb_d^{\top} \mathbf{\Omega}^{(\ell)})( \mathbf{\Omega}^{(\ell)\top} \hat\Vb_d \Db_j) \hat{\mathbf{\Lambda}}_d   \hat\Vb_d^{\top} \overset{\rm d}{=} \hat\Vb_d  \hat{\mathbf{\Lambda}}_d \hat\Vb_d^{\top} \mathbf{\Omega}^{(\ell)}\mathbf{\Omega}^{(\ell)\top} \hat\Vb_d  \hat{\mathbf{\Lambda}}_d   \hat\Vb_d^{\top} = \widehat\Yb^{(\ell)} \widehat\Yb^{(\ell) \top},
\end{align*}
where the second equality is due to the fact that diagonal matrices are commutative, and the last but one equivalence in distribution is due to the fact that $ \Db_j \hat\Vb_d^{\top} \mathbf{\Omega}^{(\ell)} \overset{\rm d}{=} \hat\Vb_d^{\top} \mathbf{\Omega}^{(\ell)}$. Also we know the top $K$ eigenvectors of $\hat\Vb_d \Db_j \hat\Vb_d^{\top} \widehat\Yb^{(\ell)} \widehat\Yb^{(\ell) \top} \hat\Vb_d \Db_j \hat\Vb_d^{\top}$ are $\hat\Vb_d \Db_j \hat\Vb_d^{\top} \hat{\Vb}^{(\ell)}$, and thus $\hat\Vb_d \Db_j \hat\Vb_d^{\top} \hat{\Vb}^{(\ell)} \overset{\rm d}{=} \hat{\Vb}^{(\ell)}$. Hence we have
\begin{align*}
    &\hat\Vb_d^{\top} \mathbb{E}\left(\widehat{\Vb}^{(\ell)} \widehat{\Vb}^{(\ell) \top} | \widehat{\Mb}\right) \hat\Vb_d =\hat\Vb_d^{\top} \hat\Vb_d \mathbf{D}_{j} \hat\Vb_d^{\top} \mathbb{E}\left(\widehat{\Vb}^{(\ell)} \widehat{\Vb}^{(\ell)\top} | \widehat{\Mb} \right) \hat\Vb_d \mathbf{D}_{j} \hat\Vb_d^{\top} \hat\Vb_d\\
    &\quad =\mathbf{D}_{j} \hat\Vb_d^{\top} \mathbb{E}\left(\widehat{\Vb}^{(\ell)} \widehat{\Vb}^{(\ell) \top} | \widehat{\Mb} \right) \hat\Vb_d \mathbf{D}_{j} = \mathbf{D}_{j} \hat\Vb_d^{\top} \bSigma' \hat\Vb_d \mathbf{D}_{j}.
\end{align*}
                
The above equation holds for any $j \in[d]$, which suggests that $\hat\Vb_d^{\top} \mathbb{E}\left(\widehat{\mathbf{V}}^{(\ell)} \widehat{\mathbf{V}}^{(\ell)\top} | \widehat{\Mb}\right) \hat\Vb_d$
is diagonal and that $\mathbf{\Sigma}^{\prime}$ and $\widehat{\Mb}$ share the same set of eigenvectors.

Now under the condition that $\left\|\mathbf{\Sigma}^{\prime}-\hat\Vb \hat\Vb^{\top}\right\|_2 <1 / 2$, for any $j \in[K],$ we denote by $\hat\vb_j$ the $j$-th column
of $\hat\Vb$, and we have
$$
\left\|\mathbf{\Sigma}^{\prime} \hat\vb_j\right\|_2 =\left\|\left(\mathbf{\Sigma}^{\prime}-\hat\Vb \hat\Vb^{\top}+\hat\Vb \hat\Vb^{\top} \right) \hat\vb_j\right\|_2 \geq 1-\left\|\mathbf{\Sigma}^{\prime}-\hat\Vb \hat\Vb^{\top}\right\|_2 >1-\frac{1}{2}=\frac{1}{2}.
$$
In other words, the corresponding eigenvalue of $\hat\vb_j$ in $ \mathbf{\Sigma}^{\prime}$ is larger than $1 / 2$. On the other hand, by Weyl's inequality \citep{franklin2012matrix}, the rest of the $d-K$ eigenvalues of ${\mathbf{\Sigma}}^\prime$ should be less than 1/2. Therefore, $\hat\Vb$ are still the leading $K$ eigenvectors for ${\mathbf{\Sigma}}^\prime$, and thus $\operatorname{Col}(\Vb^\prime) = \operatorname{Col}(\hat\Vb)$.
\end{proof}
Recall in Section~\ref{sec: theory} we discuss that the bias term has the following decomposition $\cD(\Vb^\prime, {\Vb}) \le \cD(\hat\Vb, {\Vb}) + \cD(\Vb^\prime, \widehat{\Vb})$. Lemma \ref{lm: bias control} shows that as long as $\mathbf{\Sigma}^{\prime}$ and $\Vb \Vb^{\top}$ are not too far apart, $\Vb'$ and $\hat{\Vb}$ will share the same column space. In fact, Lemma \ref{lm: control of prob} in Section~\ref{sec: proof thm error bd} will show that the probability that ${\mathbf{\Sigma}}^\prime$ and $\hat\Vb \hat\Vb^{\top}$ are not sufficiently close converges to 0, and $\cD(\Vb^\prime, {\Vb}) = \cD(\hat\Vb, {\Vb})$ with high probability. With the help of Lemma~\ref{lm: bias control}, we present the proof of the main error bound results in the following section.
\subsection{Proof of Theorem~\ref{thm: error bound}}\label{sec: proof thm error bd}
Recall the problem setting in Section~\ref{sec: problem setting}. It is not hard to see that we can write $\bLambda = \Pb_0 \bLambda^0$, where $\mathbf{\Lambda}^0 = \operatorname{diag}(|\lambda_1|,\ldots, |\lambda_K|)$ and $\Pb_0 = \diag \big([\operatorname{sgn}(\lambda_k)]_{k=1}^K\big)$. Then $\Mb = (\Vb \Pb_0) \bLambda^0 \Vb^{\top}$ is the SVD of $\Mb$. 

We begin with bounding $\left( \EE \|\widetilde\Vb \widetilde\Vb^{\top} - \Vb\Vb^{\top} \|^2_{\rm F} \right)^{1/2}$. Before delving into the detailed proof, the following two lemmas provide some important properties of the random Gaussian matrix.
\begin{lemma}\label{lm: gaussian norm}
Let $\mathbf{\Omega} \in \RR^{d \times p}$ be a random matrix with i.i.d. standard Gaussian entries, where $p \le d$. For a random variable, recall that we define the $\psi_1 $ norm to be $\|\cdot\|_{\psi_1} = \sup_{p \ge 1} (\EE |\cdot|^p)^{1/p}/p$. Then we have the following bound on the $\psi_1$ norm of the matrix $\mathbf{\Omega} / \sqrt{p}$:
\begin{equation}
    \| \|\mathbf{\Omega} /\sqrt{p}\|_2  \|_{\psi_1} \lesssim \sqrt{d/p}.
\end{equation}
\end{lemma}
\begin{lemma}\label{lm: bound min eigenvalue}
Let $\boldsymbol{\Omega} \in \RR^{K \times p}$ denote a random matrix with i.i.d. Gaussian entries, where $p \ge 2K$. For any integer $a$ such that $1 \le a \le (p-K+1)/2$, there exists a constant $C>0$ such that
\begin{equation}
    \EE\left ( \left(\sigma_{\min} (\boldsymbol{\Omega} /\sqrt{p})\right)^{-a} \right) \le C^{a}.
\end{equation}
\end{lemma}
The following lemma shows that $\|{\mathbf{\Sigma}}^\prime - \Vb \Vb^{\top} \|_2$ and $\|{\mathbf{\Sigma}}^\prime - \hat\Vb \hat\Vb^{\top} \|_2$ are bounded by a small constant with high probability.
\begin{lemma}\label{lm: control of prob}
If Assumption \ref{asp: tail prob bound} holds and $p \ge \max(2K,K+3)$, there exists a constant $c_0 >0$ such that for any $\varepsilon >0$, we have 
\[\max\left\{\PP \Big( \|{\mathbf{\Sigma}}^\prime - \Vb \Vb^{\top} \|_2 \ge \varepsilon \Big), \PP \left( \|{\mathbf{\Sigma}}^\prime - \hat\Vb \hat\Vb^{\top} \|_2 \ge \varepsilon \right)\right\}  \lesssim \exp\left( -c_0 \sqrt{\frac{p}{d}} \frac{\Delta \varepsilon }{r_1(d)}\right).\]
\end{lemma}
The proof of Lemma \ref{lm: gaussian norm}, Lemma \ref{lm: bound min eigenvalue} and Lemma \ref{lm: control of prob} are deferred to Supplementary Materials~\ref{sec: proof tec lems}. Now we can start with the proof. We first decompose the bias term into two parts,
\begin{equation}\label{eq: decomp error}
    \left( \EE |\cD(\widetilde\Vb, \Vb)|^2 \right)^{1/2} \le \underbrace{\Big( \EE |\cD(\widetilde\Vb, \Vb')|^2 \Big)^{1/2}}_{\text{I}} \!\!\!+ \underbrace{\Big( \EE |\cD(\Vb', \Vb)|^2 \Big)^{1/2}}_{\text{II}}.
\end{equation}
Term I can be regarded as the variance term, whereas term II is the bias term. We will consider the bias term first.
\subsubsection{Control of the Bias Term}
We can see that term II can be further decomposed into two terms
\begin{equation}\label{eq: term II decomp}
    \left( \EE |\cD(\Vb', \Vb)|^2 \right)^{1/2} \le \left( \EE \|\Vb^\prime \Vb^{\prime \top} - \hat\Vb \hat\Vb^{\top} \|_{\rm F}^2 \right)^{1/2} + \left( \EE \|\hat\Vb \hat\Vb^{\top} - \Vb \Vb^{\top} \|_{\rm F}^2 \right)^{1/2}.
\end{equation}
We can bound both terms separately. First note that $\|\Vb^\prime \Vb^{\prime \top} - \hat\Vb \hat\Vb^{\top} \|_{\rm F} \le \sqrt{2K}\|\Vb^\prime \Vb^{\prime \top} - \hat\Vb \hat\Vb^{\top} \|_2 \le \sqrt{2K}$. Thus we have,
\begin{align*}
     &\left( \EE \|\Vb^\prime \Vb^{\prime \top} - \hat\Vb \hat\Vb^{\top} \|_{\rm F}^2 \right)^{1/2}  \le  \left( \EE \|\Vb^\prime \Vb^{\prime \top} - \hat\Vb \hat\Vb^{\top} \|_{\rm F}^2 \II\big\{ \|{\mathbf{\Sigma}}^\prime - \hat\Vb \hat\Vb^{\top} \|_2 \ge 1/2\big\}\right)^{1/2} \\
     & \quad \quad+  \left( \EE \|\Vb^\prime \Vb^{\prime \top} - \hat\Vb \hat\Vb^{\top} \|_{\rm F}^2 \II{\big\{ \|{\mathbf{\Sigma}}^\prime - \hat\Vb \hat\Vb^{\top} \|_2 < 1/2\big\}}\right)^{1/2} \\
     &\lesssim 0 + \sqrt{K} \left(\PP \left( \|{\mathbf{\Sigma}}^\prime - \hat\Vb \hat\Vb^{\top} \|_2 \ge 1/2\right) \right)^{1/2} \lesssim \sqrt{K} \exp\left( -\frac{c_0}{4} \sqrt{\frac{p}{d}} \frac{\Delta  }{r_1(d)}\right),
\end{align*}
where the last but one inequality follows from Lemma \ref{lm: bias control}, and the last inequality is a result of Lemma \ref{lm: control of prob}. As for the second term on the RHS of \eqref{eq: term II decomp}, by Davis-Kahan’s Theorem \citep{yu2015davis}, we have
\begin{align*}
    \left( \EE \|\hat\Vb \hat\Vb^{\top} - \Vb \Vb^{\top} \|_{\rm F}^2 \right)^{1/2} & \lesssim \frac{\sqrt{K}}{\Delta}\left( \EE \|\hat{\Mb} - \Mb\|_2^2 \right)^{1/2} = \frac{\sqrt{K}}{\Delta}\left( \EE \|\Eb\|_2^2 \right)^{1/2} \\
    & \le \frac{\sqrt{K}}{\Delta} \| \|\Eb\|_2 \|_{\psi_1} \lesssim \frac{\sqrt{K}}{\Delta} r_1(d).
\end{align*}
Therefore, the bound for the bias term is
$$ \text{II} \lesssim \sqrt{K} \exp\left( -\frac{c_0}{4} \sqrt{\frac{p}{d}} \frac{\Delta  }{r_1(d)}\right) + \frac{\sqrt{K}}{\Delta} r_1(d).$$
\subsubsection{Control of the Variance Term}
Now we move on to control the variance term. Suppose that $\left\|\mathbf{\Sigma}^{\prime}-\Vb \Vb^{\top}\right\|_2 <1 / 4$. Then by Weyl's inequality \citep{franklin2012matrix} we have that $\sigma_{K} (\mathbf{\Sigma}^{\prime}) > 1- 1/4 = 3/4$ and $\sigma_{K+1} (\mathbf{\Sigma}^{\prime}) < 1/4$. Thus by Davis-Kahan theorem \citep{yu2015davis}
\begin{align*}
    &\left( \EE \left(\|\widetilde\Vb \widetilde\Vb^{\top} \!\!- \!\!\Vb^\prime \Vb^{\prime \top} \|_{\rm F}^2 \II\left\{\left\|\mathbf{\Sigma}^{\prime}-\Vb \Vb^{\top}\right\|_2<1 / 4\right\}\right)\right)^{1/2}\!\!\!\\
    &\lesssim \left( \EE \left(\frac{\|\widetilde{\mathbf{\Sigma}} - {\mathbf{\Sigma}}^\prime \|_{\rm F}^2}{\left(\sigma_K({\mathbf{\Sigma}}^\prime ) - \sigma_{K+1}({\mathbf{\Sigma}}^\prime )\right)^2} \II\left\{\left\|\mathbf{\Sigma}^{\prime}-\Vb \Vb^{\top}\right\|_2 <1 / 4\right\}\right)\right)^{1/2} \\
    &\quad \lesssim \left( \EE \left(\|\widetilde{\mathbf{\Sigma}} - {\mathbf{\Sigma}}^\prime \|_{\rm F}^2\II{\left\{\left\|\mathbf{\Sigma}^{\prime}-\Vb \Vb^{\top}\right\|_2 <1 / 4\right\}} \right) \right)^{1/2} \le \underbrace{\left( \EE \|\widetilde{\mathbf{\Sigma}} - {\mathbf{\Sigma}}^\prime \|_{\rm F}^2 \right)^{1/2}}_{\text{III}}.
\end{align*}
We will bound term III later. Also similar as previously, note that $\|\widetilde\Vb \widetilde\Vb^{\top} \!\!- \!\!\Vb^\prime \Vb^{\prime \top} \|_{\rm F} \le \sqrt{2K}$. Thus by Lemma \ref{lm: control of prob},
\begin{align*}
    &\left( \EE \left(\|\widetilde\Vb \widetilde\Vb^{\top} \!\!- \!\!\Vb^\prime \Vb^{\prime \top} \|_{\rm F}^2 \II{\left\{\left\|\mathbf{\Sigma}^{\prime}-\Vb \Vb^{\top}\right\|_2 \ge \frac{1}{4}\right\}}\right)\right)^{{1}/{2}}\!\!\! \lesssim \sqrt{K} \left(\PP \left(\left\|\mathbf{\Sigma}^{\prime}-\Vb \Vb^{\top}\right\|_2 \ge \frac{1}{4}\right) \right)^{{1}/{2}}\\
    & \quad \lesssim \sqrt{K} \exp\left( -\frac{c_0}{8} \sqrt{\frac{p}{d}} \frac{\Delta  }{r_1(d)}\right).
\end{align*}
Therefore, we have 
$$
\left( \EE \|\widetilde\Vb \widetilde\Vb^{\top} \!\!- \!\!\Vb^\prime \Vb^{\prime \top} \|_{\rm F}^2\right)^{1/2} \lesssim \sqrt{K} \exp\left( -\frac{c_0}{8} \sqrt{\frac{p}{d}} \frac{\Delta  }{r_1(d)}\right) + \underbrace{\left( \EE \|\widetilde{\mathbf{\Sigma}} - {\mathbf{\Sigma}}^\prime \|_{\rm F}^2 \right)^{1/2}}_{\text{III}}.
$$
Now we move on to bound term III. 
\begin{align*}
    \left( \EE \|\widetilde{\mathbf{\Sigma}} - {\mathbf{\Sigma}}^\prime \|_{\rm F}^2 \right)^{1/2} &= \left( \EE \left\|  \frac{1}{L}\sum_{\ell = 1}^{L} \widehat{\Vb}^{(\ell)} \widehat{\Vb}^{(\ell) \top} - \EE \left( \widehat{\Vb}^{(1)} \widehat{\Vb}^{(1)\top} | \hat{\Mb}\right) \right\|_{\rm F}^2 \right)^{1/2}\\
    &=\bigg( \EE \bigg( \EE \bigg( \bigg\|  \frac{1}{L}\sum_{\ell = 1}^{L} \widehat{\Vb}^{(\ell)} \widehat{\Vb}^{(\ell) \top} - \EE \left( \widehat{\Vb}^{(1)} \widehat{\Vb}^{(1)\top} | \hat{\Mb}\right) \bigg\|_{\rm F}^2 \bigg| \hat{\Mb} \bigg)\bigg) \bigg)^{1/2}\\
    & = \frac{1}{\sqrt{L}} \left( \EE  \left\|   \widehat{\Vb}^{(\ell)} \widehat{\Vb}^{(\ell) \top} - \EE \left( \widehat{\Vb}^{(1)} \widehat{\Vb}^{(1)\top} | \hat{\Mb}\right) \right\|_{\rm F}^2  \right)^{1/2}\\
    & \le \frac{1}{\sqrt{L}} \left( \EE  \left\|   \widehat{\Vb}^{(\ell)} \widehat{\Vb}^{(\ell) \top} -\Vb \Vb^{\top}\right\|_{\rm F}^2  \right)^{1/2} + \frac{1}{\sqrt{L}} \left( \EE  \left\|  \Vb \Vb^{\top}- {\mathbf{\Sigma}}^\prime \right\|_{\rm F}^2  \right)^{1/2}.
\end{align*}
where the last but one equality is due to the independence of estimators from different sketches conditional on $\hat{\Mb}$. By Jensen's inequality \citep{jensen1906fonctions}, we have 
$$
\frac{1}{\sqrt{L}} \left( \EE  \left\|  \Vb \Vb^{\top}- {\mathbf{\Sigma}}^\prime \right\|_{\rm F}^2  \right)^{1/2} \le \frac{1}{\sqrt{L}} \left( \EE  \left\|   \widehat{\Vb}^{(\ell)} \widehat{\Vb}^{(\ell) \top} -\Vb \Vb^{\top}\right\|_{\rm F}^2  \right)^{1/2}.
$$
Thus we have 
\begin{equation}\label{eq: bound sig tilde}
\left( \EE \|\widetilde{\mathbf{\Sigma}} - {\mathbf{\Sigma}}^\prime \|_{\rm F}^2 \right)^{1/2}  \lesssim \frac{1}{\sqrt{L}} \left( \EE  \left\|   \widehat{\Vb}^{(\ell)} \widehat{\Vb}^{(\ell) \top} -\Vb \Vb^{\top}\right\|_{\rm F}^2  \right)^{1/2},
\end{equation}
Before bounding the RHS, let's consider the matrix ${\Yb}^{(\ell)} := \Vb \Pb_0 \mathbf{\Lambda}^0 \Vb^{\top} \mathbf{\Omega}^{(\ell)}$. If $\tilde{\mathbf{\Omega}}^{(\ell)} := \Vb^{\top} \mathbf{\Omega}^{(\ell)} \in \RR^{K \times p}$ does not have full row rank, then the entries will be restricted to a linear space with dimension less than $K\times p$. Since $\tilde{\mathbf{\Omega}}^{(\ell)}$ is a $K \times p$ standard Gaussian matrix, the probability that $\tilde{\mathbf{\Omega}}^{(\ell)}$ has full row rank is 1.  And thus with probability 1, the matrix ${\Yb}^{(\ell)}$ is of rank $K$, and $\Vb$ and the top $K$ left singular vectors of ${\Yb}^{(\ell)}/\sqrt{p} $ span the same column space. In other words, if we let $\boldsymbol{\Gamma}_K^{(\ell)}$ be the left singular vectors of ${\Yb}^{(\ell)} /\sqrt{p} $, then $\boldsymbol{\Gamma}_K^{(\ell)} \boldsymbol{\Gamma}_K^{(\ell)\top} = \Vb \Vb^{\top} $.

Now consider the $K$-th singular value of ${\Yb}^{(\ell)} /\sqrt{p} $, we let $ \Ub_{\tilde{\mathbf{\Omega}}} \Db_{\tilde{\mathbf{\Omega}}} \Vb_{\tilde{\mathbf{\Omega}}}^{\top}$ be the SVD of $\tilde{\mathbf{\Omega}}^{(\ell)} /\sqrt{p}$, and we have 
\begin{align*}
    \sigma_K\left({\Yb}^{(\ell)} /\sqrt{p}\right) &= \sigma_K  \left(\Vb \Pb_0 \mathbf{\Lambda}^0  \tilde{\mathbf{\Omega}}^{(\ell)} /\sqrt{p} \right) = \sigma_K  \left( \mathbf{\Lambda}^0 \Ub_{\tilde{\mathbf{\Omega}}} \Db_{\tilde{\mathbf{\Omega}}} \right)  \\
    & = \min_{\|\bx\|_2=1} \|\mathbf{\Lambda}^0 \Ub_{\tilde{\mathbf{\Omega}}} \Db_{\tilde{\mathbf{\Omega}}} \bx\|_2 \overset{(i)}{\ge} \sigma_{\min}\left( \tilde{\mathbf{\Omega}}^{(\ell)} /\sqrt{p}\right) \min_{\|\bv_1\|_2=1} \left\|\mathbf{\Lambda}^0 \Ub_{\tilde{\mathbf{\Omega}}} \bv_1\right\|_2\\
    &\overset{(ii)}{\ge}\sigma_{\min}\left( \tilde{\mathbf{\Omega}}^{(\ell)} /\sqrt{p}\right) \min_{\|\bv_2\|_2=1} \left\|\mathbf{\Lambda}^0 \bv_2\right\|_2 \ge \Delta \sigma_{\min}\left( \tilde{\mathbf{\Omega}}^{(\ell)} /\sqrt{p}\right),
\end{align*}
where $\bv_1 =  \Db_{\tilde{\mathbf{\Omega}}} \bx / \|\Db_{\tilde{\mathbf{\Omega}}} \bx\|_2$, and $\bv_2 = \Ub_{\tilde{\mathbf{\Omega}}} \bv_1$. Inequality (i) follows because $$\|\Db_{\tilde{\mathbf{\Omega}}} \bx\|_2 \ge \sigma_{\min}\left( \tilde{\mathbf{\Omega}}^{(\ell)} /\sqrt{p}\right) \|\bx\|_2 = \sigma_{\min}\left( \tilde{\mathbf{\Omega}}^{(\ell)} /\sqrt{p}\right),$$ and inequality (ii) is because $\|\bv_2\|_2 = \|\bv_1\|_2=1$.

Now by Wedin's Theorem \citep{wedin1972wedin} we have the following bound on the RHS of \eqref{eq: bound sig tilde},
\begin{align*}
    & \frac{1}{\sqrt{L}}\left( \EE  \big|\cD(\widehat\Vb^{(\ell)}, \Vb)\big|^2  \right)^{1/2} \!\!\! \lesssim \frac{\sqrt{K}}{\sqrt{L}} \left( \EE  \left\|   \widehat{\Yb}^{(\ell)}/\sqrt{p}  -{\Yb}^{(\ell)}/\sqrt{p}  \right\|_2^2 / \left(\Delta \sigma_{\min}\big( \tilde{\mathbf{\Omega}}^{(\ell)} /\sqrt{p}\big)\right)^2 \right)^{1/2}  \\
    & \quad \le \frac{\sqrt{K}}{\Delta \sqrt{L}} \left( \EE  \left\|   \widehat{\Yb}^{(\ell)}/\sqrt{p}  - {\Yb}^{(\ell)}/\sqrt{p}  \right\|_2^4 \right)^{1/4} \left(\EE \left( \sigma_{\min} \big( \tilde{\mathbf{\Omega}}^{(\ell)}/\sqrt{p}\big)\right)^{-4} \right)^{1/4}\\
    &\quad \lesssim \frac{\sqrt{K}}{\Delta \sqrt{L}} \|\|\Eb\|_2\|_{\psi_1} \cdot \|\|\mathbf{\Omega}^{(\ell)} / \sqrt{p}\|_2 \|_{\psi_1} \lesssim  \sqrt{\frac{Kd }{\Delta^2 pL}} \|\|\Eb\|_2\|_{\psi_1} \lesssim \sqrt{\frac{Kd }{\Delta^2 pL}} r_1(d),
\end{align*}
where the last but one inequality is due to Lemma \ref{lm: bound min eigenvalue}. Therefore, we have the final error rate for the estimator $\widetilde{\Vb}$:
\[
\left( \EE \|\widetilde\Vb \widetilde\Vb^{\top} - \Vb \Vb \|_{\rm F}^2 \right)^{1/2} \lesssim \underbrace{\sqrt{K} \exp\left( -\frac{c_0}{8} \sqrt{\frac{p}{d}} \frac{\Delta  }{r_1(d)}\right) + \frac{\sqrt{K}}{\Delta} r_1(d)}_{\text{bias}} + \underbrace{\sqrt{\frac{Kd }{\Delta^2 pL}} r_1(d)}_{\text{variance}}.
\]

Now consider the function $g(x) := \exp(a_0 \sqrt{\frac{p}{d}} x) / (\sqrt{d}x^2)$, where $a_0 >0$ is a fixed constant. We have 
$$
   \frac{ d \log g(x)}{ dx} = a_0 \sqrt{\frac{p}{d}} - \frac{2}{x}> 0, \quad \text{for } x \ge \frac{2}{a_0}\sqrt{\frac{d}{p}}. 
$$
Thus $g(x)$ is increasing on $x \ge 2\sqrt{d/p}/a_0$, and if we take $x \ge C \sqrt{\frac{d}{p}}\log d$ for some large enough constant $C > 0$, we have that $g(x) \ge 1$. Then by plugging in $x = \Delta / r_1(d)$ and taking $a_0 = c_0 / 8$,  under the condition that $ (\log d)^{-1}\sqrt{p/d} \Delta/r_1(d) \ge C$ for some large enough constant $C > 0$, we have that
$$
\exp\left( -\frac{c_0}{8} \sqrt{\frac{p}{d}} \frac{\Delta  }{r_1(d)}\right) \lesssim \frac{1}{\sqrt{d}} \left(\frac{r_1(d)}{\Delta}\right)^2 = o\left(\frac{r_1(d)}{\Delta}\right),
$$
and the error rate simplifies to 
\[
\left( \EE \|\widetilde\Vb \widetilde\Vb^{\top} - \Vb \Vb \|_{\rm F}^2 \right)^{1/2} \lesssim \underbrace{ \frac{\sqrt{K}}{\Delta} r_1(d)}_{\text{bias}} + \underbrace{\sqrt{\frac{Kd }{\Delta^2 pL}} r_1(d)}_{\text{variance}}.
\]
Now we move on to bound $\left( \EE \|\widetilde{\Vb}^{\rm F}\widetilde{\Vb}^{\text{F}\top} - \Vb\Vb^{\top} \|^2_{\rm F} \right)^{1/2}$. Since$\|\cdot\|_2^{2q}$ is convex, by Jensen's inequality \citep{jensen1906fonctions}, under the condition that $p \ge \max(2K, 8q + K - 1)$ we have that there exists some constant $\eta$ such that
\begin{align*}
    \EE \| \widetilde{\mathbf{\Sigma}} - \Vb \Vb^{\top} \|_2^{2q} 
    & \le \frac{1}{L}\sum_{\ell = 1}^L \EE \| \hat\Vb^{(\ell)}\hat\Vb^{(\ell)\top} - \Vb \Vb^{\top} \|_2^{2q} = \EE \| \hat\Vb^{(1)}\hat\Vb^{(1)\top} - \Vb \Vb^{\top} \|_2^{2q}\\
    & \le \EE \bigg( \left\|   \widehat{\Yb}^{(\ell)}/\sqrt{p}  -{\Yb}^{(\ell)}/\sqrt{p}  \right\|_2^{2q} \Big/ \left(\Delta \sigma_{\min}\left( \tilde{\mathbf{\Omega}}^{(\ell)} /\sqrt{p}\right)\right)^{2q}\bigg)\\
    & \le \frac{1}{\Delta^{2q}} \left(\EE  \left\|   \widehat{\Yb}^{(\ell)}/\sqrt{p}  - {\Yb}^{(\ell)}/\sqrt{p}  \right\|_2^{4q} \right)^{1/2}\left( \EE \left( \sigma_{\min} \left( \tilde{\mathbf{\Omega}}^{(\ell)}/\sqrt{p}\right)\right)^{-4q} \right)^{1/2}\\
    &\lesssim \left( \eta q^2  \sqrt{\frac{d }{\Delta^2p}}\|\|\Eb\|_2\|_{\psi_1} \right)^{2q}.
\end{align*}
Thus by Markov's inequality, we also have
\begin{align*}
    \PP \left(\| \widetilde{\mathbf{\Sigma}} - \Vb \Vb^{\top} \|_2 \ge \frac{1}{2} \right) & = \PP \left(\| \widetilde{\mathbf{\Sigma}} - \Vb \Vb^{\top} \|_2^{2q} \ge \frac{1}{2^{2q}} \right) \le 2^{2q}\EE\big(\| \widetilde{\mathbf{\Sigma}} - \Vb \Vb^{\top} \|_2^{2q}\big)\\
    & \lesssim  \left( 2\eta q^2  \sqrt{\frac{d }{\Delta^2p}}\|\|\Eb\|_2\|_{\psi_1} \right)^{2q}.
\end{align*}
Since $\widetilde\bSigma$ is the summation of positive semi-definite matrices by construction, $\widetilde\bSigma$ is also positive semi-definite. By Weyl's inequality \citep{franklin2012matrix}, we know that $\sigma_K(\widetilde{\mathbf{\Sigma}}) \ge 1- \| \widetilde{\mathbf{\Sigma}} - \Vb \Vb^{\top} \|_2 $ and $\sigma_{K+1}(\widetilde{\mathbf{\Sigma}}) \le \| \widetilde{\mathbf{\Sigma}} - \Vb \Vb^{\top} \|_2$. 

Now if we denote the SVD of $\widetilde{\mathbf{\Sigma}}^q$ by $\widetilde{\Vb} \tilde{\mathbf{\Lambda}}_K^q \widetilde{\Vb}^{\top} + \tilde{\Vb}_{\perp} \tilde{\mathbf{\Lambda}}_{\perp}^q \tilde{\Vb}_{\perp}^{\top}$, then with probability 1, $\widetilde{\Vb} \tilde{\mathbf{\Lambda}}_K^q \widetilde{\Vb}^{\top} \bOmega^{\F}$ and $\widetilde{\Vb}$ share the same column space. By the relationship $\sigma_{k}(\widetilde{\mathbf{\Sigma}}^q)=\sigma_{k}^q(\widetilde{\mathbf{\Sigma}})$ for $k \in [d]$ and Davis-Kahan’s Theorem \citep{yu2015davis}, we have
\begin{align*}
    \EE \left(\|\widetilde{\Vb}^{\rm F}\widetilde{\Vb}^{\text{F}\top} - \widetilde\Vb \widetilde\Vb^{\top}  \|^2_{\rm F} \, | \widetilde{\mathbf{\Sigma}} \right) & \lesssim \EE \left(K\| \widetilde{\mathbf{\Sigma}}^q \mathbf{\Omega}^{\rm F} - \widetilde{\Vb} \tilde{\mathbf{\Lambda}}_K^q \widetilde{\Vb}^{\top}\mathbf{\Omega}^{\rm F}\|_2^2/\sigma^2_{\min}(\widetilde{\Vb} \tilde{\mathbf{\Lambda}}_K^q \widetilde{\Vb}^{\top}\mathbf{\Omega}^{\rm F}) \, | \widetilde{\mathbf{\Sigma}} \right)\\
    & \lesssim \left( \frac{\sqrt{K}}{\sigma_{K}^q(\widetilde{\mathbf{\Sigma}}) } \|\tilde{\Vb}_{\perp} \tilde{\mathbf{\Lambda}}_{\perp}^q \tilde{\Vb}_{\perp}^{\top}\|_2 \cdot\| \|\mathbf{\Omega}^{\rm F} /\sqrt{p'}\|_2\|_{\psi_1} \right)^2\\
    & \lesssim \frac{Kd}{p'} \frac{\| \widetilde{\mathbf{\Sigma}} - \Vb \Vb^{\top} \|_2^{2q} }{ \left(1- \| \widetilde{\mathbf{\Sigma}} - \Vb \Vb^{\top} \|_2\right)^{2q} }.
\end{align*}
Therefore we have,
\begin{align*}
    &\left( \EE \|\widetilde{\Vb}^{\rm F}\widetilde{\Vb}^{\text{F}\top} - \widetilde\Vb \widetilde\Vb^{\top}  \|^2_{\rm F} \right)^{1/2}  \lesssim \left( \EE \|\widetilde{\Vb}^{\rm F}\widetilde{\Vb}^{\text{F}\top} - \widetilde\Vb \widetilde\Vb^{\top}  \|^2_{\rm F} \II\big\{\| \widetilde{\mathbf{\Sigma}} - \Vb \Vb^{\top} \|_2 \le {1}/{2} \big\} \right)^{1/2} \\
    &\quad + \left( \EE \|\widetilde{\Vb}^{\rm F}\widetilde{\Vb}^{\text{F}\top} - \widetilde\Vb \widetilde\Vb^{\top}  \|^2_{\rm F} \II\big\{\| \widetilde{\mathbf{\Sigma}} - \Vb \Vb^{\top} \|_2 > {1}/{2} \big\} \right)^{1/2}\\
    & \lesssim 2^q \sqrt{\frac{Kd}{p'}} \left( \EE \| \widetilde{\mathbf{\Sigma}} - \Vb \Vb^{\top} \|_2^{2q}  \right)^{1/2} + \sqrt{K} \left\{ \PP \left(\| \widetilde{\mathbf{\Sigma}} - \Vb \Vb^{\top} \|_2 \ge \frac{1}{2} \right) \right\}^{1/2}\\
    & \lesssim \sqrt{\frac{Kd}{p'}}\left( 2\eta q^2  \sqrt{\frac{d }{\Delta^2p}}\|\|\Eb\|_2\|_{\psi_1} \right)^{q} + \sqrt{K}\left( 2\eta q^2  \sqrt{\frac{d }{\Delta^2p}}\|\|\Eb\|_2\|_{\psi_1} \right)^{q}\\
    & \lesssim \sqrt{\frac{Kd}{p'}}\left( 2\eta q^2  \sqrt{\frac{d }{\Delta^2p}}\|\|\Eb\|_2\|_{\psi_1} \right)^{q},
\end{align*}
where the last but one inequality is by Markov's inequality, i.e., 
$$
\PP \left(\| \widetilde{\mathbf{\Sigma}} - \Vb \Vb^{\top} \|_2 \ge \frac{1}{2} \right) \le 2^{2q}\EE \| \widetilde{\mathbf{\Sigma}} - \Vb \Vb^{\top} \|_2^{2q} \lesssim \left( 2\eta q^2  \sqrt{\frac{d }{\Delta^2p}}\|\|\Eb\|_2\|_{\psi_1} \right)^{2q}.
$$
 Thus by previous results and triangle inequality we have
\begin{align*}
    & \left( \EE \big|\cD(\widetilde\Vb^{\F}, \Vb)\big|^2 \right)^{1/2}  \lesssim \left( \EE \|\widetilde{\Vb}^{\rm F}\widetilde{\Vb}^{\text{F}\top} - \widetilde\Vb \widetilde\Vb^{\top}  \|^2_{\rm F} \right)^{1/2}  + \left( \EE \|\widetilde\Vb \widetilde\Vb^{\top} - \Vb\Vb^{\top} \|^2_{\rm F} \right)^{1/2} \\
    & \quad \lesssim { \frac{\sqrt{K}}{\Delta} r_1(d)}\! +\! {\!\sqrt{\frac{Kd }{\Delta^2 pL}} r_1(d)} \!+\! \sqrt{\frac{Kd}{p'}}\left(\!\! 2\eta q^2  \sqrt{\frac{d }{\Delta^2p}}r_1(d) \right)^{q}.
\end{align*}	

\subsection{Proof of Corollary~\ref{prop: err rate terms} and Corollary~\ref{prop: err rate terms add exms}}\label{sec: proof prop err rates}

The case-specific error rates can be calculated by computing $r_1(d)$ and studying the proper value of $q$ for each example.

\noindent $\bullet$ \textbf{Example~\ref{ex: spiked gaussian}: } we know that $\Eb = \hbSigma- \bSigma+ (\sigma^2 - \hat{\sigma}^2) \Ib$. Now consider the $K'\times K'$ submatrix of $\bSigma$ corresponding to the the index set $S$, which we denote by $\mathbf{\Sigma}_{S} = \bSigma_{[S,S]}$.
We have $\mathbf{\Sigma}_{S} = \sigma^2 \Ib_{K'} + (\Vb)_{[S,:]} \mathbf{\Lambda}  (\Vb)_{[S,:]}^{\top}$, where $(\Vb)_{[S,:]}$ is the submatrix of $\Vb$ composed of the rows in $S$. Then since $(\Vb)_{[S,:]} \mathbf{\Lambda}  (\Vb)_{[S,:]}^{\top} \succeq \mathbf{0}$ and $\operatorname{rank}\big((\Vb)_{[S,:]} \mathbf{\Lambda}  (\Vb)_{[S,:]}^{\top} \big) \le K$, we know that $\sigma_{\min} (\mathbf{\Sigma}_{S}) = \sigma^2$. By Weyl's inequality \citep{franklin2012matrix}, we know $|\sigma^2 - \hat{\sigma}^2 | \le \|\widehat{\mathbf{\Sigma}}_{S} - \mathbf{\Sigma}_{S} \|_2  \le  \| \hat{\mathbf{\Sigma}} - \mathbf{\Sigma}\|_2 $. Thus we have 
$\|\Eb \|_2 \le \|\hat{\mathbf{\Sigma}} - \mathbf{\Sigma}\|_2  + |\sigma^2 - \hat{\sigma}^2 | \le 2\|\hat{\mathbf{\Sigma}} - \mathbf{\Sigma}\|_2$. Then by Lemma 3 in \citet{fandistributed2019}, we have that there exists some constant $c \ge 1$ such that for any $t \ge 0$, we have
\begin{align*}
    \PP(\|\Eb\|_2 \ge t) \le \PP(2\|\hbSigma - \bSigma\|_2 \ge t) \le \exp(-\frac{t}{2c(\lambda_1 + \sigma^2)\sqrt{r/n}}) ,
\end{align*}
where $r = \operatorname{tr}(\bSigma) / \|\bSigma\|_2$ is the effective rank of $\bSigma$. Thus we can see that $\|\Eb\|_2$ is sub-exponential with  
$$
\|\|\Eb\|_2\|_{\psi_1} \lesssim \|\|\hbSigma - \bSigma\|_2\|_{\psi_1} \lesssim (\lambda_1 + \sigma^2)\sqrt{\frac{r}{n}},
$$
and hence we can take $r_1(d) = (\lambda_1 + \sigma^2)\sqrt{\frac{r}{n}}$. When { $n \ge C (dr/{p})\kappa_1^2 (\log d)^6$ for some large enough constant $C  > 0$}, by Theorem~\ref{thm: error bound}, we have
$$
            \left( \EE |\cD(\widetilde\Vb^{\F}, \Vb)|^2 \right)^{1/2} \lesssim  \kappa_1 \sqrt{\frac{Kr}{n}} + \kappa_1  \sqrt{\frac{Kd r}{npL}} + \sqrt{\frac{Kd}{p'}}\left( \eta q^2\kappa_1 \sqrt{\frac{d r}{np}} \right)^{q},
$$
where by \eqref{eq: q range} the third term will be dominated by the first bias term { when taking $q = \lceil \log d /\log\log d + 3\rceil$}.

\noindent $\bullet$ \textbf{Example~\ref{ex: GMM}: }From the problem setting we know that we can represent $\bW_j$ as $\bW_j = \sum_{k=1}^K \II\{k_j = k \}\btheta_k + \bZ_j$, where $\bZ_j \overset{\text{i.i.d}}{\sim} {\cN}(\mathbf{0},\Ib_n )$, $j \in [d]$. Denote  $\Zb = (\bZ_1, \ldots, \bZ_d)$, then it can be seen that $\EE(\Xb^{\top}\Xb) = \EE(\Xb)^{\top}\EE(\Xb) + \EE(\Zb^{\top}\Zb) = \Fb\bTheta^{\top}\bTheta\Fb^{\top} + n\Ib_d$, and we can write 
$$
\Eb  = \Xb^{\top}\Xb - \EE(\Xb^{\top}\Xb ) = \Fb \bTheta^{\top} \Zb + \Zb^{\top}\bTheta\Fb^{\top} + \Zb^{\top}\Zb - n\Ib_d,
$$
then we know that $\|\Eb\|_2 \le 2\|\Fb \bTheta^{\top} \Zb\|_2 + n\|\Zb^{\top}\Zb/n - \Ib_d\|_2$. We consider $\|\Fb \bTheta^{\top} \Zb\|_2$ first. We know that $\widetilde\Zb := \bTheta^{\top} \Zb =  \bTheta^{\top} (\bZ_1, \ldots, \bZ_d) = (\widetilde\bZ_1,\ldots,\widetilde\bZ_d) \in \RR^{K \times d}$, where $\widetilde\bZ_j \overset{\text{i.i.d}}{\sim} {\cN}(\mathbf{0}, \bTheta^{\top}\bTheta )$. Under the given conditions we know that $\| \bTheta^{\top}\bTheta \|_2 \le \Delta_0^2 $. Since $(\bTheta^{\top}\bTheta )^{-1/2} \widetilde\Zb$ is a $K\times d$ i.i.d. Gaussian matrix,  by Lemma \ref{lm: gaussian norm}, we have that $$\|\|\widetilde\Zb\|_2\|_{\psi_1} \lesssim \|(\bTheta^{\top}\bTheta )^{1/2} \|_2 \|\|(\bTheta^{\top}\bTheta )^{-1/2} \widetilde\Zb\|_2\|_{\psi_1} \lesssim \Delta_0\sqrt{d}.$$ 

As for $\|\Zb^{\top}\Zb/n - \Ib_d\|_2$, when $n > d$, by Lemma 3 in \citet{fandistributed2019} we know that $\|\|\Zb^{\top}\Zb/n - \Ib_d\|_2\|_{\psi_1} \lesssim \sqrt{d/n}$, and hence in summary we have 
\begin{align*}
    \|\|\Eb\|_2\|_{\psi_1} \lesssim \|\Fb\|_2 \|\|\widetilde\Zb\|_2\|_{\psi_1} + n\|\|\Zb^{\top}\Zb/n - \Ib_d\|_2\|_{\psi_1} \lesssim \Delta_0d/\sqrt{K} + \sqrt{nd},
\end{align*}
and we can take $r_1(d) = \Delta_0d/\sqrt{K} + \sqrt{nd}$.  We know that $\Delta = \sigma_{\min}(\Fb \bTheta^{\top} \bTheta \Fb^{\top}) \gtrsim  d\Delta_0^2/K$, and thus under the condition that { $\Delta_0^2 \ge C K(\log d)^3 \left(d(\log d)^3/p \vee \sqrt{n/p}\right)$} for some large enough constant $C > 0$, by Theorem~\ref{thm: error bound} we have that 

$$
         \begin{aligned}
        \left( \EE |\cD(\widetilde\Vb^{\F}, \Vb)|^2 \right)^{1/2} &\!\!\!\lesssim \left(\frac{K}{\Delta_0} \!+\!\frac{K}{\Delta_0^2} \sqrt{\frac{Kn}{d}}\right)\!+  \!\sqrt{\frac{d}{pL}}\left(\frac{K}{\Delta_0} \!+\!\frac{K}{\Delta_0^2} \sqrt{\frac{Kn}{d}}\right)\\
        &+ \sqrt{\frac{Kd}{p'}}\left(\! \eta q^2\! \left(\sqrt{\frac{dK}{p\Delta_0^2}}+ \frac{K}{\Delta_0^2} \sqrt{\frac{n}{p}}\right) \right)^{q},
    \end{aligned}
$$
{ where the third term is  dominated by the bias term when taking $q = \lceil \log d /\log\log d + 3\rceil$ according to \eqref{eq: q range}.} 

 \begin{remark}\label{rmk: exm3 rate}
In fact we can derive a slightly sharper tail bound for the convergence rate of $\|\Eb\|_2$. More specifically, for any $ t \ge \Delta_0 \sqrt{d}$, by Lemma 3 in \citet{fandistributed2019} there exists some constant $c \ge 1$ such that
\begin{align*}
    &\PP\big(\|\widetilde\Zb\|_2 \ge t \big)  = \PP\big(\|\widetilde\Zb \widetilde\Zb^{\top} \|_2 \ge t^2 \big) = \PP\big(d\|\widetilde\Zb \widetilde\Zb^{\top}/d - \bTheta^{\top}\bTheta + \bTheta^{\top}\bTheta \|_2 \ge t^2 \big)\\
    &\le \PP\big(d\|\widetilde\Zb \widetilde\Zb^{\top}/d - \bTheta^{\top}\bTheta\|_2 \ge t^2  - d\|\bTheta^{\top}\bTheta \|_2 \big) \le \PP\big(\|\widetilde\Zb \widetilde\Zb^{\top}/d - \bTheta^{\top}\bTheta\|_2 \ge t^2/d  - \Delta_0^2 \big) \\
    &\le \exp \Big( -\frac{t^2/d - \Delta_0^2}{c  \Delta_0^2\sqrt{K/d}} \Big),
\end{align*}
which indicates that $\|\widetilde\Zb\|_2 \lesssim \Delta_0 \sqrt{d}$ with probability at least $1 - d^{-10}$. Hence  under the condition that $\sqrt{K/d}\log d = O(1)$, with probability at least $1 - O(d^{-10})$ we have that $\|\Eb\|_2 \lesssim d\Delta_0/\sqrt{K} + \sqrt{dn} \log d$, which will be used as the statistical rate of $\|\Eb\|_2$ in later proofs. 
\end{remark}

\noindent $\bullet$ \textbf{Example~\ref{ex: SBM}: }  Under the problem settings we know that $\Eb = \widehat\Mb - \Mb = \Xb - \EE \Xb$. For the eigenvalues of $\Mb$, under the given conditions we know that 
$$
\sigma_K(\Mb) \gtrsim \theta \sigma_K (\Pb) \sigma_K^2 (\mathbf{\Pi}) \gtrsim d\theta/K, \quad \sigma_1(\Mb) \lesssim \theta \sigma_1 (\Pb) \sigma_1^2 (\mathbf{\Pi}) \lesssim Kd\theta \|\mathbf{\Pi}\|_{2, \infty}^2 \le Kd\theta,
$$
where the last inequality is because for $i \in [d]$, we have that $$\|\boldsymbol{\pi}_i\|_2 = \big(\sum_{k=1}^K \boldsymbol{\pi}_i(k)^2\big)^{1/2} \le \big(\sum_{k=1}^K \boldsymbol{\pi}_i(k)\big)^{1/2} = 1 \quad \text{and}\quad \|\mathbf{\Pi}\|_{2,\infty} \le 1.$$
Thus we know that $\Delta \gtrsim d\theta/K$. 

We then bound the entries of $\Mb$.  We know $\Mb_{ij} = \theta_{i} \theta_{j} \sum_{k=1}^{K} \sum_{k'=1}^{K} \boldsymbol{\pi}_{i}(k) \boldsymbol{\pi}_{j}(k') \Pb_{k k'}$, and thus we have that 
\begin{align*}
    &\Mb_{ij}  \ge \theta_{i} \theta_{j} \sum_{k=1}^{K} \sum_{k'=1}^{K} \boldsymbol{\pi}_{i}(k) \boldsymbol{\pi}_{j}(k') \min_{kk'} (\Pb_{k k'}) \\
    & = \theta_i\theta_j \min_{kk'} (\Pb_{k k'})  \sum_{k=1}^{K} \sum_{k'=1}^{K} \boldsymbol{\pi}_{i}(k) \boldsymbol{\pi}_{j}(k') = \theta_i\theta_j \min_{kk'} (\Pb_{k k'}); \\  
    &\Mb_{ij}  \le \theta_{i} \theta_{j} \sum_{k=1}^{K} \sum_{k'=1}^{K} \boldsymbol{\pi}_{i}(k) \boldsymbol{\pi}_{j}(k') \max_{kk'} (\Pb_{k k'})\\
    &= \theta_i\theta_j \max_{kk'} (\Pb_{k k'})  \sum_{k=1}^{K} \sum_{k'=1}^{K} \boldsymbol{\pi}_{i}(k) \boldsymbol{\pi}_{j}(k') = \theta_i\theta_j \max_{kk'} (\Pb_{k k'}) .
\end{align*}
Thus we can see that $\Mb_{ij} \asymp \theta$, $\max_{ij}\EE (\Eb_{ij}^2) \lesssim \theta$ and $\max_{i} \sum_{j} \EE(\Eb_{ij}^2) \lesssim d\theta$. 
By Theorem 3.1.4 in \citet{chen2020spectral}, we know that there exists some constant $c>0$ such that for any $t > 0$, 
$$
\mathbb{P}\{\|\Eb\|_2 \geq 4 \sqrt{d\theta}+t\} \leq d \exp \left(-t^{2}/c\right).
$$
Also, since for $t \ge 5\sqrt{d\theta}$, there exists a constant $c >0$ such that $ \PP(\|\Eb\|_2 \ge t) \le \exp(-t^2/c)$, we have that $\|\|\Eb\|_2\|_{\psi_1} \lesssim \sqrt{d\theta}$, and hence we can take $r_1(d) = \sqrt{d\theta}$. Besides, $\sqrt{p/d} \Delta/r_1(d) = \sqrt{p\theta}/K \gtrsim d^{\epsilon/2}$, and hence by Theorem~\ref{thm: error bound} we have 
$$
            \left( \EE |\cD(\widetilde\Vb^{\F}, \Vb)|^2 \right)^{1/2} \lesssim K\sqrt{\frac{K}{d\theta}} + K \sqrt{\frac{K}{pL\theta}} + \sqrt{\frac{Kd}{p'}}\left( \eta q^2  \frac{K}{\sqrt{p\theta}} \right)^{q},
$$
{ where the third term is  dominated by the bias term when taking $q = \lceil \log d /\log\log d + 3\rceil$ by \eqref{eq: q range}. }
\begin{remark}\label{rmk: exm2 rate}
It's worth noting that here in Example~\ref{ex: SBM} $\|\Eb\|_2$ converges faster than sub-Exponential random variables and $\|\Eb\|_2 \lesssim \sqrt{d\theta}$ with probability at least $1 - d^{-10}$, which we will take into account in later proofs.
\end{remark}
\begin{remark}\label{rmk: exm2 no-self-loop rate}
    Under the case where no self-loops are present, $\Eb$ is replaced by $\Eb' = \Eb - \diag(\Xb) = \Eb - \diag(\Eb) - \diag(\Mb)$. With similar arguments we can show that 
    $$\|\|\Eb'\|_2\|_{\psi_1} \lesssim \|\|\Eb - \diag(\Eb) \|_2\|_{\psi_1} + \|\diag(\Mb)\|_2 \lesssim \sqrt{d\theta} + \theta \lesssim \sqrt{d\theta},$$
    $$\text{and}\quad \|\Eb'\|_2 \lesssim \|\Eb - \diag(\Eb)\|_2 + \|\diag(\Mb)\|_2 \lesssim \sqrt{d\theta} + \theta \lesssim \sqrt{d\theta},$$
    with probability at least $1-d^{-10}$, and hence \eqref{eq: err bd exm 2} also holds for the no-self-loops case.
\end{remark}

\noindent $\bullet$ \textbf{Example~\ref{ex: missing mat}: } 
We define $\bar{\cE} = [\varepsilon_{ij}]$, then $\widehat\Mb = (1/\hat\theta) \cP_{\cS}(\Mb+\bar{\cE})$, where $\cP_{\cS}$ is the projection onto the subspace of matrices with non-zero entries only in $\cS$. Since $\widehat\Mb$ and $\widehat\Mb' := (\hat\theta/\theta) \widehat\Mb = ({1}/{\theta})\cP_{\cS}(\Mb+\bar{\cE})$ differ only by a positive factor, $\widehat\Mb$ and $\widehat\Mb'$ share exactly the same sequence of eigenvectors and  $\widetilde\Vb^{\F}$ can be viewed as the output by applying FADI to $\widehat\Mb'$. Thus we will establish the results for $\widehat\Mb'$ instead, and abuse the notation by denoting $\Eb := \widehat\Mb' - \Mb$.  We first study the order of $\|\Mb\|_{\max}$, where $\|\Mb\|_{\max} = \max_{i,j} |\Mb_{ij}|$ denotes the matrix max norm. When $\|\Vb\|_{2,\infty} \le \sqrt{\mu K / d}$ for some rate $\mu \ge 1$ (that may change with $d$), for any $i,j \in [d]$, we have that 
$$
|\Mb_{ij}| =| \eb_i^{\top}\Vb \bLambda (\eb_j^{\top}\Vb)^{\top}| \le \|\bLambda\|_2 \|\eb_i^{\top}\Vb\|_2 \|\eb_j^{\top}\Vb\|_2  \le |\lambda_1|\|\Vb\|_{2,\infty}^2 \le \frac{|\lambda_1|\mu K}{d}.
$$
Thus we have $\|\Mb\|_{\max} = O({|\lambda_1|\mu K}/{d})$. Also, we can write $\Eb = \Eb_1 + \Eb_2$, where $(\Eb_1)_{ij} = \Mb_{ij}(\delta_{ij}-\theta)/\theta, (\Eb_2)_{ij} = \varepsilon_{ij}\delta_{ij}/\theta$, and for $ i \le j$
$$\operatorname{Var}\big((\Eb_1)_{ij}\big) = \Mb_{ij}^2(1-\theta)/\theta \le \|\Mb\|_{\max}^2/\theta = O\Big(\frac{(\lambda_1\mu K)^2 }{d^2\theta}\Big), \quad \operatorname{Var}\big((\Eb_2)_{ij}\big) = \sigma^2/\theta.$$
It is not hard to see that $\Cov((\Eb_1)_{ij},(\Eb_2)_{ij})=0$. Also, by the setting of Example~\ref{ex: missing mat} we have that $|(\Eb_1)_{ij}| \le \|\Mb\|_{\max}/\theta = O(\frac{|\lambda_1|\mu K}{d\theta})$, and there exists a constant $C > 0$ independent of $d$ such that $|(\Eb_2)_{ij}| \le C \sigma\log d/\theta $ for all $i \le j$. Then we will study $\|\Eb_1\|_2$ and $\|\Eb_2\|_2$ separately. We denote $\nu_1 = d \|\Mb\|_{\max}^2/\theta$ and $\nu_2 = d\sigma^2/\theta$. Under the condition that $ \theta \ge d^{-1/2+ \epsilon}$ for some constant $\epsilon > 0$, by Theorem 3.1.4 in \citet{chen2020spectral}, there exists constant $c>0$ such that for any $t \ge 4$ we have 
\begin{align*}
    \PP\Big(\frac{\|\Eb_1\|_2}{2\sqrt{\nu_1}} \ge t\Big)& \le \PP\big(\|\Eb_1\|_2/\sqrt{\nu_1} \ge 4 + t\big) = \PP\big(\|\Eb_1\|_2 \ge 4\sqrt{\nu_1} + t\sqrt{\nu_1} \big) \\
    & \le d\exp\Big(-\frac{t^2 d\|\Mb\|_{\max}^2/\theta}{c\|\Mb\|_{\max}^2/\theta^2}\Big) = \exp( -d\theta t^2/c + \log d) \\
    & \le \exp( -\frac{d\theta t^2}{2c}) \le \exp(-t^2).
\end{align*}
Very similarly for $\|\Eb_2\|_2$, there exists $c'>0$ such that for any $t \ge 4$, we have 
\begin{align*}
    &\PP\Big(\frac{\|\Eb_2\|_2}{2\sqrt{\nu_2}} \ge t\Big) \le \PP\big(\|\Eb_2\|_2\ge 4\sqrt{\nu_2} + t\sqrt{\nu_2} \big) \le d\exp\Big(-\frac{t^2 d\sigma^2/\theta}{c'\sigma^2(\log d)^2/\theta^2}\Big) \\
    & = \exp\Big( -\frac{d\theta t^2}{c'(\log d)^2} + \log d\Big)
    \le \exp\Big( -\frac{d\theta t^2}{2c'(\log d)^2}\Big)  \le \exp(-t^2).
\end{align*}
Thus we can see that 
$$\|\|\Eb\|_2\|_{\psi_1} \le \|\|\Eb_1\|_2\|_{\psi_1} +\|\|\Eb_2\|_2\|_{\psi_1} \lesssim \sqrt{\nu_1} + \sqrt{\nu_2} \lesssim \frac{|\lambda_1| \mu K}{\sqrt{d\theta}} + \sqrt{\frac{d\sigma^2}{\theta}}.$$
By Theorem~\ref{thm: error bound}, under the condition that $p = \Omega(\sqrt{d})$, { $\sigma/\Delta \ll (\log d)^{-3}d^{-1}\sqrt{p\theta}$ and 
$\kappa_2\mu K \ll (\log d)^{-3} d^{\epsilon/2} $}, it holds that
    $$
        \begin{aligned}
        \left( \EE |\cD(\widetilde\Vb^{\F}, \Vb)|_2^2 \right)^{1/2} &\!\!\lesssim \sqrt{K} \left(\frac{\kappa_2 \mu K}{\sqrt{d\theta}} \!+\! \sqrt{\frac{d\sigma^2}{\Delta^2 \theta}}\right) \!\!+\!  K\sqrt{\frac{d}{pL}} \!\!\left(\frac{\kappa_2 \mu K}{\sqrt{d\theta}} \!+\! \sqrt{\frac{d\sigma^2}{\Delta^2 \theta}}\right)\\
        & + \sqrt{\frac{Kd}{p'}}\left( \eta q^2 \left(\frac{\kappa_2 \mu K}{\sqrt{p\theta}} + \sqrt{\frac{d^2\sigma^2}{p\Delta^2 \theta}}\right) \right)^{q}.
    \end{aligned}
    $$
 Furthermore, { when taking $q = \lceil \log d /\log\log d + 3\rceil$,
the third term vanishes according to \eqref{eq: q range} and \eqref{eq: err bd exm 4} holds.
}

\begin{remark}\label{rmk: err rate exm 4}
Here we can also obtain a statistical rate sharper than subexponential rate for $\|\Eb\|_2$ that would be used in later proofs. Combining the above results for any $t \ge 16\max(\sqrt{\nu_1}, \sqrt{\nu_2})$ we have
\begin{align*}
    \PP\big(\|\Eb\|_2 \ge t\big) &\le \PP\big(\|\Eb_1\|_2 \ge t/2 \big) + \PP\big(\|\Eb_2\|_2 \ge t/2 \big) \le \exp( -\frac{d\theta t^2}{32c\nu_1}) + \exp\Big(\!\! -\!\frac{d\theta t^2}{32c'(\log d)^2\nu_2}\Big) \\
    & = \exp\Big(-\frac{d^2\theta^2 t^2}{C_1(\lambda_1\mu K)^2}\Big) + \exp\Big(-\frac{\theta^2 t^2}{C_2(\log d)^2 \sigma^2} \Big),
\end{align*}
where $C_1, C_2 > 0$ are constants. 
Thus $\|\Eb\|_2 \lesssim \frac{|\lambda_1| \mu K}{\sqrt{d\theta}} + \sqrt{\frac{d\sigma^2}{\theta}}$ with probability at least $1- d^{-10}$.

\end{remark}
{ 
\subsection{Heterogeneous Case for Example~\ref{ex: spiked gaussian}}\label{sec: hetero ex 1}
In this section we consider the generalization of Example~\ref{ex: spiked gaussian} to the heterogeneous setting. 
\begin{example}\label{exm: exm 1 hetero}
Let $\bX_1, \ldots, \bX_n \in \RR^d$ be independent random vectors distributed across $m$ sites, with $\EE (\bX_i) = \mathbf{0}$ and $\EE(\bX_i \bX_i^{\top}) = \bSigma_i$ for $i = 1, \ldots, n$. Assume that there exists a rate $B_n > 0$ possibly depending on $n$, such that we have 
\begin{align}
    \PP\Big(\max_{i \in [n]}\|\bX_i \bX_i^{\top} - \bSigma_i \|_2 \ge B_n \Big) &\le q_1(n),\label{eq: tail rate 1}\\
    \max_{i \in [n]}\big\|\EE\big((\bX_i \bX_i^{\top}- \bSigma_i  ) \II\{\|\bX_i \bX_i^{\top} - \bSigma_i \|_2 \ge B_n\}\big)\big\|_2 &\le q_2(n), \label{eq: tail rate 2}\\
    \max_{i \in [n]} \|\EE [(\bX_i \bX_i^{\top}- \bSigma_i  ) (\bX_i \bX_i^{\top}- \bSigma_i  ) ]\|_2 &\le \nu_n,\label{eq: var bd M}
\end{align}
where $q_1(n), q_2(n), \nu_n > 0$ are rates possibly dependent of $n$. Additionally, suppose that there exists $ \bSigma = \Db + \Vb \bLambda \Vb^{\top} \in \RR^{d \times d}$ with $\Db = \diag(\sigma_1^2,\ldots, \sigma_d^2)$, $\bLambda = \diag(\lambda_1,\ldots, \lambda_K)$ ($\lambda_1 \ge \lambda_2 \ge \ldots \ge \lambda_K > 0$) and $\Vb^{\top} \Vb = \Ib_K$, such that 
\begin{equation}\label{eq: tail rate 3}
   \Big \|\frac{1}{n} \sum_{i=1}^n \bSigma_i  - \bSigma \Big \|_2 \le q_3(n), \quad 
   \lambda_1 \|\Vb\|_{2,\infty}^2/ \Delta \le q_4(d),
\end{equation}
where $\Delta = \lambda_K$, and $q_3(n), q_4(d) > 0$ are rates possibly dependent of $n$ and $d$. Then we let $\Mb = \Vb \bLambda \Vb^{\top}$, $\widehat\Mb = \widehat\bSigma - \diag(\hbSigma)$, where $\hbSigma = \frac{1}{n} \sum_{i = 1}^n \bX_i \bX_i^{\top}$, and we will estimate $\Vb$ by performing the FADI algorithm in Section~\ref{sec: alg}, with the resulting estimator $\widetilde\Vb^{\F}$. 
\end{example}
\begin{remark}
    Example~\ref{exm: exm 1 hetero} generalizes Example~\ref{ex: spiked gaussian} to allow for heterogeneity in the data. The conditions specified in~\eqref{eq: tail rate 1} through~\eqref{eq: var bd M} are standard scaling conditions commonly met by sub-Gaussian random vectors. Condition~\eqref{eq: tail rate 3} imposes a specific structure on the population covariance matrices. A relevant scenario where Example~\ref{exm: exm 1 hetero} applies is when data partitions stored across different machines exhibit different population covariance matrices. In particular, consider the sample split $\{\bX_i^{(s)}\}_{i=1}^{n_s}$ of size $n_s$ at the $s$-th site, and assume $\bSigma^{(s)} = \EE[\bX_i^{(s)}\bX_i^{(s)\top}] = \Db_s + \Vb \bLambda_s \Vb^{\top}$, where $\Db_s$ and $\bLambda_s$ are diagonal matrices allowed to vary across sites. Under this framework, Example~\ref{exm: exm 1 hetero} is applicable, leading to $\bSigma = \Vb (\sum_{s \in [m]} \frac{n_s}{n}  \bLambda_s) \Vb^{\top} + \sum_{s \in [m]} \frac{n_s}{n} \Db_s$ such that Condition~\eqref{eq: tail rate 3} holds with $q_3(n) = 0$.
\end{remark}
\begin{remark}
  A further generalization of Example~\ref{exm: exm 1 hetero} to the non-independent setting relies on controlling the rate of $\|\widehat\bSigma - \frac{1}{n} \sum_{i=1}^n \bSigma_i \|_2$. Several studies have examined the convergence of sample covariance matrices for non-i.i.d. data \citep{Banna_2016, fan2013poet}.
\end{remark}
We characterize the statistical rate of $\widetilde\Vb^{\F}$ by the following corollary of Theorem~\ref{thm: error bound}.
\begin{corollary}\label{col: exm 1 hetero error bound}
    Let $\bX_1, \ldots, \bX_n \in \RR^d$ be the independent centered random vectors defined in Example~\ref{exm: exm 1 hetero}, and recall $\widetilde\Vb^{\F}$ is the FADI estimator of $\Vb$.  Define 
    $$
    \kappa(d,n) =  \Delta^{-1} \left( \sqrt{\frac{2\nu_n \log d}{n} } + \frac{2\big(B_n + \!q_2(n)\big) \log d} {3n} + q_2(n) + q_3(n) \right) + q_4(d).
    $$
    Under the condition that $(\log d)^{-3}\sqrt{p/d} \kappa(d,n)^{-1} \ge C$ for some large enough constant $C > 0$, if we take $p' \ge \max(2K, K + 7)$, $q = \lceil \log d /\log\log d + 3\rceil$ and $p \ge \max(2K, K + 8q -1)$, there exists some constant $\eta > 0$ such that we have
    \begin{equation}\label{eq: exm 1 hetero error bound}
          \big( \EE|\cD (\widetilde\Vb^{\F}, \Vb) |^2\big)^{1/2} \lesssim \sqrt{K} \kappa(d,n) + \sqrt{\!\frac{Kd}{p L}} \kappa(d,n)
          + \sqrt{K} q_1(n).
    \end{equation}
\end{corollary}
\begin{remark}
   In most applications, such as when $\{\bX_i\}_{i=1}^n$ are sub-Gaussian, the rate \(\nu_n\) is typically of order \(O(d)\), the rate \(B_n\) is polylogarithmic, and \(q_1(n)\) and \(q_2(n)\) are of order \(O(n^{-c})\) for some constant \(c \ge 1\) depending on \(B_n\). The rate \(\kappa(d, n)\) is generally of order \(\widetilde{O}\big(\sqrt{d/n} + {K/d} + q_3(n)\big)\), where the \(\widetilde{O}(\cdot)\) notation suppresses polylogarithmic factors. 
\end{remark}
\begin{proof}
    Define the matrix
    $$
    \widetilde\Mb_i = \bSigma_i + (\bX_i \bX_i^{\top}- \bSigma_i ) \II\{\|\bX_i \bX_i^{\top} - \bSigma_i\|_2 < B_n\}, \quad i = 1, \ldots, n,
    $$
   and the event 
    $$
    \cA_{B_n} = \{\widetilde\Mb_i  = \bX_i \bX_i^{\top}, \text{ for all } i = 1,\ldots, n\}. 
    $$
   By \eqref{eq: tail rate 1}, we have that $\PP(\cA_{B_n}) \ge 1 - q_1(n)$. Then denote $\widehat\Mb' = \frac{1}{n} \sum_{i=1}^n \widetilde\Mb_i -\diag\big( \frac{1}{n} \sum_{i=1}^n \widetilde\Mb_i\big) $, and let $\widetilde\Vb^{\F}_1$ be the resulting estimator by performing the FADI algorithm on $\widehat\Mb'$ instead. We have the following decomposition,
   \begin{align*}
       \cD (\widetilde\Vb^{\F}, \Vb) &= \cD (\widetilde\Vb^{\F}, \Vb)\II\{\cA_{B_n}\} + \cD (\widetilde\Vb^{\F}, \Vb)\II\{\cA_{B_n}^c\} = \cD (\widetilde\Vb^{\F}_1, \Vb)\II\{\cA_{B_n}\} + \cD (\widetilde\Vb^{\F}, \Vb)\II\{\cA_{B_n}^c\} \\
       &\le  \cD (\widetilde\Vb^{\F}_1, \Vb) + 2\sqrt{K} \II\{\cA_{B_n}^c\},
   \end{align*}
   where the second equality follows from the fact that under the event $\cA_{B_n}$ we have $\widetilde\Vb^{\F} = \widetilde\Vb^{\F}_1$. Hence in turn, we have the upper bound
     \begin{align*}
      \big( \EE|\cD (\widetilde\Vb^{\F}, \Vb) |^2\big)^{1/2} \lesssim   \big( \EE|\cD (\widetilde\Vb^{\F}_1, \Vb) |^2\big)^{1/2} + \sqrt{K} \PP(\cA_{B_n}^c) \le  \big( \EE|\cD (\widetilde\Vb^{\F}_1, \Vb) |^2\big)^{1/2} + \sqrt{K}q_1(n).
   \end{align*}
   Now to bound $ \big( \EE|\cD (\widetilde\Vb^{\F}_1, \Vb) |^2\big)^{1/2}$, it suffices for us to study the concentration behavior of $\Eb' = \widehat\Mb' - \Mb$ and apply Theorem~\ref{thm: error bound}. We have that 
   \begin{align*}
          \|\Eb'\|_2 &= \left\|\frac{1}{n} \sum_{i=1}^n \widetilde\Mb_i -\diag\Big( \frac{1}{n} \sum_{i=1}^n \widetilde\Mb_i\Big)  - \bSigma + \Db \right\|_2 \\
          & = \left\|\frac{1}{n} \sum_{i=1}^n \widetilde\Mb_i -\diag\Big( \frac{1}{n} \sum_{i=1}^n \widetilde\Mb_i\Big)  - \bSigma + \diag(\bSigma) - \diag(\Vb \bLambda \Vb^{\top}) \right\|_2 \\
          & \le 2 \left\|\frac{1}{n} \sum_{i=1}^n \widetilde\Mb_i - \bSigma\right \|_2 +  \lambda_1 \|\Vb\|_{2,\infty}^2,
   \end{align*}
   where the last inequality is due to the fact that 
   $$
   \left\|\diag\Big( \frac{1}{n} \sum_{i=1}^n \widetilde\Mb_i\Big)  - \diag (\bSigma ) \right \|_2 = \max_{j \in [d]}\left|\left(\frac{1}{n} \sum_{i=1}^n \widetilde\Mb_i - \bSigma\right)_{jj} \right| \le \left\|\frac{1}{n} \sum_{i=1}^n \widetilde\Mb_i - \bSigma\right \|_2 .
   $$
   Besides, we have that 
   \begin{align*}
       &\left\|\frac{1}{n} \!\sum_{i=1}^n \widetilde\Mb_i - \bSigma\right \|_2 \!\!\le \left\|\frac{1}{n} \!\sum_{i=1}^n \widetilde\Mb_i - \frac{1}{n}\! \sum_{i=1}^n \EE ( \widetilde\Mb_i) \right \|_2  \!\!+ \left\|\frac{1}{n} \!\sum_{i=1}^n \EE ( \widetilde\Mb_i) - \frac{1}{n}\! \sum_{i = 1}^n \!\bSigma_i \right \|_2 \!\!+   \left \|\frac{1}{n}\! \sum_{i=1}^n \!\bSigma_i  - \bSigma \right \|_2 \\
       & \quad \le \left\|\frac{1}{n} \!\sum_{i=1}^n \widetilde\Mb_i - \frac{1}{n}\! \sum_{i=1}^n \EE ( \widetilde\Mb_i) \right \|_2 + \max_{i \in [n]}\big\|\EE\big((\bX_i \bX_i^{\top}- \bSigma_i  ) \II\{\|\bX_i \bX_i^{\top} - \bSigma_i \|_2 \ge B_n\}\big)\big\|_2 \\
       & \quad \quad + \left \|\frac{1}{n}\! \sum_{i=1}^n \!\bSigma_i  - \bSigma \right \|_2 \le  \left\|\frac{1}{n} \!\sum_{i=1}^n \widetilde\Mb_i - \frac{1}{n}\! \sum_{i=1}^n \EE ( \widetilde\Mb_i) \right \|_2 + q_2(n) + q_3(n),
   \end{align*}
   and we will study the convergence rate of $\left\|\frac{1}{n} \!\sum_{i=1}^n \widetilde\Mb_i - \frac{1}{n}\! \sum_{i=1}^n \EE ( \widetilde\Mb_i) \right \|_2 $. By the definition of $\widetilde\Mb_i$ and \eqref{eq: var bd M}, we have that 
   \begin{align*}
       \max_{i \in [n]}\|\widetilde\Mb_i - \EE (\widetilde\Mb_i )\|_2 &\le B_n + q_2(n),\\
        \max_{i \in [n]} \| \EE [(\widetilde\Mb_i - \EE (\widetilde\Mb_i ))(\widetilde\Mb_i - \EE (\widetilde\Mb_i ))] \|_2 &\le  \max_{i \in [n]} \|\EE [(\bX_i \bX_i^{\top}- \bSigma_i  ) (\bX_i \bX_i^{\top}- \bSigma_i  ) ]\|_2 &\le \nu_n,
   \end{align*}
   and hence by Theorem~3.1.1 of \citet{chen2020spectral}, for any $t > 0$, we have that 
   \begin{align*}
       & \PP\left(\left\|\frac{1}{n} \!\sum_{i=1}^n \widetilde\Mb_i - \frac{1}{n}\! \sum_{i=1}^n \EE ( \widetilde\Mb_i) \right \|_2   \ge t \right)  \le 2d \exp\left(-\frac{nt^2 /2}{\nu_n + (B_n + q_2(n)) t /3}\right)\\
        & \le  2d \exp\left(-\frac{nt^2 /2}{\nu_n} - \frac{nt^2 /2}{(B_n + q_2(n)) t /3}\right)  \le d \left\{ \exp\left(-\frac{nt^2 }{\nu_n} \right)+\exp\left( - \frac{nt}{(B_n + q_2(n))  /3}\right)\right\},
   \end{align*}
   where the last two inequalities follow from Jensen's inequality \citep{jensen1906fonctions}. Then when $t \ge \max\{\sqrt{2\nu_n \log d/ n} , 2(B_n + q_2(n))\log d/(3n)\}$, we further have that 
   \begin{align}
        \PP\left(\left\|\frac{1}{n} \!\sum_{i=1}^n \widetilde\Mb_i - \frac{1}{n}\! \sum_{i=1}^n \EE ( \widetilde\Mb_i) \right \|_2   \ge t \right)   &\le 2 \exp\left(-\frac{t}{ \sqrt{2\nu_n / n} + 2(B_n + q_2(n))/(3n)}\right) \notag\\
        & \le  2 \exp\!\left(-\frac{t}{ \sqrt{2\nu_n \log d/ n} + 2(B_n + \!q_2(n)\!) \!\log d/(3n)}\right).  \label{eq: exp bd tail}
   \end{align}
   Then we have the following holds for all $t > 0$,
   $$
   \PP\left(\left\|\frac{1}{n} \!\sum_{i=1}^n \widetilde\Mb_i - \frac{1}{n}\! \sum_{i=1}^n \EE ( \widetilde\Mb_i) \right \|_2   \ge t \right)  \le  2 \exp\!\left(-\frac{t/2}{ \sqrt{2\nu_n \log d/ n} + 2(B_n + \!q_2(n)\!) \!\log d/(3n)}\right),
   $$
where the equality holds when $t \ge \sqrt{2\nu_n \log d/ n} + 2(B_n + q_2(n))\log d/(3n) $ as a result of  \eqref{eq: exp bd tail}, and when $t < \sqrt{2\nu_n \log d/ n} + 2(B_n + q_2(n))\log d/(3n) $, the equality holds because 
 $$
  2 \exp\!\left(-\frac{t/2}{ \sqrt{2\nu_n \log d/ n} + 2(B_n + \!q_2(n)\!) \!\log d/(3n)}\right) \ge 2 e^{-1/2} > 1.
   $$
   Then by standard probability theory, we know that $\left\|\frac{1}{n} \!\sum_{i=1}^n \widetilde\Mb_i - \frac{1}{n}\! \sum_{i=1}^n \EE ( \widetilde\Mb_i) \right \|_2 $ is sub-exponential and there exists $r_1'(d) = \sqrt{2\nu_n \log d/ n} + 2(B_n + \!q_2(n)\!) \!\log d/(3n)$ such that $\left\|\left\|\frac{1}{n} \!\sum_{i=1}^n \widetilde\Mb_i - \frac{1}{n}\! \sum_{i=1}^n \EE ( \widetilde\Mb_i) \right \|_2\right\|_{\psi_1} \lesssim r_1(d)$. Hence we have that 
   $$
   \|\|\Eb' \|_2\|_{\psi_1} \lesssim \lambda_1 \|\Vb\|_{2,\infty}^2 + q_2(n) + q_3(n) + r_1'(n),
   $$
   and applying Theorem~\ref{thm: error bound}, if $p \ge \max(2K, K+7)$ and $ (\log d)^{-1}\sqrt{p/d} \kappa(d,n)^{-1} \ge C$ for some large enough constant $C > 0$, there exists some constant $\eta > 0$ such that we have
   \begin{align*}
       \big( \EE|\cD (\widetilde\Vb^{\F}_1, \Vb) |^2\big)^{1/2} \lesssim \sqrt{K} \kappa(d,n) + \sqrt{\frac{Kd}{p L}} \kappa(d,n) + \sqrt{\frac{Kd}{p'}} \left(\eta q^2 \sqrt{\frac{d}{p}} \kappa(d,n)\right)^q,
   \end{align*}
where, according to \eqref{eq: q range}, the third term vanishes if we further impose \((\log d)^{-3}\sqrt{p/d} \kappa(d,n)^{-1} \ge C\) for some sufficiently large constant \(C > 0\), take \(q = \lceil \log d /\log\log d + 3\rceil\), and \(p \ge \max(2K, K + 8q -1)\). Then, by the triangle inequality, \eqref{eq: exm 1 hetero error bound} holds.
\end{proof}
}
\subsection{Proof of Theorem~\ref{thm: est k}}\label{sec: proof thm est k}
We first bound the recovery probability of $\hat{K}^{(\ell)}$ for each $\ell \in [L]$. Recall that $\widehat\Yb^{(\ell)}/\sqrt{p} = \Vb \mathbf{\Lambda} \tilde{\mathbf{\Omega}}^{(\ell)} / \sqrt{p} + {\Eb} \mathbf{\Omega}^{(\ell)}/\sqrt{p}$, where $\tilde{\mathbf{\Omega}}^{(\ell)} = \Vb^{\top} \mathbf{\Omega}^{(\ell)}$.

For the residual term ${\Eb}\mathbf{\Omega}^{(\ell)}/\sqrt{p}$, by Lemma 3 in \citet{fandistributed2019}, under the condition that $\sqrt{p/d}\log d = o(1)$, with probability at least $1 - d^{-10}$ we have $\|\mathbf{\Omega}^{(\ell)} /\sqrt{p}\|_2 \le 2\sqrt{\frac{d}{p}}$.
Denote by $\cA_{\Eb} $  the event $\big\{\|{\Eb}\|_2 \le 10c_e^{-1} r_1(d)\log d \big\}$, where $c_e > 0 $ is the constant defined in Remark~\ref{rmk: Eb subexp const}. Then conditional on $\cA_{\Eb}$, we have that  $\|{\Eb}\mathbf{\Omega}^{(\ell)}/\sqrt{p}\|_2 \le 20 c_e^{-1} \sqrt{\frac{d}{p}} r_1(d) \log d$ with probability at least $1- d^{-10}$ for each $\ell \in [L]$. Recall $\eta_0 = 480 c_e^{-1} \sqrt{\frac{d}{\Delta^2 p}} r_1(d) \log d$.
From Proposition 10.4 in \citet{halkofinding2011}, we know that when $p \ge 2K$,
$$
\PP \left( \sigma_{\min} \big(\tilde{\mathbf{\Omega}}^{(\ell)} / \sqrt{p}\big) \le \frac{1}{6} \sqrt{\eta_0} \right) \le \PP \left( \sigma_{\min} \big(\tilde{\mathbf{\Omega}}^{(\ell)} / \sqrt{p}\big) \le \frac{p-K+1}{{\rm e}p} \sqrt{\eta_0}\right) \le \eta_0^{\frac{p-K+1}{2}}.
$$
Therefore, with probability at least $1-\eta_0^{{(p-K+1)}/2}$, $$\sigma_{\min} \big(\Vb \mathbf{\Lambda} \tilde{\mathbf{\Omega}}^{(\ell)} / \sqrt{p}\big) \ge \Delta  \sigma_{\min} \big(\tilde{\mathbf{\Omega}}^{(\ell)} / \sqrt{p}\big) > \Delta\sqrt{\eta_0}/6 \ge 2 \mu_0.$$ 

By Weyl's inequality \citep{franklin2012matrix}, we know that conditional on $\cA_{\Eb}$, with probability at least $1-d^{-10}$, $\sigma_{K+1}(\widehat\Yb^{(\ell)} /\sqrt{p}) \le \|{\Eb}\mathbf{\Omega}^{(\ell)}/\sqrt{p}\|_2 \le 20 c_e^{-1} \sqrt{\frac{d}{p}} r_1(d) = \Delta \eta_0/24 \le \mu_0 $ for large enough $d$, which indicates that $\sigma_{k + 1}(\widehat\Yb^{(\ell)}) - \sigma_{p}(\widehat\Yb^{(\ell)}) < \sqrt{p}\mu_0$ for any $k \ge K$. For $k \le K-1$, under the same event we have
\begin{align*}
    & \sigma_{k + 1}(\widehat\Yb^{(\ell)}) - \sigma_{p}(\widehat\Yb^{(\ell)})  \ge \sigma_{K}(\widehat\Yb^{(\ell)}) - \sigma_{p}(\widehat\Yb^{(\ell)}) \ge \sigma_{\min} \big(\Vb \mathbf{\Lambda} \tilde{\mathbf{\Omega}}^{(\ell)} \big) - 2\|{\Eb}\mathbf{\Omega}^{(\ell)}\|_2 \\
    & \quad > \sqrt{p} (\Delta\sqrt{\eta_0}/6 - \Delta\eta_0/12) \ge \sqrt{p} (\Delta\sqrt{\eta_0}/6 - \Delta\sqrt{\eta_0}/12) = \Delta\sqrt{p\eta_0}/12 \ge \sqrt{p}\mu_0.
\end{align*}
Then we have
\begin{align*}
    &\PP\left(\hat{K}^{(\ell)} = K \, \,\big| \, \cA_{\Eb}\right) \ge  \PP\left(\!\! \sigma_{K}\!\big(\!{\widehat\Yb^{(\ell)} }\!\big) - \sigma_{p}\big(\!{\widehat\Yb^{(\ell)} }\!\big) > \sqrt{p}\mu_0,\,\,  \sigma_{K+1}\!\big(\!{\widehat\Yb^{(\ell)} }\!\big) - \sigma_{p}\!\big(\!{\widehat\Yb^{(\ell)} }\!\big)\le \sqrt{p} \mu_0\,\,\Big|\,\cA_{\Eb}\right) \\
    &\quad \ge \PP\left( \sigma_{\min} \big(\Vb \mathbf{\Lambda} \tilde{\mathbf{\Omega}}^{(\ell)} / \sqrt{p}\big) \ge \Delta\sqrt{\eta_0}/6,\quad  \|{\Eb}\mathbf{\Omega}^{(\ell)}/\sqrt{p}\|_2 \le \Delta\eta_0/24 \,  \Big|\,\cA_{\Eb}\right)\\
    & \quad \ge 1-  d^{-10} - \eta_0^{\frac{p-K+1}{2}}.
\end{align*}

We know that conditional on ${\Eb}$, $\II \{\hat{K}^{(\ell)} \neq K \, | \, \cA_{\Eb}\}$ are i.i.d. Bernoulli variables with expectation $p_K := \PP (\hat{K}^{(\ell)} \neq K \, | \, \cA_{\Eb}) \le  d^{-10} + \eta_0^{\frac{p-K+1}{2}} \le 1/4$ and variance $p_K(1-p_K) \le p_K$. Since the estimators $\{\hat{K}^{(\ell)}\}_{\ell =1}^L$ are all integers, we know that if $\hat{K} \neq K$, at least half of  $\{\hat{K}^{(\ell)}\}_{\ell =1}^L$ are not equal to $K$. Then by Hoeffding's inequality, we have
\begin{align*}
    \PP(\hat{K}\! \neq\! K) &\le \PP\! \left(\sum_{\ell =1}^L \II \left\{\hat{K}^{(\ell)} \neq K\right\} \!- \!p_K L \!\ge\! \frac{L}{4}\!\right) \!= \!\EE_{{\Eb}} \!\left(\!\PP \bigg(\!\sum_{\ell =1}^L \II \left\{\hat{K}^{(\ell)} \!\neq\! K\right\}\! -\! p_K L \ge \frac{L}{4} \,\big|\,{\Eb}\bigg)\!\right)\\
    & \le \PP(\cA_{\Eb}) \exp\left\{-(L/4)^2/(2Lp_K)\right\} + 1- \PP(\cA_{\Eb})\\
    & \le \exp\left\{-L\big/\big(32 d^{-10} +32 \eta_0^{\frac{p-K+1}{2}}\big)\right\} + O(d^{-10}).
\end{align*}
We know that $32 d^{-10} \le (\log d)^{-1}$ for $d \ge 2$, and under the condition that $ \eta_0 \le (32\log d)^{-\frac{2}{p-K+1}}$ we have $ \PP(\hat{K} \neq K) \le \exp(-L \log d/2) +  O(d^{-10}) \lesssim d^{-(L \wedge 20)/2}$.
\subsection{Proof of Corollary~\ref{prop: est K} and Corollary~\ref{prop: est K add exms}}\label{sec: proof of prop est K}
\noindent $\bullet$ \textbf{Example~\ref{ex: spiked gaussian}}: From the proof of Corollary~\ref{prop: err rate terms} we know that we can take $r_1(d) = (\lambda_1 + \sigma^2) \sqrt{\frac{r}{n}} \log d $. Then by plugging in each term we know that under the condition that $(\lambda_1 + \sigma^2) \left(d (np)^{-1/2} \log d\right)^{1/4} = o(1)$ and $\Delta \gg \left({\sigma^{-2} (np)^{-1/2}}d\log d\right)^{1/3}$, we have $\Delta \eta_0 /24 \ll \mu_0 \ll \Delta \sqrt{\eta_0}/12$. Besides, under the condition that $\kappa_1 \sqrt{dr/(np)}(\log d)^2 = o(1) $,  we also have $\eta_0 \le (32 \log d)^{-\frac{2}{p - K + 1}}  $. Thus the conditions for Theorem~\ref{thm: est k} are satisfied and we have $\hat{K} = K$ with probaility at least $1 - O(d^{-(L \wedge 20)/2})$.

\noindent $\bullet$ \textbf{Example~\ref{ex: GMM}: } We know from the proof of Corollary~\ref{prop: err rate terms} and Remark~\ref{rmk: exm3 rate} that $\Delta \gtrsim d\Delta_0^2/K$ and $\|\Eb\|_2 \lesssim d\Delta_0/\sqrt{K} + \sqrt{dn} \log d$ with probability at least $1- d^{-10}$. Thus we have $\eta_0 \asymp \sqrt{d/(\Delta^2 p)} \big(d\Delta_0/\sqrt{K} + \sqrt{dn} \log d\big)$. Under the condition that $ \sqrt{K (\log d)^3}\left(n/p\right)^{1/4} \ll \Delta_0 \ll \sqrt{nK/d} \log d$, we know that $d\Delta_0/\sqrt{K} + \sqrt{dn} \log d \lesssim \sqrt{dn}\log d$, $\Delta \eta_0 \asymp d\sqrt{n/p}\log d$ and $\sqrt{\eta_0} \log d = o(1)$, and thus $\Delta \eta_0/24 \ll \mu_0 \ll \Delta \sqrt{\eta_0}/12$. By Theorem~\ref{thm: est k} the claim follows. 

\noindent $\bullet$ \textbf{Example~\ref{ex: SBM}: } We know from the proof of Corollary~\ref{prop: err rate terms add exms} that $\Delta \gtrsim d\theta/K$. Also from Remark~\ref{rmk: exm2 rate} we know that $\|\Eb\|_2 \lesssim \sqrt{d\theta}$ with probability at least $1 - d^{-10}$, and thus we have $\eta_0 \asymp  \sqrt{d/(\Delta^2 p)} \sqrt{d\theta} \lesssim K / \sqrt{p\theta} \asymp 1/\sqrt{ d^{\epsilon-1/2} p}$, $\Delta \eta_0 \asymp d \sqrt{\theta/p}$ and $\Delta \sqrt{\eta_0} \gtrsim d\theta^{3/4}p^{-1/4}K^{-1/2} $. Also recall from the proof of Corollary~\ref{prop: err rate terms add exms} that $\EE(\widehat\Mb_{ij}) = \Mb_{ij} \asymp \theta$ for any $i, j \in [d]$, and hence $d^{-2} \sum_{i \le j} \Mb_{ij} \asymp \theta$. By Hoeffding's inequality \citep{hoeffding1994probability}, we have that 
$$
\PP\left(\frac{2}{d(d-1)}\left|\sum_{i \le j} \widehat\Mb_{ij} - \sum_{i \le j} \Mb_{ij}\right| \ge \frac{\sqrt{11\log d}}{d}\right) \lesssim \exp\left(- 11 d(d-1) \log d /d^2 \right) \lesssim d^{-10}.
$$
Thus we can see with probability at least $1-O(d^{-10})$, $|\hat\theta - d^{-2} \sum_{i\le j} \Mb_{ij}| \lesssim \frac{\sqrt{\log d}}{d}$ and $\hat\theta \asymp \theta$, and in turn $\Delta \eta_0/24 \ll \mu_0 \ll \Delta \sqrt{\eta_0}/12$. Thus by Theorem~\ref{thm: est k} the claim follows.

\noindent $\bullet$ \textbf{Example~\ref{ex: missing mat}: }By Hoeffding's inequality \citep{hoeffding1994probability}, with probability at least $1- d^{-10}$ we have that $ |\hat\theta-\theta|/\hat\theta \le C \sqrt{\log d}/d\theta$. As for $\hat\sigma_0^2$, we have 
\begin{align*}
    \hat\sigma_0^2 = \frac{1}{|\cS|}\sum_{(i,j) \in \cS} (\hat\theta\widehat\Mb_{ij})^2  = \frac{1}{|\cS|}\left(\sum_{i\le j} \delta_{ij}\Mb_{ij}^2 + 2\sum_{(i,j)\in \cS} \Mb_{ij}\varepsilon_{ij} + \sum_{(i,j)\in \cS} \varepsilon_{ij}^2\right).
\end{align*}
We consider the latter two terms first. We know that $|\varepsilon_{ij}| \le C\sigma\log d$ for some constant $C > 0$ and $|\Mb_{ij}| \le |\lambda_1|\mu K/d$, for any $i\le j$. Denote by $\tilde\sigma = (|\lambda_1|\mu K/d)\vee \sigma$, then we have 
$$
\Var(\M_{ij} \varepsilon_{ij}) \le ( \frac{|\lambda_1|\mu K}{d})^2\sigma^2 \le \tilde\sigma^4,\quad |\M_{ij} \varepsilon_{ij}| \le \frac{|\lambda_1|\mu K}{d}C\sigma\log d \le C\tilde\sigma^2\log d, \quad \forall i \le j,
$$
and
$$
\Var(\varepsilon_{ij}^2) \le C^4 \sigma^4 (\log d)^4 \le C^4 \tilde\sigma^4 (\log d)^4, \quad |\varepsilon_{ij}^2| \le C^2 \sigma^2 (\log d)^2 \le C^2 \tilde\sigma^2 (\log d)^2,\quad \forall i\le j. 
$$
Thus by Bernstein inequality \citep{bernstein1924ineq}, conditional on $\cS$, with probability at least $1 - 2d^{-10}$ we have that there exists a constant $C' > 0$ independent of $\cS$ such that
\begin{equation}\label{eq: est k mis mat 1}
    \left|\frac{1}{|\cS|}\sum_{(i,j)\in \cS} \Mb_{ij}\varepsilon_{ij} \right| \le C' \left(\frac{\tilde\sigma^2 \sqrt{\log d}}{\sqrt{|\cS|}} + \frac{\tilde\sigma^2(\log d)^2}{|\cS|}\right),
\end{equation}
and
\begin{equation}\label{eq: est k mis mat 2}
    \left|\frac{1}{|\cS|} \sum_{(i,j)\in \cS} \varepsilon_{ij}^2 - \sigma^2 \right| \le C' \left(\frac{\tilde\sigma^2(\log d)^{5/2}}{\sqrt{|\cS|}} + \frac{\tilde\sigma^2(\log d)^3}{|\cS|}\right).
\end{equation}
Now we consider the first term. Since $\delta_{ij}$'s are i.i.d. Bernoulli random variables with expectation $\theta$, we have 
$$
\Var(\Mb_{ij}^2 \delta_{ij}) \le \theta\tilde\sigma^4, \quad |\Mb_{ij}^2 \delta_{ij}| \le \tilde\sigma^2, \quad i\le j.
$$
Also, we know that $\sum_{i \le j}\Mb_{ij}^2 \ge \|\Mb\|_{\F}^2/2 \ge K\Delta^2/2$ and $\sum_{i \le j}\Mb_{ij}^2 \le \|\Mb\|_{\F}^2 \le K\lambda_1^2$, and hence $K\Delta^2\theta/2 \le \EE\Big(\sum_{i\le j} \delta_{ij}\Mb_{ij}^2\Big) \le K\lambda_1^2 \theta$.
Then by Bernstein inequality \citep{bernstein1924ineq} with probability at least $1-d^{-10}$, it holds that 
\begin{equation}\label{eq: est k mis mat 3}
    \left|\Big(\sum_{i\le j} \delta_{ij}\Mb_{ij}^2\Big)  - \EE\Big(\sum_{i\le j} \delta_{ij}\Mb_{ij}^2\Big) \right| \lesssim d^2 \left(\frac{\tilde\sigma^2\sqrt{\theta\log d}}{d} + \frac{\tilde\sigma^2\log d}{d^2}\right)= \tilde\sigma^2 (d\sqrt{\theta \log d} + \log d).
\end{equation}
Thus combining \eqref{eq: est k mis mat 1}, \eqref{eq: est k mis mat 2} and \eqref{eq: est k mis mat 3} with the fact that $|\cS| \asymp d^2\theta$ with probability at least $1- d^{-10}$, under the condition that $\kappa_2^2 \mu^2 K \ll (\log d)^2$, with probability at least $1- O(d^{-10})$ we have 
$$
\tilde\sigma \ll \left(\frac{\Delta\sqrt{K}\log d}{d}\vee\sigma\right)+o(\tilde\sigma) \lesssim \hat\sigma_0 \log d \lesssim \left(\frac{|\lambda_1|\sqrt{K}\log d}{d}\vee\sigma\right)+o(\tilde\sigma) \lesssim \tilde\sigma \log d.
$$

From the proof of Corollary~\ref{prop: err rate terms add exms} and Remark~\ref{rmk: err rate exm 4},  we know that with probability at least $1 - d^{-10}$,
$$\|\widehat\Mb - \Mb\|_2 \lesssim \left\|\frac{\hat\theta}{\theta}\widehat\Mb - \widehat\Mb\right\|_2 + \left\|\frac{\hat\theta}{\theta}\widehat\Mb - \Mb\right\|_2 \lesssim \frac{|\lambda_1|\sqrt{\log d}}{d\theta} + \frac{|\lambda_1 |\mu K}{\sqrt{d \theta}} + \sqrt{\frac{d \sigma^2}{\theta}} \lesssim \sqrt{\frac{d \tilde\sigma^2}{\theta}},$$
and hence  $\eta_0 \asymp d\tilde\sigma(\Delta\sqrt{p\theta})^{-1}$ and $\Delta \eta_0 \asymp d\tilde\sigma/\sqrt{p\theta} $. 

Under the condition that $(p\theta)^{-1/4}\sqrt{d\sigma/\Delta}\log d = o(1)$, with probability at least $1-O(d^{-10})$ we have $\Delta \eta_0/24 \ll \mu_0 \ll \Delta \sqrt{\eta_0}/12$. Thus by Theorem~\ref{thm: est k} the claim follows.
\subsection{Proof of Theorem~\ref{thm: leading term}}\label{sec: proof of thm leading term}
We first decompose $\widetilde{\Vb}^{\F}\Hb - \Vb = \widetilde{\Vb}^{\F}\Hb -\widetilde{\Vb}\Hb_0 + \widetilde{\Vb}\Hb_0 - \Vb$, and we consider the term $\widetilde{\Vb}\Hb_0 - \Vb$ first.

By Lemma 8 in \citet{fandistributed2019}, we have that $\|\widetilde{\Vb} \Hb_0 - \Vb - \Pb_{\perp} (\widetilde{\mathbf{\Sigma}} - \Vb \Vb^{\top})\Vb\|_2 \lesssim \|\widetilde{\mathbf{\Sigma}} - \Vb \Vb^{\top}\|_2\|\Pb_{\perp} (\widetilde{\mathbf{\Sigma}} - \Vb \Vb^{\top})\Vb\|_2$. Note that in Lemma 8 of \citet{fandistributed2019}, the norm is Frobenius norm rather than operator norm, and the modification from Frobenius norm to operator norm is trivial and hence omitted. We first study the leading term $\Pb_{\perp} (\widetilde{\mathbf{\Sigma}} - \Vb \Vb^{\top})\Vb = \frac{1}{L}\sum_{\ell =1}^L \Pb_{\perp} (\widehat{\Vb}^{(\ell)}\widehat{\Vb}^{(\ell) \top} - \Vb \Vb^{\top})\Vb$. 

For a given $\ell \in [L]$, we know that $ \widehat{\Vb}^{(\ell)}$ is the top $K$ left singular vectors of  $\widehat{\Yb}^{(\ell)} = \widehat\Mb \mathbf{\Omega}^{(\ell)}/\sqrt{p}= \Vb \mathbf{\Lambda} \Vb^{\top}\mathbf{\Omega}^{(\ell)}/\sqrt{p} + \Eb \mathbf{\Omega}^{(\ell)}/\sqrt{p} = {\Yb}^{(\ell)} + \mathbf{\cE}^{(\ell)} $, where 
$$\Yb^{(\ell)} = \Vb \mathbf{\Lambda} \Vb^{\top}\mathbf{\Omega}^{(\ell)}/\sqrt{p}\quad\text{and}\quad\mathbf{\cE}^{(\ell)}  = \Eb \mathbf{\Omega}^{(\ell)}/\sqrt{p}.$$
By the ``symmetric dilation'' trick, we denote $$\cS(\widehat{\Yb}^{(\ell)}) = \begin{pmatrix} \mathbf{0} & \widehat{\Yb}^{(\ell)} \\ \widehat{\Yb}^{(\ell)\top} & \mathbf{0} \end{pmatrix}, \quad \cS({\Yb}^{(\ell)}) = \begin{pmatrix} \mathbf{0} & {\Yb}^{(\ell)} \\ {\Yb}^{(\ell)\top} & \mathbf{0} \end{pmatrix},$$
 $$\text{and} \quad \cS(\mathbf{\cE}^{(\ell)}) = \cS(\widehat{\Yb}^{(\ell)}) - \cS({\Yb}^{(\ell)}) = \begin{pmatrix} \mathbf{0} &  \Eb \mathbf{\Omega}^{(\ell)}/\sqrt{p} \\ \mathbf{\Omega}^{(\ell)\top} \Eb/\sqrt{p} & \mathbf{0}\end{pmatrix}.$$
We let $\mathbf{\Gamma}_K^{(\ell)} {\mathbf{\Lambda}}_K^{(\ell)} {\Ub}_K^{(\ell)\top}$ be the SVD of ${\Yb}^{(\ell)}$, and we know that with probability 1 we have $\mathbf{\Gamma}_K^{(\ell)} = \Vb \Ob_{\mathbf{\Omega}^{(\ell)}}$, where $\Ob_{\mathbf{\Omega}^{(\ell)}}$ is an orthonormal matrix depending on $\mathbf{\Omega}^{(\ell)}$. It is not hard to verify that the eigen-decomposition of $\cS({\Yb}^{(\ell)})$ is:
\[
\cS({\Yb}^{(\ell)}) = \frac{1}{\sqrt{2}} \begin{pmatrix}\mathbf{\Gamma}_K^{(\ell)} & \mathbf{\Gamma}_K^{(\ell)} \\ {\Ub}_K^{(\ell)} & -{\Ub}_K^{(\ell)} \end{pmatrix} \cdot \begin{pmatrix} {\mathbf{\Lambda}}_K^{(\ell)} & \mathbf{0}\\ \mathbf{0} & -{\mathbf{\Lambda}}_K^{(\ell)}  \end{pmatrix} \cdot  \frac{1}{\sqrt{2}} \begin{pmatrix}\mathbf{\Gamma}_K^{(\ell)} & \mathbf{\Gamma}_K^{(\ell)} \\ {\Ub}_K^{(\ell)} & -{\Ub}_K^{(\ell)} \end{pmatrix}^{\top},
\]
where ${\mathbf{\Lambda}}_K^{(\ell)} = \operatorname{diag}(\lambda_1^{(\ell)}, \ldots, \lambda_K^{(\ell)})$. First we study the eigengap $\sigma_{\min}({\mathbf{\Lambda}}_K^{(\ell)}) = \lambda_K^{(\ell)}$. Recall $\widetilde{\mathbf{\Omega}}^{(\ell)} = \Vb^{\top} \mathbf{\Omega}^{(\ell)} \in \RR^{K \times p}$, and it can be seen that the entries of $\widetilde{\mathbf{\Omega}}^{(\ell)}$ are i.i.d. standard Gaussian. By Lemma 3 in \citet{fandistributed2019}, we know that with probability at least $1-d^{-10}$, we have that $\|\widetilde{\mathbf{\Omega}}^{(\ell)} \widetilde{\mathbf{\Omega}}^{(\ell)\top} / p - \Ib_K\|_2 \lesssim \sqrt{\frac{K}{p}}\log d$, and thus $\sigma_{\min}(\widetilde{\mathbf{\Omega}}^{(\ell)} / \sqrt{p}) \ge 1-O(\sqrt{\frac{K}{p}}\log d)$ with probability at least $1-d^{-10}$. Thus under the condition that $\sqrt{\frac{K}{p}}\log d = o(1)$, under the same high probability event we have that $\sigma_{\min}({\mathbf{\Lambda}}_K^{(\ell)}) \ge \Delta/2$. Now we let $\widehat{\Ub}_K^{(\ell)}$ be the top $K$ right singular vectors of $\widehat\Yb^{(\ell)}$. For $j \in [K]$ we define 
$$
\Gb_j^{(\ell)} \!\!=\! \frac{1}{2}\!\! \begin{pmatrix} \mathbf{\Gamma}_K^{(\ell)} \\-\Ub_K^{(\ell)} \end{pmatrix}\!\! (-\mathbf{\Lambda}_K^{(\ell)} - \lambda_j^{(\ell)} \Ib_K)^{-1} \!\!\begin{pmatrix} \mathbf{\Gamma}_K^{(\ell)} \\-\Ub_K^{(\ell)} \end{pmatrix}^{\!\!\!\top}\!\!\! - \frac{1}{\lambda_j^{(\ell)}}\bigg\{\Ib_K - \frac{1}{2}\begin{pmatrix}\mathbf{\Gamma}_K^{(\ell)} & \mathbf{\Gamma}_K^{(\ell)} \\ {\Ub}_K^{(\ell)} & -{\Ub}_K^{(\ell)} \end{pmatrix} \!\!\!\begin{pmatrix}\mathbf{\Gamma}_K^{(\ell)} & \mathbf{\Gamma}_K^{(\ell)} \\ {\Ub}_K^{(\ell)} & -{\Ub}_K^{(\ell)} \end{pmatrix}^{\!\!\!\top}\!\!\bigg\}.
$$ 
Then we have $\|\Gb_j^{(\ell)}\|_2 \le 1/\lambda_K^{(\ell)} \le 2/\Delta$ with probability at least $1-d^{-10}$. Correspondingly we define the linear mapping 
$$
f: \mathbb{R}^{(d+p) \times K} \rightarrow \mathbb{R}^{(d+p) \times K}, \quad\left(\mathbf{w}_{1}, \cdots, \mathbf{w}_{K}\right) \mapsto\left(-\mathbf{G}_{1}^{(\ell)} \mathbf{w}_{1}, \cdots,-\mathbf{G}_{K}^{(\ell)}{\mathbf{w}}_{K}\right),
$$ 
and denote $\tilde{\mathbf{\Gamma}}_K^{(\ell)} = \begin{pmatrix} \mathbf{\Gamma}_K^{(\ell)} \\\Ub_K^{(\ell)} \end{pmatrix}$.  By Lemma 8 in \citet{fandistributed2019}, under the condition that $\|\cS(\cE^{(\ell)})\|_2 /\Delta =o(1)$ we have 
\begin{align*}
    &\bigg\|\begin{pmatrix}\widehat{\Vb}^{(\ell)}\\\widehat{\Ub}_K^{(\ell)}\end{pmatrix}(\widehat{\Vb}^{(\ell)\top},\widehat{\Ub}_K^{(\ell)\top}) -  \tilde{\mathbf{\Gamma}}_K^{(\ell)} \tilde{\mathbf{\Gamma}}_K ^{(\ell)\top} - f\big(\cS(\cE^{(\ell)})\tilde{\mathbf{\Gamma}}_K^{(\ell)} \big)\tilde{\mathbf{\Gamma}}_K^{(\ell)\top}  -\tilde{\mathbf{\Gamma}}_K^{(\ell)} f\big(\cS(\cE^{(\ell)})\tilde{\mathbf{\Gamma}}_K^{(\ell)} \big)^{\top} \bigg \|_2 \\
    &\le \bigg\| \begin{pmatrix}\widehat{\Vb}^{(\ell)}\widehat{\Vb}^{(\ell)\top} - \mathbf{\Gamma}_K^{(\ell)}\mathbf{\Gamma}_K^{(\ell)\top} &\quad \widehat{\Vb}^{(\ell)}\widehat{\Ub}_K^{(\ell)\top} - \mathbf{\Gamma}_K^{(\ell)}\Ub_K^{(\ell)\top}\\\widehat{\Ub}_K^{(\ell)}\widehat{\Vb}^{(\ell)\top} - \Ub_K^{(\ell)}\mathbf{\Gamma}_K^{(\ell)\top}&\quad \widehat{\Ub}_K^{(\ell)}\widehat{\Ub}_K^{(\ell)\top} - \Ub_K^{(\ell)}\Ub_K^{(\ell)\top}\end{pmatrix} \\
    &\quad - f\big(\cS(\cE^{(\ell)})\tilde{\mathbf{\Gamma}}_K^{(\ell)}\big)\tilde{\mathbf{\Gamma}}_K ^{(\ell)\top} -\tilde{\mathbf{\Gamma}}_K^{(\ell)}  f\big(\cS(\cE^{(\ell)})\tilde{\mathbf{\Gamma}}_K^{(\ell)} \big)^{\top} \bigg\|_2 \lesssim \|\cS(\cE^{(\ell)})\|_2^2/\Delta^2.
\end{align*}
 By taking the upper left block of the matrix, we have
\begin{align*}
    &\big\| \widehat{\Vb}^{(\ell)}\widehat{\Vb}^{(\ell)\top} -  \mathbf{\Gamma}_K^{(\ell)}\mathbf{\Gamma}_K^{(\ell)\top} - f\big(\cS(\cE^{(\ell)})\tilde{\mathbf{\Gamma}}_K^{(\ell)} \big)_{[1:d,:]} \mathbf{\Gamma}_K^{(\ell)\top} - \mathbf{\Gamma}_K^{(\ell)}f\big(\cS(\cE^{(\ell)})\tilde{\mathbf{\Gamma}}_K^{(\ell)} \big)_{[1:d,:]}^{\top} \big\|_2 \\
    &= \big\| \widehat{\Vb}^{(\ell)}\widehat{\Vb}^{(\ell)\top} -  \Vb\Vb^{\top} - f\big(\cS(\cE^{(\ell)})\tilde{\mathbf{\Gamma}}_K^{(\ell)} \big)_{[1:d,:]} \mathbf{\Gamma}_K^{(\ell)\top} - \mathbf{\Gamma}_K^{(\ell)}f\big(\cS(\cE^{(\ell)})\tilde{\mathbf{\Gamma}}_K^{(\ell)} \big)_{[1:d,:]}^{\top} \big\|_2\\
    &\lesssim \|\cS(\cE^{(\ell)})\|_2^2/\Delta^2.
\end{align*}
Now for $j \in [K]$, we study $\Pb_{\perp} (\Gb_j^{(\ell)})_{[1:d,:]}$. Since $\mathbf{\Gamma}_K^{(\ell)} = \Vb \Ob_{\mathbf{\Omega}^{(\ell)}}$, we have $\Pb_{\perp} \mathbf{\Gamma}_K^{(\ell)} = \mathbf{0}$. Therefore we have,
\begin{align*}
&\Pb_{\perp} \mathbf{\Gamma}_K^{(\ell)}(-\mathbf{\Lambda}_K^{(\ell)} - \lambda_j^{(\ell)} \Ib_K)^{-1} \begin{pmatrix} \mathbf{\Gamma}_K^{(\ell)} \\-\Ub_K^{(\ell)} \end{pmatrix}^{\top} = \mathbf{0}, \quad \text{and}\\
& \Pb_{\perp}\bigg\{\Ib_{d+p} - \frac{1}{2}\begin{pmatrix}\mathbf{\Gamma}_K^{(\ell)} & \mathbf{\Gamma}_K^{(\ell)} \\ {\Ub}_K^{(\ell)} & -{\Ub}_K^{(\ell)} \end{pmatrix} \!\!\!\begin{pmatrix}\mathbf{\Gamma}_K^{(\ell)} & \mathbf{\Gamma}_K^{(\ell)} \\ {\Ub}_K^{(\ell)} & -{\Ub}_K^{(\ell)} \end{pmatrix}^{\top}\!\bigg\}_{[1:d,:]} \\
&= (\Pb_{\perp},\mathbf{0})- \frac{1}{2}\Pb_{\perp} \mathbf{\Gamma}_K^{(\ell)}(\Ib_d,\Ib_d) \begin{pmatrix}\mathbf{\Gamma}_K^{(\ell)} & \mathbf{\Gamma}_K^{(\ell)} \\ {\Ub}_K^{(\ell)} & -{\Ub}_K^{(\ell)} \end{pmatrix}^{\top}\\
& = (\Pb_{\perp},\mathbf{0}) + \mathbf{0} = (\Pb_{\perp},\mathbf{0}),
\end{align*}
and as a result we have
$$
\Pb_{\perp} (\Gb_j)_{[1:d,:]} = \frac{1}{2}\cdot \mathbf{0}- \frac{1}{\lambda_j^{(\ell)}}\left\{(\Pb_{\perp},\mathbf{0})-\mathbf{0}\right\} =-\frac{1}{\lambda_j^{(\ell)}}(\Pb_{\perp},\mathbf{0}).
$$
Thus in turn,
\begin{align*}
    &\Pb_{\perp} \Big(f\big(\cS(\cE^{(\ell)})\tilde{\mathbf{\Gamma}}_K^{(\ell)} \big)_{[1:d,:]} \mathbf{\Gamma}_K^{(\ell)\top} + \mathbf{\Gamma}_K^{(\ell)}f\big(\cS(\cE^{(\ell)})\tilde{\mathbf{\Gamma}}_K^{(\ell)} \big)_{[1:d,:]}^{\top} \Big) = \Pb_{\perp} f\big(\cS(\cE^{(\ell)})\tilde{\mathbf{\Gamma}}_K^{(\ell)} \big)_{[1:d,:]} \mathbf{\Gamma}_K^{(\ell)\top}\\
    & \quad = (\Pb_{\perp},\mathbf{0})\begin{pmatrix} \mathbf{0} &  \Eb \mathbf{\Omega}^{(\ell)}/\sqrt{p} \\ \mathbf{\Omega}^{(\ell)\top} \Eb/\sqrt{p} & \mathbf{0}\end{pmatrix}\begin{pmatrix} \mathbf{\Gamma}_K^{(\ell)} \\\Ub_K^{(\ell)} \end{pmatrix} (\mathbf{\Lambda}_K^{(\ell)})^{-1}\mathbf{\Gamma}_K^{(\ell)\top}\\
    & \quad = \Pb_{\perp} \Eb (\mathbf{\Omega}^{(\ell)}/\sqrt{p}) \Ub_K^{(\ell)} (\mathbf{\Lambda}_K^{(\ell)})^{-1}\mathbf{\Gamma}_K^{(\ell)\top} = \Pb_{\perp} \Eb (\mathbf{\Omega}^{(\ell)}/\sqrt{p}) (\Yb^{(\ell)})^{\dagger}.
\end{align*}
For a given $\ell \in [L]$, under the condition that $\sqrt{p/d} \log d = O(1)$, by Lemma 3 in \citet{fandistributed2019} we have that with probability at least $1-d^{-10}$, $\|\mathbf{\Omega}^{(\ell)}\|_2 \lesssim \sqrt{d}$. Combined with previous results on the eigengap $\sigma_{\min}(\mathbf{\Lambda}_K^{(\ell)})$, we have that with probability $1-O(d^{-9})$, for a fixed constant $C >0$  
$$\|\mathbf{\Omega}^{(\ell)}\|_2 \le  C\sqrt{d}, \quad  \sigma_{\min}(\mathbf{\Lambda}_K^{(\ell)}) \ge  \Delta/2, \quad  \forall \ell \in [L].$$
Besides, under Assumption \ref{asp: tail prob bound}, we have that  $\|\Eb\|_2 \lesssim r_1(d)\log d$ with probability at least $1- d^{-10}$, and in turn by Wedin's Theorem \citep{wedin1972wedin}, with high probability for all $\ell \in [L]$ we have that $$\|\widehat{\Vb}^{(\ell)}\widehat{\Vb}^{(\ell)\top} - \Vb \Vb^{\top}\|_2 \lesssim \|\cE^{(\ell)}\|_2 / \sigma_{\min}(\mathbf{\Lambda}_K^{(\ell)}) \lesssim \|\Eb\|_2 \|\mathbf{\Omega}^{(\ell)}/\sqrt{p}\|_2 /\Delta \lesssim \frac{r_1(d)}{\Delta}\log d \sqrt{\frac{d}{p}},$$
and thus $\|\widetilde{\mathbf{\Sigma}} - \Vb \Vb^{\top}\|_2 = O_P\big( r_1(d) \log d  \sqrt{d/p} /\Delta \big)$.
Besides, we have
\begin{align*}
    & \Pb_{\perp} (\widetilde{\mathbf{\Sigma}} - \Vb \Vb^{\top})\Vb = \frac{1}{L}\sum_{\ell =1}^L \Pb_{\perp} (\widehat{\Vb}^{(\ell)}\widehat{\Vb}^{(\ell) \top} - \Vb \Vb^{\top})\Vb \\
    & \quad = \frac{1}{L}\sum_{\ell =1}^L \Pb_{\perp} \Big(f\big(\cS(\cE^{(\ell)})\tilde{\mathbf{\Gamma}}_K^{(\ell)} \big)_{[1:d,:]} \mathbf{\Gamma}_K^{(\ell)\top} + \mathbf{\Gamma}_K^{(\ell)}f\big(\cS(\cE^{(\ell)})\tilde{\mathbf{\Gamma}}_K^{(\ell)} \big)_{[1:d,:]}^{\top} \Big)\Vb + \Rb_1 (\widetilde{\mathbf{\Sigma}})\\
    & \quad = \frac{1}{L}\sum_{\ell =1}^L \Pb_{\perp} \Eb (\mathbf{\Omega}^{(\ell)}/\sqrt{p}) (\Yb^{(\ell) })^{\dagger}\Vb + R_1 (\widetilde{\mathbf{\Sigma}}) = \frac{1}{L}\sum_{\ell =1}^L \Pb_{\perp} \Eb (\mathbf{\Omega}^{(\ell)}/\sqrt{p})\Bb^{(\ell) \top} + \Rb_1 (\widetilde{\mathbf{\Sigma}})\\
    & \quad = \frac{1}{L}\Pb_{\perp} \Eb \mathbf{\Omega} \Bb_{\bOmega} + \Rb_1(\widetilde{\mathbf{\Sigma}}),
\end{align*}
where $\Rb_1(\widetilde{\mathbf{\Sigma}})$ is the residual matrix with $\|\Rb_1(\widetilde{\mathbf{\Sigma}})\|_2 = O_P(\|\cS(\cE^{(\ell)})\|_2^2/\Delta^2)$. Now we study the matrix $\Bb^{(\ell)} = (\mathbf{\Lambda}\Vb^{\top} \mathbf{\Omega}^{(\ell)}/\sqrt{p})^{\dagger}$. From previous results we know that with probability at least $1-d^{-9}$,  $1/2 \le \sigma_{\min}(\widetilde{\mathbf{\Omega}}^{(\ell)} / \sqrt{p}) \le \sigma_{\max}(\widetilde{\mathbf{\Omega}}^{(\ell)} / \sqrt{p}) \le 3/2 $ for any $\ell \in [L]$, and in turn $\frac{2}{3|\lambda_1|}\le \sigma_{\min} (\Bb^{(\ell)}) \le \sigma_{\max} (\Bb^{(\ell)}) \le \frac{2}{\Delta}, \quad \forall \ell \in [L]$. Now for any vector $\yb \in \RR^{K}$ such that $\|\yb\|_2=1$, with probability $1-O(d^{-9})$ we have that 
\begin{align*}
    \|\Bb_{\mathbf{\Omega}} \yb \|_2  &= \|(\yb^{\top} \Bb^{(1)\top}, \ldots, \yb^{\top} \Bb^{(L)\top})^{\top}\|_2 = \Big(\sum_{\ell = 1}^L \|\Bb^{(\ell)} \yb\|_2^2 \Big)^{1/2},\\
    \|\Bb_{\mathbf{\Omega}}\|_2 &= \max_{\|\yb\|_2 = 1} \|\Bb_{\mathbf{\Omega}} \yb \|_2 = \max_{\|\yb\|_2 = 1} \Big(\sum_{\ell = 1}^L \|\Bb^{(\ell)} \yb\|_2^2 \Big)^{1/2} \le \Big(\sum_{\ell = 1}^L \|\Bb^{(\ell)}\|_2^2 \Big)^{1/2} \le \frac{2 \sqrt{L}}{\Delta},\\
    \sigma_{\min}\left(\Bb_{\mathbf{\Omega}} \right)&= \min_{\|\yb\|_2 = 1} \|\Bb_{\mathbf{\Omega}} \yb \|_2 = \min_{\|\yb\|_2 = 1} \Big(\sum_{\ell = 1}^L \|\Bb^{(\ell)} \yb\|_2^2 \Big)^{1/2} \ge \Big(\sum_{\ell = 1}^L \sigma_{\min}^2(\Bb^{(\ell)}) \Big)^{1/2} \ge \frac{2 \sqrt{L}}{3|\lambda_1|}.
\end{align*}
Now since we know that the entries of $\sqrt{p} \mathbf{\Omega}$ are i.i.d. standard Gaussian, similar as before, under the condition that $Lp \ll d$, by Lemma 3 in \citet{fandistributed2019} we have with high probability that $\frac{1}{2}\sqrt{\frac{d}{p}} \le \sigma_{\min} (\mathbf{\Omega}) \le \sigma_{\max} (\mathbf{\Omega}) \le \frac{3}{2}\sqrt{\frac{d}{p}}$. Therefore, we have the following upper bound on the norm of the leading term
\begin{align*}
   \| \Pb_{\perp} (\widetilde{\mathbf{\Sigma}} - \Vb \Vb^{\top})\Vb \|_2 & \lesssim \|\frac{1}{L}\Pb_{\perp} \Eb \mathbf{\Omega} \Bb_{\bOmega} \|_2+ \|\Rb_1(\widetilde{\mathbf{\Sigma}})\|_2 \le \frac{1}{L} \|\Eb\|_2 \|\mathbf{\Omega}\|_2\|\Bb_{\mathbf{\Omega}}\|_2 + \|\Rb_1(\widetilde{\mathbf{\Sigma}})\|_2 \\
   & = O_P\Big( \sqrt{\frac{d}{Lp}}\frac{r_1(d)\log d}{\Delta} + r_1(d)^2(\log d)^2 \frac{d}{p\Delta^2}\Big).
\end{align*}
Thus we have the following decomposition 
\begin{align*}
    \widetilde{\Vb} \Hb_0 - \Vb &=  \Pb_{\perp} (\widetilde{\mathbf{\Sigma}} - \Vb \Vb^{\top})\Vb + \Rb_0(\widetilde{\mathbf{\Sigma}})\\
    & =  \frac{1}{L}\Pb_{\perp} \Eb \mathbf{\Omega} \Bb_{\bOmega} + \Rb_1(\widetilde{\mathbf{\Sigma}}) + \Rb_0(\widetilde{\mathbf{\Sigma}})\\
    & =  \frac{1}{L}\Pb_{\perp} \Eb_0 \mathbf{\Omega} \Bb_{\bOmega} + \frac{1}{L}\Pb_{\perp} \Eb_b \mathbf{\Omega} \Bb_{\bOmega} + \Rb_1(\widetilde{\mathbf{\Sigma}})+ \Rb_0(\widetilde{\mathbf{\Sigma}}),
\end{align*}
where $\Rb_0(\widetilde{\mathbf{\Sigma}})$ is a residual matrix with 
\begin{align*}
    \|\Rb_0(\widetilde{\mathbf{\Sigma}})\|_2 &= O_P(\|\widetilde{\mathbf{\Sigma}} - \Vb \Vb^{\top}\|_2\|\Pb_{\perp} (\widetilde{\mathbf{\Sigma}} - \Vb \Vb^{\top})\Vb\|_2) \\
    &= O_P\Big(\frac{r_1(d)^2 (\log d)^2 {d }}{\sqrt{L}p\Delta^2}\Big )+ o_P\Big(r_1(d)^2 (\log d)^2 \frac{d}{p\Delta^2}\Big ).
\end{align*}
 Thus 
$$
\|\Rb_0(\widetilde{\mathbf{\Sigma}})+\Rb_1(\widetilde{\mathbf{\Sigma}})\|_2 = O_P\Big(r_1(d)^2 (\log d)^2 \frac{d}{p\Delta^2}  \Big).
$$

Next we consider the term $\widetilde{\Vb}^{\F} \Hb - \tilde{\Vb}\Hb_0$. We denote the SVD of $\widetilde{\bSigma}^q$ by $\widetilde{\Vb} \widetilde{\mathbf{\Lambda}}_K^q \widetilde{\Vb}^{\top} + \widetilde{\Vb}_{\perp} \widetilde{\mathbf{\Lambda}}_{\perp}^q \widetilde{\Vb}_{\perp}^{\top}$, and by Weyl's inequality \citep{franklin2012matrix}, we know that $\|\widetilde{\mathbf{\Lambda}}_{\perp}\|_2 \le \|\widetilde{\mathbf{\Sigma}} - \Vb \Vb^{\top}\|_2 =O_P\big( r_1(d)\log d \sqrt{d/p}/\Delta \big)$ and $\sigma_K(\widetilde{\mathbf{\Lambda}}_K) \ge 1 - \|\widetilde{\mathbf{\Sigma}} - \Vb \Vb^{\top}\|_2 \ge 1 - O_P(r_1(d)\log d \sqrt{d/p}/\Delta)$. Thus under the condition that $r_1(d) \log d \sqrt{d/p} /\Delta = o(1)$, for large enough $d$ with high probability we have 
$$\|\widetilde{\mathbf{\Lambda}}_{\perp}^q\|_2 \le (r_1(d)\log d \sqrt{d/p}/\Delta)^q \quad\text{and}\quad \sigma_K(\widetilde{\mathbf{\Lambda}}_K^q) \ge (1- O(r_1(d)\log d \sqrt{d/p}/\Delta))^q \ge (1/2)^q.$$ Similar as before, we know that with probability 1 the left singular vector space of $\widetilde{\Vb} \widetilde{\mathbf{\Lambda}}_K^q \widetilde{\Vb}^{\top} \mathbf{\Omega}^{\F}  = \widetilde{\Vb} \widetilde{\mathbf{\Lambda}}_K^q \widetilde{\mathbf{\Omega}}^{\F}$ and the column space of $\widetilde{\Vb}$ are the same, where $\widetilde{\mathbf{\Omega}}^{\F} := \widetilde{\Vb}^{\top} \mathbf{\Omega}^{\F} \in \RR^{K \times p'}$ is still a Gaussian test matrix with i.i.d. entries. By Lemma 3 in \citet{fandistributed2019}, we have with probability at least $1-d^{-10}$,  $ \sigma_{\min}(\widetilde{\mathbf{\Omega}}^{\F} / \sqrt{p'}) \ge 1-O(\sqrt{\frac{K}{p'}}\log d) $. When $ \sqrt{\frac{K}{p'}}\log d = o(1)$, by Wedin's Theorem \citep{wedin1972wedin}, there exists a constant $\eta>0$ such that with high probability we have
\begin{align*}
    \|\widetilde{\Vb}^{\F} \Hb - \tilde{\Vb}\Hb_0\|_2 & = \|\widetilde{\Vb}^{\F} \Hb_1 - \tilde{\Vb}\|_2 \lesssim \|\widetilde{\Vb}_{\perp} \widetilde{\mathbf{\Lambda}}_{\perp}^q \widetilde{\Vb}_{\perp}^{\top}{\mathbf{\Omega}}^{\F}/\sqrt{p'} \|_2 / \sigma_K (\widetilde{\Vb} \widetilde{\mathbf{\Lambda}}_K^q \widetilde{\mathbf{\Omega}}^{\F}/\sqrt{p'} )\\
    & \le \frac{\|\widetilde{\mathbf{\Lambda}}_{\perp}\|_2^q\|{\mathbf{\Omega}}^{\F}/\sqrt{p'}\|_2}{\sigma_K (\widetilde{\mathbf{\Lambda}}_K^q)\sigma_K( \widetilde{\mathbf{\Omega}}^{\F}/\sqrt{p'} )} \lesssim \left(\frac{2\eta r_1(d)\log d \sqrt{d/p}}{\Delta}\right)^q \sqrt{\frac{d}{p'}}.
\end{align*}
Denote  $r' :=  2\eta r_1(d)\log d \sqrt{d/p}/\Delta = o\big((\log d)^{-1/4}\big)$. Then it can be seen that when 
$$
q \ge \log d \gg 2 + \frac{\log d}{\log\log d} \ge  2 + \frac{\log \sqrt{d/p'}}{\log (1/r')},
$$
we have that $(r')^q \sqrt{d/p'} = o\big((r')^2\big)$ and 
$\|\widetilde{\Vb}^{\F} \Hb - \tilde{\Vb}\Hb_0\|_2 = O_P\Big(r_1(d)^2(\log d)^2 \frac{d}{p\Delta^2}  \Big)$. 

Now for a given $j \in [d]$, recall that with high probability  $\sigma_{\min}(\mathbf{\Sigma}_j) = \Omega\big(\eta_2(d)\Big)$. 
Therefore, under the condition that ${d^2 r_1(d)^4(\log d)^4}\big({p^2 \Delta^4 \eta_2(d)}\big)^{-1} = o(1)$ and ${d r_2(d)^2}\big({Lp \Delta^2 \eta_2(d)}\big)^{-1} = o(1)$, we have with probability $1-O(d^{-9})$, $\|\frac{1}{L}\Pb_{\perp} \Eb_b \mathbf{\Omega} \Bb_{\bOmega}\|_2 = O_P\Big(\sqrt{\frac{d}{\Delta^2Lp}}r_2(d)\Big) = o_P\big((\sigma_{\min}(\mathbf{\Sigma}_j))^{1/2}\big)$, and $\|\Rb_0(\widetilde\bSigma) + \Rb_1(\widetilde\bSigma)\|_2 = o_P\big((\sigma_{\min}(\mathbf{\Sigma}_j))^{1/2}\big)$. Then under Assumption \ref{asp: clt}, we have 
\begin{align*}
    &\mathbf{\Sigma}_j^{-1/2}(\widetilde{\Vb}^{\F} \Hb - \Vb)^{\top} \eb_j = \mathbf{\Sigma}_j^{-1/2}(\widetilde{\Vb}^{\F} \Hb - \widetilde\Vb\Hb_0 + \widetilde\Vb\Hb_0  - \Vb)^{\top} \eb_j\\
    & = \mathbf{\Sigma}_j^{-1/2}(\frac{1}{L}\Bb_{\mathbf{\Omega}}^{\top}\mathbf{\Omega}^{\top} \Eb_0 \Pb_{\perp}\eb_j)\!+\! \mathbf{\Sigma}_j^{-1/2}(\widetilde{\Vb}^{\F} \Hb \!-\! \widetilde\Vb\Hb_0 \!+\!\Rb_0(\widetilde\bSigma) \!+\! \Rb_1(\widetilde\bSigma)\! +\! \frac{1}{L}\Pb_{\perp} \Eb_b \mathbf{\Omega} \Bb_{\bOmega} )^{\top}\eb_j\\
    &=\mathbf{\Sigma}_j^{-1/2}\cV (\Eb_0)^{\top}\eb_j + o_P(1) \overset{d}{\rightarrow} {\cN}(\mathbf{0}, \Ib_K).
\end{align*}

\subsection{Proof of Corollary~\ref{col: guassian case}}\label{sec: proof col gaussian case}
To prove Corollary~\ref{col: guassian case}, it suffices for us to show that Assumptions \ref{asp: tail prob bound}, \ref{asp: stat rate biased error} and \ref{asp: clt} are met. From the proof of Corollary~\ref{prop: err rate terms}, we know that Assumption \ref{asp: tail prob bound} is satisfied. We move on to show that Assumption \ref{asp: stat rate biased error} is met. 
Define $\Vb_d = (\Vb, \Vb^{\perp})$ as the stacking of eigenvectors for the covariance matrix $\mathbf{\Sigma}$. Note that $\Vb^{\perp}$ is not identifiable under the spiked covariance model and is unique up to orthogonal transformation. Let $\bZ_i = \Vb_d^{\top} \bX_i$, and $\bZ_i \sim {\cN}(\mathbf{0}, \mathbf{\Lambda}_d)$, where $\mathbf{\Lambda}_d = \operatorname{diag}(\mathbf{\Lambda} + \sigma^2 \Ib_K, \sigma^2\Ib_{d-K})$. We let $\mathbf{\Gamma}_S = (\ub_1,\ldots,\ub_{K+1})$ be the stacking of eigenvectors for the matrix $\bSigma_{S}$, and let $\tilde{\sigma}_1 \ge \tilde{\sigma}_2 \ge \ldots \ge \tilde{\sigma}_{K+1}$ be the $K+1$ eigenvalues of $\bSigma_S$. Correspondingly, let $\hat{\sigma}_1 \ge \ldots \ge \hat\sigma_{K+1} = \hat{\sigma}^2$ be the eigenvalues of the sample covariance matrix $\widehat{\bSigma}_{S}$. Since $\bSigma_S = (\Vb)_{[S,:]}\bLambda(\Vb)_{[S,:]}^{\top} + \sigma^2 \Ib_{K+1}$, we know that $\tilde{\sigma}_{K+1} = \sigma^2$ and $\delta = \tilde{\sigma}_{K} - \tilde{\sigma}_{K+1} \ge \Delta\sigma_{\min}^2 \big((\Vb)_{[S,:]}\big)$. We define $\tilde\cb = (\Vb^{\perp}_{[S,:]})^{\top}\ub_{K+1}$, and denote $\tilde{\cb}_0 = (\mathbf{0},\Ib_{d-K})^{\top}\tilde{\cb} \in \RR^d$.  Then by the proof of Lemma 6.2 in \citet{fan2017eigenasymp}, we know that 
$$
\hat{\sigma}^2 - \sigma^2 =\tilde{\cb}_0^{\top}(\frac{1}{n}\sum_{i=1}^n \bZ_i \bZ_i^{\top} - \mathbf{\Lambda}_d )\tilde{\cb}_0 + \frac{1}{n}O_P\left(M_{K+1}-\sigma^2 W_{K+1}\right),
$$
where $M_{K+1}=\sum_{k \le K} f_{k}^{2}\left(\tilde\sigma_{k}+ (\hat{\sigma}_k - \tilde\sigma_{k})\right), W_{K+1}=\sum_{k \le K} f_{ k}^{2}$ and $f_{ k}$ is the $(K+1)$-th element of the $k$-th eigenvector of $\mathbf{\Gamma}_S^{\top} \hbSigma_S \mathbf{\Gamma}_S $ multiplied by $\sqrt{n}$ for $k \le K$. We let $\fb = (f_1, \ldots, f_K)^{\top}/\sqrt{n}$. By Wedin's Theorem \citep{wedin1972wedin} and Lemma 3 in \citet{fandistributed2019}, we have that with probability at least $1-d^{-10}$, $|\hat{\sigma}_k - \tilde{\sigma}_k| \le \|\widehat{\bSigma}_S - \bSigma_S\|_2 \lesssim \tilde{\sigma}_1 \log d \sqrt{\frac{K}{n}}$ for $k \le K$. If we denote by $\Fb_S := (\Ib_K, \mathbf{0})^{\top}$ the stacked top $K$ eigenvectors of $\mathbf{\Gamma}_S^{\top} \bSigma_S \mathbf{\Gamma}_S $, and by $\widehat\Fb_S$ the stacked top $K$ eigenvectors of $\mathbf{\Gamma}_S^{\top} \hbSigma_S \mathbf{\Gamma}_S $, then we know that $\fb$ is the $(K+1)$-th row of $\widehat\Fb_S$. By Davis-Kahan’s Theorem \citep{yu2015davis}, we also know that there exists an orthonormal matrix $\Ob_S \in \RR^{K \times K}$ such that $\|\fb\|_2 = \|\Ob_S^{\top} \fb -\mathbf{0}\|_2 \le \|\widehat\Fb_S\Ob_S  - \Fb_S \|_2 \lesssim \frac{\tilde{\sigma}_1\log d}{\delta}\sqrt{\frac{K}{n}}$, and thus 
$$W_{K+1} = \sum_{k \le K} f_{ k}^{2} = n\|\fb\|_2^2 \lesssim \frac{\tilde{\sigma}_1^2K (\log d)^2}{\delta^2 },$$
$$\text{and} \quad M_{K+1} \le \tilde{\sigma}_1 \sum_{k \le K} f_{ k}^{2} + (\sum_{k \le K} f_{ k}^{2})\|\widehat{\bSigma}_S - \bSigma_S\|_2 \lesssim \frac{\tilde{\sigma}_1^3K}{\delta^2 }(\log d)^2 .$$ 
Thus we can write $\hat{\sigma}^2 - \sigma^2 = \tilde{\cb}_0^{\top}(\frac{1}{n}\sum_{i=1}^n \bZ_i \bZ_i^{\top} - \mathbf{\Lambda}_d )\tilde{\cb}_0 + O_P\big(\frac{\tilde{\sigma}_1^3K}{\delta^2 n} (\log d)^2\big)$. 

 Now we take $\Eb_0 = \widehat{\mathbf{\Sigma}} - \mathbf{\Sigma} -(\tilde{\cb}_0^{\top}(\frac{1}{n}\sum_{i=1}^n \bZ_i \bZ_i^{\top} - \mathbf{\Lambda}_d )\tilde{\cb}_0)\Ib_d$, and from previous results we know that with high probability $\|\Eb_b\|_2 = \|\Eb - \Eb_0\|_2 \lesssim \frac{\tilde{\sigma}_1^3K}{\delta^2 n}(\log d)^2$, such that we have $r_2(d) \asymp \frac{\tilde{\sigma}_1^3K}{\delta^2 n}(\log d)^2 = o\left(r_1(d)\right)$ and Assumption \ref{asp: stat rate biased error} is satisfied. 
 
 Now we move on to study the statistical rate $\eta_2(d)$. For any $j \in [d]$, we first study the covariance of $\Eb_0 \Pb_{\perp} \eb_j $. We denote $\widetilde{\Eb} = \bZ_1 \bZ_1^{\top} - \bLambda_d$, then it's not hard to verify that $\Cov(\widetilde{\Eb}_{st}, \widetilde{\Eb}_{gh} ) = \lambda_s(\bSigma) \lambda_t (\bSigma) (\II\{s=g, t=h\} + \II\{s=h, t=g\})$. Since $\Eb_0 \Pb_{\perp} \eb_j$ and $\Vb_d^{\top}\Eb_0 \Pb_{\perp} \eb_j$ share the same eigenvalues, we can study the covariance of $\Vb_d^{\top}\Eb_0 \Pb_{\perp} \eb_j$ instead. Then $\operatorname{Cov}(\Vb_d^{\top}\Eb_0 \Pb_{\perp} \eb_j) $ can be calculated as following
\begin{align*}
     & \operatorname{Cov}\Big\{\Vb_d^{\top}\big(\frac{1}{n}\sum_{i=1}^n \bX_i \bX_i^{\top} - \mathbf{\Sigma}  \big)\Vb^{\perp} (\Vb^{\perp})^{\top} \eb_j - \big(\tilde{\cb}_0^{\top}(\frac{1}{n}\sum_{i=1}^n \bZ_i \bZ_i^{\top} - \mathbf{\Lambda}_d )\tilde{\cb}_0\big)\Vb_d^{\top} \Pb_{\perp} \eb_j\Big\}\\
    & = \operatorname{Cov}\Big\{\Vb_d^{\top}\big(\frac{1}{n}\sum_{i=1}^n \bX_i \bX_i^{\top} - \mathbf{\Sigma}  \big)\Vb_d (\mathbf{0},\Ib_{d-K})^{\top} \tilde{\eb} - \big(\tilde{\cb}_0^{\top}(\frac{1}{n}\sum_{i=1}^n \bZ_i \bZ_i^{\top} - \mathbf{\Lambda}_d )\tilde{\cb}_0\big)\tilde{\eb}_0 \Big\}\\
    & = \operatorname{Cov}\Big\{\big(\frac{1}{n}\sum_{i=1}^n \bZ_i \bZ_i^{\top} - \mathbf{\Lambda}_d \big) \tilde{\eb}_0 - \big(\tilde{\cb}_0^{\top}(\frac{1}{n}\sum_{i=1}^n \bZ_i \bZ_i^{\top} - \mathbf{\Lambda}_d )\tilde{\cb}_0\big)\tilde{\eb}_0 \Big\},
    \end{align*}
    where $\tilde{\eb} = (\Vb^{\perp})^{\top}\eb_j$ and $\tilde{\eb}_0 = (\mathbf{0},\Ib_{d-K})^{\top}\tilde{\eb}$. Then we have
    \begin{align*}
        & \operatorname{Cov}(\Vb_d^{\top}\Eb_0 \Pb_{\perp} \eb_j)  = \frac{1}{n}\Cov(\widetilde\Eb \tilde\eb_0 - \tilde{\cb}_0^{\top}\widetilde\Eb \tilde{\cb}_0 \tilde{\eb}_0) \\
        &\quad = \frac{1}{n} \Big\{\Cov(\widetilde\Eb \tilde\eb_0 ) + \operatorname{Var}\big(\tilde{\cb}_0^{\top}\widetilde\Eb \tilde{\cb}_0\big)\tilde\eb_0 \tilde\eb_0^{\top} - \Cov(\widetilde\Eb \tilde\eb_0,\tilde{\cb}_0^{\top}\widetilde\Eb \tilde{\cb}_0)\tilde\eb_0^{\top} - \tilde\eb_0 \Cov(\widetilde\Eb \tilde\eb_0,\tilde{\cb}_0^{\top}\widetilde\Eb \tilde{\cb}_0)^{\top}\Big\}\\
        & \quad= \frac{1}{n}\{ \|\tilde{\eb}_0\|_2^2 \sigma^2 \mathbf{\Lambda}_d+3\sigma^4 \tilde{\eb}_0 \tilde{\eb}_0^{\top} -2\sigma^4\langle\tilde\cb,\tilde\eb\rangle(\tilde\cb_0\tilde\eb_0^{\top} + \tilde\eb_0 \tilde\cb_0^{\top}) \}.
    \end{align*}
    
     Thus it can be seen that the covariance matrix is block-diagonal:
 \[
 \operatorname{Cov}(\Vb_d^{\top}\Eb_0 \Pb_{\perp} \eb_j) = \frac{1}{n}\begin{pmatrix} \|\tilde{\eb}_0\|_2^2 \sigma^2 (\mathbf{\Lambda} +\sigma^2 \Ib_K) & \mathbf{0}\\ \mathbf{0} & \|\tilde{\eb}_0\|_2^2 \sigma^4 (\Ib_{d-K} + 3 \mathbf{\tau}_1\mathbf{\tau}_1^{\top}- 2\rho \tilde\cb \mathbf{\tau}_1^{\top}-2\rho \mathbf{\tau}_1 \tilde\cb^{\top}) \end{pmatrix},
 \]
 where $\mathbf{\tau}_1 = \tilde{\eb}/\|\tilde{\eb}\|_2$ and $\rho = \langle \tilde{\cb}, \mathbf{\tau}_1 \rangle$. Then following basic algebra, we can write $\operatorname{Cov}(\Eb_0 \Pb_{\perp} \eb_j)$ as:
 $$\frac{1}{n}\!\!\left\{\!\sigma^2 \|\tilde\eb_0\|_2^2\bSigma \!+\! 3\sigma^4 \Pb_{\perp}\eb_j \eb_j^{\top}\Pb_{\perp} \!\!-\!2\sigma^4\rho\|\tilde\eb_0\|_2\big[\!(\Pb_{\perp})_{[:,S]}\ub_{K+1}\eb_j^{\top}\! \Pb_{\perp}\!\!+\!\Pb_{\perp}\eb_j(\ub_{K+1})^{\!\top}\!(\Pb_{\perp})_{[S,:]}\big]\!\!\right\} .$$

To study $\eta_2(d)$, we will first define $\bSigma_j'$ as following
$$
    \bSigma_j' = \frac{1}{nL^2}\Bb_{\mathbf{\Omega}}^{\top}\mathbf{\Omega}^{\top}\Big\{\sigma^2 \bSigma+3\sigma^4\eb_j\eb_j^{\top}-2\sigma^4\rho\|\tilde\eb_0\|_2\big((\Ib_d)_{[:,S]}\ub_{K+1}\eb_j^{\top}+ \eb_j\ub_{K+1}^{\top}(\Ib_d)_{[S,:]}\big)\Big\}\mathbf{\Omega}\Bb_{\mathbf{\Omega}}.
$$
We know that $\|\tilde\eb_0\|^2_2 = \|\Pb_{\perp}\eb_j \|_2^2 = 1 - O(\mu K/d)$, thus we have
 \begin{align*}
     &\Big\|\sigma^2 \|\Pb_{\perp} \eb_j\|_2^2\bSigma - \sigma^2 \bSigma\Big\|_2  \le O\big(\frac{\mu K\sigma^2}{d}(\sigma^2+\lambda_1)\big), \\
     &\|3\sigma^4 \Pb_{\perp}\eb_j \eb_j^{\top}\Pb_{\perp} - 3\sigma^4\eb_j \eb_j^{\top}\|_2  \le 3\sigma^4 \|(\Pb_{\perp}\eb_j - \eb_j)\eb_j^{\top}\Pb_{\perp}\|_2 + 3\sigma^4\|\eb_j(\Pb_{\perp}\eb_j - \eb_j)^{\top}\|_2 \lesssim \sigma^4 \sqrt{\frac{\mu K}{d}},\\
     &\|(\Pb_{\perp})_{[:,S]}\ub_{K+1}\eb_j^{\top} \Pb_{\perp} - (\Ib_d)_{[:,S]}\ub_{K+1}\eb_j^{\top} \|_2 \le \|[(\Pb_{\perp})_{[:,S]} - (\Ib_d)_{[:,S]}]\ub_{K+1}\eb_j^{\top}\Pb_{\perp}\|_2 \\
     &\quad + \|(\Ib_d)_{[:,S]}]\ub_{K+1}\eb_j^{\top}(\Pb_{\perp}-\Ib_d)\|_2  \lesssim K\sqrt{\frac{\mu}{d}} + \sqrt{\frac{\mu K}{d}}\lesssim K\sqrt{\frac{\mu}{d}}, \\
     & 2\sigma^4\rho\|\tilde\eb_0\|_2\Big\|\!\big[(\Pb_{\perp})_{[:,S]}\ub_{K+1}\eb_j^{\top} \Pb_{\perp}\!\!+\!\Pb_{\perp}\eb_j(\ub_{K+1})^{\top}\!(\Pb_{\perp})_{[S,:]}\big] \!\!-\!\big[\!(\Ib_d)_{[:,S]}\ub_{K+1}\eb_j^{\top} \!\!\!+\!\eb_j(\ub_{K+1})^{\top}\!(\Ib_d)_{[S,:]}\big]\!\Big\|_2\\
     &\quad \lesssim K\sigma^4\sqrt{\frac{\mu}{d}},
 \end{align*}
 and in summary we have $\|\bSigma_j - {\bSigma}_j'\|_2 = O_P\Big( \frac{Kd \sigma^4}{n\Delta^2Lp}\sqrt{\frac{\mu}{d}} \Big)= O_P\Big(\frac{K\lambda_1^2}{\Delta^2}\sqrt{\frac{\mu}{d}}\Big) \frac{d\sigma^4}{nLp\lambda_1^2} = o_P\big(\frac{d\sigma^4}{nLp\lambda_1^2}\big)$. Now we study $\|\bSigma_j' - \widetilde\bSigma_j\|_2$. Since the entries of  $\sqrt{p}\mathbf{\Omega}$ are i.i.d. standard Gaussian, by Lemma 3 in \citet{fandistributed2019}, we know that with probability $1-O(d^{-9})$, we have 
 $$
 \|\mathbf{\Omega}\|_{2,\infty} \lesssim \sqrt{L}, \quad \text{and} \quad  \|\mathbf{\Omega}_{[S, :]}\|_2 \lesssim \sqrt{L}.
 $$
    Therefore, under the condition that $\frac{\lambda_1^2 Lp}{\Delta^2 d} = o(1)$ we have 
    \begin{align*}
        \|\bSigma_j' - \widetilde\bSigma_j\|_2& =  \sigma^4\Big\| \frac{1}{nL^2}\Bb_{\mathbf{\Omega}}^{\top}\mathbf{\Omega}^{\top}\Big(3\eb_j\eb_j^{\top}-2\rho\big((\Ib_d)_{[:,S]}\ub_{K+1}\eb_j^{\top}+ \eb_j\ub_{K+1}^{\top}(\Ib_d)_{[S,:]}\big)\Big)\mathbf{\Omega}\Bb_{\mathbf{\Omega}}\Big\|_2\\
        & \lesssim \frac{\sigma^4}{nL^2} \|\Bb_{\mathbf{\Omega}}\|_2^2  \|\mathbf{\Omega}\|_{2,\infty} \big( \|\mathbf{\Omega}\|_{2,\infty} +\|\mathbf{\Omega}_{[S,:]}\|_2 \big) = O_P\big(\frac{\sigma^4}{n\Delta^2}\big) = o_P(\frac{d\sigma^4}{nLp\lambda_1^2}).
    \end{align*}
    As for $\widetilde\bSigma_j$, by Lemma 3 in \citet{fandistributed2019} with high probability we have that $\sigma_K(\bOmega^{\top}\Vb) \gtrsim \sqrt{L}$ and in turn
    \begin{align*}
        \sigma_K(\widetilde\bSigma_j) &\gtrsim \frac{\sigma^2}{nL^2} \big(\sigma_K(\Bb_{\bOmega})\big)^2 \left(\big(\sigma_K(\bOmega^{\top}\Vb\big)^2\Delta +\big(\sigma_K(\bOmega\big)^2 \sigma^2\right) \gtrsim \frac{d\sigma^4}{nLp\lambda_1^2}  + \frac{\sigma^2\Delta}{n\lambda_1^2}.
    \end{align*}
    Therefore, combining the previous results, we have that by Weyl's inequality \citep{franklin2012matrix}, with high probability
    \begin{align*}
            \lambda_{K}\big(\bSigma_j\big)& \ge \lambda_{K}\big(\widetilde\bSigma_j\big) - \|\bSigma_j - \bSigma_j'\|_2 - \|\bSigma_j' - \widetilde\bSigma_j\|_2 \\
            & \gtrsim \frac{d\sigma^4}{nLp\lambda_1^2}  + \frac{\sigma^2\Delta}{n\lambda_1^2} - o(\frac{d\sigma^4}{nLp\lambda_1^2})  \gtrsim \frac{d\sigma^4}{nLp\lambda_1^2}  + \frac{\sigma^2\Delta}{n\lambda_1^2}.
    \end{align*}

 Thus we know $\eta_2(d) \asymp d \sigma^4/(nLp\lambda_1^2) + \sigma^2\Delta/(n\lambda_1^2)$.

 Recall from the proof of Corollary~\ref{prop: err rate terms} with probability $1-O(d^{-10})$ we have $\|\Eb_0\|_2 \lesssim (\lambda_1+\sigma^2)\log d\sqrt{\frac{r}{n}}$. Also recall that $ r_2(d) \asymp \frac{\tilde{\sigma}_1^3K}{\delta^2 n}(\log d)^2 $. Therefore, under the condition that
 $$n \gg \frac{\kappa_1^4\lambda_1 d r^2 (\log d)^4 }{p \sigma^2 }\left(\kappa_1\frac{d}{p}\wedge\frac{\lambda_1}{\sigma^2}L\right) \quad\text{and}\quad \frac{\tilde{\sigma}_1^6K^2}{\delta^4\sigma^4 n}(\log d)^4 \ll (\frac{\Delta}{\lambda_1})^2,$$
 we have ${d^2 r_1(d)^4(\log d)^4}\big({p^2 \Delta^4 \eta_2(d)}\big)^{-1} = o(1)$ and ${d r_2(d)^2}\big({Lp \Delta^2 \eta_2(d)}\big)^{-1} = o(1)$. 
 
 Now we need to verify Assumption \ref{asp: clt}. It can be seen that the randomness of the leading term comes from $\bOmega$ and $\Eb_0$ both. We will first establish the results conditional on $\bOmega$. In fact, we will first show a more general CLT that will also cover the case of the leading term under the regime $Lp \gg d$. More specifically, we will show that
 for any matrix $\Ab \in \RR^{d \times K}$ that satisfies the following two conditions: (1) $  \sigma_{\max}(\Ab)/\sigma_{\min}(\Ab) \le C|\lambda_1|/\Delta$; (2) $\lambda_K \big(\Cov(\Ab^{\top} \Eb_0 \Pb_{\perp} \eb_j ) \big) \ge c n^{-1}\sigma^4 \big(\sigma_{\min}(\Ab)\big)^2$, where $C, c > 0$ are fixed constants irrelevant to $\Ab$ and we abuse the notation by denoting $\bSigma_j := \Cov(\Ab^{\top} \Eb_0 \Pb_{\perp} \eb_j )$, it holds that 
\begin{equation}\label{eq: asp clt 1}
\mathbf{\Sigma}_j^{-1/2}\Ab^{\top} \Eb_0 \Pb_{\perp} \eb_j \overset{d}{\rightarrow} {\cN}(\mathbf{0}, \Ib_K).
\end{equation}
 Now for any matrix $\Ab \in \RR^{d \times K}$ satisfying the aforementioned conditions,  
 to show that $\Ab^{\top} \Eb_0 \Pb_{\perp} \eb_j$ is asymptotically normal, we only need to show that $\ab^{\top} \bSigma_j^{-1/2}\Ab^{\top} \Eb_0 \Pb_{\perp} \eb_j \overset{d}{\rightarrow} {\cN}(0,1)$ for any $\ab \in \RR^K$ with $\|\ab\|_2 = 1$. We can write 
 \begin{align*}
     &\ab^{\top}\bSigma_j^{-1/2} \Ab^{\top} \Eb_0 \Pb_{\perp} \eb_j  = \frac{1}{n}\sum_{i=1}^n \ab^{\top}\bSigma_j^{-1/2}\Ab^{\top}\{\bX_i\bX_i^{\top} - \bSigma -\tilde\cb_0^{\top}(\bZ_i\bZ_i^{\top} - \mathbf{\Lambda}_d)\tilde\cb_0 \Ib_d\}\Pb_{\perp} \eb_j\\
     & =  \frac{1}{n}\sum_{i=1}^n \Big\{\ab^{\top}\bSigma_j^{-1/2}\Ab^{\top}(\bX_i\bX_i^{\top} - \bSigma )\Pb_{\perp} \eb_j - \tilde\cb_0^{\top}(\bZ_i\bZ_i^{\top} - \mathbf{\Lambda}_d)\tilde\cb_0 (\ab^{\top}\bSigma_j^{-1/2}\Ab^{\top}\Pb_{\perp}\eb_j)\Big\}.
 \end{align*}
 We let $x_i = \ab^{\top}\bSigma_j^{-1/2}\Ab^{\top}(\bX_i\bX_i^{\top} - \bSigma )\Pb_{\perp} \eb_j$ and $y_i = \tilde\cb_0^{\top}(\bZ_i\bZ_i^{\top} - \mathbf{\Lambda}_d)\tilde\cb_0 (\ab^{\top}\bSigma_j^{-1/2}\Ab^{\top}\Pb_{\perp}\eb_j)$. For $\bSigma_j$, we have that $\|\bSigma_j^{-1/2}\|_2 \le \sigma_{\min}(\bSigma_j)^{-1/2} \le  \sqrt{n}/\big(\sigma^2 \sigma_{\min}(\Ab)\big)$.
  Then we have
 \begin{align*}
 &\EE|x_i|^3  \lesssim  \EE|\ab^{\top}\bSigma_j^{-1/2}\Ab^{\top}\bX_i \bX_i^{\top}\Pb_{\perp}\eb_j|^3 \le  \sqrt{\EE|\ab^{\top}\bSigma_j^{-1/2}\Ab^{\top}\bX_i|^6 \EE|\eb_j^{\top}\Pb_{\perp}\bX_i|^6} \\
 & \quad \lesssim \|\bSigma_j^{-1/2}\|_2^3\sqrt{(\lambda_1+\sigma^2)^3 \sigma^6 \|\Ab\|_2^6},\\
 &\EE|y_i|^3  \lesssim (\ab^{\top}\bSigma_j^{-1/2}\Ab^{\top}\Pb_{\perp}\eb_j)^3 \EE |\tilde\cb_0^{\top}\bZ_i\bZ_i^{\top}\tilde\cb_0|^3 \le \|\bSigma_j^{-1/2}\Ab\|_2^3\EE|\tilde{\cb}_0^{\top}\bZ_i|^6 \\
 & \quad \lesssim 
  \|\bSigma_j^{-1/2}\|_2^3(\lambda_1+\sigma^2)^3 \|\Ab\|_2^3,\\
     &\EE|x_i - y_i|^3  \lesssim \EE|x_i|^3 + \EE|y_i|^3 \lesssim \|\bSigma_j^{-1/2}\|_2^3\Big(\sqrt{(\lambda_1+\sigma^2)^3 \sigma^6 \|\Ab\|_2^6} +(\lambda_1+\sigma^2)^3 \|\Ab\|_2^3\Big)\\
     &\lesssim n^{3/2}(\lambda_1+\sigma^2)^3 \|\Ab\|_2^3/\big(\sigma^2 \sigma_{\min}(\Ab)\big)^3. 
 \end{align*}
 Thus 
 $$
 \frac{\sum_{i=1}^n \EE|x_i - y_i|^3}{\operatorname{Var}\Big\{\sum_{i=1}^n (x_i - y_i)\Big\}^{3/2}} \lesssim \frac{n (\lambda_1 + \sigma^2)^3 \|\Ab\|_2^3}{n^{3/2}\sigma^6 \sigma_{\min}(\Ab)^3} \lesssim \frac{(\lambda_1+\sigma^2)^3 \lambda_1^3}{\sqrt{n} \sigma^6\Delta^3} = o(1).
 $$
 Thus the Lyapunov’s condition is met and \eqref{eq: asp clt 1} holds. Then we take $\Ab = \bOmega\Bb_{\bOmega}$, and define the following event
\begin{align*}
    \cA_{\bOmega} &= \bigg\{1/2 \le \sigma_{\min}(\widetilde{\mathbf{\Omega}}^{(\ell)} / \sqrt{p}) \le \sigma_{\max}(\widetilde{\mathbf{\Omega}}^{(\ell)} / \sqrt{p}) \le 3/2, \quad  \forall \ell \in [L]\bigg\}\\
    & \quad \cap \bigg\{\frac{1}{2}\sqrt{\frac{d}{p}} \le \sigma_{\min} (\mathbf{\Omega}) \le \sigma_{\max} (\mathbf{\Omega}) \le \frac{3}{2}\sqrt{\frac{d}{p}}, \quad \forall \ell \in [L] \bigg\}.
\end{align*}
Then from previous results we know that $\PP((\cA_{\bOmega})^c) = o(1)$, and under the event $\cA_{\bOmega}$ we have $$\frac{\sigma_{\max}(\bOmega\Bb_{\bOmega})}{\sigma_{\min}(\bOmega\Bb_{\bOmega})} \le 9 \lambda_1/\Delta, \quad \lambda_K(\bSigma_j) \ge \frac{\sigma^4}{2n} \big(\sigma_{\min}(\bOmega\Bb_{\bOmega})\big)^2. $$

Thus from the above proof, for any vector $\tb \in \RR^K$, we have
$\PP\big(\mathbf{\Sigma}_j^{-1/2}\cV (\Eb_0)^{\top} \eb_j \le \tb|\cA_{\bOmega}\big) - \Phi(\tb) = o(1)$, where $\Phi(\cdot)$ is the CDF for ${\cN}(0,\Ib_K)$. Then we have
\begin{align*}
    & \PP\big(\mathbf{\Sigma}_j^{-1/2}\cV (\Eb_0)^{\top} \eb_j \le \tb\big) = \EE\Big(\PP\big(\mathbf{\Sigma}_j^{-1/2}\cV (\Eb_0)^{\top} \eb_j \le \tb | \mathbf{\Omega}\big)\Big)\\
    & \quad = \PP\big(\mathbf{\Sigma}_j^{-1/2}\cV (\Eb_0)^{\top} \eb_j \le \tb | \mathbf{\Omega} \in \cA_{\bOmega} \big)\PP( \cA_{\bOmega})+\PP\big(\mathbf{\Sigma}_j^{-1/2}\cV (\Eb_0)^{\top} \eb_j \le \tb | \mathbf{\Omega} \in \cA_{\bOmega}^c\big)\PP(\cA_{\bOmega}^c)\\
    & \quad = \big(\Phi(\tb) + o(1) \big)\big(1-o(1)\big) + o(1) = \Phi(\tb) + o(1).
\end{align*}
 Hence we have that Assumption \ref{asp: clt} holds and \eqref{eq: col gaussian L small 2} follows.
 Next we need to show that the result also holds for $\widetilde{\bSigma}_j$. From previous discussion we already know that $\|\bSigma_j - \widetilde\bSigma_j\|_2 = o_P\big(\lambda_K(\widetilde\bSigma_j)\big)$, then by Lemma 13 in \citet{Chen2019matcompinf} we have that 
 $\|\widetilde\bSigma_j^{-1/2}\bSigma_j^{1/2} - \Ib_d\|_2 = O_P\big(\|\widetilde\bSigma_j^{-1/2}\|_2\|\bSigma_j^{1/2} - \widetilde\bSigma_j^{1/2} \|_2\big)=  O_P\big(\lambda_K(\widetilde\bSigma_j)^{-1}\|\bSigma_j - \widetilde\bSigma_j\|_2\big) =  o_P(1) $.
 Then by Slutsky's Theorem, we have  
$$
\widetilde\bSigma_j^{-1/2}(\widetilde\Vb^{\F}\Hb - \Vb)^{\top}\eb_j = (\widetilde\bSigma_j^{-1/2}\bSigma_j^{1/2})\bSigma_j^{-1/2}(\widetilde\Vb^{\F}\Hb - \Vb)^{\top}\eb_j \overset{d}{\rightarrow} {\cN}(0,\Ib_K).
$$

Finally, we move on to verify the validity of the estimator $\hbSigma_j$ for the asymptotic covariance matrix. 
From Lemma 7 in \citet{fandistributed2019}, it can be seen that with probability $1- o(1)$, $\Hb$ is orthonormal. When $\Hb$ is orthonormal, by Slutsky's Theorem we have that 
 $$
 \Hb \widetilde\bSigma_j^{-1/2}(\widetilde\Vb^{\F}\Hb - \Vb)^{\top}\eb_j = \Hb \widetilde\bSigma_j^{-1/2} \Hb^{\top} (\widetilde\Vb^{\F} - \Vb\Hb^{\top})^{\top}\eb_j \overset{d}{\rightarrow} {\cN}(\mathbf{0}, \Ib_K),
 $$
where it can be seen that $\Hb \widetilde\bSigma_j^{-1/2} \Hb^{\top} = (\Hb \widetilde\bSigma_j \Hb^{\top})^{-1/2}$. Therefore, it suffices to show that $\|\widehat\bSigma_j - \Hb \widetilde\bSigma_j \Hb^{\top}\|_2 = o_P\big(\lambda_K(\widetilde\bSigma_j)\big)$, and the results will hold by Slutsky's Theorem. Recall from the proof of Corollary~\ref{prop: err rate terms}, we have the following bounds
$$
    \|\bSigma - \widehat\bSigma\|_2 = O_P\Big((\lambda_1 + \sigma^2)\sqrt{\frac{r}{n}}\Big), \quad |\hat\sigma^2 - \sigma^2| = O_P(\tilde{\sigma}_1\sqrt{\frac{K}{n}}), 
$$
We will bound the components of $\|\widehat\bSigma_j - \widetilde\bSigma_j\|_2$ respectively. We have
\begin{align*}
    &\| \sigma^2 \bSigma - \hat\sigma^2 \hbSigma\|_2 \lesssim |\hat\sigma^2 - \sigma^2|\|\bSigma\|_2 + \sigma^2\|\bSigma - \hbSigma\|_2 = O_P\Big(\tilde{\sigma}_1(\lambda_1+\sigma^2)\sqrt{\frac{K}{n}}\Big)\\
    &\quad + O_P\Big(\sigma^2(\lambda_1 + \sigma^2)\sqrt{\frac{r}{n}}\Big) = O_P\Big(\sigma^2(\lambda_1 + \sigma^2)\sqrt{\frac{r}{n}}\Big),
\end{align*}

Also, from proof of Theorem~\ref{thm: leading term}, we have that with high probability
$$\|\widetilde\Vb^{\F} \Hb - \Vb \|_2 = \|\widetilde\Vb^{\F} - \Vb\Hb^{\top}\|_2 \lesssim \|\Eb_0 \|_2 \|\mathbf{\Omega}\|_2\|\Bb_{\mathbf{\Omega}}\|_2/L = O_P(\kappa_1 \sqrt{\frac{dr}{npL}}),$$
and $\|\hbSigma^{\text{tr}} - \Vb \bLambda \Vb^{\top} \|_2 = O_P\Big((\lambda_1 + \sigma^2)\sqrt{\frac{r}{n}}\Big)$, where $\hbSigma^{\text{tr}} 
 = \hbSigma - \hat\sigma^2\Ib_d$.Then with high probability, for all $\ell \in [L]$ we have that 
\begin{align*}
    &\big\|(\widetilde\Vb^{\F\top} \hbSigma^{\text{tr}}-\Hb \bLambda \Vb^{\top})\mathbf{\Omega}^{(\ell)}/\sqrt{p}\big\|_2 \lesssim \sqrt{\frac{d}{p}}\big(\|\hbSigma^{\text{tr}} - \Vb\bLambda\Vb^{\top}\|_2 + \lambda_1 \|\widetilde\Vb^{\F} - \Vb\Hb^{\top}\|_2\big) \\
    &= O_P\left(\kappa_1 \lambda_1 \sqrt{\frac{d^2r}{np^2L}}\right) = o_P(\Delta),
\end{align*}
and thus by Theorem 3.3 in \citet{stewart1977pertinv}, with high probability for all $\ell \in [L]$ we have that
\begin{align*}
    &\|\widehat\Bb^{(\ell)} - \Bb^{(\ell)}\Hb^{\top}\|_2 = \big\|(\widetilde\Vb^{\F\top}\hbSigma^{\text{tr}}\mathbf{\Omega}^{(\ell)}/\sqrt{p})^{\dagger} - (\Hb\bLambda\Vb^{\top}\mathbf{\Omega}^{(\ell)}/\sqrt{p})^{\dagger}\big\|_2 \\
    &= O_P\bigg(\Delta^{-2}\kappa_1 \lambda_1 \sqrt{\frac{d^2r}{np^2L}}\bigg) ,
\end{align*}
and in turn we have $\|\widehat\Bb_{\mathbf{\Omega}} - \Bb_{\mathbf{\Omega}}\Hb^{\top}\|_2 = O_P\bigg(\Delta^{-2}\kappa_1 \lambda_1 \sqrt{\frac{d^2r}{np^2L}}\bigg) \sqrt{L} = O_P\bigg(\Delta^{-2}\kappa_1 \lambda_1 \sqrt{\frac{d^2r}{np^2}}\bigg) $.

Thus combining the above results, under the condition that $\frac{\lambda_1\kappa_1^4}{\sigma^2}\sqrt{\frac{d^2 r}{np^2L}} = o(1)$, following basic algebra we have 
\begin{align*}
   & \|\hbSigma_j -\Hb\widetilde\bSigma_j\Hb^{\top}\|_2 \!\lesssim O_P\Big(\sigma^2(\lambda_1 \!+\! \sigma^2)\sqrt{\frac{r}{n}}\Big)\frac{d}{nLp\Delta^2} \! +\! O_P\bigg(\frac{d\sqrt{L}}{nL^2p\Delta^3} \sigma^2(\sigma^2\!+\!\lambda_1) \kappa_1 \lambda_1 \sqrt{\frac{d^2r}{np^2}}\bigg)\\
   &= O_P\Big(\frac{\lambda_1^2}{\Delta^2\sigma^2}(\lambda_1 + \sigma^2)\sqrt{\frac{r}{n}}\Big)\frac{d\sigma^4}{nLp\lambda_1^2} + O_P\Big(\frac{\lambda_1\kappa_1^4}{\sigma^2}\sqrt{\frac{d^2 r}{np^2L}}\Big)\frac{d\sigma^4}{nLp\lambda_1^2}= o_P\big(\lambda_K(\widetilde\bSigma_j)\big).
\end{align*}
Therefore, by Slutsky's Theorem, under the event $\cB := \{\Hb \text{ is orthonormal}\}$, for any vector $\tb \in \RR^K$, we have that $\PP(\widehat{\bSigma}_j^{-1/2}(\widetilde{\Vb}^{\F} - \Vb\Hb^{\top})^{\top} \eb_j \le \tb |\cB) - \Phi(\tb) = o(1)$, and thus 
 \begin{align*}
     &\PP(\widehat{\bSigma}_j^{-1/2}(\widetilde{\Vb}^{\F} - \Vb\Hb^{\top})^{\top} \eb_j \le \tb) =  \PP(\widehat{\bSigma}_j^{-1/2}(\widetilde{\Vb}^{\F} - \Vb\Hb^{\top})^{\top} \eb_j \le \tb |\cB)\PP(\cB)\\
     &\quad+ \PP(\widehat{\bSigma}_j^{-1/2}(\widetilde{\Vb}^{\F} - \Vb\Hb^{\top})^{\top} \eb_j \le \tb | \cB^c) \PP(\cB^c) \\
     &= \PP\big(\widehat{\bSigma}_j^{-1/2}(\widetilde{\Vb}^{\F} - \Vb\Hb^{\top})^{\top} \eb_j \le \tb |\cB\big)\big(1-o(1)\big) + o(1) = \Phi(\tb)+o(1).
 \end{align*}
 Hence the claim follows.

\subsection{Proof of Corollary~\ref{col: GMM L small}}\label{sec: proof col gmm L small}
We will verify that Assumptions \ref{asp: tail prob bound}, \ref{asp: stat rate biased error}, \ref{asp: incoh} and \ref{asp: clt} hold. First, it is not hard to see that there exists some orthonormal matrix $\Ob \in \RR^{K \times K}$ such that $\Vb = \Fb \Cb^{-1} \Ob$, where $\Cb = \diag(\sqrt{d_1}, \ldots, \sqrt{d_K})$. From the problem setting of Example~\ref{ex: GMM} we also know that there exists a constant $C > 0$ such that
$$
C^{-1}K\max_k d_k \le K\min_k d_k \le d \le K \max_k d_k, \quad d_1 \asymp \ldots \asymp d_K \asymp d/K,
$$
and thus that $\sqrt{d/K} \lesssim \sigma_{K}(\Cb) \le \|\Cb\|_2 \lesssim \sqrt{d/K}$. Then $\|\Vb\|_{2,\infty} \lesssim \sqrt{\frac{K}{d}}\|\Fb\|_{2,\infty} = \sqrt{\frac{K}{d}}$.  Thus Assumption \ref{asp: incoh} holds with $\mu = O(1)$.

From the proof of Corollary~\ref{prop: err rate terms} we know that Assumption \ref{asp: tail prob bound} is satisfied. Besides, recall from Remark~\ref{rmk: exm3 rate},
under the condition that $\sqrt{K/d}\log d = O(1)$,  with probability at least $1 - d^{-10}$ we have that $\|\Eb\|_2 \lesssim d\Delta_0/\sqrt{K} + \sqrt{dn}\log d := r_1'(d)$, which is sharper than $r_1(d)\log d$.  Since $\Eb_b = 0$, we have $r_2(d) = 0$ and Assumption \ref{asp: stat rate biased error} holds trivially. Now we move on to study the minimum covariance eigenvalue rate $\eta_2(d)$. From the proof of Corollary~\ref{prop: err rate terms}, we know that 
$$
\Eb = \Eb_0 = \Fb \bTheta^{\top} \Zb + \Zb^{\top}\bTheta\Fb^{\top} + \Zb^{\top}\Zb - n\Ib_d = \sum_{i = 1}^n \big\{ \Qb_i \Zb_{i.}^{\top} + \Zb_{i.}\Qb_i^{\top} + \Zb_{i.}\Zb_{i.}^{\top} -\Ib_d \big\},
$$
where $\Qb_i = \Fb\bTheta_{i.} \in \RR^d$ with $\bTheta_{i.}$ being the $i$-th row of $\bTheta$, $\Zb_{i.}$ is the $i$-th row of $\Zb$ and $\Zb_{i.} \overset{\text{i.i.d}}{\sim} {\cN}(\mathbf{0}, \Ib_d)$. Then for $j \in [d]$, we have 
\begin{align*}
 \Cov(\Eb_0\Pb_{\perp}\eb_j) &= \Cov\Big(\sum_{i = 1}^n \big\{ \Qb_i \Zb_{i.}^{\top} + \Zb_{i.}\Qb_i^{\top} + \Zb_{i.}\Zb_{i.}^{\top} -\Ib_d \big\} \Pb_{\perp}\eb_j\Big)\\
 &= \sum_{i = 1}^n\Cov\Big( \big\{ \Qb_i \Zb_{i.}^{\top} + \Zb_{i.}\Qb_i^{\top} + \Zb_{i.}\Zb_{i.}^{\top} -\Ib_d \big\} \Pb_{\perp}\eb_j\Big)\\
 &= \sum_{i = 1}^n\Cov\Big( \big\{ \Qb_i \Zb_{i.}^{\top} + \Zb_{i.}\Zb_{i.}^{\top} -\Ib_d \big\} \Pb_{\perp}\eb_j\Big),
\end{align*}
where the last equality is due to the fact that $ \Pb_{\perp}\Qb_i = \Pb_{\perp} \Fb\bTheta_{i.} = \mathbf{0}$. Now for $i \in [n]$, we calculate $\Cov\Big( \big\{ \Qb_i \Zb_{i.}^{\top} + \Zb_{i.}\Zb_{i.}^{\top} -\Ib_d \big\} \Pb_{\perp}\eb_j\Big)$. Following basic algebra, we have that 
\begin{align*}
    &\Cov\Big( \big\{ \Qb_i \Zb_{i.}^{\top}\! + \!\Zb_{i.}\Zb_{i.}^{\top} \!-\!\Ib_d \big\} \Pb_{\perp}\eb_j\Big) \\
    &=\! \EE\Big(\big\{ \Qb_i \Zb_{i.}^{\top} \!+\! \Zb_{i.}\Zb_{i.}^{\top} \big\} \Pb_{\perp}\eb_j\eb_j^{\top} \Pb_{\perp}\big\{\Zb_{i.}\Qb_i ^{\top}\! + \!\Zb_{i.}\Zb_{i.}^{\top} \big\}\Big)\! -\! \Pb_{\perp}\eb_j \eb_j^{\top} \Pb_{\perp}\\
    & = \|\Pb_{\perp}\eb_j\|_2^2(\Qb_i\Qb_i^{\top} + \Ib_d) + 2 \Pb_{\perp}\eb_j \eb_j^{\top} \Pb_{\perp} -\Pb_{\perp}\eb_j \eb_j^{\top} \Pb_{\perp} \\
    &= \|\Pb_{\perp}\eb_j\|_2^2(\Qb_i\Qb_i^{\top} + \Ib_d) +  \Pb_{\perp}\eb_j \eb_j^{\top} \Pb_{\perp},
\end{align*}
and thus 
\begin{align*}
     \Cov(\Eb_0\Pb_{\perp}\eb_j) &= \sum_{i=1}^n \left(\|\Pb_{\perp}\eb_j\|_2^2(\Qb_i\Qb_i^{\top}\! +\! \Ib_d) \!+\!  \Pb_{\perp}\eb_j \eb_j^{\top} \Pb_{\perp} \right)\\
     &=\! \|\Pb_{\perp}\eb_j\|_2^2(\sum_{i=1}^n\Qb_i\Qb_i^{\top} \!+\! n\Ib_d) \!+\!  n\Pb_{\perp}\eb_j \eb_j^{\top} \Pb_{\perp} \\
     & = \|\Pb_{\perp}\eb_j\|_2^2(\Fb \bTheta^{\top} \bTheta\Fb^{\top} + n\Ib_d)+\!  n\Pb_{\perp}\eb_j \eb_j^{\top} \Pb_{\perp}.
\end{align*}
Then since $\|\Pb_{\perp}\eb_j\|_2 = 1-K/d = 1-o(1)$, we have that $\lambda_d\big( \Cov(\Eb_0\Pb_{\perp}\eb_j) \big) \gtrsim n$, and hence we have $\eta_2(d) \asymp dn/(\lambda_1^2 Lp)$. Then under the condition that { $n \gg d^3L/p$ and $\Delta_0^2 \gg K (\log d)^2 \sqrt{dnL/p}$}, we have that 
$$
\frac{r_1'(d)^4}{\eta_2(d)} \lesssim  \frac{\lambda_1^2 Lp}{d}\left(\frac{d^4\Delta_0^4}{K^2n} + d^2 n (\log d)^4 \right) \ll \frac{\lambda_1^2 p^2 d^2\Delta_0^4}{K^2 d^2  } \asymp \frac{p^2 \Delta^4}{d^2}, \quad \frac{d^2 r_1'(d)^4}{p^2 \Delta^4 \eta_2(d)}  = o(1).
$$

Now we move on to check Assumption \ref{asp: clt}. Similar as in the proof of Corollary~\ref{col: guassian case}, we will first show the results conditional on $\bOmega$ by establishing a more general CLT . More specifically, we will show that for any $\ab \in \RR^K$ with $\|\ab\|_2 = 1$, and $\Ab \in \RR^{d \times K}$ such that $\lambda_K \big( \Cov(\Ab^{\top} \Eb_0 \Pb_{\perp} \eb_j )\big) \ge c n \sigma_{\min}(\Ab)^2$ and $\sigma_{\max}(\Ab) / \sigma_{\min}(\Ab) \le C$, where $C,c >0$ are constants irrelevant to $\Ab$ and we abuse the notation by denoting $\bSigma_j := \Cov(\Ab^{\top} \Eb_0 \Pb_{\perp} \eb_j )$,  we have $\ab^{\top}\bSigma_j^{-1/2} \Ab^{\top} \Eb_0 \Pb_{\perp}\eb_j \overset{d}{\rightarrow} {\cN}(0,1)$.  
Define $\Qb = \bTheta \Fb^{\top}$. We know that 
$$\|\Qb\|_{2,\infty} = \max_{i \in [n]}\|\Qb_i\|_2 \le \|\Fb\|_2\|\bTheta\|_{2,\infty} \lesssim \mu_{\theta}\Delta_0\sqrt{\frac{d}{n}}, \, \text{and}$$ 
$$\ab^{\top}\bSigma_j^{-1/2} \Ab^{\top} \Eb_0 \Pb_{\perp}\eb_j = \sum_{i = 1}^n \big\{\ab^{\top}\bSigma_j^{-1/2} \Ab^{\top} (\Qb_i \Zb_{i.}^{\top} + \Zb_{i.}\Zb_{i.}^{\top} -\Ib_d)\Pb_{\perp}\eb_j\big\},$$ and we denote
$$
x_i = \ab^{\top}\bSigma_j^{-1/2} \Ab^{\top} \Qb_i \Zb_{i.}^{\top} \Pb_{\perp}\eb_j, \quad y_i = \ab^{\top}\bSigma_j^{-1/2} \Ab^{\top} (\Zb_{i.}\Zb_{i.}^{\top} -\Ib_d) \Pb_{\perp}\eb_j.
$$
Then we have 
\begin{align*}
    \EE|x_i + y_i|^3 & \lesssim \EE|x_i|^3 + \EE|y_i|^3 \lesssim \|\bSigma_j^{-1/2} \Ab^{\top} \Qb_i \|_2^3   + \|\bSigma_j^{-1/2} \Ab^{\top} \|_2^3 \\
    &\le \|\bSigma_j^{-1/2}\|_2^3\|\Ab\|_2^3 (\|\Qb\|_{2,\infty}^3 + 1) \lesssim n^{-3/2}\big\{\frac{\|\Ab\|_2}{\sigma_{\min}(\Ab)}\big\}^3 \big\{\mu_{\theta}^3\Delta_0^3\big(\frac{d}{n}\big)^{3/2} +1\big\}\\
    &\lesssim  n^{-3/2} \mu_{\theta}^3\Delta_0^3\big(\frac{d}{n}\big)^{3/2} + n^{-3/2}.
\end{align*} 
Then 
\begin{align*}
    \frac{\sum_{i=1}^n \EE|x_i + y_i|^3}{\operatorname{Var}\big\{\sum_{i=1}^n (x_i+y_i)\big\}^{3/2}}& = \sum_{i=1}^n \EE|x_i + y_i|^3 \lesssim  n^{-2} \mu_{\theta}^3\Delta_0^3d^{3/2} + n^{-1/2} .
\end{align*}
Then under the condition that $\Delta_0^2 \ll n^{4/3}/(\mu_{\theta}^2d)$, we have that 
$$n^{-2} \mu_{\theta}^3\Delta_0^3d^{3/2}=o(1) \quad \text{and}\quad  \left({\sum_{i=1}^n \EE|x_i + y_i|^3}\right){\operatorname{Var}\big(\sum_{i=1}^n (x_i+y_i)\big)^{-3/2}} = o(1).$$
Thus the Lyapunov's condition is met and the CLT holds. Also recall from previous arguments, there exists a fixed constant $C > 0$ such that with high probability we have
$$\frac{\sigma_{\max}(\bOmega\Bb_{\bOmega})}{\sigma_{\min}(\bOmega\Bb_{\bOmega})} \le 9 \left(\frac{\sigma_1 (\bTheta)}{\sigma_K (\bTheta)}\right)^2 \le C, \quad \lambda_K(\bSigma_j) \ge \frac{n}{2} \big(\sigma_{\min}(\bOmega\Bb_{\bOmega})\big)^2. $$
Then by taking $\Ab = \bOmega \Bb_{\bOmega}$ and following similar steps as in the proof of Corollary~\ref{col: guassian case}, we know that Assumption \ref{asp: clt} is satisfied. Then by Theorem~\ref{thm: leading term}, \eqref{eq: general clt for vk small L} holds. 

We move on to prove \eqref{eq: col GMM L small 2}. It suffices to show that $\|\bSigma_j - \widetilde\bSigma_j\|_2=o_P\big( \lambda_K(\widetilde\bSigma_j) \big)$. 
When $\Delta_0^2 \ll n$ and $K\ll d$, with high probability we have
\begin{align*}
    &\|\bSigma_j - \widetilde\bSigma_j\|_2  \lesssim \frac{d}{L\Delta^2 p} \big\{(n+\Delta)(1\!-\!\|\Pb_{\perp}\eb_j\|_2^2) + n\|\Pb_{\perp}\eb_j\eb_j^{\top}\Pb_{\perp} - \eb_j\eb_j^{\top}\|_2\big\}\\
    & \quad + \frac{n}{L\Delta^2}\|\mathbf\Omega\|_{2,\infty}^2 \lesssim \frac{d}{L\Delta^2 p} \Big(\frac{Kn}{d} +\Delta_0^2 + n \sqrt{\frac{K}{d}}\Big) + \frac{n}{\Delta^2}= o(\frac{dn}{L\lambda_1^2 p})  = o_P\big(\lambda_K(\widetilde\bSigma_j)\big).
\end{align*} 
Thus \eqref{eq: col GMM L small 2} holds. 

Last we verify the validity of $\hbSigma_j$. Similar as in the proof of Corollary~\ref{col: guassian case}, it suffices to show that $\|\hbSigma_j - \Hb\widetilde\bSigma_j\Hb^{\top}\|_2=o_P\big(\lambda_K(\widetilde\bSigma_j)\big)$. Recall with high probability $\|\widehat\Mb - \Mb\|_2 \lesssim r_1'(d) = d\Delta_0/\sqrt{K} + \sqrt{dn}\log d$. 

Also, from the proof of Theorem~\ref{thm: leading term}, we have that 
\begin{align*}
    \|\widetilde\Vb^{\F} \Hb - \Vb\|_2= \|\widetilde\Vb^{\F} - \Vb\Hb^{\top}\|_2 = \frac{1}{L} O_P(\|\widehat\Mb - \Mb\|_2\|\mathbf{\Omega}\|_2\|\Bb_{\mathbf{\Omega}}\|_2) = O_P\Big(\sqrt{\frac{d}{pL}} \frac{r_1'(d)}{\Delta}\Big),
\end{align*}
Then with high probability, for all $\ell \in [L]$ we have that 
\begin{align*}
    &\big\|(\widetilde\Vb^{\F\top} \widehat\Mb -\Hb \bLambda \Vb^{\top})\mathbf{\Omega}^{(\ell)}/\sqrt{p}\big\|_2 \lesssim \sqrt{\frac{d}{p}}\big(\|\widehat\Mb - \Mb\|_2 + \lambda_1 \|\widetilde\Vb^{\F} - \Vb\Hb^{\top}\|_2\big) \\
    &= O_P\left( \sqrt{\frac{d^2}{p^2L}}r_1'(d) \right) = o_P(\Delta),
\end{align*}
and thus by Theorem 3.3 in \citet{stewart1977pertinv}, we have that 
\begin{align*}
    \|\widehat\Bb^{(\ell)} - \Bb^{(\ell)}\Hb^{\top}\|_2 = \big\|(\widetilde\Vb^{\F\top}\widehat\Mb\mathbf{\Omega}^{(\ell)}/\sqrt{p})^{\dagger} - (\Hb\bLambda\Vb^{\top}\mathbf{\Omega}^{(\ell)}/\sqrt{p})^{\dagger}\big\|_2 = O_P\left( \sqrt{\frac{d^2}{p^2L}}\frac{r_1'(d)}{\Delta^2} \right),
\end{align*}
and in turn we have $\|\widehat\Bb_{\mathbf{\Omega}} - \Bb_{\mathbf{\Omega}}\Hb^{\top}\|_2 =O_P\left( \sqrt{\frac{d^2}{p^2L}}\frac{r_1'(d)}{\Delta^2} \right)\sqrt{L} = O_P\left( \frac{dr_1'(d)}{p\Delta^2} \right) $.

Therefore, under the condition that $\Delta_0^2 \ll {KLp^2n^2}/{d^4}$, we have
\begin{align*}
    &\|\hbSigma_j - \Hb\widetilde\bSigma_j\Hb^{\top}\|_2 \lesssim \frac{d}{L\Delta^2 p}\|\widehat\Mb - \Mb\|_2 + (n+\lambda_1)\frac{d}{pL\Delta}O_P\left( \frac{d}{p\sqrt{L}}\frac{r_1'(d)}{\Delta^2} \right)   \\
    &= o_P\big(\frac{dn}{L\Delta^2 p} \big) = o_P\big(\lambda_K(\widetilde\bSigma_j)\big).
\end{align*}
Thus the claim follows.

\subsection{Proof of Theorem~\ref{thm: leading term L big}}\label{sec: proof thm leading term L big}
We will first decompose $\widetilde{\Vb}^{\F}\Hb - \Vb = (\widetilde{\Vb}^{\F}\Hb -\widetilde{\Vb}\Hb_1\Hb_0) + (\widetilde{\Vb}\Hb_1\Hb_0 - \widehat{\Vb}\Hb_0) +(\widehat{\Vb}\Hb_0 - \Vb)$. We will show that when $L$ is sufficiently large the first two terms are negligible, and we will consider the third term $\widehat{\Vb}\Hb_0 - \Vb$ first. We will first study $\|\widehat{\Vb}\Hb_0 - \Vb - \Pb_{\perp} \Eb_0 \Vb \mathbf{\Lambda}^{-1}\|_{2, \infty}$ by conducting decomposition of the error term.
For the convenience of notations, we let $\Pb = \Vb^{\top} \Vb$ for short. If we define $\widehat\Hb_0 = \widehat{\Vb}^{\top} \Vb $, we can decompose 
\begin{align*}
    &\widehat{\Vb}\Hb_0  - \Vb - \Pb_{\perp} \Eb_0 \Vb \mathbf{\Lambda}^{-1} \\
    &= \Pb_{\perp}\widehat{\Vb}\widehat{\Hb}_0 \!-\! \Pb_{\perp} \Eb_0 \Vb\mathbf{\Lambda}^{-1}\! +\! \Pb_{\perp}\widehat\Vb({\Hb}_0\!-\!\widehat{\Hb}_0) \!+\! (\Pb \widehat{\Vb}\Hb_0\! -\!  \Vb ).
\end{align*}
Under the condition that $\|\Eb\|_2/\Delta =O_P\big( r_1(d)/\Delta \big) = o_P(1)$, we have that $\Hb_0$ is a full-rank orthonormal matrix with probability $1-o(1)$. Then we have with probability $1-o(1)$ that
\begin{align*}
    & \|\Pb_{\perp}\widehat\Vb({\Hb}_0-\widehat{\Hb}_0)\|_{2,\infty}  = \|(\Ib - \Vb \Vb^{\top})(\widehat\Vb\Hb_0 - \Vb)\Hb_0^{\top}({\Hb}_0-\widehat{\Hb}_0)\|_{2,\infty} \\
    &\le \|(\widehat\Vb\Hb_0 - \Vb)\Hb_0^{\top}({\Hb}_0-\widehat{\Hb}_0)\|_{2,\infty} + \|\Vb \Vb^{\top}(\widehat\Vb\Hb_0 - \Vb)\Hb_0^{\top}({\Hb}_0-\widehat{\Hb}_0)\|_{2,\infty}\\
    &\le \|\widehat\Vb\Hb_0 - \Vb\|_{2,\infty}\|{\Hb}_0-\widehat{\Hb}_0\|_2+ \|\Vb\|_{2,\infty}\|\widehat\Vb\Hb_0 - \Vb\|_2 \|{\Hb}_0-\widehat{\Hb}_0\|_2 \\
    &\lesssim \Big( r_3(d) + \sqrt{\frac{\mu K}{d}} \frac{\|\Eb\|_2}{\Delta}\Big)\|{\Hb}_0-\widehat{\Hb}_0\|_2.
\end{align*}
From Lemma 7 in \citet{fandistributed2019}, we know that $ \|{\Hb}_0-\widehat{\Hb}_0\|_2  \lesssim \|\widehat\Vb\widehat\Vb^{\top} - \Vb \Vb^{\top} \|_2^2 \lesssim (\|\Eb\|_2/\Delta)^2 =O_P( r_1(d)^2/\Delta^2)$, and thus we have
$$
 \|\Pb_{\perp}\widehat\Vb({\Hb}_0-\widehat{\Hb}_0)\|_{2,\infty} = O_P \left( \Big( r_3(d) + \sqrt{\frac{\mu K}{d}}\frac{r_1(d)}{\Delta}\Big)r_1(d)^2/\Delta^2\right).
$$
We move on to bound $\|\Pb \widehat{\Vb}\Hb_0 -  \Vb \|_{2,\infty}$,
\begin{align*}
    \|\Pb \widehat{\Vb}\Hb_0 -  \Vb \|_{2,\infty} &  = \|\Vb( \widehat{\Hb}_0^{\top}\Hb_0 -  \Ib_K  )\|_{2,\infty} \le \|\Vb\|_{2,\infty} \|{\Hb}_0-\widehat{\Hb}_0\|_2 \\
    & = O_P\left( \sqrt{\frac{\mu K}{d}}r_1(d)^2/\Delta^2\right).
\end{align*}
Finally, we consider the term $\Pb_{\perp}\widehat{\Vb}\widehat{\Hb}_0 - \Pb_{\perp} \Eb_0 \Vb\mathbf{\Lambda}^{-1}$. We can decompose

{
\begin{align*}
    &\Pb_{\perp}\widehat{\Vb}\widehat{\Hb}_0 - \Pb_{\perp} \Eb_0 \Vb\mathbf{\Lambda}^{-1}  = \Pb_{\perp} \widehat\Vb \widehat\Hb_0 \bLambda \bLambda^{-1} - \Pb_{\perp} \Eb_0 \Vb\mathbf{\Lambda}^{-1} \\
    &= \Pb_{\perp}\big(\Eb\widehat{\Vb}\widehat{\Hb}_0 - \Eb_0 \Vb + \widehat{\Vb}(\mathbf{\Lambda} - \widehat{\bLambda} )\widehat{\Hb}_0 + \widehat{\Vb} (\widehat{\Hb}_0  \mathbf{\Lambda} - \mathbf{\Lambda} \widehat{\Hb}_0)\big) \mathbf{\Lambda}^{-1}.
\end{align*}
}

We bound the three terms separately, with high probability
\begin{align*}
    & \|\!\Pb_{\perp}\!\big(\Eb\widehat{\Vb}\widehat{\Hb}_0 \!-\! \Eb_0 \Vb\big) \!\mathbf{\Lambda}^{-1}\!\|_{2,\infty}  \!\!\le\! \|\!\Pb_{\perp}\Eb_0(\widehat{\Vb}\widehat{\Hb}_0 \!-\!\Vb\!) \mathbf{\Lambda}^{-1}\!\|_{2,\infty} \!\!+\! \|\!\Pb_{\perp}\Eb_b\! \widehat{\Vb}\!\widehat{\Hb}_0 \mathbf{\Lambda}^{-1}\!\|_{2,\infty}\\
    & \quad \le \|\Eb_0(\widehat{\Vb}\widehat{\Hb}_0 -\Vb) \mathbf{\Lambda}^{-1}\|_{2,\infty}+ \|\Vb \Vb^{\top}\Eb_0(\widehat{\Vb}\widehat{\Hb}_0 -\Vb) \mathbf{\Lambda}^{-1} \|_{2,\infty} + \|\Eb_b\|_2/\Delta\\
    & \quad \le \|\Eb_0(\widehat{\Vb}\widehat{\Hb}_0 -\Vb)\|_{2,\infty}/\Delta + \|\Vb\|_{2,\infty}\|\Eb_0\|_2\|\widehat{\Vb}\widehat{\Hb}_0 -\Vb\|_2/\Delta + r_2(d)/\Delta\\
    & \quad = O_P\left( r_4(d,\mathbf{\Lambda} )/\Delta + \sqrt{\frac{\mu K}{d}}r_1(d)^2/\Delta^2 + r_2(d)/\Delta\right). 
\end{align*}

As for $\Pb_{\perp}\widehat{\Vb}({\mathbf{\Lambda}} - \widehat{\mathbf{\Lambda}} )\widehat{\Hb}_0\mathbf{\Lambda}^{-1}$, we have
\begin{align*}
    & \|\Pb_{\perp}\widehat{\Vb}(\widehat{\mathbf{\Lambda}} - \mathbf{\Lambda} )\widehat{\Hb}_0\mathbf{\Lambda}^{-1}\|_{2,\infty}  \le \|(\widehat{\Vb}\Hb_0 - \Vb)\Hb_0^{\top}(\widehat{\mathbf{\Lambda}} - \mathbf{\Lambda} )\widehat{\Hb}_0\mathbf{\Lambda}^{-1}\|_{2,\infty} \\
    &\quad \quad + \|\Vb \Vb^{\top}(\widehat{\Vb}\Hb_0 - \Vb)\Hb_0^{\top}(\widehat{\mathbf{\Lambda}} - \mathbf{\Lambda} )\widehat{\Hb}_0\mathbf{\Lambda}^{-1}\|_{2, \infty}\\
    &\quad \le \|\widehat{\Vb}\Hb_0 - \Vb\|_{2,\infty}\|\Eb_0\|_2/\Delta + \|\Vb\|_{2,\infty}\|\widehat{\Vb}\Hb_0 - \Vb\|_2\|\Eb_0\|_2/\Delta\\ 
    &\quad = O_P\bigg\{r_3(d)r_1(d)/\Delta + \sqrt{\frac{\mu K}{d}} r_1(d)^2/\Delta^2\bigg\},
\end{align*}
and finally
\begin{align*}
    &\|\Pb_{\perp}\widehat{\Vb} (\mathbf{\Lambda} \widehat{\Hb}_0  -\widehat{\Hb}_0  \mathbf{\Lambda})  \mathbf{\Lambda}^{-1}\|_{2,\infty} \le \|(\widehat{\Vb} \Hb_0 - \Vb)\Hb_0^{\top}(\mathbf{\Lambda} \widehat{\Hb}_0  -\widehat{\Hb}_0  \mathbf{\Lambda})  \mathbf{\Lambda}^{-1}\|_{2,\infty} \\
    & \quad \quad + \|\Vb \Vb^{\top}(\widehat{\Vb} \Hb_0 - \Vb)\Hb_0^{\top} (\mathbf{\Lambda} \widehat{\Hb}_0  -\widehat{\Hb}_0  \mathbf{\Lambda})  \mathbf{\Lambda}^{-1}\|_{2,\infty}\\
    & \quad =O_P \left( \big(r_3(d) + \sqrt{\frac{\mu K}{d}}r_1(d)/\Delta\big)\|\mathbf{\Lambda} \widehat{\Hb}_0  -\widehat{\Hb}_0  \mathbf{\Lambda}\|_2/\Delta\right)\\
    & \quad = O_P\left( \big(r_3(d) + \sqrt{\frac{\mu K}{d}}r_1(d)/\Delta\big)r_1(d)/\Delta\right),
\end{align*}
where the last inequality is due to the fact that 
{
\begin{align*}
    \|\mathbf{\Lambda} \widehat{\Hb}_0  -\widehat{\Hb}_0  \mathbf{\Lambda}\|_2 & = \|\bLambda \widehat\Vb^{\top} \Vb \Vb^{\top} - \widehat\Vb^{\top} \Vb \bLambda \Vb^{\top}\|_2 \\
    &= \|\bLambda \widehat\Vb^{\top} \Vb \Vb^{\top} - \widehat\Vb^{\top} \Mb \Vb \Vb^{\top}\|_2 \\
    & \le \|\bLambda \widehat\Vb^{\top} \Vb \Vb^{\top} - \widehat\Vb^{\top} \widehat\Mb \Vb \Vb^{\top} \|_2 + \| \widehat\Vb^{\top} \Eb \Vb \Vb^{\top} \|_2 \\
    & = \|(\bLambda - \widehat\bLambda)\widehat\Vb^{\top} \Vb \Vb^{\top} \|_2 +  \| \widehat\Vb^{\top} \Eb \Vb \Vb^{\top} \|_2 \le 2\|\Eb\|_2.
\end{align*}
}

Thus in summary, we have
\begin{align*}
    &\|\widehat{\Vb}\Hb_0 \!\!- \!\!\Vb \!\!- \!\Pb_{\perp} \Eb_0 \Vb \mathbf{\Lambda}^{-1}\|_{2,\infty} \! =\! O_P\bigg\{\!\! \frac{r_3(d)r_1(d)}{\Delta} \!+\!\! \sqrt{\frac{\mu K}{d}}\!\frac{r_1(d)^2}{\Delta^2} \!+\! \frac{r_2(d)\!+\! r_4(d)}{\Delta}\!\!\bigg\},
\end{align*}
Now we move on to bound $\|\widetilde{\Vb} \Hb_1 \Hb_0 - \widehat{\Vb} \Hb_0\|_2$. By Theorem~\ref{thm: error bound}, we know that 
\begin{align*}
    \|\widetilde{\Vb} \Hb_1 \Hb_0 - \widehat{\Vb} \Hb_0\|_2 &  \le  \|\widetilde{\Vb} \Hb_1- \widehat{\Vb}\|_2  \lesssim \|\widetilde{\Vb}\widetilde{\Vb}^{\top} - \widehat{\Vb}\widehat{\Vb}^{\top}\|_{2} \\
    &\le \|\widetilde{\Vb}\widetilde{\Vb}^{\top} - {\Vb}^{\prime}{\Vb}^{\prime\top}\|_{2} + \|{\Vb}^{\prime}{\Vb}^{\prime\top} - \widehat{\Vb}\widehat{\Vb}^{\top}\|_{2} \\
    & \le \|\widetilde{\Vb}\widetilde{\Vb}^{\top} - {\Vb}^{\prime}{\Vb}^{\prime\top}\|_{\F} + \|{\Vb}^{\prime}{\Vb}^{\prime\top} - \widehat{\Vb}\widehat{\Vb}^{\top}\|_{2} \\
    & = O_P\Big( \frac{1}{\sqrt{d}}\frac{r_1(d)^2}{\Delta^2}+ \sqrt{\frac{Kd }{\Delta^2 pL}} r_1(d) \Big).
\end{align*}
Finally, we consider $\|\widetilde{\Vb}^{\F} \Hb - \widetilde{\Vb} \Hb_1 \Hb_0\|_2 $. From the proof of Theorem~\ref{thm: error bound}, we know that 
\begin{align*}
    &\|\widetilde{\Vb}^{\F} \Hb - \widetilde{\Vb} \Hb_1 \Hb_0\|_2 \le \|\widetilde{\Vb}^{\F} \Hb_2 - \widetilde{\Vb}\|_2  \lesssim \|\widetilde{\Vb}^{\F}\widetilde{\Vb}^{\text{F}\top} -\widetilde{\Vb}\widetilde{\Vb}^{\top} \|_2 \\
    &\quad  = O_P\big(\EE(\|\widetilde{\Vb}^{\F}\widetilde{\Vb}^{\text{F}\top} -\widetilde{\Vb}\widetilde{\Vb}^{\top} \|_2^2|\widetilde{\bSigma})^{1/2}\big) \lesssim \sqrt{\frac{d}{p'}} \frac{\| \widetilde{\mathbf{\Sigma}} - \Vb \Vb^{\top} \|_2^{q} }{ \left(1- \| \widetilde{\mathbf{\Sigma}} - \Vb \Vb^{\top} \|_2\right)^{q} }.
\end{align*}
From the proof of Theorem~\ref{thm: error bound}, we know that with probability converging to 1, there exists some constant $\eta > 0$ such that $ \|\widetilde{\mathbf{\Sigma}} - \Vb \Vb^{\top}\|_2 \le \eta r_1(d)\log d \sqrt{d/p} /\Delta = o(1)$, and thus that 
$$
\|\widetilde{\Vb}^{\F} \Hb_2 - \widetilde{\Vb}\|_2 = O_P\left( \sqrt{\frac{d}{p'}} \left(2\eta r_1(d) \log d \sqrt{\frac{d}{\Delta^2 p}}\right)^q \right).
$$
When we choose $q$ to be large enough, i.e., 
$$
q \ge 2 + \frac{\log (Ld)}{\log \log d} \gg 1 + \frac{\log \left(\log d\sqrt{Ld/(Kp')}\right)}{\log \left((2\eta \log d)^{-1}\Delta/r_1(d)\sqrt{p/d}\right)},
$$
we have $\|\widetilde{\Vb}^{\F} \Hb_2 - \widetilde{\Vb}\|_2 = O_P ( \sqrt{\frac{Kd }{\Delta^2 pL}} r_1(d) )$. Therefore, if we denote 
$$r(d)  := \Delta^{-1}\Big(\sqrt{\frac{Kd }{ pL}} r_1(d) + r_3(d)r_1(d) \!+\! \sqrt{\frac{\mu K}{d\Delta^2}}r_1(d)^2\!+\! r_2(d)\!+\! r_4(d)\Big),$$  
we can write
\begin{equation}\label{eq: 2 inf norm Lp gg d}
    \widetilde{\Vb}^{\F} \Hb -\Vb = \Pb_{\perp}\Eb_0 \Vb \mathbf{\Lambda}^{-1} + \Rb(d),
\end{equation}
where $\|\Rb(d)\|_{2,\infty} = O_P\big(r(d)\big)$. Then under the condition that $\eta_1(d)^{-1/2} r(d) = o(1)$, we have that $\|\Rb(d)\|_{2,\infty} = o_P\Big(\big(\sigma_{\min}\big(\bSigma_j)\big)^{1/2}\Big)$. Thus by Assumption \ref{asp: clt},
$$
\bSigma_j^{-1/2}(\widetilde{\Vb}^{\F} \Hb -\Vb)^{\top}\eb_j = \bSigma_j^{-1/2}(\mathbf{\Lambda}^{-1}\Vb^{\top}\Eb_0 \Pb_{\perp}\eb_j ) + o_P(1) \overset{d}{\rightarrow} {\cN}(\mathbf{0},\Ib_K).
$$
{
\subsection{Sparsity Assumption on the Eigenspace for Improved Estimation Rates When $Lp \gg d$}\label{sec: sparse incorp}

In this section, we discuss how to improve the estimation performance of FADI  by incorporating sparsity constraints on the leading eigenvectors $\Vb$. For any matrix $\Ub \in \RR^{d \times K}$, we define its row support as 
$$
\cS (\Ub) = \{j \in [d]: \|\Ub^{\top} \eb_j\|_2 > 0\},
$$
which is the set of indices corresponding to the non-zero rows of $\Ub$.  Notably, this definition is invariant to orthonormal transformations of $\Ub$, i.e., $\cS (\Ub) = \cS (\Ub \Ob)$ for any orthonormal matrix $\Ob \in \RR^{K \times K}$. To formalize the sparsity structure of $\Vb$, we impose the following assumption:

\begin{assumption}[Sparsity and Signal Strength Assumption]\label{asp: sparse}
    There exists $s^* \ge K$ such that $|\cS(\Vb)| \le s^*$. In addition, there exists $\eta_3(d) > 0$ such that $\min_{j \in \cS(\Vb)} \|\Vb^{\top} \eb_j\|_2 \ge \eta_3(d)$.
\end{assumption}
When Assumption~\ref{asp: sparse} holds for the leading eigenvectors $\Vb$, we propose the following truncated estimator, $\widetilde\Vb^{\rm Tr}$, by thresholding the row norms of $\widetilde\Vb^{\F}$: 
\begin{equation}\label{eq: thresholding V}
    \eb_j^{\top}\widetilde\Vb^{\rm Tr} = \left\{\begin{array}{ll}
       \eb_j^{\top}\widetilde\Vb^{\rm F}    ,& \text{if } \|\eb_j^{\top}\widetilde\Vb^{\rm F} \|_2 \ge \mu_1 \\
       \mathbf{0}^{\top} , & \text{if } 
 \|\eb_j^{\top}\widetilde\Vb^{\rm F} \|_2 < \mu_1 
    \end{array}\right., \quad j \in [d],
\end{equation}
where $\mu_1 < 1$ is a thresholding parameter. In the following theorem, we show that under proper scaling conditions, with the proper choice of $\mu_1$, we will have an improved statistical rate for $\widetilde\Vb^{\rm Tr}$ compared to $\widetilde\Vb^{\F}$.
\begin{theorem}
When $d \gg 3$ and $Lp \gg d$, suppose that Assumptions~\ref{asp: tail prob bound}, \ref{asp: stat rate biased error}, \ref{asp: eigenspace}, and~\ref{asp: sparse} hold, and that there exist statistical rates $r_5(d), r_6(d) > 0$ such that
$$
\lim_{d \rightarrow \infty} \PP(\|\Eb_0 \Vb\|_{2,\infty} \le r_5(d), \quad \|\Vb^{\top} \Eb_0 \Vb\|_2 \le r_6(d)) = 1.
$$
Define the rate
$\tilde{r}(d) = \Delta^{-1}\big(\sqrt{\frac{Kd }{ pL}} r_1(d) + r_3(d)r_1(d) \!+\! r_2(d)\!+\! r_4(d)\big)$. 
Under the conditions that $\Delta^{-1}r_1(d)(\log d)^2\sqrt{d/p} = o(1)$ and $\Delta^{-1} r_5(d) + \tilde{r}(d)= o(\eta_3(d))$, if we take 
    $$q \ge 2 + \log (Ld)/\log\log d, \quad p' \ge \max(2K, K+7) \quad\text{and}\quad p \ge \max(2K, K+8q-1),$$
    and $r(d) + \Delta^{-1} r_5(d) \ll \mu_1 \le \eta_3(d)/2$,  then with probability $1 - o(1)$, we have
\begin{equation}\label{eq: trunc rate}
    \| \widetilde{\Vb}^{\rm Tr} \Hb -\Vb\|_{\rm F} \lesssim  \sqrt{s^*}\left( \tilde{r}(d) +  \frac{r_5(d)}{\Delta}+ \frac{r_6(d)}{\Delta}\right).
\end{equation}
\end{theorem}
\begin{remark}
    The statistical rate on the right-hand side of \eqref{eq: trunc rate} is typically much smaller than that of \eqref{eq: error VF Vk} in Theorem~\ref{thm: error bound} when $s^*$ is sufficiently small. For instance, consider Example~\ref{ex: spiked gaussian}. Under the same scaling conditions as in Corollary~\ref{col: inf gaussian case L big} and assuming that $\eta_3(d) \gtrsim \sqrt{K/s^*}$, by taking $\mu_1 = c \sqrt{K/s^*}$ for some properly large $c >0$, a straightforward application of Lemma~3 in \citet{fandistributed2019} gives that $r_5(d), r_6(d) \lesssim (\lambda_1 + \sigma^2)\log d \sqrt{K/n}$ and $\tilde{r}(d) = o(r_5(d))$. Consequently, the rate in \eqref{eq: trunc rate} reduces to 
    \begin{equation}\label{eq: trunc rate ex spiked}
          \| \widetilde{\Vb}^{\rm Tr} \Hb -\Vb\|_{\rm F} \lesssim \kappa_1 \sqrt{\frac{K s^* }{n}} \log d.
    \end{equation}
   If the sparsity parameter $s^* \ll r = \mathrm{tr}(\bSigma) / |\bSigma|_2$, the rate in \eqref{eq: trunc rate ex spiked} offers a substantial improvement over that in \eqref{eq: err bd exm 1}.
\end{remark}
\begin{proof}
   We first study the rate $\|\widetilde{\Vb}^{\F} \Hb -\Vb\|_{2, \infty}$. Define
$$\tilde{r}(d)  := \Delta^{-1}\Big(\sqrt{\frac{Kd }{ pL}} r_1(d) + r_3(d)r_1(d) + r_2(d) +  r_4(d)\Big).$$
Then with minor modifications to the steps of showing  \eqref{eq: 2 inf norm Lp gg d} in Section~\ref{sec: proof thm leading term L big}, we have that 
$$
\|\eb_j^{\top}(\widetilde{\Vb}^{\F} \Hb -\Vb - \Pb_{\perp}\Eb_0 \Vb \mathbf{\Lambda}^{-1})\|_{2} = \left\{ \begin{array}{ll}
 O_P\left(\tilde{r}(d) + \|\Vb^{\top} \eb_j\|_2 \cdot r_1(d)^2/\Delta^2\right),    & \text{if } j \in \cS(\Vb)\\
 O_P\left(\tilde{r}(d)\right),     & \text{if } j \notin \cS(\Vb)
\end{array}\right. .
$$ 
Hence we will focus on the term $\Pb_{\perp}\Eb_0 \Vb \mathbf{\Lambda}^{-1}$. For $j \notin \cS(\Vb)$, we have that 
   \begin{align*}
       \|\eb_j^{\top} \Pb_{\perp}\Eb_0 \Vb \mathbf{\Lambda}^{-1}\|_2 & = \|\eb_j^{\top} \Eb_0 \Vb \mathbf{\Lambda}^{-1} - (\Vb^{\top}\eb_j)^{\top} \Vb^{\top}\Eb_0 \Vb \mathbf{\Lambda}^{-1}\|_2 = \|\eb_j^{\top} \Eb_0 \Vb \mathbf{\Lambda}^{-1} \|_2  \\
       & \le \Delta^{-1} \| \Eb_0 \Vb \|_{2, \infty}  = O_P(\Delta^{-1} r_5(d)),
   \end{align*}
   where the first equality follows from the fact that $\Vb^{\top}\eb_j = \mathbf{0}$ for $j \notin \cS(\Vb)$. Under the condition that $\Delta^{-1} r_6(d) = o(1)$ and $\Delta^{-1} r_5(d) = o(\eta_3(d))$, for $j \in \cS(\Vb)$, we have  
   \begin{align*}
       \|\eb_j^{\top} \Pb_{\perp}\Eb_0 \Vb \mathbf{\Lambda}^{-1}\|_2 & = \|\eb_j^{\top} \Eb_0 \Vb \mathbf{\Lambda}^{-1} - (\Vb^{\top}\eb_j)^{\top} \Vb^{\top}\Eb_0 \Vb \mathbf{\Lambda}^{-1}\|_2 \\
       & \le \Delta^{-1} \| \Eb_0 \Vb \|_{2, \infty} + \Delta^{-1} \|\Vb^{\top}\Eb_0 \Vb\|_2 \cdot \|\Vb^{\top}\eb_j\|_2 \\
       &= O_P(\Delta^{-1} r_5(d) + \Delta^{-1} r_6(d) \cdot \|\Vb^{\top}\eb_j\|_2 ) = o_P(\|\Vb^{\top}\eb_j\|_2 ).
   \end{align*}
   Hence for $j \in \cS(\Vb)$, if we further assume $\tilde{r}(d)= o(\eta_3(d))$ and $r_1(d)/\Delta = o(1)$, we have that
   \begin{align*}
      \|\eb_j^{\top} \widetilde{\Vb}^{\F}\|_2 &= \!\|(\widetilde{\Vb}^{\F} \Hb)^{\top} \!\eb_j\|_2 \!\ge \!\|\Vb^{\top} \eb_j\|_2   - \!\|\eb_j^{\top}\!(\widetilde{\Vb}^{\F} \Hb -\!\Vb - \Pb_{\perp}\Eb_0 \Vb \!\mathbf{\Lambda}^{-1})\|_{2}\! -   \|\eb_j^{\top} \!\Pb_{\perp}\Eb_0 \Vb \mathbf{\Lambda}^{-1}\|_2\\
       & \ge \|\Vb^{\top} \eb_j\|_2 - O_P(\tilde{r}(d)) - o_P(\|\Vb^{\top} \eb_j\|_2 ) \ge \eta_3(d) /2 \ge \mu_1,
   \end{align*}
   and for $j \notin \cS(\Vb)$, we have 
   \begin{align*}
        \|\eb_j^{\top} \widetilde{\Vb}^{\F}\|_2 & = \|(\widetilde{\Vb}^{\F} \Hb)^{\top} \eb_j\|_2 \le \|\eb_j^{\top}(\widetilde{\Vb}^{\F} \Hb -\Vb - \Pb_{\perp}\Eb_0 \Vb \mathbf{\Lambda}^{-1})\|_{2}+  \|\eb_j^{\top} \Pb_{\perp}\Eb_0 \Vb \mathbf{\Lambda}^{-1}\|_2\\
        & = O_P(\tilde{r}(d)) + O_P(\Delta^{-1} r_5(d)) \ll \mu_1.
   \end{align*}
   Hence when $\tilde{r}(d) + \Delta^{-1} r_5(d) \ll \mu_1 \le \eta_3(d)/2$, by the definition of $\widetilde\Vb^{\rm Tr}$ in \eqref{eq: thresholding V}, with high probability we have that $\cS(\widetilde\Vb^{\rm Tr}) = \cS(\Vb)$. In turn, with high probability we have
   \begin{align*}
       \| \widetilde{\Vb}^{\rm Tr} \Hb -\Vb\|_{\F} \le \sqrt{s^*} \max_{j \in \cS(\Vb)}\| \eb_j^{\top}(\widetilde{\Vb}^{\F} \Hb -\Vb)\|_{2} \lesssim  \sqrt{s^*}\left( \tilde{r}(d) +  \frac{r_5(d)}{\Delta}+ \frac{r_6(d)}{\Delta}\right).
   \end{align*}
\end{proof}
}
\subsection{Proof of Corollary~\ref{col: inf gaussian case L big}}\label{sec: proof col inf gaussian case L big}
We define $\Eb_0$ and $\Eb_b$ the same as in the proof of Corollary~\ref{col: guassian case}. Then Assumptions \ref{asp: tail prob bound} and \ref{asp: stat rate biased error} are satisfied as been proven for Corollary~\ref{col: guassian case}. As for Assumption \ref{asp: clt},  we have shown that under the condition that ${\kappa_1^3 (\lambda_1/\sigma^2)^3} = o(\sqrt{n})$, the results \eqref{eq: asp clt 1} holds for any matrix $\Ab \in \RR^{d \times  K}$ such that $\sigma_{\max}(\Ab)/\sigma_{\min}(\Ab) \le C |\lambda_1|/\Delta$ and $\lambda_K \big(\Cov(\Ab^{\top} \Eb_0 \Pb_{\perp} \eb_j) \big) \ge c n^{-1}\sigma^4 \big(\sigma_{\min}(\Ab)\big)^2$ in the proof of Corollary~\ref{col: guassian case}.  Under the regime $Lp \gg d$, the leading term $\cV (\Eb_0) = \Pb_{\perp} \Eb_0 \Vb \bLambda^{-1}$, and by taking $\Ab = \Vb \bLambda^{-1}$, it can be seen that 
$$\sigma_{\max}(\Vb \bLambda^{-1})/\sigma_{\min}(\Vb \bLambda^{-1}) = \sigma_{\max}(\bLambda) / \sigma_{\min}(\bLambda) \le |\lambda_1|/\Delta,$$ 
and if we can show that $\eta_1(d) \ge (2 n)^{-1}\lambda_1^{-2} \sigma^4$, we have $\lambda_K ( \bSigma_j) \ge \eta_1(d) = (2n)^{-1}\sigma^4 \big(\sigma_{\min}(\Vb \bLambda^{-1})\big)^2$ and Assumption \ref{asp: clt} is satisfied. Thus we only need to verify Assumption \ref{asp: eigenspace} and the conditions for $\eta_1(d)$. 
Recall from the proof of Corollary~\ref{col: guassian case} we have the following rates
$$
r_1(d) = (\lambda_1 + \sigma^2) \sqrt{\frac{r}{n}}, \quad  r_2(d) \asymp \frac{\tilde\sigma_1^3K}{\delta^2n}(\log d)^2, 
$$
and we can further derive that the following bounds hold with high probability 
\begin{align*}
    &\|\widehat{\Vb} \operatorname{sgn}(\widehat{\Vb}^{\top} \Vb) - \Vb\|_{2,\infty} \le\|\widehat{\Vb} \operatorname{sgn}(\widehat{\Vb}^{\top} \Vb) - \Vb\|_2 \lesssim \|\Eb_0\|_2/\Delta \lesssim r_1(d)\log d/\Delta ;\\
    &\|\Eb_0(\widehat{\Vb} (\widehat{\Vb}^{\top} \Vb) - \Vb)\|_{2,\infty} \lesssim \|\Eb_0\|_2\|\widehat{\Vb} \operatorname{sgn}(\widehat{\Vb}^{\top} \Vb) - \Vb\|_2 \lesssim r_1(d)^2(\log d)^2/\Delta.
\end{align*}
Thus we know $r_3(d) \asymp \kappa_1 \log d  \sqrt{{r}/{n}}$ and $r_4(d) \asymp r_1(d)^2(\log d)^2/\Delta = \kappa_1(\lambda_1+\sigma^2) (\log d)^2 r/n$. 

From the proof of Corollary~\ref{col: guassian case}, we know that $\bSigma_j = n^{-1}\bLambda^{-1} \Vb^{\top} \bSigma_j^0 \Vb\bLambda^{-1}$, where
$$\mathbf{\Sigma}_j^0 \!=\!  \Big\{\!\sigma^2 \|\Pb_{\perp} \eb_j\|_2^2\bSigma \!+ 3\sigma^4 \Pb_{\perp}\eb_j \eb_j^{\top}\Pb_{\perp} \!\!-\!2\sigma^4\rho\|\Pb_{\perp} \eb_j\|_2\!\big[(\Pb_{\perp})_{[:,S]}\ub_{K+1}\eb_j^{\top}\! \Pb_{\perp}\!\!+\!\Pb_{\perp}\eb_j(\ub_{K+1})^{\top}\!(\Pb_{\perp})_{[S,:]}\big]\!\!\Big\}.$$

Similar as in the proof of Corollary~\ref{col: guassian case}, we will first define $\bSigma_j'$ as following
$$
    \bSigma_j' = \frac{1}{n}\bLambda^{-1}\Vb^{\top}\Big\{\sigma^2 \bSigma+3\sigma^4\eb_j\eb_j^{\top}-2\sigma^4\rho\|\Pb_{\perp} \eb_j\|_2\big((\Ib_d)_{[:,S]}\ub_{K+1}\eb_j^{\top}+ \eb_j\ub_{K+1}^{\top}(\Ib_d)_{[S,:]}\big)\Big\}\Vb \bLambda^{-1}.
$$
Then following similar arguments as in the proof of Corollary~\ref{col: guassian case}, we have that 
$$\|\bSigma_j - {\bSigma}_j'\|_2 = O\Big( \frac{K \sigma^4}{n\Delta^2}\sqrt{\frac{\mu}{d}} \Big)= O\Big(\frac{K\lambda_1^2}{\Delta^2}\sqrt{\frac{\mu}{d}}\Big) \frac{\sigma^4}{n\lambda_1^2} = o\big(\frac{\sigma^4}{n\lambda_1^2}\big).$$
Besides, under the condition that $\mu^2\kappa_1^4 K^3 \ll d^2$ we have 
\begin{align*}
    &\|\bSigma_j' - \widetilde\bSigma_j\|_2 \lesssim \frac{\sqrt{K}\sigma^4}{n\Delta^2} \|\Vb\|_{2,\infty}^2  \lesssim \frac{\mu K\sqrt{K}\sigma^4}{dn\Delta^2} = O\left(\frac{\mu\kappa_1^2 K\sqrt{K}}{d}\right)\frac{\sigma^4}{n\lambda_1^2} = o\Big(\frac{\sigma^4}{n\lambda_1^2}\Big).
\end{align*}

{ Then we know that $\lambda_K\big(\bSigma_j\big) \ge \frac{\sigma^4}{2n\lambda_1^2} + \frac{\sigma^2}{2n\lambda_1}$ and we can take $\eta_1(d) = \frac{\sigma^4}{2n\lambda_1^2} + \frac{\sigma^2}{2n\lambda_1}$. Thus Assumption \ref{asp: clt} holds. }Then by plugging in the above rates, we can derive the rate $r(d)$ as 
\begin{align*}
    r(d)  &= \sqrt{\frac{Kd}{pL}} \frac{r_1(d)}{\Delta}+ r_3(d)r_1(d)/\Delta + \sqrt{\frac{\mu K}{d}}r_1(d)^2/\Delta^2 + \big(r_2(d)+ r_4(d)\big)/\Delta \\
    & \lesssim \kappa_1\sqrt{\frac{Kdr}{npL}} + \frac{\kappa_1^2 (\log d)^2 r}{n} + \frac{\tilde\sigma_1^3K}{\delta^2n \Delta}(\log d)^2.
\end{align*}
Then under the condition that $L \gg \frac{Kdr}{p}\kappa_1^2(\frac{\lambda_1}{\sigma^2})$, $n \gg \kappa_1^4(\log d)^4 r^2 (\frac{\lambda_1}{\sigma^2})$ and $K(\frac{\tilde\sigma_1}{\delta})^2 \ll \kappa_1 r$, we have $\eta_1(d)^{-1/2}r(d) = o(1)$,
and hence the condition for $\eta_1(d)$ is satisfied and \eqref{eq: general clt vk large L} holds. Also recall from the above proof that  $\|\widetilde\bSigma_j - \bSigma_j\|_2 = o\big(\lambda_K(\bSigma_j)\big)$, and \eqref{eq: col gaussian L large 2} holds. 

Now we verify the validity of $\hbSigma_j$. Similar as in the proof of Corollary~\ref{col: guassian case}, it suffices to show that $\|\widehat\bSigma_j - \Hb \widetilde\bSigma_j \Hb^{\top}\|_2 = o_P\big(\lambda_K(\widetilde\bSigma_j)\big)$, and the results will hold by Slutsky's Theorem. From proof of Corollary~\ref{col: guassian case}, we have 
$$
\|\hbSigma^{\text{tr}} - \Vb\bLambda\Vb^{\top}\|_2  = O_P\Big((\lambda_1 + \sigma^2)\sqrt{\frac{r}{n}}\Big).
$$
Also, we know that with high probability
\begin{align*}
\|\widetilde\Vb^{\F} - \Vb \Hb^{\top}\|_2 &= \|\widetilde\Vb^{\F}\Hb - \Vb\|_2 \lesssim  \sqrt{\frac{Kd }{\Delta^2 pL}} r_1(d)\log d + r_1(d)\log d/\Delta    \\
& \lesssim r_1(d)\log d/\Delta   \lesssim \kappa_1 \log d \sqrt{\frac{r}{n}} .
\end{align*}
Then we have 
\begin{align*}
   & \|\widetilde\bLambda - \Hb \bLambda \Hb^{\top}\|_2 \le \|\widetilde\Vb^{\F\top}(\hbSigma^{\text{tr}} - \Vb\bLambda\Vb^{\top})\widetilde\Vb^{\F}\|_2 + \|(\widetilde\Vb^{\F}-\Vb\Hb^{\top})^{\top}(\Vb\bLambda\Vb^{\top})\widetilde\Vb^{\F}\|_2 \\
    &\quad + \|\Hb\Vb^{\top}(\Vb\bLambda\Vb^{\top})(\widetilde\Vb^{\F}-\Vb\Hb^{\top})\|_2 =O_P\Big( \lambda_1\kappa_1 \log d \sqrt{\frac{r}{n}}\Big).
\end{align*}

Then if we denote $\Db_{\bLambda} = (\widetilde{\bLambda} - \Hb \bLambda \Hb^{\top})\Hb \bLambda^{-1} \Hb^{\top}$, we have that $\|\Db_{\bLambda}\|_2 = O_P(\kappa_1^2 \log d \sqrt{\frac{r}{n}}) = o_P(1)$, and thus we have 
 \begin{align*}
     \|\widetilde{\bLambda}^{-1} - \Hb \bLambda^{-1} \Hb^{\top}\|_2 & = \|(\Hb \bLambda \Hb^{\top} + \widetilde{\bLambda} - \Hb \bLambda \Hb^{\top})^{-1} - (\Hb \bLambda \Hb^{\top})^{-1}\|_2\\
     & = \Big\|\Hb \bLambda^{-1} \Hb^{\top}\big[ (\Ib_K + \Db_{\bLambda} )^{-1} - \Ib_K\big]\Big\|_2 \le \|\bLambda^{-1}\|_2 \big\|\sum_{i=1}^{\infty}(-\Db_{\bLambda})^i\big\|_2 \\
     & = O_P\left(\kappa_1^2 \log d \sqrt{\frac{r}{n}}\right)\Delta^{-1},
 \end{align*}
 and furthermore, we have
 \begin{align*}
     \|\widetilde{\bLambda}^{-2} - \Hb \bLambda^{-2} \Hb^{\top}\|_2 & \lesssim \|\bLambda^{-1}\|_2\|\widetilde{\bLambda}^{-1} - \Hb \bLambda^{-1} \Hb^{\top}\|_2= O_P\left(\kappa_1^2 \log d \sqrt{\frac{r}{n}}\right)\Delta^{-2}.
 \end{align*}
Then following basic algebra, under the condition that $n \gg \kappa_1^4(\log d)^4r^2(\lambda_1/\sigma^2)^2$ we have 
\begin{align*}
    &\|\Hb\widetilde\bSigma_j\Hb^{\top} - \hbSigma_j\|_2 = \frac{1}{n} \| \Hb(\sigma^2 \bLambda^{-1} + \sigma^4 \bLambda^{-2})\Hb^{\top} - (\hat\sigma^2 \widetilde\bLambda^{-1} + \hat\sigma^4 \widetilde\bLambda^{-2}) \|_2 \\
    &\le \frac{1}{n} \left(\|\sigma^2 \Hb \bLambda^{-1} \Hb^{\top}- \hat\sigma^2 \widetilde\bLambda^{-1}\|_2  + \|\sigma^4 \Hb \bLambda^{-2} \Hb^{\top}- \hat\sigma^4 \widetilde\bLambda^{-2}\|_2\right)\\
    & =  O_P\left(\kappa_1^2 \log d \frac{\sigma^2}{n\Delta} \sqrt{\frac{r}{n}}\right) + O_P\left(\frac{\tilde\sigma_1}{n\Delta}\sqrt{\frac{K}{n}}\right) +O_P\left(\kappa_1^2 \log d \frac{\sigma^4}{n\Delta^2}\sqrt{\frac{r}{n}}\right) + O_P\left(\frac{\tilde{\sigma}_1\sigma^2}{n\Delta^2}\sqrt{\frac{K}{n}}\right)\\
    & = O_P\Big(\kappa_1^2\log d \big(\frac{\Delta}{\sigma^2}\big)\sqrt{\frac{r}{n}}\Big)\frac{\sigma^4}{n\Delta^2} = O_P\Big(\kappa_1^2\log d \big(\frac{\lambda_1}{\sigma^2}\big)\big(\frac{\lambda_1}{\Delta}\big)\sqrt{\frac{r}{n}}\Big)\frac{\sigma^4}{n\lambda_1^2} \\
    & = O_P\Big(\kappa_1^3 \log d \big(\frac{\lambda_1}{\sigma^2}\big) \sqrt{\frac{r}{n}} \Big)\frac{\sigma^4}{n\lambda_1^2} = o_P\big(\lambda_K(\widetilde\bSigma_j)\big).
\end{align*}
Therefore, by Slutsky's Theorem, the claim follows.

\subsection{Proof of Corollary~\ref{col: GMM L big}}\label{sec: proof col gmm L big}
From the proof of Corollary~\ref{col: GMM L small}, we have verified Assumptions \ref{asp: tail prob bound}-\ref{asp: incoh}.
It can be checked that $\Vb \bLambda^{-1}$ satisfies the two conditions for the general CLT results in the proof of Corollary~\ref{col: GMM L small}, then under the condition that $\Delta_0^2 \ll n^{4/3}/(\mu_{\theta}^2d)$, Assumption \ref{asp: clt} is also satisfied.

Now we move on to check the conditions for $\eta_1(d)$. Recall from the proof of Corollary~\ref{col: GMM L small}, we have 
$$
\Cov(\Eb_0 \Pb_{\perp} \eb_j) = \|\Pb_{\perp}\eb_j\|_2^2(\Fb \bTheta^{\top} \bTheta\Fb^{\top} + n\Ib_d)+\!  n\Pb_{\perp}\eb_j \eb_j^{\top} \Pb_{\perp}.
$$
Then we have 
\begin{align*}
    \|\widetilde\bSigma_j - \bSigma_j\|_2 \lesssim \frac{K}{d\Delta^2} (n+\Delta) \lesssim O\left(\frac{K}{dn}(n+\Delta)\right)\frac{n}{\lambda_1^2} = o\left(\frac{n}{\lambda_1^2}\right).
\end{align*}
Besides, it can be seen that $\lambda_K(\widetilde\bSigma_j) \ge n/\lambda_1^2 + 1/\lambda_1$, and hence we can take $\eta_1(d) = \lambda_1^{-2}n/2 + \lambda_1^{-1}/2$. Next we move on to verify the statistical rates $r_3(d)$ and $r_4(d)$.  By Davis-Kahan’s Theorem \citep{yu2015davis}, we have that with high probability
\begin{align*}
    \|\widehat\Vb \operatorname{sgn}(\widehat\Vb^{\top}\Vb) - \Vb\|_{2,\infty} \le  \|\widehat\Vb \operatorname{sgn}(\widehat\Vb^{\top}\Vb) - \Vb\|_2 \lesssim \|\Eb\|_2/\Delta \lesssim r_1'(d)/\Delta,
\end{align*}
where $r_1'(d) = d\Delta_0/\sqrt{K} + \sqrt{dn}\log d$ as defined in the proof of Corollary~\ref{col: GMM L small}, and thus we know that $r_3(d) \asymp r_1'(d)/\Delta $. Besides, with high probability we have 
\begin{align*}
    \|\Eb_0(\widehat\Vb (\widehat\Vb^{\top}\Vb) - \Vb) \|_{2,\infty} \le  \|\Eb_0(\widehat\Vb (\widehat\Vb^{\top}\Vb) - \Vb) \|_2 \lesssim r_1'(d)^2/\Delta,
\end{align*}
and we have $r_4(d) \asymp r_1'(d)^2/\Delta$. Thus Assumption \ref{asp: eigenspace} is satisfied. Then we have
\begin{align*}
    r(d)  &= \sqrt{\frac{Kd}{pL}}\frac{r_1(d)}{\Delta}+ r_3(d)r_1(d)/\Delta + \sqrt{\frac{ K}{d}}r_1(d)^2/\Delta^2 + \big(r_2(d)+ r_4(d)\big)/\Delta \\
    & \lesssim \sqrt{\frac{Kd}{pL}}\frac{r_1(d)}{\Delta} + r_1'(d)^2/\Delta^2 \lesssim \frac{K}{\Delta_0^2} + \frac{K^2 n (\log d)^2}{d \Delta_0^4} + \sqrt{\frac{Kd}{pL}}\Big(\frac{\sqrt{K}}{\Delta_0} + \frac{K\sqrt{n}}{\sqrt{d}\Delta_0^2}\Big).
\end{align*}
Therefore, under the conditions that $\Delta_0^2 \gg K\sqrt{n}(\log d)^2$, $n \gg d^2$ and $L \gg Kd^2/p$, we have  $\eta_1(d)^{-1/2}r(d) = o(1)$.
Thus by Theorem~\ref{thm: leading term L big}, \eqref{eq: general clt vk large L} holds. As for \eqref{eq: col GMM L big 2}, from the above arguments we have $\|\widetilde\bSigma_j - \bSigma_j\|_2 = o\big(\lambda_K(\bSigma_j)\big)$,
and hence \eqref{eq: col GMM L big 2} holds. 

Now we need to check the validity of $\hbSigma_j$. Similar as before, it suffices for us to prove that $\|\hbSigma_j - \Hb\widetilde\bSigma_j \Hb^{\top}\|_2 = o_P\big(\lambda_K(\widetilde\bSigma_j) \big)$. From Corollary~\ref{col: GMM L big}, we have that $\|\widehat\Mb - \Mb\|_2 = O_P(d\Delta_0/\sqrt{K} + \sqrt{dn})$ and
$ \|\widetilde\Vb^{\F}\Hb - \Vb\|_2 = \|\widetilde\Vb^{\F} - \Vb\Hb^{\top}\|_2  = O_P\big(K\sqrt{\frac{n}{d}}/\Delta_0^2\big) $. Then we have

\begin{align*}
   & \|\widetilde\bLambda - \Hb \bLambda \Hb^{\top}\|_2 \le \|\widetilde\Vb^{\F\top}(\widehat\Mb - \Mb)\widetilde\Vb^{\F}\|_2 + \|(\widetilde\Vb^{\F}-\Vb\Hb^{\top})^{\top}\Mb\widetilde\Vb^{\F}\|_2 \\
   &\quad + \|\Hb\Vb^{\top}\Mb(\widetilde\Vb^{\F}-\Vb\Hb^{\top})\|_2 =O_P\Big(d\Delta_0/\sqrt{K} + \sqrt{dn} \Big).
\end{align*}

Then if we denote $\Db_{\bLambda} = (\widetilde{\bLambda} - \Hb \bLambda \Hb^{\top})\Hb \bLambda^{-1} \Hb^{\top}$, we have that $$\|\Db_{\bLambda}\|_2 = O_P\left(K\sqrt{\frac{n}{d}}\Delta_0^{-2}\right) = o_P(1),$$ and thus we have 
$$
     \|\widetilde{\bLambda}^{-1} - \Hb \bLambda^{-1} \Hb^{\top}\|_2  \lesssim \|\bLambda^{-1}\|_2 \|\Db_{\bLambda}\|_2  = O_P\left(K\sqrt{\frac{n}{d}}\Delta_0^{-2}\right) \Delta^{-1} =o_P(n/\lambda_1^2) = o_P\big(\lambda_K(\widetilde\bSigma_j) \big),
     $$
 and furthermore, we have
 \begin{align*}
     n\|\widetilde{\bLambda}^{-2} \!- \!\Hb \bLambda^{-2} \Hb^{\top}\|_2 &\! \lesssim \!n\|\bLambda^{-1}\|_2\!\|\widetilde{\bLambda}^{-1}\! -\! \Hb \bLambda^{-1} \Hb^{\top}\!\|_2\!=\! O_P\big(K\sqrt{\frac{n}{d}}/\Delta_0^2\big)n\Delta^{-2} \!=\! o_P\big(\lambda_K(\widetilde\bSigma_j) \big).
 \end{align*}
Combining the above results, we have $\|\hbSigma_j - \Hb\widetilde\bSigma_j \Hb^{\top}\|_2 = o_P\big(\lambda_K(\widetilde\bSigma_j) \big)$, and hence \eqref{eq: col GMM L big 2} holds with $\widetilde\bSigma_j$ replaced by $\hbSigma_j$.

\subsection{Proof of Corollary~\ref{col: sbm}}\label{sec: proof col sbm}
The proof for the case where no self-loops are present is almost identical to the case where there are self-loops except for some modifications. We will first prove the results for the case when self-loops are present, then in the end we will discuss how to modify the proof for the case where self-loops are absent. 

We only need to verify that Assumptions \ref{asp: tail prob bound} to \ref{asp: clt} hold. Recall from the proof of Corollary~\ref{prop: err rate terms add exms} that we have 
$ \|\|\Eb\|_2\|_{\psi_1}  \lesssim r_1(d) = \sqrt{d\theta}$,
and thus we know that Assumption \ref{asp: tail prob bound} is satisfied. Also Assumption \ref{asp: stat rate biased error} holds trivially due to the unbiasedness of $\Eb$. We will then verify Assumption \ref{asp: incoh} holds under the model.
We know that $\mathbf{\Theta}\mathbf\Pi$ and $\Vb$ share the same column space, and thus there exists a non-singular matrix $\Cb \in \RR^{K\times K}$ such that $\mathbf{\Theta}\mathbf\Pi = \Vb\Cb$ and $\Vb = \mathbf{\Theta}\mathbf\Pi\Cb^{-1}$. Then we can see that $\sigma_{\min}(\Cb) = \sigma_{\min}(\mathbf{\Theta}\mathbf\Pi) \gtrsim \sqrt{d\theta/K} $, and $\|\Cb^{-1}\|_2 \lesssim \sqrt{K/d\theta}$. Hence we have $\|\Vb\|_{2,\infty} \le \|\mathbf{\Theta}\mathbf\Pi\|_{2,\infty}\|\Cb^{-1}\|_2 \lesssim \sqrt{\theta}\sqrt{K/d\theta} = \sqrt{K/d} $. Thus we can see that Assumption \ref{asp: incoh} is satisfied with $\mu = O(1)$. 

Now we move on to verify Assumption \ref{asp: eigenspace}. Recall from the proof of Corollary~\ref{prop: err rate terms add exms} that $\Delta \gtrsim d\theta/K$, $\|\Mb\|_2 \lesssim Kd\theta$, $\Mb_{ij} \asymp \theta$ and $\max_{ij}\EE (\Eb_{ij}^2) \lesssim \theta$.
By Theorem 4.2.1 in \citet{chen2020spectral}, we have that with probability $1-O(d^{-5})$,
$$
\|\widehat{\Vb} \operatorname{sgn}(\widehat{\Vb}^{\top} \Vb) - \Vb\|_{2,\infty} \lesssim \frac{K^3\sqrt{ K} + K\sqrt{K\log d}}{d\sqrt{\theta}},\quad r_3(d) \asymp \frac{K^3\sqrt{ K} + K\sqrt{K\log d}}{d\sqrt{\theta}}, 
$$
and by the proof of Theorem 4.2.1 in \citet{chen2020spectral}, we further have that with probability $1-O(d^{-7})$,
\begin{align*}
    &\|\Eb(\widehat{\Vb} (\widehat{\Vb}^{\top} \Vb) \! -\! \Vb)\|_{2,\infty} \lesssim\! \frac{K\!\sqrt{K\theta\log d}}{d\theta}\|\Eb\|_2 \!+ \!  r_3(d)(\log d \!\!+\!\!\sqrt{d\theta})\\
    &\lesssim  r_3(d)(\log d + \sqrt{d\theta}) + K\sqrt{K\log d/d}\\
    &\lesssim  \frac{K^3\sqrt{ K} + K\sqrt{K\log d}}{\sqrt{d}}, r_4(d) \asymp   \frac{K^3\sqrt{ K} + K\sqrt{K\log d}}{\sqrt{d}}.
\end{align*}
Thus Assumption \ref{asp: eigenspace} is met and now we move on to study the order of $\eta_1(d)$. Before we continue with the proof, we state the following elementary lemma that helps study the operator norm of a covariance matrix.
\begin{lemma}\label{lm: cov op norm up bd}
$\xb_1, \xb_2 \in \RR^d$ are two random vectors, then we have
$$
\|\Cov(\xb_1,\xb_2)\|_2 = \|\Cov(\xb_2,\xb_1)\|_2  \le \sqrt{\|\Cov(\xb_1)\|_2\|\Cov(\xb_2)\|_2},
$$
and
$$
\|\Cov(\xb_1 + \xb_2)\|_2 \le 2\|\Cov(\xb_1)\|_2 + 2\|\Cov(\xb_2)\|_2.
$$
\end{lemma}

The proof of Lemma \ref{lm: cov op norm up bd} can be found in Supplementary Materials~\ref{sec: proof lm cov op norm}. With the help of Lemma \ref{lm: cov op norm up bd}, we first decompose $\Eb = \Eb_1 + \Eb_2$, where $\Eb_1 = [\Eb_{ij} \mathbb{I}\{i\le j\}]$ is composed of the diagonal and upper triangular entries of $\Eb$ and $\Eb_2 = [\Eb_{ij} \mathbb{I}\{i > j\}]$ is composed of the off-diagonal lower triangular entries of $\Eb$. Then it can be seen that both $\Eb_1$ and $\Eb_2$ have independent entries. Now for $j \in [d]$, we can write 
$$\Eb \Pb_{\perp} \eb_j = \Eb \eb_j - \Eb\Vb \Vb^{\top} \eb_j = \Eb \eb_j - (\Eb_1\Vb \Vb^{\top} \eb_j + \Eb_2 \Vb \Vb^{\top} \eb_j).$$
Then we study the covariance of the three terms separately. We have
\begin{align*}
    &\Cov(\Eb \eb_j)  = \Cov(\Eb_{.j}) = \diag\big(\Mb_{1j}(1-\Mb_{1j}), \ldots, \Mb_{dj}(1-\Mb_{dj})\big);\\
    &\Cov(\Eb_1 \Vb \Vb^{\top} \eb_j)  = \diag\Big(\big[\sum_{k=1}^d \Mb_{ik}(1-\Mb_{ik})(\Pb_{\Vb}\eb_j)_k^2 \mathbb{I}\{i\le k\}\big]_{i=1}^d\Big);\\
    &\Cov(\Eb_2 \Vb \Vb^{\top} \eb_j)  = \diag\Big(\big[\sum_{k=1}^d \Mb_{ik}(1-\Mb_{ik})(\Pb_{\Vb}\eb_j)_k^2 \mathbb{I}\{i > k\}\big]_{i=1}^d\Big).
\end{align*}
Then we have $ \theta \lesssim \lambda_d \big(\Cov(\Eb \eb_j)\big) \le \|\Cov(\Eb \eb_j) \|_2  \le \max_{ij} \EE(\Eb_{ij}^2) \lesssim \theta$ and
\begin{align*}
    &\|\Cov(\Eb_1 \Vb \Vb^{\top} \eb_j)\|_2 \le \max_{i \in [d]} \sum_{k=1}^d \Mb_{ik}(1-\Mb_{ik})(\Pb_{\Vb}\eb_j)_k^2 \mathbb{I}\{i\le k\} \\
    &\le \max_{ik}\EE(\Eb_{ik})^2 \sum_{k=1}^d (\Pb_{\Vb}\eb_j)_k^2 \lesssim \theta \|\Pb_{\Vb}\eb_j\|_2^2 \le \theta \|\Vb\|_{2,\infty}^2 \le \frac{\theta  K}{d},
\end{align*}
and very similarly we also have $\|\Cov(\Eb_2 \Vb \Vb^{\top} \eb_j)\|_2 \lesssim \theta  K/d$. Thus by Lemma \ref{lm: cov op norm up bd}, we know that $\|\Cov(\Eb_1\Vb \Vb^{\top} \eb_j + \Eb_2 \Vb \Vb^{\top} \eb_j)\|_2 \lesssim \theta  K/d$ and 
$$
\|\Cov(\Eb_1\Vb \Vb^{\top} \eb_j + \Eb_2 \Vb \Vb^{\top} \eb_j, \Eb \eb_j)\|_2 \lesssim \sqrt{\theta^2  K/d} = \theta \sqrt{ K/d}.
$$
Therefore, we can write
\begin{align*}
    & \|\Cov(\Eb \Pb_{\perp}\eb_j) - \Cov(\Eb \eb_j)\|_2  \le 2\|\Cov(\Eb_1\Vb \Vb^{\top} \eb_j + \Eb_2 \Vb \Vb^{\top} \eb_j, \Eb \eb_j)\|_2\\
    & \quad +\|\Cov(\Eb_1\Vb \Vb^{\top} \eb_j + \Eb_2 \Vb \Vb^{\top} \eb_j)\|_2 \lesssim \theta \sqrt{ K/d}.
\end{align*}
Thus we have $\lambda_d \big( \Cov(\Eb \Pb_{\perp}\eb_j)\big) \ge \lambda_d \big(\Cov(\Eb \eb_j)\big) - \|\Cov(\Eb \Pb_{\perp}\eb_j) - \Cov(\Eb \eb_j)\|_2 \gtrsim \theta$, and we have $\eta_1(d) \asymp \lambda_1^{-2}\theta$. Therefore, { when $\theta = K^2 d^{-1/2+\epsilon}$ for some constant $\epsilon > 0$, $p = \Omega(\sqrt{d})$ and $L \gg K^5 d^2/p$, $K= o(d^{1/18})$}, 
we have that 
\begin{align*}
r(d) & = \Delta^{-1}\Big(\sqrt{\frac{Kd }{ pL}} r_1(d) + r_3(d)r_1(d) + \sqrt{\frac{\mu K}{d\Delta^2}}r_1(d)^2+ r_2(d)+ r_4(d)\Big)\\
&\lesssim  \frac{K^4\sqrt{ K} + K^2\sqrt{K\log d}}{d^{3/2}\theta} + K\sqrt{\frac{K}{\theta p L}} \ll \frac{1}{Kd\sqrt{\theta}} \lesssim \eta_1(d)^{1/2}.
\end{align*}
Thus $\eta_1(d)^{-1/2}r(d) = o(1)$ and the condition for the asymptotic covariance matrix is satisfied. Now we need to verify Assumption \ref{asp: clt}, and similar as in the proof of Corollary~\ref{col: guassian case}, 
we can verify the following more general result. 

Given $j \in [d]$, for any matrix $\Ab \in \RR^{d \times K}$ that satisfies the following two conditions: (1)$\|\Ab\|_{2,\infty}/\sigma_{\min} (\Ab)\le C \sqrt{\lambda_1^2 \mu K/(d \Delta^2 )} $; (2) $\lambda_K \big(\bSigma_j \big) \ge c \theta\big(\sigma_{\min}(\Ab)\big)^2$, where $\bSigma_j := \Cov(\Ab^{\top} \Eb_0 \Pb_{\perp} \eb_j )$ and $C, c > 0$ are fixed constants independent of $\Ab$, it holds that 
\begin{equation}\label{eq: asp clt2}
\mathbf{\Sigma}_j^{-1/2}\Ab^{\top} \Eb_0 \Pb_{\perp} \eb_j \overset{d}{\rightarrow} {\cN}(\mathbf{0}, \Ib_K).
\end{equation}
It can be checked from the previous proof that $\Ab = \Vb \bLambda^{-1} $ satisfies the two conditions. To show \eqref{eq: asp clt2}, we need to show that $\ab^{\top} \bSigma_j^{-1/2}\Ab^{\top} \Eb \Pb_{\perp}\eb_j \overset{d}{\rightarrow} {\cN}(0,1)$ for any $\ab \in \RR^K, \|\ab\|_2=1$. We will first study the entries of $\Pb_{\perp}\eb_j$ and $\Ab \bSigma_j^{-1/2}\ab$. It holds that 
\begin{align*}
    &|(\Pb_{\perp}\eb_j)_j|  = |\big((\Ib_d - \Vb \Vb^{\top})\eb_j\big)_j| \le 1 + \|\Vb\|_{2,\infty}^2 = 1+o(1);\\
    &\max_{i \neq j} |(\Pb_{\perp}\eb_j)_i|  = \max_{i \neq j}|\eb_i^{\top}\eb_j - \eb_i^{\top}\Vb \Vb^{\top}\eb_j| \le 0+ \|\Vb\|_{2,\infty}^2 = \frac{ K}{d};\\
    &\|\Ab \bSigma_j^{-1/2}\ab\|_{\infty} \le \|\Ab\|_{2,\infty} \|\bSigma_j^{-1/2}\|_2 \lesssim \theta^{-1/2} \|\Ab\|_{2,\infty}/\sigma_{\min}(\Ab) \lesssim K^2\sqrt{\frac{ K}{d\theta}}.
\end{align*}
 Then we know that 
 \begin{align*}
     &\ab^{\top} \bSigma_j^{-1/2}\Ab^{\top} \Eb \Pb_{\perp}\eb_j = \sum_{ik} \Eb_{ik} (\Ab \bSigma_j^{-1/2}\ab)_i (\Pb_{\perp}\eb_j)_k = \sum_{i=1}^d \Eb_{ii} (\Ab \bSigma_j^{-1/2}\ab)_i (\Pb_{\perp}\eb_j)_i\\
     & \quad + \sum_{i < k} \Eb_{ik} \big[ (\Ab \bSigma_j^{-1/2}\ab)_i (\Pb_{\perp}\eb_j)_k + (\Ab \bSigma_j^{-1/2}\ab)_k (\Pb_{\perp}\eb_j)_i \big].
 \end{align*}
Then for the diagonal entries we have
 \begin{align*}
     & \sum_{i=1}^d \EE |\Eb_{ii} (\Ab \bSigma_j^{-1/2}\ab)_i (\Pb_{\perp}\eb_j)_i|^3 \\
     &= \EE |\Eb_{jj} (\Ab \bSigma_j^{-1/2}\ab)_j (\Pb_{\perp}\eb_j)_j|^3 + \sum_{i \neq j} \EE |\Eb_{ii} (\Ab \bSigma_j^{-1/2}\ab)_i (\Pb_{\perp}\eb_j)_i|^3\\
     & \quad \lesssim \theta \|\Ab \bSigma_j^{-1/2}\ab\|_{\infty}^3 + d\theta \|\Ab \bSigma_j^{-1/2}\ab\|_{\infty}^3 \max_{i \neq j} |(\Pb_{\perp}\eb_j)_i|^3 \lesssim \frac{K^6}{d}\sqrt{\frac{K^3}{d \theta}},
 \end{align*}
 and for the off-diagonal entries, when $K = o(d^{1/26})$ it holds that
 \begin{align*}
      & \sum_{i < k} \EE\Big|\Eb_{ik} \big[ (\Ab \bSigma_j^{-1/2}\ab)_i (\Pb_{\perp}\eb_j)_k + (\Ab \bSigma_j^{-1/2}\ab)_k (\Pb_{\perp}\eb_j)_i \big]\Big|^3  \lesssim d \theta \|\Ab \bSigma_j^{-1/2}\ab\|_{\infty}^3\\
      &\quad + d^2 \theta \|\Ab \bSigma_j^{-1/2}\ab\|_{\infty}^3 \big(\frac{ K}{d}\big)^3 \lesssim K^6 \sqrt{\frac{K^3}{d \theta}} = o(1).
 \end{align*}
Moreover, since $\operatorname{Var}(\ab^{\top} \bSigma_j^{-1/2}\Ab^{\top} \Eb \Pb_{\perp}\eb_j ) = 1$, by the Lyapunov's condition and plugging in $\Ab = \Vb \bLambda^{-1}$, Assumption \ref{asp: clt} is met and \eqref{eq: general clt vk large L} follows. 
 
 Now we only need to verify that the result also holds when replacing $\bSigma_j$ by $\widetilde{\bSigma}_j$. From previous discussion we learnt that 
 \begin{align*}
     &\|\widetilde{\bSigma}_j - \bSigma_j\|_2 \le \|\Vb\mathbf{\Lambda}^{-1}\|_2^2 \|\Cov(\Eb \Pb_{\perp}\eb_j) - \Cov(\Eb \eb_j)\|_2 \\
     &\le \frac{K^2}{d^2\theta} \sqrt{\frac{ K}{d}} \lesssim K^4 \sqrt{\frac{ K}{d}} \lambda_K(\widetilde{\bSigma}_j) = o(\lambda_K(\widetilde{\bSigma}_j)).
 \end{align*}
 Then by Slutsky's Theorem, \eqref{eq: col SBM 2} holds.

 Now we verify the validity of $\hbSigma_j$. Similar as in the proof of Corollary~\ref{col: inf gaussian case L big}, $\Hb$ is orthonormal with probability $1-o(1)$, and we will start by showing that $\|\hbSigma_j - \Hb\widetilde\bSigma_j\Hb^{\top}\|_2 = o_P\big(\lambda_K(\widetilde\bSigma_j)\big)$. From previous discussion we have the following bounds
 $$
 \|\widehat\Mb - \Mb\|_2 = O_P(\sqrt{d\theta}),\quad \|\widetilde\Vb^{\F} - \Vb \Hb^{\top}\|_2 =  \|\widetilde\Vb^{\F}\Hb - \Vb \|_2 = O_P(\frac{K}{\sqrt{d\theta}}), 
 $$
 and 
 \begin{align*}
      \|\widetilde\Vb^{\F}\Hb - \Vb \|_{2,\infty} &\le \|\widetilde\Vb^{\F}\Hb - \widehat\Vb \Hb_0\|_2 + \|\widehat\Vb \Hb_0 - \Vb\|_{2,\infty} = o_P(\frac{1}{Kd\sqrt{\theta}})\\
      &\quad + O_P(\frac{K^3\sqrt{ K}+ K\sqrt{K\log d}}{d\sqrt{\theta}}) = O_P(\frac{K^3\sqrt{ K} + K\sqrt{K\log d}}{d\sqrt{\theta}}).
 \end{align*}
With the help of the above results, we will study the components of $\hbSigma_j - \Hb\widetilde\bSigma_j\Hb^{\top}$ separately. In the following proof, we will base the discussion on the event that $\Hb$ is orthonormal. We first study $\widetilde\Mb = (\widetilde\Vb^{\F}\widetilde\Vb^{\F\top}) \widehat\Mb (\widetilde\Vb^{\F}\widetilde\Vb^{\F\top}) = \widetilde\Vb^{\F}\Hb (\Hb^{\top}\widetilde\Vb^{\F\top} \widehat\Mb \widetilde\Vb^{\F} \Hb)\Hb^{\top}\widetilde\Vb^{\F\top}$. We have that 
\begin{align*}
    & \|\Hb^{\top}\widetilde\Vb^{\F\top} \widehat\Mb \widetilde\Vb^{\F}\Hb- \bLambda\|_2  \le \|\Hb^{\top}\widetilde\Vb^{\F\top} \widehat\Mb \widetilde\Vb^{\F}\Hb - \Hb^{\top}\widetilde\Vb^{\F\top} \Mb \widetilde\Vb^{\F}\Hb\|_2 \\
    &\quad + \|\Hb^{\top}\widetilde\Vb^{\F\top} \Mb (\widetilde\Vb^{\F}\Hb - \Vb)\|_2 + \|(\widetilde\Vb^{\F}\Hb - \Vb)^{\top}\Mb\Vb\|_2 \\
    &\le \|\widehat\Mb - \Mb\|_2 + 2 \|\Mb\|_2\|\widetilde\Vb^{\F}\Hb - \Vb\|_2 = O_P(K^2\sqrt{d\theta}).
\end{align*}
Then for $i,k \in [d]$, we have
\begin{align*}
    & |\widetilde\Mb_{ik}-\Mb_{ik}|  = |(\widetilde\Vb^{\F}\Hb)_i^{\top} (\Hb^{\top}\widetilde\Vb^{\F\top} \widehat\Mb \widetilde\Vb^{\F}\Hb)(\widetilde\Vb^{\F}\Hb)_k - \Mb_{ik}|\\
    &\le |(\widetilde\Vb^{\F}\Hb)_i^{\top} (\Hb^{\top}\widetilde\Vb^{\F\top} \widehat\Mb \widetilde\Vb^{\F}\Hb - \bLambda)(\widetilde\Vb^{\F}\Hb)_k| + |(\widetilde\Vb^{\F}\Hb -\Vb)_i \bLambda (\widetilde\Vb^{\F}\Hb)_k| \\
    &\quad + |(\Vb)_i \bLambda (\widetilde\Vb^{\F}\Hb - \Vb)_k|. 
\end{align*}
It is not hard to see that 
\begin{align*}
     &|(\widetilde\Vb^{\F}\Hb)_i^{\top} (\Hb^{\top}\widetilde\Vb^{\F\top} \widehat\Mb \widetilde\Vb^{\F}\Hb - \bLambda)(\widetilde\Vb^{\F}\Hb)_k| \lesssim \|\Hb^{\top}\widetilde\Vb^{\F\top} \widehat\Mb \widetilde\Vb^{\F}\Hb- \bLambda\|_2 \|\widetilde\Vb^{\F}\Hb\|_{2,\infty}^2 \\
     & = O_P(K^2\sqrt{d\theta}\|\widetilde\Vb^{\F}\Hb\|_{2,\infty}^2 ) = O_P\left( K^3 \sqrt{\frac{\theta}{d}}\right),\\
     & |(\widetilde\Vb^{\F}\Hb -\Vb)_i \bLambda (\widetilde\Vb^{\F}\Hb)_k| + |(\Vb)_i \bLambda (\widetilde\Vb^{\F}\Hb - \Vb)_k|  \\
     & = O_P(Kd\theta\|\Vb\|_{2,\infty}\|\widehat\Vb \Hb_0 - \Vb\|_{2,\infty}) = O_P\left(K^3( K^2 + \sqrt{ \log d})\sqrt{\frac{\theta}{d}}\right),
\end{align*}
and in turn we have the upper bound 
\begin{align*}
    & |\widetilde\Mb_{ik}-\Mb_{ik}|  = O_P\left( K^3 \sqrt{\frac{\theta}{d}}\right) + O_P\left(K^3( K^2 + \sqrt{ \log d})\sqrt{\frac{\theta}{d}}\right)\\
    & = O_P\Big(\frac{K^3( K^2 + \sqrt{ \log d})}{\sqrt{d\theta}}\Big)\theta = o_P(\theta) = o_P(\Mb_{ik}). 
\end{align*}
 Thus we have 
 $$
 \|\diag\big([\widetilde\Mb_{ij}(1-\widetilde\Mb_{ij})]_{i=1}^d\big) - \diag\big([\Mb_{ij}(1-\Mb_{ij})]_{i=1}^d\big) \|_2 = O_P\Big(\frac{K^3( K^2 + \sqrt{ \log d})}{\sqrt{d\theta}}\theta\Big). 
 $$
 Then we move on to study $\widetilde{\bLambda}$. We have
 \begin{align*}
     & \|\widetilde{\bLambda} - \Hb \bLambda \Hb^{\top}\|_2 \le \|\widetilde\Vb^{\F \top}( \widehat\Mb - \Mb) \widetilde{\Vb}^{\F} \|_2 + \|(\widetilde\Vb^{\F} - \Vb\Hb^{\top})^{\top}\Mb\widetilde\Vb^{\F}\|_2\\
     &\quad + \|\Hb\Vb^{\top}\Mb(\widetilde\Vb^{\F} - \Vb\Hb^{\top})\|_2 = O_P(\sqrt{d\theta}) + O_P(K^2\sqrt{d\theta}) = O_P(K^2\sqrt{d\theta}).
 \end{align*}
 Then if we denote $\Db_{\bLambda} = (\widetilde{\bLambda} - \Hb \bLambda \Hb^{\top})\Hb \bLambda^{-1} \Hb^{\top}$, we have that $\|\Db_{\bLambda}\|_2 = O_P(K^3/\sqrt{d\theta}) = o_P(1)$, and thus we have 
 \begin{align*}
     \|\widetilde{\bLambda}^{-1} - \Hb \bLambda^{-1} \Hb^{\top}\|_2 & = \|(\Hb \bLambda \Hb^{\top} + \widetilde{\bLambda} - \Hb \bLambda \Hb^{\top})^{-1} - (\Hb \bLambda \Hb^{\top})^{-1}\|_2\\
     & = \Big\|\Hb \bLambda^{-1} \Hb^{\top}\big[ (\Ib_K + \Db_{\bLambda} )^{-1} - \Ib_K\big]\Big\| \le \|\bLambda^{-1}\|_2 \big\|\sum_{i=1}^{\infty}(-\Db_{\bLambda})^i\big\|_2 \\
     & = O_P\big({K^4}/{(d\theta)^{3/2}} \big).
 \end{align*}
  Thus, following basic algebra we have the following bounds
 \begin{align*}
      &\|\widetilde\Vb^{\F\top} \diag\big([\widetilde\Mb_{ij}(1-\widetilde\Mb_{ij})]_{i=1}^d\big)\widetilde\Vb^{\F} - \Hb \Vb^{\top}\diag\big([\Mb_{ij}(1-\Mb_{ij})]_{i=1}^d\big)\Vb\Hb^{\top} \|_2 \\
      &\le  \|\widetilde\Vb^{\F\top}\Big(\diag\big([\widetilde\Mb_{ij}(1-\widetilde\Mb_{ij})]_{i=1}^d\big) - \diag\big([\Mb_{ij}(1-\Mb_{ij})]_{i=1}^d\big)\Big)\widetilde\Vb^{\F} \|_2 \\
      &\quad + 2\|\widetilde\Vb^{\F} - \Vb\Hb^{\top}\|_2\|\diag\big([\Mb_{ij}(1-\Mb_{ij})]_{i=1}^d\big)\|_2 = O_P\Big(\frac{K^3( K^2 + \sqrt{ \log d})}{\sqrt{d\theta}}\Big)\theta,
 \end{align*}
 and further, under the condition that $K = o(d^{1/32})$, we have
 \begin{align*}
     \|\widehat\bSigma_j - \Hb \widetilde\bSigma_j\Hb^{\top}\|_2 & \lesssim O_P\Big(\frac{K^3( K^2 + \sqrt{ \log d})}{\sqrt{d\theta}}\theta\Big)\|\widetilde\bLambda^{-1}\|_2^2 + \theta\|\bLambda^{-1}\|_2\|\widetilde{\bLambda}^{-1} - \Hb \bLambda^{-1} \Hb^{\top}\|_2\\
     & = O_P\Big(\frac{K^7( K^2 + \sqrt{ \log d})}{\sqrt{d\theta}}\Big)\frac{1}{K^2d^2\theta} + O_P(\frac{K^7}{\sqrt{d\theta}})\frac{1}{K^2d^2\theta}\\
     &=O_P\Big(\frac{K^7( K^2 + \sqrt{ \log d})}{\sqrt{d\theta}}\Big)\frac{1}{K^2 d^2\theta} = o_P\big(\lambda_K(\widetilde\bSigma_j)\big).
 \end{align*}
 Thus with similar arguments as in the proof of Corollary~\ref{col: inf gaussian case L big}, the claim follows.
 \begin{remark}\label{rmk: exm2 inf no-self-loop}
 The inferential results also hold for the case where self-loops are absent. Recall that under the no-self-loop case, the observed matrix is 
 $$\widehat\Mb = \Xb - \diag(\Xb) = \Mb + \Eb - \diag(\Mb + \Eb) = \Mb + \Eb - \diag(\Eb) -\diag(\Mb),$$
 where $\Eb = \Xb - \Mb$ is the error matrix between the adjacency matrix with self-loops and its expectation. We define $\widehat\Mb' = \Mb + \Eb - \diag(\Eb)$ and denote by $\widehat\Vb'$ its $K$ leading eigenvectors. By Weyl's inequality \citep{franklin2012matrix} we know that with probability at least $1 - d^{-10}$ we have that 
 $\sigma_K(\widehat\Mb') - \sigma_{K+1}(\widehat\Mb') \ge \Delta - O(\sqrt{d\theta}) \gtrsim d\theta/K$, and hence by Davis-Kahan’s Theorem \citep{yu2015davis} we have 
 $$
 \|\widehat\Vb \widehat\Vb^{\top} - \widehat\Vb' \widehat\Vb'^{\top}\|_2 \le \|\diag(\Mb)\|_2/\big(\sigma_K(\widehat\Mb') - \sigma_{K+1}(\widehat\Mb')\big) \lesssim K/d,
 $$
 with probability at least $1 - d^{-10}$.
 The verification of Assumptions \ref{asp: tail prob bound}, \ref{asp: incoh} and \ref{asp: clt} when self-loops are present can also be applied to the no-self-loop case. For Assumption~\ref{asp: stat rate biased error}, we can take $\Eb_0 = \Eb - \diag(\Eb)$ and $\Eb_b = -\diag(\Mb)$. Then $r_2(d) = \|\diag(\Mb)\|_2 \lesssim \theta = o(r_1(d))$ and Assumption~\ref{asp: stat rate biased error} is satisfied. As for Assumption~\ref{asp: eigenspace}, by Lemma~7 in \citet{fandistributed2019}, we have 
 $$
 \|\sgn(\widehat\Vb'^{\top}\Vb) - \widehat\Vb'^{\top}\Vb\|_2 \lesssim \|\widehat\Vb'\widehat\Vb'^{\top} - \Vb\Vb^{\top}\|_2^2 \lesssim \frac{K^2}{d\theta}, 
 $$
 $$
 \|\sgn(\widehat\Vb^{\top}\widehat\Vb') - \widehat\Vb^{\top}\widehat\Vb'\|_2 \lesssim \|\widehat\Vb\widehat\Vb^{\top} - \widehat\Vb'\widehat\Vb'^{\top}\|_2^2\lesssim\frac{K^2}{d^2}.
 $$
 With similar arguments as in the self-loop case, for $\widehat\Vb'$ with high probability we have 
 $$
 \|\widehat{\Vb}' \operatorname{sgn}(\widehat{\Vb}'^{\top} \Vb) - \Vb\|_{2,\infty} \lesssim \frac{K^3\sqrt{ K} + K\sqrt{K\log d}}{d\sqrt{\theta}},
$$
$$
\|\Eb_0(\widehat{\Vb}' (\widehat{\Vb}'^{\top} \Vb) \! -\! \Vb)\|_{2,\infty} \lesssim \frac{K^3\sqrt{ K} + K\sqrt{K\log d}}{\sqrt{d}}.
 $$
 Then for $\widehat\Vb$, with high probability we have that 
 \begin{align*}
     &\|\widehat{\Vb} \operatorname{sgn}(\widehat{\Vb}^{\top} \Vb) - \Vb\|_{2,\infty}  \le \|\widehat{\Vb} (\widehat{\Vb}^{\top} \Vb) - \Vb\|_{2,\infty} + \|\widehat{\Vb} \big(\operatorname{sgn}(\widehat{\Vb}^{\top} \Vb) - \widehat{\Vb}^{\top} \Vb\big) - \Vb\|_{2,\infty}\\
     & \le \|\widehat{\Vb}\big (\widehat{\Vb}^{\top}(\Ib_d - \widehat\Vb'\widehat\Vb'^{\top})\Vb\big) - \Vb\|_{2,\infty} + \|\widehat{\Vb} (\widehat{\Vb}^{\top} \widehat\Vb'\widehat\Vb'^{\top}\Vb) - \Vb\|_{2,\infty}  + O(\frac{K^2}{d\theta})\|\widehat\Vb\|_{2,\infty}\\
     &  \le  \|\widehat{\Vb} (\widehat{\Vb}^{\top} \widehat\Vb') \!- \!\widehat\Vb'\|_{2,\infty}\! +\!  \|\widehat{\Vb}' (\widehat\Vb'^{\top}\Vb) \!-\! \Vb\|_{2,\infty} \!+\! O(\frac{K^2}{d\theta})\|\widehat\Vb\|_{2,\infty}  \!+\! \|\widehat\Vb \widehat\Vb^{\top} \!-\! \widehat\Vb' \widehat\Vb'^{\top}\|_2\|\widehat\Vb\|_{2,\infty} \\
     &  \le O(\frac{K^2}{d\theta}) \!\!\left(\|\Vb\|_{2,\infty} \!+\! \|\widehat{\Vb} \operatorname{sgn}(\widehat{\Vb}^{\top} \Vb) - \Vb\|_{2,\infty}\!\right) \!+\! \|\widehat\Vb \widehat\Vb^{\top}\! -\! \widehat\Vb' \widehat\Vb'^{\top}\|_2 \!+\!  \|\widehat{\Vb}' (\widehat\Vb'^{\top}\Vb) \!-\! \Vb\|_{2,\infty} ,
 \end{align*}
 where in the last two inequalities we use the fact that 
 $$
 \|(\Ib_d - \widehat\Vb\widehat\Vb^{\top})\widehat\Vb'\|_2 = \|(\Ib_d - \widehat\Vb'\widehat\Vb'^{\top})\widehat\Vb\|_2 = \|\widehat\Vb_{\perp}^{\top}\widehat\Vb'\|_2 = \|\widehat\Vb_{\perp}'^{\top}\widehat\Vb\|_2  = \|\widehat\Vb \widehat\Vb^{\top} - \widehat\Vb' \widehat\Vb'^{\top}\|_2,
 $$
with $\widehat\Vb_{\perp}$ and $\widehat\Vb'_{\perp}$ being the orthogonal complement of $\widehat\Vb$ and $\widehat\Vb'$ respectively. Since $K^2/(d\theta) = o(1)$, for large enough $d$ we further get 
\begin{align*}
    &\frac{1}{2} \|\widehat{\Vb} \operatorname{sgn}(\widehat{\Vb}^{\top} \Vb) - \Vb\|_{2,\infty} \le \left(1 -  O\left({K^2}/{(d\theta)}\right)\right)\|\widehat{\Vb} \operatorname{sgn}(\widehat{\Vb}^{\top} \Vb) - \Vb\|_{2,\infty}\\
    &\quad \le  O(\frac{K^2}{d\theta})\|\Vb\|_{2,\infty} + O(\frac{K}{d}) + \|\widehat{\Vb}' \sgn(\widehat\Vb'^{\top}\Vb) - \Vb\|_{2,\infty} + O(\frac{K^2}{d\theta})\|\widehat\Vb'\|_{2,\infty} \\
    &\quad \lesssim \frac{K^2}{d\theta}\sqrt{\frac{K}{d}} + \frac{K}{d} + \frac{K^3\sqrt{ K} + K\sqrt{K\log d}}{d\sqrt{\theta}} \lesssim \frac{K^3\sqrt{ K} + K\sqrt{K\log d}}{d\sqrt{\theta}}. 
\end{align*}
Hence $r_3(d) = K\sqrt{K}(K^2 + \sqrt{\log d})/(d\sqrt{\theta})$. We also have 
 \begin{align*}
     \|\Eb(\widehat{\Vb} (\widehat{\Vb}^{\top} \Vb)  - \Vb)\|_{2,\infty} & \lesssim \|\Eb_0(\widehat{\Vb} (\widehat{\Vb}^{\top} \Vb)  - \Vb)\|_{2,\infty} + \frac{r_2(d)r_1(d)}{\Delta}\\
     & \lesssim \|\Eb_0(\widehat{\Vb}' (\widehat{\Vb}'^{\top} \Vb)  - \Vb)\!\|_{2,\infty} + \frac{r_2(d)r_1(d)}{\Delta }\\
     &\lesssim \frac{K^3\sqrt{ K} + K\sqrt{K\log d}}{\sqrt{d}} + K\sqrt{\frac{\theta}{d}} \lesssim \frac{K^3\sqrt{ K} + K\sqrt{K\log d}}{\sqrt{d}},
 \end{align*}
 and hence we can take $r_4(d) = K\sqrt{K}(K^2 + \sqrt{\log d})/\sqrt{d}$. Now to get a sharper rate for $r(d)$, we take into consideration the diagonal structure of $\Eb_b$ and derive the following bound
 \begin{align*}
     \|\!\Pb_{\perp}\Eb_b \widehat{\Vb}\widehat{\Hb}_0 \mathbf{\Lambda}^{-1}\!\|_{2,\infty} & \le  \|\Vb\Vb^{\top}\Eb_b \widehat{\Vb}\widehat{\Hb}_0 \mathbf{\Lambda}^{-1}\!\|_{2,\infty} +  \|\Eb_b \widehat{\Vb}\!\widehat{\Hb}_0 \mathbf{\Lambda}^{-1}\|_{2,\infty}\\
     & \le \frac{r_2(d)\|\Vb\|_{2,\infty}}{\Delta} +\frac{\|\diag(\Mb) \widehat\Vb\|_{2,\infty}}{\Delta}\\
     & \lesssim \frac{K}{d}\sqrt{\frac{K}{d}} + \frac{\|\diag(\Mb)\|_2\| \widehat\Vb\|_{2,\infty}}{\Delta} \lesssim \frac{K}{d}\sqrt{\frac{K}{d}}.
 \end{align*}
Then from the proof of Theorem~\ref{thm: leading term L big} we have that $$r(d) \lesssim \frac{K^4\sqrt{ K} + K^2\sqrt{K\log d}}{d^{3/2}\theta} + K\sqrt{\frac{K}{\theta p L}} + \frac{K}{d}\sqrt{\frac{K}{d}} \ll \frac{1}{Kd\sqrt{\theta}},$$
 and we are only left to verify the minimum eigenvalue condition of $\bSigma_j$ by showing that the order of $\eta_1(d)$ is the same as when there are self-loops. With the same arguments, we know that 
$$
\|\Cov(\bLambda^{-1}\Vb^{\top}\Eb'\Pb_{\perp}\eb_j ) - \Cov(\bLambda^{-1}\Vb^{\top}\Eb'\eb_j )\| \le O\left(K^4\sqrt{\frac{K}{d}}\right)\frac{1}{K^2d^2\theta}. 
$$
Besides, we also have 
\begin{align*}
    &\|\Cov(\bLambda^{-1}\Vb^{\top}\Eb'\eb_j ) - \widetilde\bSigma_j\|_2 = \big\|\bLambda^{-1}\Vb^{\top} \big(\Mb_{jj}(1-\Mb_{jj})\eb_j\eb_j^{\top}\big)\Vb\bLambda^{-1}\big\|_2\\
    & \lesssim \Mb_{jj}\|\bLambda^{-1}\|_2^2 \|\Vb\|_{2,\infty}^2  \lesssim \frac{K^2}{d^2 \theta} \frac{K}{d} = O(\frac{K^5}{d}) \frac{1}{K^2d^2 \theta} = o\big(\lambda_K(\widetilde\bSigma_j)\big).
\end{align*}
Thus we also have $\|\Cov(\bLambda^{-1}\Vb^{\top}\Eb'\Pb_{\perp}\eb_j ) - \widetilde\bSigma_j\|_2 = o\big(\lambda_K(\widetilde\bSigma_j)\big)$, and thereby 
$$ \lambda_K\big(\Cov(\bLambda^{-1}\Vb^{\top}\Eb'\Pb_{\perp}\eb_j ) \big) = \lambda_K\big(\widetilde\bSigma_j\big)\big(1 + o(1)\big) \gtrsim \frac{\theta}{K^2d^2\theta^2}.$$ 
Thus we  still have $\eta_1(d) = \lambda_1^{-2}\theta$ for the case where self-loops are absent. The condition for $\eta_1(d)$ also holds for the no-self-loop case and both \eqref{eq: general clt vk large L} and \eqref{eq: col SBM 2} hold.  The verification of \eqref{eq: est cov SBM} is almost identical to the self-loop case and is hence omitted.
\end{remark}
 
\subsection{Proof of Corollary~\ref{col: missing mat}}\label{sec: proof col missing mat}
Recall that $\widehat\Mb =  ({1}/{\hat\theta})\cP_{\cS}(\Mb + \bar{\cE}) $ and $ \widehat\Mb' = ({1}/{\theta})\cP_{\cS}(\Mb + \bar{\cE}) $ share exactly the same sequence of eigenvectors, and we can treat $\widetilde\Vb^{\F}$ as the FADI estimator applied to $ \widehat\Mb'$. We will abuse the notation and denote $\Eb := \widehat\Mb' - \Mb$.

To show that \eqref{eq: general clt vk large L} holds, we need to verify that Assumptions \ref{asp: tail prob bound} to \ref{asp: clt} hold and the minimum eigenvalue conditions hold for the asymptotic covariance matrix. We know from Corollary~\ref{prop: err rate terms add exms} that Assumption \ref{asp: tail prob bound} and Assumption \ref{asp: stat rate biased error} are satisfied, and that $r_1(d) = |\lambda_1| \mu K/\sqrt{d\theta} + \sqrt{d\sigma^2/\theta}$ and $r_2(d) = 0$. Define $\tilde{\sigma} = (|\lambda_1|\mu K/d)\vee \sigma$, we have from the proof of Corollary~\ref{prop: err rate terms add exms} that $\operatorname{Var}(\Eb_{ij}) \asymp \tilde\sigma^2/\theta $ and $|\Eb_{ij}| = O(\tilde{\sigma}\log d/\theta)$ for $i,j \in [d]$. From Theorem 4.2.1 in \citet{chen2020spectral}, we have that with probability $1-O(d^{-5})$
\begin{align*}
    &\|\widehat\Vb\operatorname{sgn}(\widehat\Vb^{\top}\Vb) - \Vb\|_{2,\infty} \lesssim \frac{\kappa_2\tilde{\sigma}\sqrt{\mu K/\theta} + \tilde{\sigma}\sqrt{K \log d/\theta}}{\Delta},
\end{align*}
and thus we know $ r_3(d) \asymp \big(\kappa_2\tilde{\sigma}\sqrt{\mu K/\theta} + \tilde{\sigma}\sqrt{K \log d/\theta}\big)/{\Delta}$. Besides, by the proof of Theorem 4.2.1 in \citet{chen2020spectral}, with probability $1-O(d^{-7})$, we have
\begin{align*}
    \big\|\Eb \big(\widehat\Vb(\widehat\Vb^{\top}\Vb) \!- \!\!\Vb\big)\!\big\|_{2,\infty} \!&\lesssim\! \frac{\sqrt{dK}\tilde{\sigma}^2}{\Delta{\theta}}\big(\!\sqrt{\log d}\!+\!\!\sqrt{\mu }\big)\! \!+\! \tilde{\sigma}\sqrt{\frac{d}{\theta}}r_3(d)\!+\!\frac{\tilde{\sigma}}{\Delta}\!\sqrt{K\frac{\log d}{\theta}}\|\Eb\|_2\\
    & \quad \lesssim \frac{\sqrt{d}\tilde{\sigma}^2}{\Delta{\theta}}\big(\sqrt{K\log d}+\kappa_2\sqrt{\mu K}\big),
\end{align*}
and thus $r_4(d) \asymp \frac{\sqrt{d}\tilde{\sigma}^2}{\Delta{\theta}}\big(\sqrt{K\log d}+\kappa_2\sqrt{\mu K}\big)$. Therefore, Assumption \ref{asp: eigenspace} is met and we have 
\begin{align*}
    r(d)  &= \sqrt{\frac{Kd}{pL}}\frac{r_1(d)}{\Delta}+ r_3(d)r_1(d)/\Delta + \sqrt{\frac{\mu K}{d}}r_1(d)^2/\Delta^2 + \big(r_2(d)+ r_4(d)\big)/\Delta\\
    & \lesssim \left(\frac{\sqrt{d}\tilde{\sigma}}{\Delta \sqrt{\theta}}\right)\left(\left(\frac{\kappa_2 \sqrt{\mu K} + \sqrt{K \log d}}{\Delta}\right)\frac{\tilde{\sigma}}{\sqrt{\theta}} + \sqrt{\frac{Kd}{pL}}\right).
\end{align*}

Now we will study the statistical rate $\eta_1(d)$. We know that $\Eb_{ij} = \Eb_{ji}$ are i.i.d. across $i \le j$ and $\operatorname{Var}(\Eb_{ij}) \asymp \tilde{\sigma}^2/\theta$, then by Lemma \ref{lm: cov op norm up bd}, with almost identical arguments as in the proof of Corollary~\ref{col: sbm}, for $j \in [d]$ we have that $\|\Cov(\Eb\Pb_{\perp}\eb_j) - \Cov(\Eb\eb_j)\|_2 \lesssim \tilde{\sigma}^2/\theta \sqrt{\mu K/d}$, and thus $\lambda_d\big( \Cov(\Eb\Pb_{\perp}\eb_j) \big) \gtrsim \lambda_d\big( \Cov(\Eb\eb_j) \big)  \gtrsim \tilde{\sigma}^2/\theta$ and we have $\eta_1(d) \asymp \lambda_1^{-2} \theta^{-1 }\tilde{\sigma}^2$. Therefore, under the condition that { $L \gg \kappa_2^2 Kd^2/p$ and $\tilde{\sigma}/\Delta \sqrt{d/\theta} \ll \min\left(\big(\kappa_2^2\sqrt{\mu K} + \kappa_2\sqrt{K \log d}\big)^{-1}, \sqrt{p/d}\right)$}, we have that $\eta_1(d)^{-1/2}r(d) = o(1)$.


Now we move on to verify Assumption \ref{asp: clt}. More specifically, we will show that the following results hold:

Given $j \in [d]$, for any matrix $\Ab \in \RR^{d \times K}$ that satisfies the following two conditions: (1)$\|\Ab\|_{2,\infty}/\sigma_{\min} (\Ab)\le C\sqrt{\lambda_1^2 \mu K/(d \Delta^2 )} $; (2) $\lambda_K \big(\bSigma_j \big) \ge c \tilde\sigma^2\theta^{-1}\big(\sigma_{\min}(\Ab)\big)^2$, where $\bSigma_j := \Cov(\Ab^{\top} \Eb_0 \Pb_{\perp} \eb_j )$ and $C,c >0$ are fixed constants independent of $\Ab$, it holds that 
\begin{equation}\label{eq: asp clt2 mismat}
\mathbf{\Sigma}_j^{-1/2}\Ab^{\top} \Eb_0 \Pb_{\perp} \eb_j \overset{d}{\rightarrow} {\cN}(\mathbf{0}, \Ib_K).
\end{equation}
To prove \eqref{eq: asp clt2 mismat}, it suffices to show that $\ab^{\top} \bSigma_j^{-1/2}\Ab^{\top} \Eb \Pb_{\perp}\eb_j \overset{d}{\rightarrow} {\cN}(0,1)$ for any $\ab \in \RR^K, \|\ab\|_2=1$. We will first study $\Pb_{\perp}\eb_j$, $\Ab \bSigma_j^{-1/2}\ab$ and $\max_{ik}\EE|\Eb_{ik}|^3$. It holds that 
\begin{align*}
    &|(\Pb_{\perp}\eb_j)_j|  = |\big((\Ib_d - \Vb \Vb^{\top})\eb_j\big)_j| \le 1 + \|\Vb\|_{2,\infty}^2 = 1+o(1);\\
    &\max_{i \neq j} |(\Pb_{\perp}\eb_j)_i|  = \max_{i \neq j}|\eb_i^{\top}\eb_j - \eb_i^{\top}\Vb \Vb^{\top}\eb_j| \le 0+ \|\Vb\|_{2,\infty}^2 = \frac{\mu K}{d};\\
    &\|\Ab \bSigma_j^{-1/2}\ab\|_{\infty} \le \|\Ab\|_{2,\infty} \|\bSigma_j^{-1/2}\|_2 \lesssim (\tilde\sigma^2/{\theta})^{-1/2} \|\Ab\|_{2,\infty}/{\sigma}_{\min}(\Ab) \lesssim \kappa_2 \sqrt{\frac{\mu K}{d}} \frac{\sqrt{\theta}}{\tilde{\sigma}};\\
    &\max_{ik}\EE|\Eb_{ik}|^3 \lesssim \frac{\|\Mb\|_{\max}^3}{\theta^3}\theta + \frac{{\sigma}^3 (\log d)^3}{\theta^3}\theta \lesssim \frac{\tilde{\sigma}^3}{\theta^2}(\log d)^3. 
\end{align*}
 Then we know that 
 \begin{align*}
     &\ab^{\top} \bSigma_j^{-1/2}\Ab^{\top} \Eb \Pb_{\perp}\eb_j = \sum_{ik} \Eb_{ik} (\Ab \bSigma_j^{-1/2}\ab)_i (\Pb_{\perp}\eb_j)_k = \sum_{i=1}^d \Eb_{ii} (\Ab \bSigma_j^{-1/2}\ab)_i (\Pb_{\perp}\eb_j)_i\\
     & \quad + \sum_{i < k} \Eb_{ik} \big[ (\Ab \bSigma_j^{-1/2}\ab)_i (\Pb_{\perp}\eb_j)_k + (\Ab \bSigma_j^{-1/2}\ab)_k (\Pb_{\perp}\eb_j)_i \big].
 \end{align*}
Then for the diagonal entries we have
 \begin{align*}
     & \sum_{i=1}^d \EE |\Eb_{ii} (\Ab \bSigma_j^{-1/2}\ab)_i (\Pb_{\perp}\eb_j)_i|^3 \\
     &= \EE |\Eb_{jj} (\Ab \bSigma_j^{-1/2}\ab)_j (\Pb_{\perp}\eb_j)_j|^3 + \sum_{i\neq j} \EE |\Eb_{ii} (\Ab \bSigma_j^{-1/2}\ab)_i (\Pb_{\perp}\eb_j)_i|^3\\
     & \lesssim \EE|\Eb_{jj}|^3 \|\Ab \bSigma_j^{-1/2}\ab\|_{\infty}^3 + d\max_{i}\EE|\Eb_{ii}|^3 \|\Ab \bSigma_j^{-1/2}\ab\|_{\infty}^3 \max_{i \neq j} |(\Pb_{\perp}\eb_j)_i|^3 \\
     & \lesssim \frac{\kappa_2^3 K\mu}{d}\sqrt{\frac{\mu K}{d \theta}}(\log d)^3,
 \end{align*}
 and for the off-diagonal entries, under the condition { $\kappa_2^6 K^3\mu^3=o(d^{1/2})$ }it holds that
 \begin{align*}
      & \sum_{i < k} \EE\Big|\Eb_{ik} \big[ (\Ab \bSigma_j^{-1/2}\ab)_i (\Pb_{\perp}\eb_j)_k + (\Ab \bSigma_j^{-1/2}\ab)_k (\Pb_{\perp}\eb_j)_i \big]\Big|^3  \lesssim d \frac{\tilde{\sigma}^3}{\theta^2} \|\Ab \bSigma_j^{-1/2}\ab\|_{\infty}^3(\log d)^3\\
      &\quad + d^2 \frac{\tilde{\sigma}^3}{\theta^2}(\log d)^3 \|\Ab \bSigma_j^{-1/2}\ab\|_{\infty}^3 \big(\frac{\mu K}{d}\big)^3 \lesssim {\kappa_2^3 K\mu}\sqrt{\frac{\mu K}{d \theta}} (\log d)^3= o(1).
 \end{align*}
 Moreover, since $\operatorname{Var}(\ab^{\top} \bSigma_j^{-1/2}\Ab^{\top} \Eb \Pb_{\perp}\eb_j ) = 1$, by the Lyapunov's condition, \eqref{eq: asp clt2 mismat} holds and Assumption \ref{asp: clt} is satisfied by plugging in $\Ab = \Vb \bLambda^{-1}$. By Theorem~\ref{thm: leading term L big}, we have that \eqref{eq: general clt vk large L} follows. 
 
 To show that \eqref{eq: missing mat 2} holds we need to show that $\|\widetilde{\bSigma}_j - \bSigma_j\|_2  = o(\lambda_K(\widetilde{\bSigma}_j))$.
 From previous discussion we learnt that 
 \begin{align*}
     &\|\widetilde{\bSigma}_i - \bSigma_j\|_2 \le \|\Vb\mathbf{\Lambda}^{-1}\|_2^2 \|\Cov(\Eb \Pb_{\perp}\eb_j) - \Cov(\Eb \eb_j)\|_2 \\
     &\le \frac{1}{\Delta^2} \sqrt{\frac{\mu K}{d}} \frac{\tilde{\sigma}^2}{\theta}\lesssim \kappa_2^2 \sqrt{\frac{\mu K}{d}} \lambda_K(\widetilde{\bSigma}_j) = o(\lambda_K(\widetilde{\bSigma}_j)).
 \end{align*}
 
 Then by Slutsky's Theorem, \eqref{eq: missing mat 2} holds. 
 
 Last we verify that the distributional convergence still holds when we plug in the estimator $\hbSigma_j$. Similar as in the previous proof, it suffices for us to prove that $\|\widehat\bSigma_j - \Hb \widetilde\bSigma_j\Hb^{\top}\|_2 = o_P\big(\lambda_K(\widetilde\bSigma_j)\big)$. In the following proof, we will base the discussion on the event that $\Hb$ is orthonormal. We will first bound $\|\widetilde\Mb - \Mb\|_{\max}$.  From previous discussion we have the following bounds
 $$
 \|\widehat\Mb' - \Mb\|_2 = O_P(\sqrt{d\tilde\sigma^2/\theta}),\quad \|\widetilde\Vb^{\F} - \Vb \Hb^{\top}\|_2 =  \|\widetilde\Vb^{\F}\Hb - \Vb \|_2 = O_P(\frac{1}{\Delta}\sqrt{d\tilde\sigma^2/\theta}), 
 $$
 and 
 \begin{align*}
     & \|\widetilde\Vb^{\F}\Hb - \Vb \|_{2,\infty} \le \|\widetilde\Vb^{\F}\Hb - \widehat\Vb \Hb_0\|_2 + \|\widehat\Vb \Hb_0 - \Vb\|_{2,\infty} = o_P(\frac{\tilde\sigma}{|\lambda_1|\sqrt{\theta}})\\
      &\quad + O_P(\frac{\kappa_2\tilde\sigma\sqrt{\mu K/\theta} + \tilde\sigma\sqrt{K \log d/\theta}}{\Delta}) =O_P(\frac{\kappa_2\tilde\sigma\sqrt{\mu K/\theta} + \tilde\sigma\sqrt{K \log d/\theta}}{\Delta}).
 \end{align*}
Now we can study $\widetilde\Mb = (\widetilde\Vb^{\F}\widetilde\Vb^{\F\top}) \widehat\Mb (\widetilde\Vb^{\F}\widetilde\Vb^{\F\top}) = \widetilde\Vb^{\F}\Hb (\frac{\theta}{\hat\theta}\Hb^{\top}\widetilde\Vb^{\F\top} \widehat\Mb' \widetilde\Vb^{\F} \Hb)\Hb^{\top}\widetilde\Vb^{\F\top}$. Recall by Hoeffding's inequality \citep{hoeffding1994probability}, with probability $1-O(d^{-10})$ we have that $|\hat\theta - \theta| \lesssim \frac{\sqrt{ \log d}}{d}$ and $|\cS| =\Omega(d^2\theta)$, and we have that 
\begin{align*}
    & \|\frac{\theta}{\hat\theta}\Hb^{\top}\widetilde\Vb^{\F\top} \widehat\Mb' \widetilde\Vb^{\F}\Hb- \bLambda\|_2  \le \|\Hb^{\top}\widetilde\Vb^{\F\top} \widehat\Mb' \widetilde\Vb^{\F}\Hb - \Hb^{\top}\widetilde\Vb^{\F\top} \Mb \widetilde\Vb^{\F}\Hb\|_2 \\
    & \quad + \|\Hb^{\top}\widetilde\Vb^{\F\top} \Mb (\widetilde\Vb^{\F}\Hb - \Vb)\|_2  + \|(\widetilde\Vb^{\F}\Hb - \Vb)^{\top}\Mb\Vb\|_2 + O_P\Big(\frac{\sqrt{\log d}}{d{\theta}}|\lambda_1|\Big) \\
    &\lesssim \|\widehat\Mb' - \Mb\|_2 + 2 \|\Mb\|_2\|\widetilde\Vb^{\F}\Hb - \Vb\|_2 + O_P\Big(\frac{\sqrt{\log d}}{d{\theta}}|\lambda_1|\Big)  = O_P(\kappa_2\sqrt{d\tilde\sigma^2/\theta}).
\end{align*}
Then for any $i,k \in [d]$, we have
\begin{align*}
    & |\widetilde\Mb_{ik}\!-\!\Mb_{ik}|  = |(\widetilde\Vb^{\F}\Hb)_i^{\top} (\frac{\theta}{\hat\theta}\Hb^{\top}\widetilde\Vb^{\F\top} \widehat\Mb' \widetilde\Vb^{\F}\Hb)(\widetilde\Vb^{\F}\Hb)_k - \Mb_{ik}| \\
    &\le |(\widetilde\Vb^{\F}\Hb)_i^{\top} (\frac{\theta}{\hat\theta}\Hb^{\top}\widetilde\Vb^{\F\top} \widehat\Mb' \widetilde\Vb^{\F}\Hb - \bLambda)(\widetilde\Vb^{\F}\Hb)_k| + |(\widetilde\Vb^{\F}\Hb -\Vb)_i \bLambda (\widetilde\Vb^{\F}\Hb)_k| \\
    &\quad + |(\Vb)_i \bLambda (\widetilde\Vb^{\F}\Hb - \Vb)_k| = O_P(\kappa_2\sqrt{d\tilde\sigma^2/\theta}\|\widetilde\Vb^{\F}\Hb\|_{2,\infty}^2 ) \\
    &\quad + O_P(|\lambda_1|\|\Vb\|_{2,\infty}\|\widehat\Vb \Hb_0 - \Vb\|_{2,\infty}) =O_P(\frac{\sqrt{d}\tilde\sigma}{\Delta\sqrt{\theta}})\frac{|\lambda_1|\mu K}{d} = O_P\left(\frac{\kappa_2\mu K}{\sqrt{d\theta}}\right) \tilde\sigma,
\end{align*}
and in turn we have    
\begin{align*}
    |\widetilde\Mb_{ik}^2\!-\!\Mb_{ik}^2| \lesssim \frac{|\lambda_1| \mu K}{d} |\widetilde\Mb_{ik}\!-\!\Mb_{ik}| = O_P\left(\frac{\sqrt{d}\tilde\sigma}{\Delta\sqrt{\theta}}\right) \Big(\frac{|\lambda_1|\mu K}{d}\Big)^2, \quad \forall i,k \in [d].   
\end{align*}
 Now we move on to bound the error of $\hat\sigma^2$. We know from the setting of Example~\ref{ex: missing mat} that $\varepsilon_{ik}$'s are sub-Gaussian with variance proxy of order $O(\sigma^2 (\log d)^2)$, and thus 
\begin{align*}
    |\hat\sigma^2 \!- \!\sigma^2| &\!=\! \Big|\!\!\sum_{(i,k) \in \cS}\!\!\!(\Mb_{ik}\! +\! \varepsilon_{ik}\!-\!\widetilde\Mb_{ik})^2/ |\cS|\! - \!\sigma^2\Big|\! =\! \Big|\!\!\sum_{(i,k) \in \cS}\!\!\!(\Mb_{ik} \!+\! \varepsilon_{ik}\!-\!\Mb_{ik}\! +\! \Mb_{ik} \!-\! \widetilde\Mb_{ik})^2/ |\cS|\! -\! \sigma^2 \Big|\\
    & \lesssim \Big|\frac{1}{|\cS|}\sum_{(i,k) \in \cS} \varepsilon_{ik}^2 - \sigma^2\Big| + \|\widetilde\Mb - \Mb\|_{\max}^2 = O_P\big(\frac{\sigma^2(\log d)^2}{\sqrt{|\cS|}}\big) + O_P\Big(\frac{\kappa_2^2\mu^2K^2}{d\theta}\Big)\tilde\sigma^2\\
    &= O_P\Big(\frac{(\log d)^2}{d\sqrt{\theta}}\Big)\sigma^2 +O_P\Big( \frac{\kappa_2^2\mu^2K^2}{d\theta}\Big)\tilde\sigma^2.
\end{align*}
Then for any $i \in [d]$, we have that
\begin{align*}
    & \Big|\frac{\widetilde\Mb_{ij}^2(1-\hat\theta)}{\hat\theta}\! +\! \frac{\hat\sigma^2}{\hat\theta} \!-\! \frac{\Mb_{ij}^2(1\!-\!\theta)}{\theta}\!-\!\frac{\sigma^2}{\theta}\Big| \!\lesssim \! |\widetilde\Mb_{ij}|^2\Big|\frac{1}{\hat\theta}\! -\! \frac{1}{\theta}\Big| \!+\! \frac{|\widetilde\Mb_{ij}^2\! -\! \Mb_{ij}^2|}{\theta} \!+\! \hat\sigma^2\Big|\frac{1}{\hat\theta} \!-\! \frac{1}{\theta}\Big|\! +\! \frac{|\hat\sigma^2 \!-\! \sigma^2|}{\theta}\\
    & = O_P\left(\frac{\sqrt{d}\tilde\sigma}{\Delta\sqrt{\theta}}\right)\theta^{-1} \Big(\frac{|\lambda_1|\mu K}{d}\Big)^2 + O_P\Big(\frac{(\log d)^2}{d\sqrt{\theta}} + \frac{\kappa_2^2\mu^2K^2}{d\theta}\Big)\frac{\tilde\sigma^2}{\theta} =O_P\left(\frac{\sqrt{d}\tilde\sigma}{\Delta\sqrt{\theta}}\right)\frac{\tilde\sigma^2}{\theta} ,
\end{align*}
and thus we have that 
$$
\|\diag\big([\widetilde\Mb_{ij}^2(1-\hat\theta)/\hat\theta + \hat\sigma^2/\hat\theta]_{i=1}^d\big) - \diag\big([\Mb_{ij}^2(1-\theta)/\theta + \sigma^2/\theta]_{i=1}^d\big)\|_2 = O_P\Big(\frac{\sqrt{d}\tilde\sigma}{\Delta\sqrt{\theta}}\Big)\frac{\tilde\sigma^2}{\theta}.
$$
Also, we have shown that 
$$\|\widetilde\bLambda - \Hb\bLambda\Hb^{\top}\|_2 = \|\frac{\theta}{\hat\theta}\Hb^{\top}\widetilde\Vb^{\F\top}\widehat\Mb'\widetilde\Vb^{\F}\Hb - \bLambda\|_2 = O_P\Big(\kappa_2 \frac{\sqrt{d}\tilde\sigma}{\Delta\sqrt{\theta}}\Big)\Delta,$$
then we have $\|\widetilde\bLambda^{-1} - \Hb \bLambda^{-1}\Hb^{\top}\|_2 = O_P\Big(\kappa_2 \frac{\sqrt{d}\tilde\sigma}{\Delta\sqrt{\theta}}\Big)\frac{1}{\Delta}$, and hence
 \begin{align*}
     &\|\widetilde\Vb^{\F}\widetilde\bLambda^{-1} - \Vb \bLambda^{-1}\Hb^{\top}\|_2 \le \|\widetilde{\bLambda}^{-1} - \Hb \bLambda^{-1} \Hb^{\top}\|_2 + \|\bLambda^{-1}\|_2\|\widetilde\Vb^{\F} - \Vb \Hb^{\top}\|_2 \\
     & = O_P\Big(\kappa_2 \frac{\sqrt{d}\tilde\sigma}{\Delta\sqrt{\theta}}\Big)\frac{1}{\Delta} + O_P\Big( \frac{\sqrt{d}\tilde\sigma}{\Delta\sqrt{\theta}}\Big)\frac{1}{\Delta} = O_P\Big(\kappa_2 \frac{\sqrt{d}\tilde\sigma}{\Delta\sqrt{\theta}}\Big)\frac{1}{\Delta}.
 \end{align*}
 Then following basic algebra we have that with high probability
 \begin{align*}
     \|\hbSigma_j - \Hb\widetilde\bSigma_j \Hb^{\top}\|_2 \lesssim O_P\Big(\frac{\sqrt{d}\tilde\sigma}{\Delta\sqrt{\theta}}\Big)\frac{\tilde\sigma^2}{\Delta^2\theta} + O_P\Big(\kappa_2\frac{\sqrt{d}\tilde\sigma}{\Delta\sqrt{\theta}}\Big)\frac{\tilde\sigma^2}{\Delta^2\theta} = O_P\Big(\kappa_2\frac{\sqrt{d}\tilde\sigma}{\Delta\sqrt{\theta}}\Big)\frac{\tilde\sigma^2}{\Delta^2\theta}.
 \end{align*}
 Then under the condition that { $\kappa_2^3\frac{\sqrt{d}\tilde\sigma}{\Delta\sqrt{\theta}} = o(1)$,} we have that $$\|\hbSigma_j - \Hb\widetilde\bSigma_j \Hb^{\top}\|_2 = O_P(\kappa_2^3\frac{\sqrt{d}\tilde\sigma}{\Delta\sqrt{\theta}} ) \frac{\tilde\sigma^2}{\lambda_1^2\theta} = o_P\big(\lambda_K(\widetilde\bSigma_j)\big).$$
 \section{Proof of Technical Lemmas}\label{sec: proof tec lems}
In this section, we provide proofs of the technical lemmas used in the proofs of the main theorems. 
\subsection{Proof of Lemma \ref{lm: gaussian norm}}\label{sec: proof lm gaussian norm}
It can be easily seen that 
$$\|\mathbf{\Omega} / \sqrt{p}\|_2 = (\|\mathbf{\Omega} \mathbf{\Omega}^{ \top} / p \|_2 )^{1/2}= \left((d/p) \|\mathbf{\Omega}^{\top} \mathbf{\Omega} / d \|_2 \right)^{1/2}.$$ 
By Lemma 3 in \citet{fandistributed2019}, we know that $\|\|\mathbf{\Omega}^{\top} \mathbf{\Omega} / d - \Ib_p \|_2 \|_{\psi_1} \lesssim \sqrt{p/d}$, and thus $\|\|\mathbf{\Omega}^{\top} \mathbf{\Omega} / d \|_2 \|_{\psi_1} \lesssim 1 + \sqrt{p/d} = O(1)$. Therefore, we have $\|\|\mathbf{\Omega} \mathbf{\Omega}^{ \top} / p\|_2 \|_{\psi_1} \lesssim d/p$. By Jensen's inequality, we in turn get $\|\|\mathbf{\Omega}/ \sqrt{p}\|_2 \|_{\psi_1} \lesssim \sqrt{d/p}$.
\subsection{Proof of Lemma \ref{lm: bound min eigenvalue}}\label{sec: proof lm bound min eig}
By Proposition 10.4 in \citet{halkofinding2011}, we know that for any $t \ge 1$, we have 
\begin{equation}
\mathbb{P}\left(\left\|\boldsymbol{\Omega}^{\dagger}\right\|_2 \geq \frac{{\rm e} \sqrt{p}}{p-K+1} \cdot t\right) \leq t^{-(p-K+1)}.    
\end{equation}
Since $p \ge 2K$, there exists a constant $c$ such that $\frac{{\rm e}p}{p-K+1} \le c$, and thus 
\begin{equation}
\mathbb{P}\left(
\sqrt{p}\left\|\boldsymbol{\Omega}^{\dagger}\right\|_2 \geq c t\right) \leq t^{-(p-K+1)}.  
\end{equation}
Therefore, we have 
\begin{align*}
    &\EE\left ( \left(\sigma_{\min} (\boldsymbol{\Omega}/\sqrt{p})\right)^{-a} \right)  = \EE \left( \left\|\sqrt{p}\boldsymbol{\Omega}^{\dagger}\right\|_2^{a}\right) = \int_{u \ge 0} \PP \left( \left\|\sqrt{p}\boldsymbol{\Omega}^{\dagger}\right\|_2^{a} \ge u \right ) d u\\
    &\quad  = \int_{ 0 \le u \le c^{a}} \PP \left( \left\|\sqrt{p}\boldsymbol{\Omega}^{\dagger}\right\|_2^{a} \ge u \right ) d u + \int_{u \ge c^{a}} \PP \left( \left\|\sqrt{p}\boldsymbol{\Omega}^{\dagger}\right\|_2^{a} \ge u \right ) d u\\
    &\quad  \le c^{a} + \int_{u \ge c^{a}} \!\!\! \PP \left( \left\|\sqrt{p}\boldsymbol{\Omega}^{\dagger}\right\|_2 \ge u^{1/a} \right ) d u  \le c^{a} + \int_{u \ge c^{a}} \!\!\! \left(u^{1/a}/c \right)^{-(p-K+1)} d u\\
    & \quad = c^{a}\left(1+\frac{1}{(p-K+1)/a-1} \right).
\end{align*}
Since $1+\frac{1}{(p-K+1)/a-1} \le 2$, the claim follows. 

\subsection{Proof of Lemma \ref{lm: control of prob}}\label{sec: proof lm control prob}
We first consider the probability $\PP \left( \|{\mathbf{\Sigma}}^\prime - \Vb \Vb^{\top} \|_2  \ge \varepsilon \right) $. Recall the matrix ${\Yb}^{(\ell)} := \Vb \Pb_0 \mathbf{\Lambda}^0 \Vb^{\top} \mathbf{\Omega}^{(\ell)}$. 
Now by Jensen's inequality and Wedin's Theorem \citep{wedin1972wedin}, we have
\begin{align*}
    &\|{\mathbf{\Sigma}}^\prime - \Vb \Vb^{\top} \|_2  = \|  \mathbb{E}\left(\widehat{\Vb}^{(\ell)} \widehat{\Vb}^{(\ell) \top} | \hat{\Mb} \right)- \Vb \Vb^{\top}\|_2 \le \mathbb{E}\left(\left\|  \widehat{\Vb}^{(\ell)} \widehat{\Vb}^{(\ell) \top} - \Vb \Vb^{\top}\right\|_2\Big| \hat{\Mb} \right)\\
    & \quad \lesssim \EE \left( \|\widehat\Yb^{(\ell)}/\sqrt{p} - {\Yb}^{(\ell)}/\sqrt{p}\|_2/ \sigma_K\left({\Yb}^{(\ell)} /\sqrt{p}\right)| \hat{\Mb}\right)
    \le \frac{\|\Eb\|_2}{\Delta} \EE \left( \frac{\|\mathbf{\Omega}^{(\ell)}/\sqrt{p}\|_2}{\sigma_{\min}\left( \tilde{\mathbf{\Omega}}^{(\ell)} /\sqrt{p}\right)} \quad \bigg| \hat{\Mb}\right)\\
    &\quad = \frac{\|\Eb\|_2}{\Delta} \EE \left( \frac{\|\mathbf{\Omega}^{(\ell)}/\sqrt{p}\|_2}{\sigma_{\min}\left( \tilde{\mathbf{\Omega}}^{(\ell)} /\sqrt{p}\right)}\right)
    \le \frac{\|\Eb\|_2}{\Delta} \EE \left( \|\mathbf{\Omega}^{(\ell)}/\sqrt{p}\|_2^2\right)^{1/2} \EE\left ( \left(\sigma_{\min} (\boldsymbol{\Omega}^{(\ell)} /\sqrt{p})\right)^{-2} \right)^{1/2}\\
    & \quad \lesssim \frac{\|\Eb\|_2}{\Delta} \|\|\mathbf{\Omega}^{(\ell)}/\sqrt{p}\|_2\|_{\psi_1} \lesssim \frac{\|\Eb\|_2}{\Delta} \sqrt{d/p},
\end{align*}
where the last but one inequality is due to Lemma \ref{lm: bound min eigenvalue} under the condition that $p \ge \max(2K, K + 3)$, and the last inequality is due to Lemma \ref{lm: gaussian norm}. Therefore, by Assumption \ref{asp: tail prob bound}, there exist constants $c_0, c_0' >0$ such that 
$$
    \PP \left( \|{\mathbf{\Sigma}}^\prime - \Vb \Vb^{\top} \|_2  \ge \varepsilon \right)  \le \PP \left( \frac{\|\Eb\|_2}{\Delta} \sqrt{d/p} \ge c_0' \varepsilon \right)\le 2\exp\left( -c_0 \sqrt{\frac{p}{d}} \frac{\Delta \varepsilon }{r_1(d)}\right).
$$
Similarly, we consider the probability $\PP \left( \|{\mathbf{\Sigma}}^\prime - \hat\Vb \hat\Vb^{\top} \|_2 \ge \varepsilon \right)$. By Assumption \ref{asp: tail prob bound}, there exist constants $c_0'', c_0'''>0$ such that 
\begin{align*}
    &\PP \left( \|{\mathbf{\Sigma}}^\prime - \hat\Vb \hat\Vb^{\top} \|_2 \ge \varepsilon \right)  \le \PP \left( \|{\mathbf{\Sigma}}^\prime - \Vb \Vb^{\top} \|_2  \ge \varepsilon/2 \right) + \PP \left( \|\hat\Vb \hat\Vb^{\top}- \Vb \Vb^{\top} \|_2  \ge \varepsilon/2 \right)\\
    & \quad \lesssim \exp\left( -c_0 \sqrt{\frac{p}{d}} \frac{\Delta \varepsilon }{2r_1(d)}\right) + \PP \left(\frac{\|\Eb\|_2}{\Delta} \ge c_0''' \varepsilon\right) \lesssim \exp\left( -c_0 \sqrt{\frac{p}{d}} \frac{\Delta \varepsilon }{2r_1(d)}\right)   + \exp\left(-\frac{c_0'' \Delta \varepsilon}{r_1(d)}\right) \\
    & \quad \lesssim \exp\left( -c_0 \sqrt{\frac{p}{d}} \frac{\Delta \varepsilon }{2r_1(d)}\right). 
\end{align*}
Therefore, the claim follows.
\subsection{Proof of Lemma \ref{lm: cov op norm up bd}}\label{sec: proof lm cov op norm}
We know that $\Cov(\xb_1 + \xb_2) = \Cov(\xb_1) + \Cov(\xb_2) + \Cov(\xb_1, \xb_2) + \Cov(\xb_2, \xb_1)$, where $\Cov(\xb_1, \xb_2) = \EE(\xb_1 - \EE \xb_1)(\xb_2 - \EE \xb_2)^{\top}$, and $$\|\Cov(\xb_i)\|_2 = \max_{\|\vb\|_2=1} \vb^{\top} \Cov(\xb_i) \vb = \max_{\|\vb\|_2=1} \operatorname{Var}\big(\vb^{\top}\xb_i\big),$$ for $i =1, 2$. Therefore, we have
\begin{align*}
    \|\Cov(\xb_1, \xb_2)\|_2 & = \max_{\|\vb\|_2=1,\|\ub\|_2=1} \vb^{\top} \Cov(\xb_1, \xb_2) \ub = \max_{\|\vb\|_2=1,\|\ub\|_2=1} \Cov (\vb^{\top}\xb_1, \ub^{\top}\xb_2)\\
    & \le \max_{\|\vb\|_2=1,\|\ub\|_2=1} \sqrt{\operatorname{Var}(\vb^{\top}\xb_1) } \sqrt{\operatorname{Var}(\vb^{\top}\xb_2) }= \sqrt{\|\Cov(\xb_1)\|_2\|\Cov(\xb_2)\|_2} \\
    & \le \frac{1}{2}\|\Cov(\xb_1)\|_2 + \frac{1}{2}\|\Cov(\xb_2)\|_2.
\end{align*}
Thus we have 
\begin{align*}
\|\Cov(\xb_1 + \xb_2) \|_2 &\le \|\Cov(\xb_1)\|_2 + \|\Cov(\xb_2)\|_2 + \|\Cov(\xb_1, \xb_2)\|_2 + \|\Cov(\xb_2, \xb_1)\|_2 \\
&\le 2\|\Cov(\xb_1)\|_2 + 2\|\Cov(\xb_2)\|_2.    
\end{align*}

\section{Wedin's Theorem}\label{sec: supp}
\begin{lemma}[Modified Wedin's Theorem]\label{lm: wedin}
Let $\Mb^{\star}$ and $\Mb=\Mb^{\star}+\Eb$ be two matrices in $\mathbb{R}^{n_{1} \times n_{2}}$ (without loss of generality, we assume $n_{1} \leq n_{2}$ ), whose SVDs are given respectively by
$$
\Mb^{\star}=\sum_{i=1}^{n_{1}} \sigma_{i}^{\star} \ub_{i}^{\star} \vb_{i}^{\star \top}=\left[\begin{array}{ll}
\Ub^{\star} & \Ub_{\perp}^{\star}
\end{array}\right]\left[\begin{array}{ccc}
\boldsymbol{\Sigma}^{\star} & \mathbf{0} & \mathbf{0} \\
\mathbf{0} & \boldsymbol{\Sigma}_{\perp}^{\star} & \mathbf{0}
\end{array}\right]\left[\begin{array}{c}
\Vb^{\star \top} \\
\Vb_{\perp}^{\star \top}
\end{array}\right], 
$$
$$
\Mb=\sum_{i=1}^{n_{1}} \sigma_{i} \ub_{i} \vb_{i}^{\top}=\left[\begin{array}{ll}
\Ub & \Ub_{\perp}
\end{array}\right]\left[\begin{array}{ccc}
\boldsymbol{\Sigma} & \mathbf{0} & \mathbf{0} \\
\mathbf{0} & \boldsymbol{\Sigma}_{\perp} & \mathbf{0}
\end{array}\right]\left[\begin{array}{l}
\Vb^{\top} \\
\Vb_{\perp}^{\top}
\end{array}\right].
$$

Here, $\sigma_{1} \geq \cdots \geq \sigma_{n_{1}}$ (resp. $\sigma_{1}^{\star} \geq \cdots \geq \sigma_{n_{1}}^{\star}$) stand for the singular values of $\Mb$ (resp. $\Mb^{\star}$) arranged in descending order, $\ub_{i}$ (resp. $\left.\ub_{i}^{\star}\right)$ denotes the left singular vector associated with the singular value $\sigma_{i}$ (resp. $\sigma_{i}^{\star}$), and $\vb_{i}$ (resp. $\vb_{i}^{\star}$) represents the right singular vector associated with $\sigma_{i}$ (resp. $\sigma_{i}^{\star}$). $\Ub$ and $\Ub^{\star}$ stand for the top $r$ eigenvectors of $\Mb$ and $\Mb^{\star}$ respectively. Then, 

\begin{equation}\label{eq: wedin 1}
\max \left\{\|\Ub \Ub^{\top} - \Ub^{\star} \Ub^{\star \top} \|_2, \|\Vb\Vb^{\top} - \Vb^{\star}\Vb^{\star \top} \|_2 \right\}  \lesssim \frac{2\|\Eb\|}{\sigma_{r}^{\star}-\sigma_{r+1}^{\star}}, 
\end{equation}
and 
\begin{equation}\label{eq: wedin 2}
\max \left\{\|\Ub \Ub^{\top} - \Ub^{\star} \Ub^{\star \top} \|_{\rm F}, \|\Vb\Vb^{\top} - \Vb^{\star}\Vb^{\star \top} \|_{\rm F} \right\}  \lesssim \frac{2 \sqrt{r}\|\Eb\|}{\sigma_{r}^{\star}-\sigma_{r+1}^{\star}}.
\end{equation}
\end{lemma}
\begin{proof}
By Wedin's Theorem \citep{wedin1972wedin}, if $\|\Eb\|_2 <(1-1 / \sqrt{2})\left(\sigma_{r}^{\star}-\sigma_{r+1}^{\star}\right)$, \eqref{eq: wedin 1} and \eqref{eq: wedin 2} are true. When $\|\Eb\|_2 \ge (1-1 / \sqrt{2})\left(\sigma_{r}^{\star}-\sigma_{r+1}^{\star}\right)$, the RHS of \eqref{eq: wedin 1} are larger than or equal to $2 - \sqrt{2}$, whereas the LHS are bounded by 1. Thus \eqref{eq: wedin 1} follows trivially, and so is \eqref{eq: wedin 2}.
\end{proof}
\section{Supplementary Figures}\label{sec: supp figs}
We provide in this section additional figures deferred from the main paper.
\begin{figure}[htbp]
		\centering
		\begin{tabular}{cc}
 		     {\small (a) Example~\ref{ex: spiked gaussian}: Spiked Covariance Model }&{\small (b) Example~\ref{ex: GMM}: Gaussian Mixture Models }\\
 		    \includegraphics[height=0.27\textwidth]{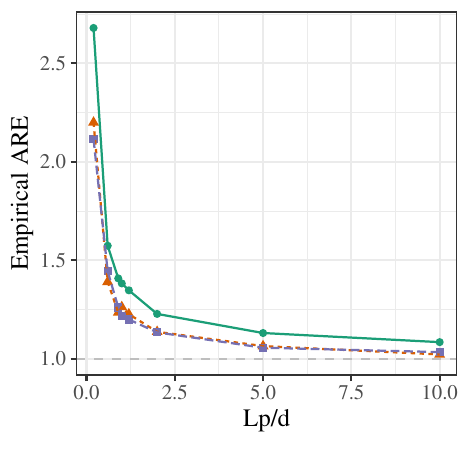}  & \includegraphics[height=0.27\textwidth]{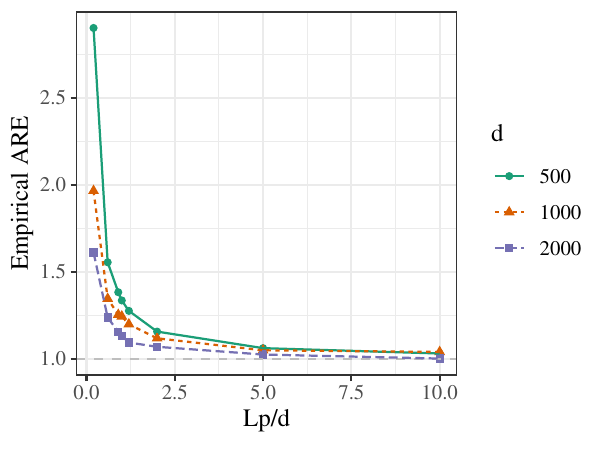} 
 		\end{tabular}
		\caption{\small \small Asymptotic relative efficiency (ARE) between the FADI estimator and the traditional PCA estimator under the spiked covariance model and GMM, where the ARE is measured by  $\det(\hbSigma^{\rm FADI})^{1/K}\cdot \det(\hbSigma^{\rm PCA})^{-1/K}$ with 
		$\hbSigma^{\rm FADI}$ and $\hbSigma^{\rm PCA}$  being the empirical covariance matrices for   the FADI and traditional PCA estimators \citep{serfling2009approximation}. 
  The results suggest that when $Lp/d>1$ and increases, the ARE between FADI and the traditional PCA approaches 1.
  } \label{fig: ARE exm 1 3}
	\end{figure}
\begin{figure}[H]
	    \centering
	    \includegraphics[width=0.7\textwidth]{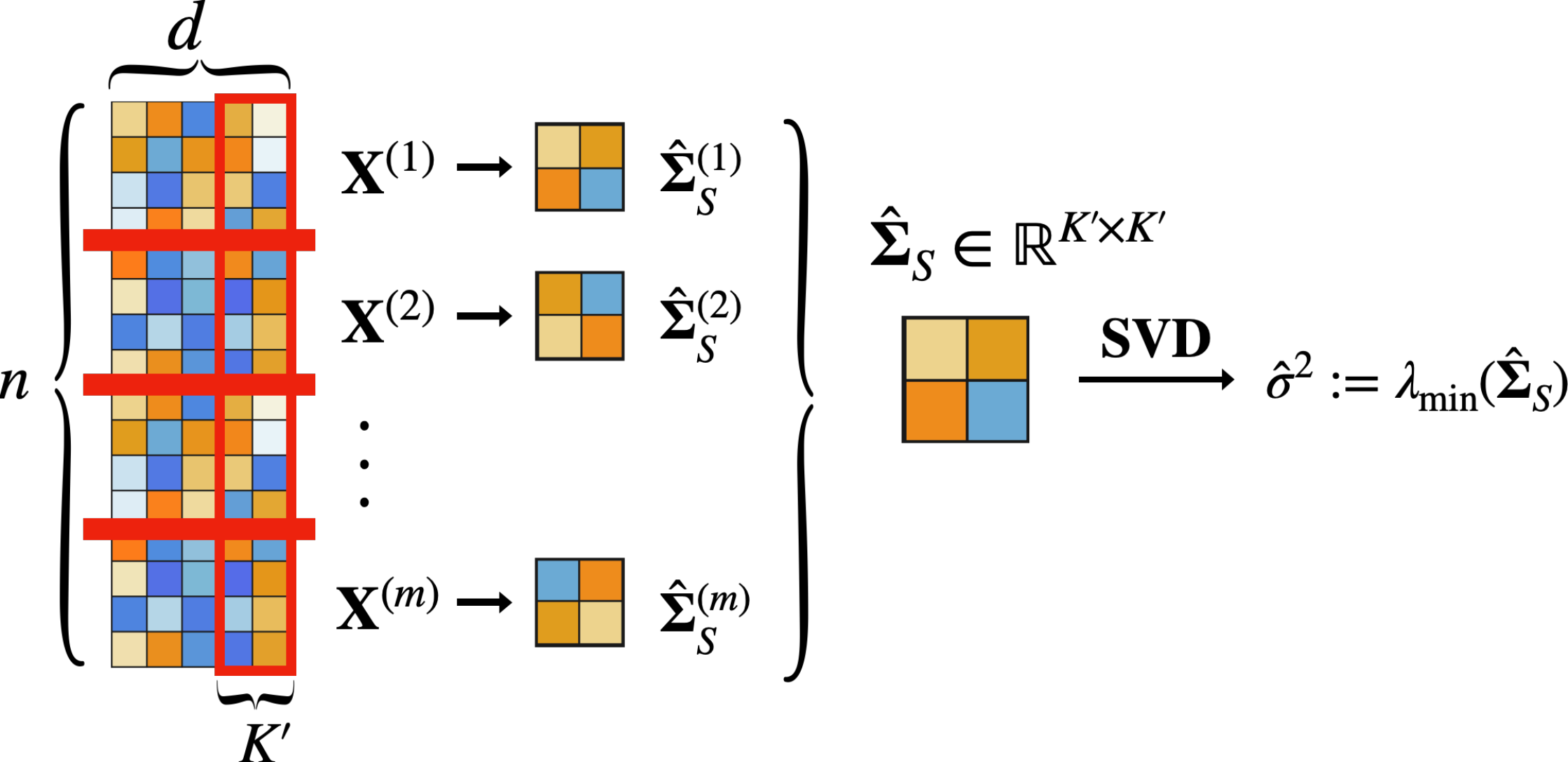}
	    \caption{\small Illustration of Step 0 for Example~\ref{ex: spiked gaussian}. $\hbSigma_S^{(s)} = \Xb_{[:,S]}^{(s)\top}\Xb_{[:,S]}^{(s)}$ is calculated by the data columns in the set $S$ for the $s$-th split ($s \in [m]$), and $\hbSigma_S = n^{-1} \sum_{s \in [m]} \hbSigma_S^{(s)}$.}
	    \label{fig: step 0 exm 1}
	\end{figure}
 \begin{figure}[H]
		\centering
		\includegraphics[height=0.27\textwidth]{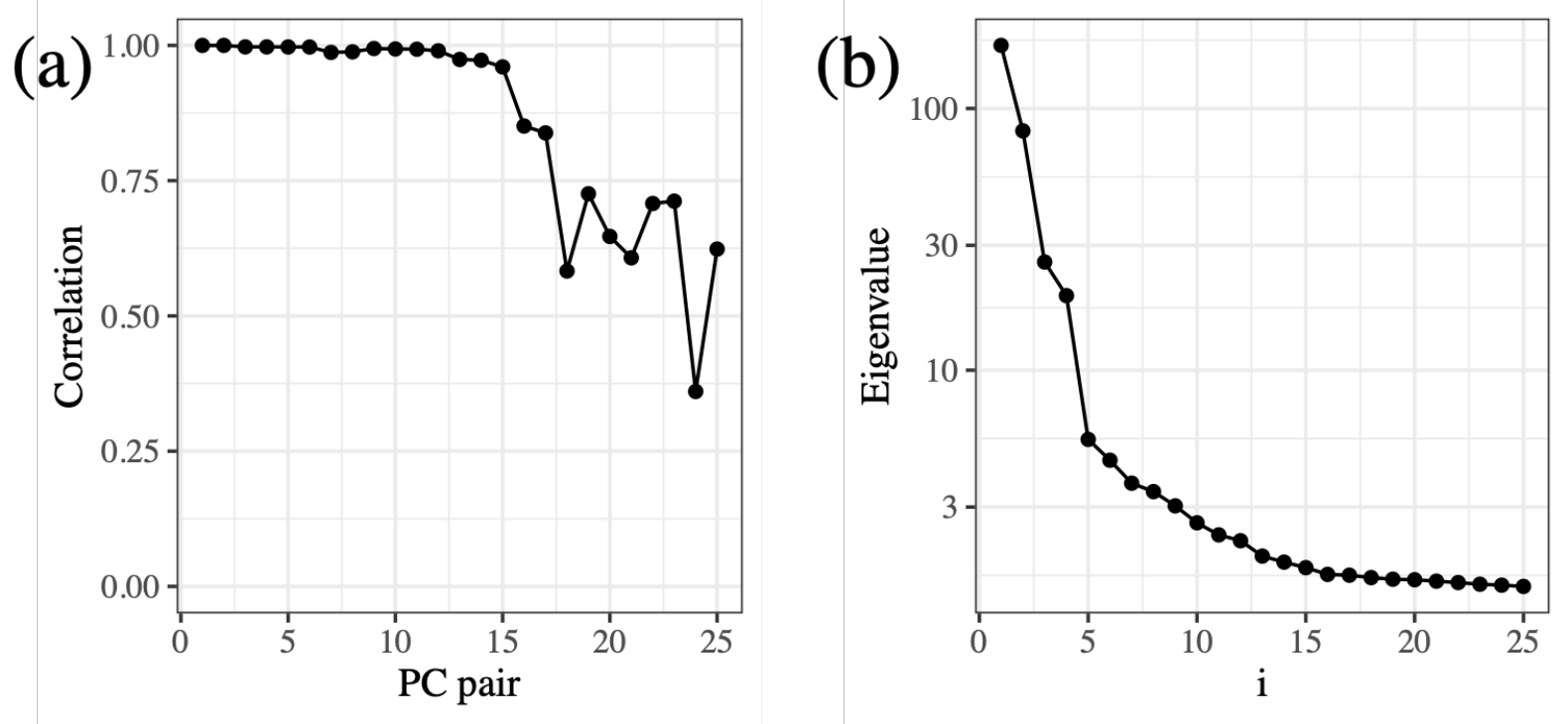} 
		\caption{\small (a) Correlations between the 25 leading PCs calculated by FADI and by full sample PCA on  the 1000 Genomes Data; (b) Top 25 eigenvalues for the sample covariance matrix of the 1000 Genomes Data. We can see that for the 15 leading PCs, the results calculated by FADI are highly correlated to the results calculated by the traditional full sample PCA, whereas the correlations drop afterward. This can be attributed to the fact that the top 15 eigenvalues are well-separated for the sample covariance matrix of the 1000 Genomes Data, and the eigengaps get smaller after the 15th eigenvalue. 
  }\label{fig: 1000g corr}
	\end{figure}
 \begin{figure}[H]
		\centering
			\includegraphics[width=0.9\textwidth]{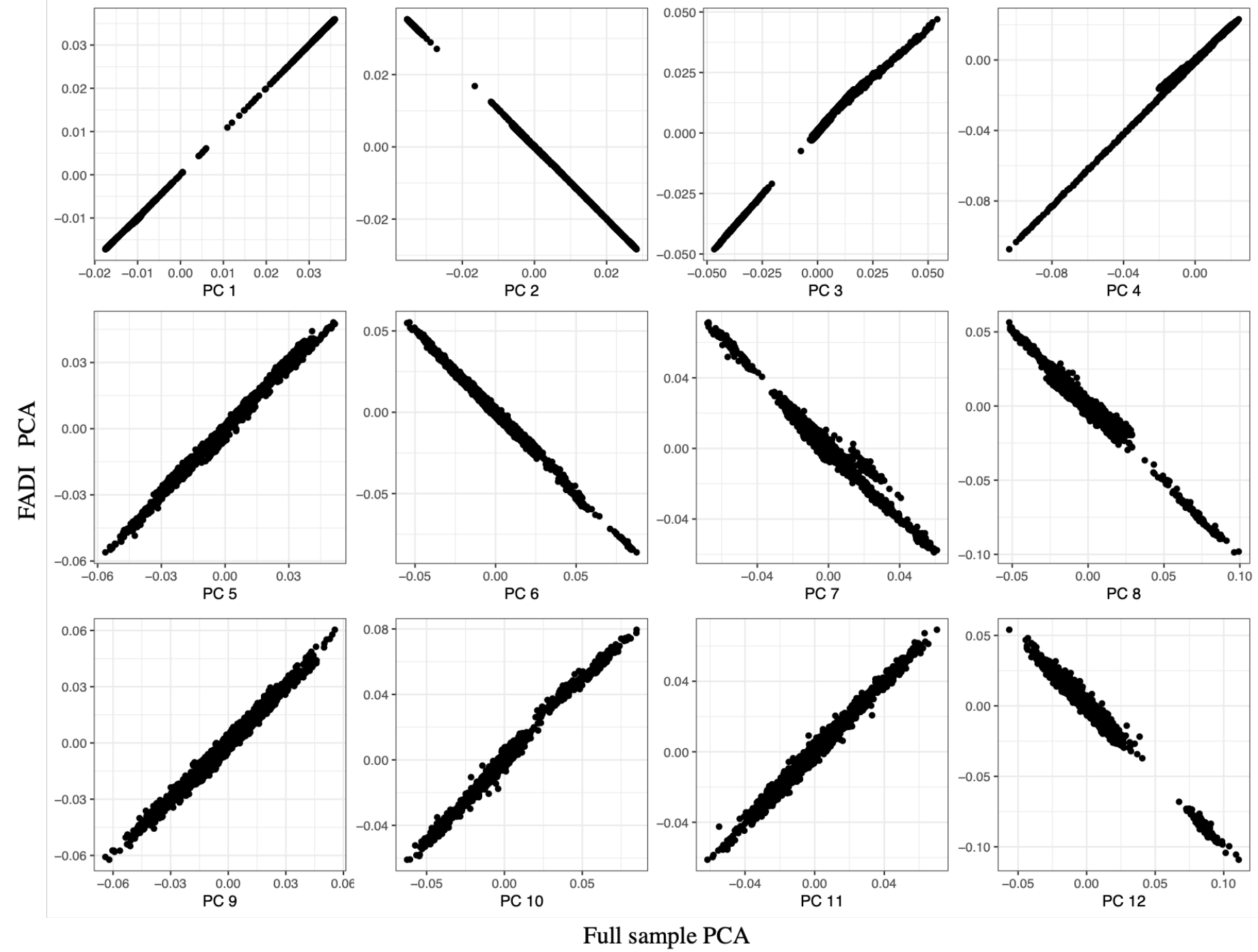}
		\caption{\small Comparison of the top 12 PCs of the 1000 Genomes Data calculated by full sample traditional PCA and by FADI.}\label{fig: pc corr scatter plot 1000g}
	\end{figure}
\end{document}